\documentclass[12pt, a4paper, openany, oneside]{memoir}
\usepackage[T1]{fontenc}
\usepackage[latin1]{inputenc}
\usepackage{pifont}
\usepackage{pslatex}
\usepackage{yfonts}
\usepackage[mathscr]{euscript}
\usepackage{latexsym}
\usepackage{stmaryrd}
\usepackage{amssymb}
\usepackage{amsmath}
\usepackage{graphicx}
\usepackage{pst-all}
\usepackage{verbatim} 
\usepackage{fancyvrb}%Pour utiliser l'environnement "Verbatim"
\usepackage{mathrsfs}
\usepackage{vmargin}
\usepackage{titletoc}
\usepackage{bbding}
\usepackage{wasysym}
\usepackage{ifthen}
\usepackage[final]{pdfpages}
\usepackage[natbibapa]{apacite}
\usepackage{morewrites}
\usepackage{tocbibind}
\usepackage{multicol}
 
\setpapersize{A4}
\DeclareMathAlphabet{\mathpzc}{OT1}{pzc}{m}{it}

\newcommand\QEDBox
{{\leavevmode\unskip\nobreak\hfil\penalty50\hskip.75cm%
    \hbox{} \nobreak\hfil $\Box$ \parfillskip=0pt
    \finalhyphendemerits=0
    \par
}}
\def\qed{\QEDBox}
\makeatletter
\usepackage{color}
\usepackage{colortbl}

\definecolor{greenyellow}   {cmyk}{0.15, 0   , 0.69, 0   }
\definecolor{yellow}        {cmyk}{0   , 0   , 1   , 0   }
\definecolor{goldenrod}     {cmyk}{0   , 0.10, 0.84, 0   }
\definecolor{dandelion}     {cmyk}{0   , 0.29, 0.84, 0   }
\definecolor{apricot}       {cmyk}{0   , 0.32, 0.52, 0   }
\definecolor{peach}         {cmyk}{0   , 0.50, 0.70, 0   }
\definecolor{melon}         {cmyk}{0   , 0.46, 0.50, 0   }
\definecolor{yelloworange}  {cmyk}{0   , 0.42, 1   , 0   }
\definecolor{orange}        {cmyk}{0   , 0.61, 0.87, 0   }
\definecolor{burntorange}   {cmyk}{0   , 0.51, 1   , 0   }
\definecolor{bittersweet}   {cmyk}{0   , 0.75, 1   , 0.24}
\definecolor{redorange}     {cmyk}{0   , 0.77, 0.87, 0   }
\definecolor{mahogany}      {cmyk}{0   , 0.85, 0.87, 0.35}
\definecolor{maroon}        {cmyk}{0   , 0.87, 0.68, 0.32}
\definecolor{brickred}      {cmyk}{0   , 0.89, 0.94, 0.28}
\definecolor{red}           {cmyk}{0   , 1   , 1   , 0   }
\definecolor{orangered}     {cmyk}{0   , 1   , 0.50, 0   }
\definecolor{rubinered}     {cmyk}{0   , 1   , 0.13, 0   }
\definecolor{wildstrawberry}{cmyk}{0   , 0.96, 0.39, 0   }
\definecolor{salmon}        {cmyk}{0   , 0.53, 0.38, 0   }
\definecolor{carnationpink} {cmyk}{0   , 0.63, 0   , 0   }
\definecolor{magenta}       {cmyk}{0   , 1   , 0   , 0   }
\definecolor{violetred}     {cmyk}{0   , 0.81, 0   , 0   }
\definecolor{rhodamine}     {cmyk}{0   , 0.82, 0   , 0   }
\definecolor{mulberry}      {cmyk}{0.34, 0.90, 0   , 0.02}
\definecolor{redviolet}     {cmyk}{0.07, 0.90, 0   , 0.34}
\definecolor{fuchsia}       {cmyk}{0.47, 0.91, 0   , 0.08}
\definecolor{lavender}      {cmyk}{0   , 0.48, 0   , 0   }
\definecolor{thistle}       {cmyk}{0.12, 0.59, 0   , 0   }
\definecolor{orchid}        {cmyk}{0.32, 0.64, 0   , 0   }
\definecolor{darkorchid}    {cmyk}{0.40, 0.80, 0.20, 0   }
\definecolor{purple}        {cmyk}{0.45, 0.86, 0   , 0   }
\definecolor{plum}          {cmyk}{0.50, 1   , 0   , 0   }
\definecolor{violet}        {cmyk}{0.79, 0.88, 0   , 0   }
\definecolor{royalpurple}   {cmyk}{0.75, 0.90, 0   , 0   }
\definecolor{blueviolet}    {cmyk}{0.86, 0.91, 0   , 0.04}
\definecolor{periwinkle}    {cmyk}{0.57, 0.55, 0   , 0   }
\definecolor{cadetblue}     {cmyk}{0.62, 0.57, 0.23, 0   }
\definecolor{cornflowerblue}{cmyk}{0.65, 0.13, 0   , 0   }
\definecolor{midnightblue}  {cmyk}{0.98, 0.13, 0   , 0.43}
\definecolor{navyblue}      {cmyk}{0.94, 0.54, 0   , 0   }
\definecolor{royalblue}     {cmyk}{1   , 0.50, 0   , 0   }
\definecolor{blue}          {cmyk}{1   , 1   , 0   , 0   }
\definecolor{cerulean}      {cmyk}{0.94, 0.11, 0   , 0   }
\definecolor{cyan}          {cmyk}{1   , 0   , 0   , 0   }
\definecolor{processblue}   {cmyk}{0.96, 0   , 0   , 0   }
\definecolor{skyblue}       {cmyk}{0.62, 0   , 0.12, 0   }
\definecolor{turquoise}     {cmyk}{0.85, 0   , 0.20, 0   }
\definecolor{tealblue}      {cmyk}{0.86, 0   , 0.34, 0.02}
\definecolor{aquamarine}    {cmyk}{0.82, 0   , 0.30, 0   }
\definecolor{bluegreen}     {cmyk}{0.85, 0   , 0.33, 0   }
\definecolor{emerald}       {cmyk}{1   , 0   , 0.50, 0   }
\definecolor{junglegreen}   {cmyk}{0.99, 0   , 0.52, 0   }
\definecolor{seagreen}      {cmyk}{0.69, 0   , 0.50, 0   }
\definecolor{green}         {cmyk}{1   , 0   , 1   , 0   }
\definecolor{forestgreen}   {cmyk}{0.91, 0   , 0.88, 0.12}
\definecolor{pinegreen}     {cmyk}{0.92, 0   , 0.59, 0.25}
\definecolor{limegreen}     {cmyk}{0.50, 0   , 1   , 0   }
\definecolor{yellowgreen}   {cmyk}{0.44, 0   , 0.74, 0   }
\definecolor{springgreen}   {cmyk}{0.26, 0   , 0.76, 0   }
\definecolor{olivegreen}    {cmyk}{0.64, 0   , 0.95, 0.40}
\definecolor{rawsienna}     {cmyk}{0   , 0.72, 1   , 0.45}
\definecolor{sepia}         {cmyk}{0   , 0.83, 1   , 0.70}
\definecolor{brown}         {cmyk}{0   , 0.81, 1   , 0.60}
\definecolor{tan}           {cmyk}{0.14, 0.42, 0.56, 0   }
\definecolor{gray}          {cmyk}{0   , 0   , 0   , 0.50}
\definecolor{black}         {cmyk}{0   , 0   , 0   , 1   }
\definecolor{white}         {cmyk}{0   , 0   , 0   , 0   } 

\usepackage{memhfixc}

\newcommand{\myPrintChapterLabel}[1]{
	\ifthenelse{\equal{#1}{A}}
		{\myAppendixLabel} 
		{\ifthenelse{\equal{#1}{B}}
			{\myAppendixLabel} 
			{\ifthenelse{\equal{#1}{C}}
				{\myAppendixLabel} 
				{\ifthenelse{\equal{#1}{D}}
					{\myAppendixLabel} 
					{\ifthenelse{\equal{#1}{E}}
						{\myAppendixLabel} 
						{\ifthenelse{\equal{#1}{F}}
							{\myAppendixLabel} 
							{\ifthenelse{\equal{#1}{G}}
								{\myAppendixLabel} 
								{\myChapterLabel}}}}}}} 
}

\newcommand{\myChapterStyle}[1]{
	\chapterstyle{#1}
}

\makechapterstyle{fieldset}{%
    \renewcommand{\chaptitlefont}{\fontsize{25pt}{30pt}\Huge\bfseries}%
    \renewcommand{\printchaptertitle}[1]{%
		  \vspace*{-45.3pt}
		  \begin{center}
			  \chaptitlefont %\hrule height 1.5pt
			  \begin{center}\textcolor{black}{\textsc{{##1}}}\end{center}
			  \vspace{-2mm}
			  \textcolor{black}{\hrule height 2.5pt}%
		  \end{center}
        }%
		
}

\makechapterstyle{titleontopright}{%
    \renewcommand{\chaptitlefont}{\Huge\bfseries}%
    \renewcommand{\printchaptertitle}[1]{%
		  \vspace*{-28.3pt}
        \chaptitlefont \hrule height 1.0pt
        \begin{flushright}\textcolor{black}{{##1}}\end{flushright}
		  \hrule height 1.0pt%
        }%
		
}

\makechapterstyle{bringhurst}{%

\renewcommand{\printchaptertitle}[1]{%
\raggedright\Large\scshape\MakeLowercase{##1}}

}
\newlength{\headindent}
\newlength{\rightblock}
\makechapterstyle{southall}{%
\setlength{\headindent}{36pt}
\setlength{\rightblock}{\textwidth}
\addtolength{\rightblock}{-\headindent}
\setlength{\beforechapskip}{2\baselineskip}
\setlength{\afterchapskip}{5\baselineskip}
\setlength{\midchapskip}{0pt}
\renewcommand{\chaptitlefont}{\huge\rmfamily\raggedright}

\renewcommand{\printchaptertitle}[1]{%
\hfill\begin{minipage}[t]{\rightblock}
{\vspace{0pt}\chaptitlefont ##1\par}\end{minipage}}

}
\makechapterstyle{chappell}{
\setlength\beforechapskip{0pt}

\renewcommand*\chaptitlefont{\Large\itshape}
\renewcommand*\printchaptertitle[1]{%
\hrule\vskip\onelineskip\centering\chaptitlefont ##1}
}

\makechapterstyle{fieldset-rev}{%
    \renewcommand{\chaptitlefont}{\fontsize{25pt}{30pt}\Huge\bfseries}%
    \renewcommand{\printchaptertitle}[1]{%
		  \vspace*{-45.3pt}
		  \begin{center}
			  \chaptitlefont %\hrule height 1.5pt
			  \begin{center}\textcolor{black}{\textsc{{##1}}}\end{center}
			  \vspace{-4mm}
			  \textcolor{black}{\hrule height 1.5pt}%
			  \vspace{1mm}
			  \textcolor{black}{\hrule height 4.5pt}
		  \end{center}
        }%
		
}
\newcommand{\timesContentFontSize}{\fontsize{13pt}{18pt}}

\newcommand{\FonteSectionI}{\sffamily\bfseries\raggedright\fontsize{16pt}{20.7pt}\selectfont}%

\renewcommand{\section}{%
   \par\vspace{20pt}
   \vspace{1.5mm}
   \renewcommand{\@seccntformat}[1]{\fontsize{16pt}{20.7pt}\thesection.\hspace{0.7em}}
   \@startsection{section}  % nom de l'inter
   {1}%                     % niveau de l'inter
   {0pt}%                   % l'indentation du titre et du texte suivant
   {4pt}% beforeskip %
   {6pt}% afterskip
   {\FonteSectionI}%        % style
}

\newcommand{\FonteSectionII}{\sffamily\bfseries\raggedright\fontsize{15pt}{19.4pt}\selectfont}%

\renewcommand{\thesubsection}{\thesection.\arabic{subsection}}%
\renewcommand{\subsection}{%
\vspace{3mm}
  \renewcommand{\@seccntformat}[1]%
               {{\fontsize{15pt}{19.4pt}\thesubsection.\hspace{0.7em}}}%
  \@startsection%
   {subsection}%            % nom de l'inter
   {2}%                     % niveau de l'inter
   {0pt}%                   % l'indentation du titre et du texte suivant
   {3pt}
   {5pt}
   {\FonteSectionII}}%      % style

\newcommand{\FonteSectionIII}{\sffamily\bfseries\fontsize{14pt}{18.2pt}\raggedright\selectfont}%

\renewcommand{\thesubsubsection}{\thesubsection.\arabic{subsubsection}}%
\renewcommand{\subsubsection}{%
\vspace{2mm}
  \renewcommand{\@seccntformat}[1]{%
               {\sffamily\bfseries\fontsize{14pt}{18.2pt}\thesubsubsection.\hspace{0.7em}}}%
  \@startsection%
   {subsubsection}%         % nom de l'inter
   {3}%                     % niveau de l'inter
   {0pt}%                   % l'indentation du titre et du texte suivant
   {3pt}
   {3pt}
   {\FonteSectionIII}}%     % style

\makeevenhead{ruled}{\small\textsc{\rightmark}}{}{\thepage}
\makeoddhead{ruled}{\small\textsc{\rightmark}}{}{\thepage}
\makeoddfoot{ruled}{\hrule height 0.3mm \small\textsc{\phdthesislabel}}{}{\small\textsc{\studentlab}}
\makeevenfoot{ruled}{\hrule height 0.3mm \small\textsc{\phdthesislabel}}{}{\small\textsc{\studentlab}}

\newcommand{\@Myfnmark}{
      \mbox{\fontsize{8}{11}\sffamily\arabic{footnote}. }%
}

\renewcommand{\@makefntext}[1]{%
      \noindent\@Myfnmark#1%
}%

\def\@thefnmark{\arabic{footnote}}

\settocdepth{subsubsection}
\setsecnumdepth{subsubsection}
\maxsecnumdepth{subsubsection}
\settocdepth{subsubsection}
\maxtocdepth{subsubsection}

\newcommand{\lof}{false}
\renewcommand{\numberline}[1]{
\ifthenelse{\equal{\lof}{false}}{\hspace{-0.6cm}$\mathrm{{#1}}$ -} {$\mathrm{{#1}}$ -}
}

\let\oldcontentsline=\contentsline
\renewcommand{\contentsline}[4]{
	\vspace{0.8mm}
	\oldcontentsline{#1}{#2}{#3}{#4}
	\vspace{0.8mm}
}

\newcommand{\minitoclevel}{section}
\newcommand{\minitocstyle}{titleontopright}

\newcommand{\myMiniToc}[2]{
	\ifthenelse{\equal{#1}{}}
	{\renewcommand{\minitoclevel}{section}}
	{\renewcommand{\minitoclevel}{#1}}
	\startcontents[chapters]
	\ifthenelse{\equal{\minitocstyle}{fieldset}}
	{\myMiniTocFieldset{#2}}
	{
		\ifthenelse{\equal{\minitocstyle}{titleontopright}}
		{\myMiniTocTitleOnTopRight{#2}}
		{}
	}
}%

\newcommand{\myMiniTocTitleOnTopRight}[1]{
	\vspace{-0.7mm}
	\vspace{20pt}
	\begin{minipage}{\textwidth}
	\begin{flushright}\noindent\textcolor{black}{\textbf{#1}} \vspace{5pt} \hrule height 0.08mm \end{flushright}
	\par
	\printcontents[chapters]{}{1}{}
	\par
	\begin{flushright} \vspace{0pt} \hrule height 0.08mm \vspace{30pt} \end{flushright}
	\end{minipage}
}

\newcommand{\myMiniTocFieldset}[1]{
	\vspace{5pt}
	\begin{minipage}{\textwidth}
		\parbox[c]{.07\textwidth}{
			\centering
			\textcolor{black}{\hrule height 0.4mm}
		}
		\parbox[c]{.2\textwidth}{
			\centering
			\textcolor{black}{\textsc{\textbf{#1}}}
		}
		\parbox[c]{.71\textwidth}{
			\centering
			\textcolor{black}{\hrule height 0.4mm}
		}
		\par
		\printcontents[chapters]{}{1}{}
		\par
		\begin{flushright} 
			\vspace{0pt} 
			\hrule height 0.4mm 
			\vspace{30pt} 
		\end{flushright}
	\end{minipage}
}

\newcommand{\myMiniTocStyle}[1]{
	\renewcommand{\minitocstyle}{#1}
}

\newcommand{\myBibliography}[2]{
	\bibliographystyle{#1}
	\bibliography{#2}
	\myCleanStarChapterEnd
}

\def\bibi[#1]{\item[\@biblabel{#1}\hfill]} % @ special
\newenvironment{myBiblio}{%BEGIN
   \list{}{
         \usecounter{enumiv}%
         \let\p@enumiv\@empty
         \renewcommand\newblock{\hskip .11em \@plus.33em \@minus.07em}%
         \setlength{\leftmargin}{0mm}%%%
         \setlength{\labelsep}{2mm}%%%
         \setlength{\labelwidth}{10mm}%%%
          \setlength{\topsep}{0pt}%
          \setlength{\parskip}{6pt}%
          \setlength{\itemsep}{5pt}%
          \setlength{\partopsep}{0pt}%
          \setlength{\parsep}{3pt}%
         \sloppy\clubpenalty4000\widowpenalty4000%
         \sfcode`\.=\@m
         }%
  }{%END
      \def\@noitemerr{\@latex@warning{Empty 'thebibliography' environment}}
      \endlist%
}

\renewcommand{\cite}[1]{\Citep{#1}}

\usepackage[format=hang,font=small,labelfont=bf,textfont=it,skip=5pt,labelsep=endash]{caption}
\newcommand{\tocsetted}{false}

\let\oldmainmatter=\mainmatter
\renewcommand{\mainmatter}{
	\oldmainmatter
	\renewcommand{\thechapter}{\Roman{chapter}}
	\renewcommand{\thetable}{\Roman{table}}
	\renewcommand{\thefigure}{\arabic{figure}}
	
	\counterwithout*{figure}{chapter}
	\counterwithout*{table}{chapter}
}

\let\oldappendix=\appendix
\renewcommand{\appendix}{
	\oldappendix
	\renewcommand{\thetable}{\Roman{table}}
	\renewcommand{\thefigure}{\arabic{figure}}
	
	\counterwithout*{figure}{chapter}
	\counterwithout*{table}{chapter}
}

\newcommand{\myChapterLabel}{Chapter}
\newcommand{\myAppendixLabel}{Appendix}
\newcommand{\lifa}{Laboratoire d'Informatique Fondamentale et Appliquée (LIFA)}
\newcommand{\myBibliographyTitle}{Bibliography}
\newcommand{\losname}{List of Symbols}
\newcommand{\loaname}{List of Acronyms}

\newcommand{\doctypethesis}{Thèse de Doctorat en}
\newcommand{\doctypemaster}{Mémoire de Master en}
\newcommand{\doctype}{\ifthenelse{\equal{\doclevel}{\master}}{\doctypemaster}{\doctypethesis}}
\newcommand{\phd}{PhD}
\newcommand{\master}{Master}
\newcommand{\doclevel}{\phd}
\newcommand{\level}[1]{
	\renewcommand{\doclevel}{#1}
}
\newcommand{\phdthesislabel}{\doctype $~$ \studentspeciality $~$, Université de Dschang}
\newcommand{\studentspeciality}{\computerScience}
\newcommand{\speciality}[1]{
	\renewcommand{\studentspeciality}{#1}
}
\newcommand{\computerScience}{Informatique}
\newcommand{\mathematics}{Mathématiques}
\newcommand{\studentlab}{LIFA}
\newcommand{\lab}[1]{
	\renewcommand{\studentlab}{#1}
}

\newcommand{\myTableOfContents}[1]{
	\ifthenelse{\equal{\tocsetted}{false}}
	{\clearpage}{}
	\mySaveMarks
	\ifthenelse{\equal{#1}{}}{}
	{\renewcommand{\contentsname}{#1}}
	\addcontentsline{toc}{section}{\myNumberLine{\contentsname}}
	\renewcommand{\leftmark}{\contentsname}
	\renewcommand{\rightmark}{\contentsname}
	\tableofcontents*
	\myCleanStarChapterEnd
	\renewcommand{\tocsetted}{true}
}

\newcommand{\myTableOfContentsStar}[1]{
	\ifthenelse{\equal{\tocsetted}{false}}
	{\clearpage}{}
	\mySaveMarks
	\ifthenelse{\equal{#1}{}}{}
	{\renewcommand{\contentsname}{#1}}
	\renewcommand{\leftmark}{\contentsname}
	\renewcommand{\rightmark}{\contentsname}
	\tableofcontents*
	\myCleanStarChapterEnd
	\renewcommand{\tocsetted}{true}
}

\newcommand{\myListOfSymbols}[1]{
	\ifthenelse{\equal{#1}{}}{}
	{\renewcommand{\losname}{#1}}
	\myChapterStar{\losname}{}{section}
	\begin{center}
	\begin{tabular}[t]{rp{5mm}p{12cm}}
		$EC$ & &  The Editor in Chief in the running example \\
		$AE$ & &  The Associated Editor in the running example \\
		$R1$ & &  The first Referee in the running example \\
		$R2$ & &  The second Referee in the running example \\
		$\mathbb{G}$ & &  A grammatical model of workflow \\
		$t_{i_f}$ & & A global artefact obtained after merging a set of artefacts
	\end{tabular}
\end{center}

	\myCleanStarChapterEnd
	\renewcommand{\tocsetted}{true}
}

\newcommand{\myListOfSymbolsStar}[1]{
	\ifthenelse{\equal{#1}{}}{}
	{\renewcommand{\losname}{#1}}
	\myChapterStar{\losname}{}{false}
	\begin{center}
	\begin{tabular}[t]{rp{5mm}p{12cm}}
		$EC$ & &  The Editor in Chief in the running example \\
		$AE$ & &  The Associated Editor in the running example \\
		$R1$ & &  The first Referee in the running example \\
		$R2$ & &  The second Referee in the running example \\
		$\mathbb{G}$ & &  A grammatical model of workflow \\
		$t_{i_f}$ & & A global artefact obtained after merging a set of artefacts
	\end{tabular}
\end{center}

	\myCleanStarChapterEnd
	\renewcommand{\tocsetted}{true}
}

\newcommand{\myListOfAcronyms}[1]{
	\ifthenelse{\equal{#1}{}}{}
	{\renewcommand{\loaname}{#1}}
	\myChapterStar{\loaname}{}{section}
	\begin{center}
	\begin{tabular}[t]{rp{5mm}p{10cm}}
		AST & & Abstract Syntax Tree \\
		BPM & & Business Process Management \\
		BPMN & & Business Process Model and Notation \\
		CDML & & Component Description Meta Language \\
		CSCW & & Computer-Supported Cooperative Work \\
		P2P & & Peer to Peer \\
		DS(E)L & & Domain Specific (Embeded) Language \\
		DTD & & Document Type Definition \\
		GMAWfP & & a Grammatical Model of Administrative Workflow Process \\
		GMWf & & Grammatical Model of Workflow \\
		LSAWfP & & a Language for the Specification of Administrative Workflow Processes \\
		(L)WfE & & (Local) Workflow Engine \\
		P2PTinyWfMS & & a Peer-to-Peer Tiny Workflow Management System \\
		SOA & & Service Oriented Architecture \\
		SON & & Shared-Overlay Network \\
		TinyCE & & a Tiny Cooperative Editor \\
		WfM(S) & & Workflow Management (System) \\
		WF-Net & & Workflow Net \\
		(WS-)BPEL & & (Web Services) Business Process Execution Language \\
		XML & & eXtensible Markup Language \\
		YAWL & & Yet Another Workflow Language \\
	\end{tabular}
\end{center}

	\myCleanStarChapterEnd
	\renewcommand{\tocsetted}{true}
}

\newcommand{\myListOfAcronymsStar}[1]{
	\ifthenelse{\equal{#1}{}}{}
	{\renewcommand{\loaname}{#1}}
	\myChapterStar{\loaname}{}{false}
	\begin{center}
	\begin{tabular}[t]{rp{5mm}p{10cm}}
		AST & & Abstract Syntax Tree \\
		BPM & & Business Process Management \\
		BPMN & & Business Process Model and Notation \\
		CDML & & Component Description Meta Language \\
		CSCW & & Computer-Supported Cooperative Work \\
		P2P & & Peer to Peer \\
		DS(E)L & & Domain Specific (Embeded) Language \\
		DTD & & Document Type Definition \\
		GMAWfP & & a Grammatical Model of Administrative Workflow Process \\
		GMWf & & Grammatical Model of Workflow \\
		LSAWfP & & a Language for the Specification of Administrative Workflow Processes \\
		(L)WfE & & (Local) Workflow Engine \\
		P2PTinyWfMS & & a Peer-to-Peer Tiny Workflow Management System \\
		SOA & & Service Oriented Architecture \\
		SON & & Shared-Overlay Network \\
		TinyCE & & a Tiny Cooperative Editor \\
		WfM(S) & & Workflow Management (System) \\
		WF-Net & & Workflow Net \\
		(WS-)BPEL & & (Web Services) Business Process Execution Language \\
		XML & & eXtensible Markup Language \\
		YAWL & & Yet Another Workflow Language \\
	\end{tabular}
\end{center}

	\myCleanStarChapterEnd
	\renewcommand{\tocsetted}{true}
}

\newcommand{\myListOfFigures}[1]{
	\ifthenelse{\equal{\tocsetted}{false}}
	{\clearpage}{}
	\mySaveMarks
	\ifthenelse{\equal{#1}{}}{}
	{\renewcommand{\listfigurename}{#1}}
	\addcontentsline{toc}{section}{\myNumberLine{\listfigurename}}
	\renewcommand{\leftmark}{\listfigurename}
	\renewcommand{\rightmark}{\listfigurename}
	\renewcommand{\lof}{true}
	\listoffigures*
	\renewcommand{\lof}{false}
	\myCleanStarChapterEnd
	\renewcommand{\tocsetted}{true}
}

\newcommand{\myListOfFiguresStar}[1]{
	\ifthenelse{\equal{\tocsetted}{false}}
	{\clearpage}{}
	\mySaveMarks
	\ifthenelse{\equal{#1}{}}{}
	{\renewcommand{\listfigurename}{#1}}
	\renewcommand{\leftmark}{\listfigurename}
	\renewcommand{\rightmark}{\listfigurename}
	\renewcommand{\lof}{true}
	\listoffigures*
	\renewcommand{\lof}{false}
	\myCleanStarChapterEnd
	\renewcommand{\tocsetted}{true}
}

\newcommand{\myListOfTables}[1]{
	\ifthenelse{\equal{\tocsetted}{false}}
	{\clearpage}{}
	\mySaveMarks
	\ifthenelse{\equal{#1}{}}{}
	{\renewcommand{\listtablename}{#1}}
	\addcontentsline{toc}{section}{\myNumberLine{\listtablename}}
	\renewcommand{\leftmark}{\listtablename}
	\renewcommand{\rightmark}{\listtablename}
	\renewcommand{\lof}{true}
	\listoftables*
	\renewcommand{\lof}{false}
	\myCleanStarChapterEnd
	\renewcommand{\tocsetted}{true}
}

\newcommand{\myListOfTablesStar}[1]{
	\ifthenelse{\equal{\tocsetted}{false}}
	{\clearpage}{}
	\mySaveMarks
	\ifthenelse{\equal{#1}{}}{}
	{\renewcommand{\listtablename}{#1}}
	\renewcommand{\leftmark}{\listtablename}
	\renewcommand{\rightmark}{\listtablename}
	\renewcommand{\lof}{true}
	\listoftables*
	\renewcommand{\lof}{false}
	\myCleanStarChapterEnd
	\renewcommand{\tocsetted}{true}
}

\newcommand{\myListOfAlgorithms}[1]{
	\ifthenelse{\equal{\tocsetted}{false}}
	{\clearpage}{}
	\mySaveMarks
	\ifthenelse{\equal{#1}{}}{}
	{\renewcommand{\listalgorithmname}{#1}}
	\addcontentsline{toc}{section}{\myNumberLine{\listalgorithmname}}
	\renewcommand{\leftmark}{\listalgorithmname}
	\renewcommand{\rightmark}{\listalgorithmname}
	\listofalgorithms
	\myCleanStarChapterEnd
	\renewcommand{\tocsetted}{true}
}

\newcommand{\myListOfAlgorithmsStar}[1]{
	\ifthenelse{\equal{\tocsetted}{false}}
	{\clearpage}{}
	\mySaveMarks
	\ifthenelse{\equal{#1}{}}{}
	{\renewcommand{\listalgorithmname}{#1}}
	\renewcommand{\leftmark}{\listalgorithmname}
	\renewcommand{\rightmark}{\listalgorithmname}
	\renewcommand{\lof}{true}
	\listofalgorithms
	\renewcommand{\lof}{false}
	\myCleanStarChapterEnd
	\renewcommand{\tocsetted}{true}
}

\newcommand{\myChapter}[2]{
	\chapter[#2]{#1}
}

\newcommand{\shortTitle}{}

\newcommand{\mySaveMarks}{
	\let\oldleftmark=\leftmark
	\let\oldrightmark=\rightmark
}

\newcommand{\myNumberLine}[1]{
	\hspace{-0.55cm}#1
}

\newcommand{\myChapterNumberLine}[1]{
	\hspace{-0.25cm}#1
}

\newcommand{\myChapterStar}[3]{
	\mySaveMarks
	\ifthenelse{\equal{#2}{}}
	{\renewcommand{\shortTitle}{#1}}
	{\renewcommand{\shortTitle}{#2}}
	\renewcommand{\leftmark}{\shortTitle}
	\renewcommand{\rightmark}{\shortTitle}
	\chapter*{#1}
	\ifthenelse{\equal{#3}{false}}
	{}
	{
		\ifthenelse{\equal{#3}{}}
		{\addcontentsline{toc}{chapter}{\myChapterNumberLine{\shortTitle}}}
		{
			\ifthenelse{\equal{#3}{true}}
			{\addcontentsline{toc}{chapter}{\myChapterNumberLine{\shortTitle}}}
			{
				\ifthenelse{\equal{#3}{chapter}}
				{\addcontentsline{toc}{#3}{\myChapterNumberLine{\shortTitle}}}
				{\addcontentsline{toc}{#3}{\myNumberLine{\shortTitle}}}
			}
		}
	}
}

\newcommand{\mySection}[2]{
	\resumecontents[chapters]
	\ifthenelse{\equal{#2}{}}
	{
		\renewcommand{\shortTitle}{#1}
		\Needspace{5\baselineskip}
		\section{#1}
	}
	{
		\renewcommand{\shortTitle}{#2}
		\Needspace{5\baselineskip}
		\section[#2]{#1}
	}
	\hrule height 0.5mm
	\vspace{5mm}
	\stopcontents[chapters]
}

\newcommand{\mySectionStar}[3]{
	\resumecontents[chapters]
	\ifthenelse{\equal{#2}{}}
	{\renewcommand{\shortTitle}{#1}}
	{\renewcommand{\shortTitle}{#2}}
	\renewcommand{\rightmark}{\shortTitle}
	\section*{#1}
	\hrule height 0.5mm
	\vspace{5mm}
	\ifthenelse{\equal{#3}{false}}
	{}
	{
		\ifthenelse{\equal{#3}{}}
		{\addcontentsline{toc}{section}{\myNumberLine{\shortTitle}}}
		{
			\ifthenelse{\equal{#3}{true}}
			{\addcontentsline{toc}{section}{\myNumberLine{\shortTitle}}}
			{
				\ifthenelse{\equal{#3}{chapter}}
				{\addcontentsline{toc}{#3}{\myChapterNumberLine{\shortTitle}}}
				{\addcontentsline{toc}{#3}{\myNumberLine{\shortTitle}}}
			}
		}
	}
	\stopcontents[chapters]
}

\newcommand{\mySubSection}[2]{
	\ifthenelse{\equal{\minitoclevel}{section}}{}
	{\resumecontents[chapters]}
	\ifthenelse{\equal{#2}{}}
	{
		\subsection{#1}
	}
	{
		\subsection[#2]{#1}
	}
	\stopcontents[chapters]
}

\newcommand{\mySubSectionStar}[3]{
	\ifthenelse{\equal{\minitoclevel}{section}}{}
	{\resumecontents[chapters]}
	\ifthenelse{\equal{#2}{}}
	{\renewcommand{\shortTitle}{#1}}
	{\renewcommand{\shortTitle}{#2}}
	\renewcommand{\rightmark}{\shortTitle}
	\subsection*{#1}
	\ifthenelse{\equal{#3}{false}}
	{}
	{
		\ifthenelse{\equal{#3}{}}
		{\addcontentsline{toc}{subsection}{\myNumberLine{\shortTitle}}}
		{
			\ifthenelse{\equal{#3}{true}}
			{\addcontentsline{toc}{subsection}{\myNumberLine{\shortTitle}}}
			{
				\ifthenelse{\equal{#3}{chapter}}
				{\addcontentsline{toc}{#3}{\myChapterNumberLine{\shortTitle}}}
				{\addcontentsline{toc}{#3}{\myNumberLine{\shortTitle}}}
			}
		}
	}
	\stopcontents[chapters]
}

\newcommand{\mySubSubSection}[2]{
	\ifthenelse{\equal{\minitoclevel}{subsubsection}}
	{\resumecontents[chapters]}{}
	\ifthenelse{\equal{#2}{}}
	{
		\renewcommand{\shortTitle}{#1}
		\subsubsection{#1}
	}
	{
		\renewcommand{\shortTitle}{#2}
		\subsubsection[#2]{#1}
	}
	\stopcontents[chapters]
}

\newcommand{\mySubSubSectionStar}[3]{
	\ifthenelse{\equal{\minitoclevel}{subsubsection}}
	{\resumecontents[chapters]}{}
	\ifthenelse{\equal{#2}{}}
	{\renewcommand{\shortTitle}{#1}}
	{\renewcommand{\shortTitle}{#2}}
	\renewcommand{\rightmark}{\shortTitle}
	\subsubsection*{#1}
	\ifthenelse{\equal{#3}{false}}
	{}
	{
		\ifthenelse{\equal{#3}{}}
		{\addcontentsline{toc}{subsubsection}{\myNumberLine{\shortTitle}}}
		{
			\ifthenelse{\equal{#3}{true}}
			{\addcontentsline{toc}{subsubsection}{\myNumberLine{\shortTitle}}}
			{
				\ifthenelse{\equal{#3}{chapter}}
				{\addcontentsline{toc}{#3}{\myChapterNumberLine{\shortTitle}}}
				{\addcontentsline{toc}{#3}{\myNumberLine{\shortTitle}}}
			}
		}
	}
	\stopcontents[chapters]
}

\newcommand{\myRestoreMarks}{
	\let\leftmark=\oldleftmark
	\let\rightmark=\oldrightmark
}

\newcommand{\myCleanStarChapterEnd}{
	\clearpage
	\myRestoreMarks
}

\newcommand{\currentlanguage}{english}

\newcommand{\switchLanguage}[1]{
	\renewcommand{\currentlanguage}{#1}
	\ifthenelse{\equal{#1}{français}}
	{
		\usepackage[frenchb]{babel}
		\renewcommand{\myChapterLabel}{Chapitre}
		\renewcommand{\myAppendixLabel}{Annexe}
		\renewcommand{\lifa}{Laboratoire d'Informatique Fondamentale et Appliquée (LIFA)}
		\renewcommand{\myBibliographyTitle}{Bibliographie}
		\renewcommand{\phdthesislabel}{Thèse de Doctorat en $~$ \studentspeciality $~$, Université de Dschang}
		\renewcommand{\computerScience}{Informatique}
		\renewcommand{\mathematics}{Mathématiques}
		\renewcommand{\studentlab}{URIFIA}
		\renewcommand{\doctypethesis}{Thèse de Doctorat en}
		\renewcommand{\doctypemaster}{Mémoire de Master en}
		\renewcommand{\losname}{Liste des Symboles}
		\renewcommand{\loaname}{Liste des Acronymes}
	}{
		\usepackage[english]{babel}
		\renewcommand{\myChapterLabel}{Chapter}
		\renewcommand{\myAppendixLabel}{Appendix}
		\renewcommand{\lifa}{Laboratoire d'Informatique Fondamentale et Appliquée (LIFA)}
		\renewcommand{\myBibliographyTitle}{Bibliography}
				\renewcommand{\phdthesislabel}{PhD Thesis in $~$ \studentspeciality $~$, University of Dschang}
		\renewcommand{\computerScience}{Computer Science}
		\renewcommand{\mathematics}{Mathematics}
		\renewcommand{\studentlab}{URIFIA}
		\renewcommand{\doctypethesis}{PhD Thesis in}
		\renewcommand{\doctypemaster}{Master Report in}
		\renewcommand{\losname}{List of Symbols}
		\renewcommand{\loaname}{List of Acronyms}
	}
}

\newcommand{\documentType}[1]{
	\ifthenelse{\equal{#1}{numerical}}
	{
		\ifpdf
			\pdfcompresslevel=9
				\usepackage[plainpages=false,pdfpagelabels,bookmarksnumbered,%
				colorlinks=true,%
				linkcolor=blue,%
				citecolor=blue,%
				filecolor=forestgreen,%
				urlcolor=midnightblue,%
				pdftex,%
				unicode]{hyperref}
			\pdfimageresolution=600
			\usepackage{thumbpdf} 
		\else
			\usepackage{hyperref}
		\fi
	}
	{
		\ifpdf
			\pdfcompresslevel=9
				\usepackage[plainpages=false,pdfpagelabels,bookmarksnumbered,%
				colorlinks=true,%
				linkcolor=black,%
				citecolor=black,%
				filecolor=black,%
				urlcolor=black,%
				pdftex,%
				unicode]{hyperref}
			\pdfimageresolution=600
			\usepackage{thumbpdf} 
		\else
			\usepackage{hyperref}
		\fi
	}
}

\letcountercounter{sidefootnote}{footnote}

\DeclareMathVersion{normal2}

\documentType{numerical}

\switchLanguage{english}

\usepackage{algorithm}
\usepackage{algpseudocode}
\usepackage{listings}% http://ctan.org/pkg/listings  Pour ecrire des math dans un verbatim \begin{lstlisting} ... \end{lstlisting}
\renewcommand{\listalgorithmname}{List of Algorithms}
\floatname{algorithm}{Algorithm}

\let\oldComment=\Comment
\renewcommand{\Comment}[1]{\oldComment{{\scriptsize#1}}}
\algdef{SE}[DOWHILE]{DoWhile}{EndDoWhile}{\algorithmicdo}[1]{\algorithmicwhile\ #1}%

\newtheorem{theorem}{Theorem}
\newtheorem{definition}[theorem]{Definition}
\newtheorem{proposition}[theorem]{Proposition}

\newtheorem{example}[theorem]{Example}

\newtheorem{corollary}[theorem]{Corollary}

\newenvironment{proof}[1][{\textbf{Proof}}]{
	\par
	\normalfont
	\topsep6\p@\@plus6\p@ \trivlist
	\item[\hskip\labelsep\itshape
	#1\@addpunct{.}]\ignorespaces
}{%
	\qed\endtrivlist
}

\usepackage[framemethod=TikZ]{mdframed}
\mdfdefinestyle{MyFrame}{%
	linecolor=white,
	outerlinewidth=0pt,
	roundcorner=0pt,
	innertopmargin=4pt,
	innerbottommargin=4pt,
	innerrightmargin=4pt,
	innerleftmargin=4pt,
	leftmargin=4pt,
	rightmargin=4pt
}

\usepackage{csquotes}

\level{\phd}

\speciality{\computerScience}

\lab{URIFIA}

\myChapterStyle{fieldset-rev}
\myMiniTocStyle{fieldset}

\title{A Grammatical Approach to Peer-to-Peer Cooperative Editing on a Service-Oriented Architecture}
\author{ZEKENG NDADJI Milliam Maxime}
\date{\today}

\begin{document}
{
	% Taille de la police du texte
	\timesContentFontSize

	% Code pour inclure des documents au format PDF
	\includepdf[pages=-, offset=72 -72]{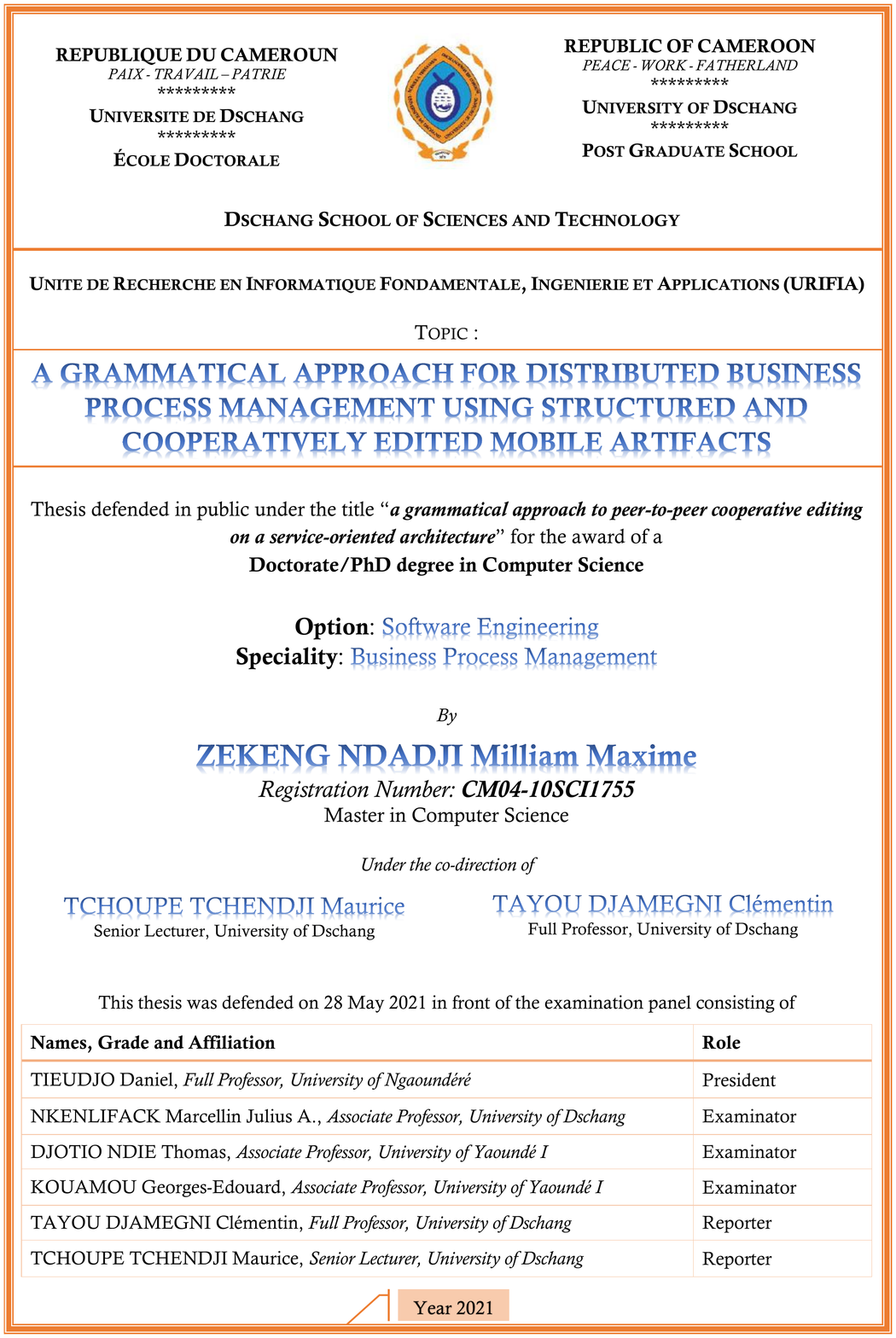} %Insertion de la couverture
	%\includepdf[pages=-, offset=72 -72]{Originalite.pdf} %Attestation d'originalité
	%\includepdf[pages=-, offset=72 -72]{Correction.pdf} %Attestaion de correction

	\pagestyle{ruled}
	\nouppercaseheads
	\normalfont

	\frontmatter
	
	% Inclusion des fichiers de dédicaces et de remerciements
	%\input{declaration-authorship/declaration-authorship}
	%\myChapterStar{Titre}{Titre court}{Ajouter à la table des matières? (false|true|chapter|section|subsection|subsubsection -chapter par défaut-)}
\myChapterStar{Dedication}{}{section}
\vspace*{5cm}
\begin{flushright}
\textit{I dedicate this work to all those who one day, saw the efforts they had put into a project (especially an intense love relationship) being wiped out without a logical explanation and who took it upon themselves, to ride the waves unleashed by these storms, came out ten times stronger and developed an incredible desire to live rather than die. Like your daily efforts, this work is in large part the result of the knowledge that suffering has brought to me and thus, it contributes to prove Friedrich Nietzsche's aphorism: "\textbf{what doesn't kill you, makes you stronger}".}
\end{flushright}

	%\myChapterStar{Titre}{Titre court}{Ajouter à la table des matières? (false|true|chapter|section|subsection|subsubsection -chapter par défaut-)}
\myChapterStar{Acknowledgements}{}{section}
Along the paths I followed during this thesis, I have gained a unique experience and, I hope, the necessary maturity to aspire to the title of Doctor/PhD. Throughout these years of work, I have come to realise that producing a thesis is far from being a solitary labour. In order to go through this tunnel, one must be under the benevolent escort of the Almighty Lord and of souls of good faith. I have benefited from a multiform accompaniment of many persons, both physical and moral. I would like them to find in this section, the expression of my deep gratitude. I therefore thank very warmly:

~

\noindent- The \textsc{LORD} \textit{our God, creator of heaven and earth}: in addition to my days, he offers me every day, the necessary grace and strength to continue to glorify him. Glory be to you, Holy Father!

~

\noindent- Dr \textsc{TCHOUP{\'E} TCHENDJI Maurice}, \textit{Senior Lecturer at the Department of Mathematics and Computer Science of the Faculty of Science of the University of Dschang}: he took the time to mentor me throughout this work, inculcating in me his sense of commitment and organisation. Sir, I reiterate that you are my model.

~

\noindent- Pr \textsc{TAYOU DJAMEGNI Cl{\'e}mentin}, \textit{Head of the Computer Engineering Department of the Institute of Technology - Fotso Victor of Bandjoun}: it was under difficult conditions that he agreed to supervise this work and to reward us with his legendary positivism. I'm more than honoured to be one of your students, Sir.

~

\noindent- All the imminent members of my pre-hearing and defence juries: they agreed to objectively evaluate this work. Thank you for the chance you are giving me, gentlemen.

~

\noindent- Dr \textsc{PARIGOT Didier}, \textit{Senior Researcher on Programming Languages at INRIA}: he accompanied me during this thesis by his reviews, his points of view, his recommendations and his know-how in terms of software programming. I learned a lot from your multiple contributions, Sir.

~

\noindent- The \textit{INRIA/LIRIMA FUCHSIA Associate Research Team}: they welcomed me when I needed it most and gave me a chance to express myself. Working with you guys is a lifelong dream.

~

\noindent- The \textit{lecturers at the University of Dschang}: they contributed enormously to my education. As education is priceless, I can only express my gratitude.

~

\noindent- My \textit{parents}, \textsc{NDADJI Emmanuel} and \textsc{MAFFOZEMTSOP Marie}: they gave me everything. I love you infinitely.

~

\noindent- My \textit{brothers, sisters, friends and colleagues (Nani, Chance, Alex, Arnold, Brice, Ari\`ege, Brel, Doris, Audrey, Fabrice, Rodrigue, Virginie, Emeric, Nestor, Preston, Yann, Max, Ange, Lionel, etc.)}: I live only by you and for you (uh, for me too ;-)).

~

\noindent- My dear \textsc{BANGUKET T. Nina}: she proofread and corrected the typos in all my English documents. Now I can write a correct sentence in English, thanks to you.

\myCleanStarChapterEnd

	% Génération de la table des matières (avec "Table of Contents" comme titre), de la liste des symboles, de la liste des acronymes, de la liste des tableaux et de la liste des figures
	% Des alternatives à ces commandes sont dispos: il suffit juste d'ajouter le suffixe "Star" pour ne pas les mettre dans la table de matière
	\myTableOfContents{Table of Contents}
	
	% Resumé et abstract
	\let\oldprintchaptertitle=\printchaptertitle
\renewcommand{\printchaptertitle}[1]{%
	\vspace*{-75pt}
	\oldprintchaptertitle{#1}
}%
\myChapterStar{Abstract}{}{section}
\let\printchaptertitle=\oldprintchaptertitle
In this thesis, we focus on the proposal of distributed workflow systems dedicated to the automation of administrative business processes. We propose an approach to build such systems by relying on the concepts of multiagent systems, Peer to Peer (P2P) architecture, Service-Oriented Architecture (SOA) and structured documents (artifacts) cooperative edition. 
Indeed, we develop mathematical tools that allow any workflow systems designer, to express each administrative process in the form of an attributed grammar whose symbols represent tasks to be executed, productions specify a scheduling of these tasks, and instances (the derivation trees that conform to it) represent the different execution scenarios leading to business goal states. The obtained grammatical model is then introduced into a proposed P2P system which is in charge of carrying out the completely decentralised execution of the underlying process's instances. 
The said system orchestrates a process's instance execution as a choreography during which, various software agents driven by human agents (actors), coordinate themselves through artifacts that they collectively edit. The exchanged artifacts represent the system's memory: they provide information on already executed tasks, on those ready to be executed and on their executors. The software agents are autonomous and identical: they execute the same unique protocol each time they receive an artifact. This protocol allows them to identify the tasks they must immediately execute, to execute them, to update the artifact and to disseminate it if necessary, for the continuation of the execution. 
Moreover, actors potentially have only a partial perception of processes in which they are involved. In practice, this means that certain tasks can be carried out confidentially: this property makes it possible to offer automatic management of administrative processes that is a little closer to their non-computerised management.

\vspace{1cm}

\myCleanStarChapterEnd

	\clearpage
	
	%\myListOfSymbols{}
	\myListOfAcronyms{}
	\myListOfTables{}
	\myListOfFigures{}
	\myListOfAlgorithms{}

	% *********** Partie principale ***********
	\mainmatter
	
	% Inclusion des différents chapitres
	%\myChapterStar{Titre}{Titre court}{Ajouter à la table des matières? (false|true|chapter|section|subsection|subsubsection -chapter par défaut-)}
\myChapterStar{General Introduction}{}{true}
\label{chap0:introduction}
\myMiniToc{section}{Contents}
% If no minitoc then
% \startcontents[chapters]

\mySectionStar{The Emergence of Business Process Management}{}{true}
\label{chap0:sec:bpm-emergence}
Business Process Management (BPM) has received considerable attention in recent years due to its potential for significantly increasing productivity and saving costs. It is defined by Wil M. P. Van Der Aalst \citeyearpar{van2013business} as \textit{"the discipline that combines knowledge from information technology and knowledge from management sciences and applies this to operational business processes"}. BPM aims to improve business processes by focusing on their automation, their analysis, their involvement in decision-making operations (management) and the organisation of work. Hence, BPM is often accompanied by software to manage, control and support business processes: these software systems are called \textit{Workflow Management Systems} (WfMS) \cite{workflow95, ima}.

The BPM discipline emerged in the 1980s in a more restrictive form known as \textit{Workflow Management} (WfM). Before WfM was developed, information systems were built from scratch; it means that, all components of such systems had to be programmed, including data storage and retrieval \cite{van1998application}. Software vendors soon realised that many information systems had similar data management requirements. This generic functionality was therefore outsourced to a data management system. Subsequently, the generic functionality related to user interaction (forms, buttons, graphics, etc.) was outsourced to user interfaces generators. The trend to outsource recurring functionalities to generic tools has continued in different areas. It is in this context that WfM has been introduced. Precisely, a WfMS automatically manages the process-related aspects \cite{workflow95, van2013business} of information systems. The aim of WfM in the design of information systems is to simplify as much as possible, the modelling and management of the business processes they automate: traditionally, this modelling and management are reduced to the specification of processes using simple graphical languages called \textit{workflow languages}.

Across the years, WfM has evolved into BPM. While WfM focused primarily on the automation of business processes, i.e. it was not fundamentally interested in other issues such as the analysis, verification and maintenance of their specifications, BPM made it its foundation \cite{van2016don}. With this evolution, many tools and techniques have been developed and have allowed the BPM field to mature. Today, its relevance is recognised by practitioners (users, managers, analysts, consultants and software developers) and academics. The availability of many BPM systems (WfMS) and a series of BPM-related conferences is proof of this \cite{van2013business}.

With the evolution of WfM and the development of new concepts related to the implementation of collaborative software systems, namely \textit{Peer-to-Peer (P2P) computing}, \textit{multiagent} paradigm and \textit{Service-Oriented Architecture} (SOA), the way of designing and implementing WfMS has also evolved. There has been a shift from centralised systems implemented according to a reference architectural model \cite{workflowModel}, to fully distributed systems offering decentralised workflow execution \cite{meilin1998workflow, junYan06, fakas04}. New paradigms of specification and management of business processes have been developed too. Among these, we find the \textit{document-centric} \cite{krishnan2001xdoc, marchetti2005xflow, badouelTchoupeCmcs}, the \textit{email-based} \cite{burkhart2012context, gazze2012workmail}, the \textit{database-based} \cite{actionWorkflow, miao2011realization}, the \textit{artifact-centric} \cite{nigam2003business, hull2009facilitating, lohmann2010artifact}, the \textit{data-centric} \cite{damaggio2013equivalence, badouel2015active} paradigms, etc. The \textit{artifact-centric} paradigm has been the subject of several studies over the last two decades \cite{abi2016towards, deutsch2014automatic, hull2009facilitating, lohmann2010artifact, assaf2017continuous, assaf2018generating, boaz2013bizartifact, lohmann2011artifact, estanol2012artifact}; it has been very successful because it has enabled the development of much more flexible workflow languages, that treat process data as first-class citizens, as opposed to the existing languages (BPMN - \textit{Business Process Model and Notation} \cite{BPMN} -, YAWL - \textit{Yet Another Workflow Language} \cite{van1998application, van2005yawl} -, etc.), that were only concerned with process tasks scheduling and assignment to actors. All these concepts' evolution and this involvement of many technologies in workflow systems' engineering, has made the BPM field, one of the most attractive for software vendors and software engineering researchers.

\mySectionStar{The Mitigated Use of Business Process Management}{}{true}
\label{chap0:sec:bpm-mitigated-use}
The BPM discipline has quickly established itself as an indispensable solution to the process automation needs of large firms, which often involve production lines \cite{van2016don}. Indeed, the workflows encountered in this context are generally highly structured and their tasks are almost entirely automated (executed by machines); this greatly "simplifies" their management by BPM. Another field of application in which BPM tends to naturally impose itself, is that of science. Indeed, the management (storage, distribution, computation, analysis, etc.) of generally very voluminous scientific data, involves several human and material resources that are often distributed across organisations \cite{bell2009beyond}. With the help of cloud computing, BPM in this context, serves to organise these resources for a more efficient management of scientific data \cite{juveGideon, ludascher2006scientific}. In these two application contexts, the systems' complexity due to the large amount of data to be managed and to the time-consuming and intensive computations requiring computer support, as well as the multiplicity and the distributivity of the involved resources, have somehow "imposed" the use of BPM.

Although in its current state BPM could be applied very successfully in the context of organisations with so-called administrative business processes (i.e. processes whose tasks are often semi-automated and therefore require the expertise of human agents - they are then more complex to automate using generic frameworks) \cite{boukhedouma2015adaptation}, there is less enthusiasm for it \cite{dumas2015models, van2016don}; information systems using classical database management concepts and tailored to the use cases involved, are often preferred. In some cases, such information systems are inspired by some BPM concepts or they embed workflow engines to make limited use of them: this is the case of Enterprise Resource Planning (ERP) systems such as SAP S/4HANA\footnote{Official website of SAP: \url{https://www.sap.com/}, visited the 19/03/2020.} and Oracle Fusion Applications (OFA)\footnote{Official website of Oracle ERP: \url{https://www.oracle.com/applications/erp/}, visited the 19/03/2020.} \cite{van2016don, van2013business}. Therefore, the application of "pure" BPM is still limited to specific industries such as banking and insurance. This mitigated use of BPM for the automation of quite common administrative processes can be explained by the following factors:
\begin{itemize}
	\item \textit{Building process management systems is considerably more "tricky" than building information systems using classical database management}. In database-based information systems, a specific number of processes whose behaviours are known in advance are designed and automated, whereas in process management, systems are designed and implemented to offer a generic management of an arbitrary set of processes with behavioural similarities. To quote Wil M. P. Van Der Aalst \citeyearpar{van2013business}, "\textit{BPM is multifaceted, complex, and difficult to demarcate}"; it is therefore not accessible to everyone and requires from its practitionners : great capacities of reasoning, logical abstraction, generalisation and architectural design \cite{workflow95}.
	\item \textit{Existing BPM solutions are too abstract and generic}. As a result, compared to traditional information systems, they are less suitable for certain applications such as the exclusive management of administrative processes, that have a variety of specifications depending on organisations \cite{borger2012approaches, zur2013much, van2013business}.
	\item \textit{There is no real consensus on BPM implementation techniques and tools}. BPM is composed by a multitude of paradigms and tools. Even if an effort of standardisation has been made in recent years, WfMS vendors seem to prefer the development of multiple proprietary solutions \cite{van2013business}; in our opinion, this has the effect of reducing BPM credibility.
\end{itemize} 

Because they offer enormous benefits in terms of time saving, implementation cost and system complexity management, we believe that there is a need to better adapt workflow solutions for administrative processes' management, which seem to be the most frequently encountered \cite{mcCready, van1998application, dumas2005process}. Naturally, we are not the first to take an interest in this issue; the authors of \cite{dumas2015models} and \cite{van2016don} have established that the potential obstacle to the popularisation of workflow solutions, is the imbalance of research work on the different BPM aspects. Indeed, they believe that research has focused too much on BPM specific artifacts (such as process models) rather than on improving organisations' business processes: this does not meet the real use and needs of BPM practitioners. They therefore propose to address now the issue of business process improvement, in order to give even more reasons to organisations regarding the choice of BPM for their processes management. This perspective is shared by a plurality of researchers; hence their growing interest in the new field of \textit{process mining}\footnote{Process mining is one of the leading techniques conducting workflows' event data automatic analysis, for possible improvement of its corresponding workflow model \cite{van2013business}.} \cite{van2011process}. 

Although the idea of improving business processes by analysing the data produced during their execution is interesting, we believe that it does not answer the question that we are struggling with: \textit{how can we get information system users (organisations and software vendors) to systematically opt for BPM technology to automate their processes} ? Indeed, data analysis is posterior to the users' choice of a technology to produce these data. Moreover, the question of improving business processes seems really interesting for use cases where these processes are very complex; it is not always the case for administrative processes.

\mySectionStar{Our Global Vision}{}{true}
\label{chap0:sec:our-global-vision}
Because the benefits of BPM are now widely recognised and unanimously accepted, we believe that it would be more interesting to tackle the problem of its popularisation for administrative business processes automation, by improving its technology as it stands at present. Therefore, it would be a matter of:
\begin{enumerate}
	\item Making more accessible, the automation of administrative business processes using BPM. We believe that if the implementation of workflow systems for administrative business processes is simplified, then more and more software vendors will choose it to produce organisations' information systems.
	\item Adapt BPM technology so that it is less abstract and evasive, so that it is closer to that of classical information systems, and so that it can respond to specific problems of organisations.
\end{enumerate}

We believe that one way to achieve these goals, is to apply the \textit{domain-specific software engineering} \cite{bryant2010domain} to BPM and thus, to focus on a kind of \textit{Domain-Specific BPM}: that is the application of BPM techniques in specific domains of activity, thus respecting the constraints imposed by them, and offering a framework which best fits the needs of practitioners in these domains. In this thesis therefore, we must use BPM exclusively for the automation of administrative business processes: it is not a matter of reinventing the wheel (even if some new concepts are added) but rather of reproducing and adapting in the more constrained context of administrative business processes management, what is already being done in BPM, while introducing new concepts and making judicious choices to achieve the desired goals.

As far as we know, there are no other scientific studies that have focused exclusively on the automation of administrative processes. On the other hand, there is a growing effort among workflow solution providers to popularise administrative BPM. Precisely, they increasingly offer to organisations, cloud-based and flexible process management solutions such as Metatask\footnote{Official website of Metatask: \url{https://metatask.io/}, visited the 29/03/2020.}, Samepage\footnote{Official website of Samepage: \url{https://www.samepage.io/}, visited the 29/03/2020.}, Digital Business Transformation Suite\footnote{Official website of Digital Business Transformation Suite: \url{https://www.interfacing.com/}, visited the 29/03/2020.}, Favro\footnote{Official website of Favro: \url{https://learn.favro.com/}, visited the 29/03/2020.}, etc. Let us mention that there are studies focused on adapting BPM technologies to scientific data management exclusively (Domain-specific BPM): these gave birth to the field of scientific workflows \cite{bell2009beyond, juveGideon, ludascher2006scientific}.

\mySectionStar{The Challenge Addressed in this Thesis}{}{true}
\label{chap0:sec:thesis-challenge}
In this thesis, we are interested in the automation of administrative business processes (exclusively) using BPM technology. The idea is to use the most up-to-date paradigms already developed in the implementation of workflow systems, to produce a new approach for the specification and execution of administrative business processes. Nowadays, the most common paradigms and concepts used in the implementation of workflow management approaches include:
\begin{itemize}
	\item \textit{Artifact-centric paradigm}: it is the most successful BPM paradigm of the last two decades; it simplifies the specification of a business process, to the instantiation of a data structure called artifact, whose state evolves into a business goal state when executing process instances.
	\item \textit{Cooperative edition of documents}: in the artifact-centric paradigm, an artifact can be seen as a structured document containing the execution state of a process instance at a given time. Because the evolution of the state contained in it is the consequence of actions carried out by the different actors involved in the considered instance's execution, this execution can be assimilated to the cooperative edition of structured documents.
	\item \textit{Peer to Peer computing}: workflow systems increasingly rely on P2P architectures; as opposed to centralised systems, they facilitate scalability and fault tolerance \cite{fakas04, junYan06, theseMounir}.
	\item \textit{Multiagent paradigm}: the multiagent paradigm was developed to facilitate the creation of distributed systems, especially those based on P2P architectures. It is also increasingly used in workflow systems, where process tasks are now executed by agents that communicate through messages; these agents have a high degree of autonomy to allow better decentralisation of process management.
	\item \textit{Service-Oriented Architecture}: in decentralised workflow management systems based on P2P architectures, the concept of SOA is generally used to define communication protocols between agents and to increase their autonomy by implementing a loose coupling between them.
\end{itemize}

We combine these concepts to produce a workflow solution that best suits the automation of administrative business processes. Because this project is much too voluminous to be addressed in the context of a single PhD thesis, we will focus herein, only on the fundamental aspects: administrative process modelling and their distributed execution. One could summarise the main objective of this thesis by saying that it focuses on:
\begin{displayquote}
\textit{The proposal of a new artifact-centric framework, facilitating the modelling of administrative business processes and the completely decentralised execution of the resulting workflows; this completely decentralised execution being provided by a P2P system conceived as a set of agents communicating asynchronously by service invocation so that, the execution of a given workflow instance is technically assimilated to the cooperative edition of (mobile) structured documents called artifacts.}
\end{displayquote}
This justifies the title of this thesis. We should however admit that, a title like "\textit{\textbf{yet another approach to facilitate administrative workflows design and distributed execution using structured and cooperatively edited mobile artifacts}}" would certainly have better reflected the work done\footnote{Although this is the title that best suits our work, academic constraints have "forced" us to keep the one currently in use (A Grammatical Approach to Peer-to-Peer Cooperative Editing on a Service-Oriented Architecture). Indeed, it is with this last one that we applied for a thesis in our doctoral school; at the time when we were certain of the judicious title for our work, the legal texts governing the PhD thesis in our institution, no longer allowed us to make changes to the title of our thesis.}.

\mySectionStar{A Synoptic View of our Methodology and Engineering}{}{true}
\label{chap0:sec:engineering-overview}
\begin{figure}[ht!]
	\noindent
	\makebox[\textwidth]{\includegraphics[scale=0.16]{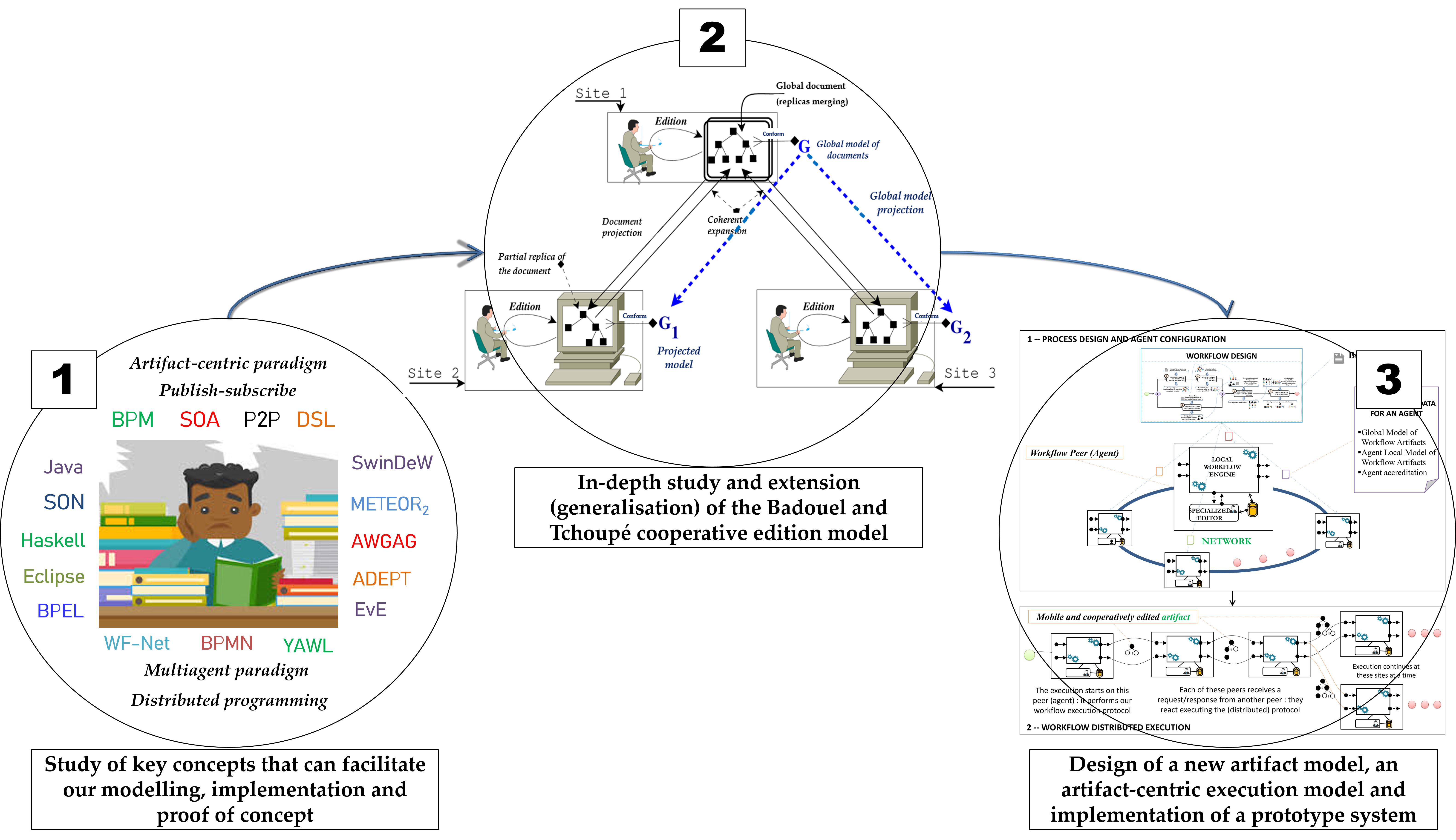}}
	\caption{A synoptic view of our methodology.}
	\label{chap0:fig:overview-methodology}
\end{figure}

Figure \ref{chap0:fig:overview-methodology} summarises the methodology and chronology we used to complete this thesis work. We started by studying a plethora of concepts including, the key concepts presented in chapter \ref{chap1:artifact-centric-bpm}. Then, we reinforced our knowledge on some other concepts, especially those related to the implementation of distributed systems (see fig. \ref{chap0:fig:overview-methodology}(1)). Having a better knowledge on BPM, especially on the artifact-centric paradigm, we undertook the design of a new artifacts' model; the search for such a model led us to an extensive study of structured cooperative editing workflows. In particular, we studied the structured documents asynchronous cooperative editing model proposed by Badouel and Tchoup\'e \citeyearpar{badouelTchoupeCmcs} so well that, we extended it with three new contributions (see chapter \ref{chap2:structured-editing-artifact-type}). Based on the document model manipulated by Badouel and Tchoup\'e, we derived an artifact model (an attributed grammar) that can be used to specify administrative business processes. We then designed a P2P system based on communicating agents using a service-oriented model, capable of executing in a completely decentralised manner, administrative processes specified using the proposed artifact model. To demonstrate the concreteness of our artifact-centric model, we finally implemented a fully functional prototype system allowing to experiment the proposed models on various administrative processes (see chapter \ref{chap3:choreography-workflow-design-execution}). As a result, we count at least six major contributions in this thesis:
\begin{enumerate}
	\item An extension of the merge algorithm proposed by Badouel and Tchoup\'e for the cooperative edition of a structured document, so that it can be applied in the more general case where edition conflicts might appear. It should be noted that this contribution was initiated in our Master's work before being fully matured during our first year of thesis. Some of the elements presented here are therefore part of our Master's dissertation;
	
	\item A generic system architecture that can be used to produce workflow systems for the cooperative editing of structured documents based on the Badouel and Tchoup\'e extended model;
	
	\item A workflow system prototype referred to as TinyCE v2 (a Tiny Cooperative Editor version 2), coded in Java and Haskell following the proposed system architecture and a novel cross-fertilisation protocol. It allowed us to test all the proposed algorithms related to the merge of documents' replicas;
	
	\item Another tree-based model of "business artifact" for administrative processes modelling, which makes it possible to better assimilate them to structured documents edited cooperatively;
	
	\item A choreography-oriented artifact-centric workflow execution model in which geographically dispersed agents execute the same and unique update (editing of artifact upon receipt) and diffusion (dissemination of updates) protocol;
	
	\item A prototype of a distributed system referred to as P2PTinyWfMS (a Peer-to-Peer Tiny Workflow Management System), allowing to fully experiment the artifact-centric approach investigated in this thesis.
\end{enumerate}

More concretely, in the asynchronous cooperative editing workflow as perceived by Badouel and Tchoup\'e, each of the co-authors receives in the different phases of the editing process, a copy of the edited document to insert his contribution. 
Since the collectively edited document is structured, it may in some cases, be preferable for reasons of confidentiality for example, that a co-author has access only to certain information; meaning that, he only has access to certain parts of the document, belonging to certain given types (\textit{sorts}\footnote{A \textit{sort} is a datum used to define the structuring rules (syntax) in a document model. Example: a \textit{non-terminal symbol} in a context free grammar, an \textit{ELEMENT} in a \textit{Document Type Definition} (DTD).}) of the document model. Thus, the replica $t_{\mathcal{V}_i}$ edited by co-author $c_i$ from the site $i$ may be only a \textit{partial replica} of the (global\footnote{We designate by \textit{global document} or simply \textit{document} when there is no ambiguity, the document including all parts.}) document $t$, obtained via a \textit{projection operation}, which conveniently eliminates from the global document $t$, parts which are not accessible to the co-author in question. Badouel and Tchoup\'e call \textit{view} of a co-author, the set of \textit{sorts} that he can access. 

When asynchronous local editions are done on partial replicas, it can be assumed that each co-author has on his site a local document model guiding him in his edition. This local model can help to ensure that for any update $t_{\mathcal{V}_i}^{maj}$ of a partial replica $t_{\mathcal{V}_i}$ (conform to the considered model), there is at least one document $t_c$ conform to the global model so that $t_{\mathcal{V}_i}^{maj}$ is a partial replica of $t_c$: for this purpose, the local model should be coherent towards the global one\footnote{Intuitively, a local model of document is \textit{coherent} towards a global model if any partial document $t_{\mathcal{V}_i}$ that is conform to it, is the partial replica of at least one (global) document $t_c$ conform to the global model.}. Thus, because of the edition's asynchronism, the only inconsistencies that we can have when the synchronisation time arrives are those from the concurrent edition of the same node\footnote{Manipulated documents are structured, they can be intentionally represented by a tree. Intuitively, a node is an identifiable part in the document (a section, a subsection, an image, a table, \ldots): it is the instance of a sort.}  (in the point of view of the global document) by several co-authors: the partial replicas concerned are said to be in conflict. 
In its first contribution, this thesis proposes an approach of detection and resolution of such conflicts by \textit{consensus},
using a tree automaton said of consensus, to represent all the documents that are the consensus of competing editions realised on different partial replicas.

A structured document $t$ is intentionally represented by a tree that possibly contains buds\footnote{A \textit{bud} is a leaf node of a tree indicating that an edition must be done at that level in the tree. Edit a bud consists to replace it by a subtree using the productions of the grammar of the document.}. Intuitively, synchronise or merge consensually the updates $t_{\mathcal{V}_1}^{maj}, \ldots, t_{\mathcal{V}_n}^{maj}$ of $n$ partial replicas of a document $t$, consists in finding a document $t_{c}$ conform to the global model, integrating all the nodes of $t_{\mathcal{V}_i}^{maj}$ not in conflict and in which, all the conflicting nodes are replaced by buds. Consensus documents are therefore the largest prefixes without conflicts in merged documents. 
Technically, the process for obtaining the documents forming part of the consensus is: 
(1) For each update $t_{\mathcal{V}_i}^{maj}$ of a partial replica $t_{\mathcal{V}_i}$, we associate a tree automaton with \textit{exit states} $\mathcal{A}^{(i)}$ recognising the trees (conform to the global model) for which $t_{\mathcal{V}_i}^{maj}$ is a  projection. 
(2) The consensual automaton $\mathcal{A}_{(sc)}$ generating the consensus documents is obtained by performing a \textit{synchronous product} of the automata $\mathcal{A}^{(i)}$ with a commutative and associative operator noted $\otimes$  that we define. It is such that: $\mathcal{A}_{(sc)}=\otimes\mathcal{A}^{(i)}$. 
(3) It only remains to generate the set of trees (or those most representative) accepted by the automaton $\mathcal{A}_{(sc)}$, to obtain the consensus documents.

The concept of structured document as perceived by Badouel and Tchoup\'e can be adapted to correspond to that of artifact in the sense of artifact-centric workflow systems. In this sense, a structured document can be seen as an artifact (annotated tree) that can be exchanged between the different agents (actors) involved in the execution of a given particular business process case; during its life, it is edited appropriately to make the system converge towards the achievement of one of the considered process's business goals. Its buds materialise the tasks to be executed or which are being executed and, an attributed grammar called the \textit{Grammatical Model of Workflow} (GMWf) is used as \textit{artifact type}. The sorts of a given GMWf represent the process tasks and each of its productions represents a scheduling of a subset of these tasks.
When a task is executed on a given site, the corresponding bud in the artifact is closed accordingly; then, one of the GMWf's production having the considered task as left hand side is choosen by the local stakeholder to expand the bud into a subtree highlighting in the form of new buds, the new tasks to be executed. To enrich the notion of access to different parts of artifacts, we add to GMWf, organisational information called \textit{accreditations} (similar to views) that offer a simple mechanism for modelling the generally different perceptions that actors have on processes and their data. The couple (GMWf, accreditations) is then the proposed model of "business artifact".

The execution of an administrative process's instance modelled using the couple (GMWf, accreditations), is a choreography in which the agents are reactive autonomous software components, communicating in peer-to-peer mode and driven by human agents (actors) in charge of executing tasks. An agent's reaction to the reception of a message (an artifact) consists in the execution of a five-step protocol clearly described in this thesis. 
This protocol allows it to: (1) \textit{merge} the received artifact with the one it hosts locally in order to consider all updates, (2) \textit{project} the artifact resulting from the merger in order to hide the parts to which the local actor may not have access and highlight the tasks to be locally executed, (3) make the local actor \textit{execute} the revealed tasks and thus edit the potentially partial replica of the artifact obtained after the projection, (4) integrate the new updates into the artifact through an operation called \textit{expansion-pruning} and finally, (5) \textit{diffuse} the updated artifact to other sites for further execution of the process if necessary. 
The agents' operational capabilities allow that, for the execution of a given process, an artifact created by one of them (initially reduced to an open node), moves from site to site to indicate tasks that are ready to be executed at the appropriate time and to provide necessary data (created by other agents) for that execution; the mobile artifact, cooperatively edited by agents, thus "grows" as it transits through the distributed system so formed.

\mySectionStar{The Organisation of this Manuscript}{}{true}
\label{chap0:sec:manuscript-organisation}
The rest of this thesis manuscript consists of four chapters and two appendices organised as follows:

~

\noindent\textbf{Chapter \ref{chap1:artifact-centric-bpm} - A State of the Art in Business Process Management: the Artifact-Centric Modelling}: we present some basic concepts related to the field of BPM, P2P, SOA, as well as the multiagent and artifact-centric paradigms. We also do a survey of some P2P and artifact-centric workflow management systems.

~

\noindent\textbf{Chapter \ref{chap2:structured-editing-artifact-type} - A Workflow for Structured Documents' Cooperative Editing : Key Principles and Algorithms}: after a brief presentation of Badouel and Tchoup\'e's asynchronous cooperative edition model, we present our algorithm for reconciling partial replicas of a structured document as well as a generic architecture of administrative workflow management systems.

~

\noindent\textbf{Chapter \ref{chap3:choreography-workflow-design-execution} - A Choreography-like Workflow Design and Distributed Execution Framework Based on Structured Mobile Artifacts Cooperative Editing}: we propose a new choreography-like approach to address administrative workflow design and their distributed execution. We also present a fully functional prototype system built according to the proposed approach.

~

\noindent\textbf{General Conclusion}: we summarise our work and present some possible venues for its further development.

~

\noindent\textbf{Appendix \ref{appendice1:algorithms-implementations} - Implementations of Some Important Algorithms Presented in this Thesis}: we present an implementation of all this thesis' key algorithms in the Haskell language.

~

\noindent\textbf{Appendix \ref{appendice2:article-appendice} - List of Scientific Communications Issued from the Work Presented in this Thesis}: we list the various scientific papers we produced during this thesis.

\myCleanStarChapterEnd

	%\myChapter{Titre}{Titre court}
\myChapter{A State of the Art in Business Process Management: the Artifact-Centric Modelling}{}
\label{chap1:artifact-centric-bpm}
\myMiniToc{section}{Contents}
% If no minitoc then
% \startcontents[chapters]
\mySection{Introduction}{}
\label{chap1:sec:introduction}
%\subsection*{Contexte et d�finitions}
%\label{sec:contexte}
The work we are doing in this thesis falls within the domain of Computer-Supported Cooperative Work (CSCW); it is a sub-domain of Software Engineering. Software engineering can be seen as a field of engineering that enables the design, the implementation and the maintenance of quality software systems \cite{barais2005construire}. 
Research in the field of software engineering has led to the implementation of new design and even programming paradigms (e.g. \textit{object-oriented programming} \cite{meyer2000conception}, \textit{component-based programming} \cite{heineman2001component}, etc.), new methods and processes (\textit{Object Modeling Technique - OMT -} \cite{rumbaugh1991object}, \textit{Unified Modeling Language - UML -} \cite{booch2000guide}, etc.), new technologies (\textit{\textbf{workflow}}, etc.), etc. The importance of software and the ever-increasing complexity of systems keep the field of software engineering among the most important in computer science.

Since the start of the 1980s, workflow technology knows an ever-growing success near companies and researchers in the field of computer-aided production. 
This success can be justified by the fact that, workflow enables firms to reduce their production costs as well as to quickly and easily develop new products and services: their competitiveness is therefore increased. 
Workflow technology offers methods and tools (notations, management systems, etc.) for the specification, optimisation, automation and monitoring of business processes \cite{workflow95, van2015business}.   
Workflow technology tools are logically orchestrated within complex systems called \textit{Workflow Management Systems} (WfMS) \cite{workflowModel, ima}. 
The purpose of WfMS is not only to automate at best workflows, but also to provide an appropriate framework for facilitating collaboration between actors involved in the execution of a given business process. 

The search for Workflow Management (WfM) / Business Process Management\footnote{BPM can be considered as an extension of classical Workflow Management (WfM) systems and approaches \cite{van2015business}.} (BPM) techniques has been densely conducted over the past two decades and a clear interest has been given to the \textit{artifact-centric} paradigm \cite{nigam2003business} proposed by International Business Machines Corporation (IBM). This one, revisited in several works \cite{abi2016towards, deutsch2014automatic, hull2009facilitating, lohmann2010artifact, assaf2017continuous, assaf2018generating, boaz2013bizartifact, lohmann2011artifact, estanol2012artifact}, proposes a new approach to BPM by focusing on both automated processes (tasks and their sequencing) and data manipulated through the concept of "\textit{business artifact}" (\textit{artifact-centric modelling}). 
%This contrasts with traditional BPM approaches (BPMN - \textit{Business Process Model and Notation}\footnote{BPMN was initiated by the \textit{Business Process Management Initiative} (BPMI) which merged with \textit{Object Management Group} (OMG) in 2005.} \cite{BPMN} -, YAWL - \textit{Yet Another Workflow Language}\footnote{YAWL allows processes to be represented using an extension of WF-Net (workflow net), a formalism derived from that of \textit{Petri Nets}.} \cite{van1998application, van2005yawl} -, BPEL - \textit{Business Process Execution Language}\footnote{BPEL allows to formalise the behaviour of business processes by choreographing web services.} \cite{jordan2007web} -) which are mainly concerned only with task scheduling and messages to be exchanged \cite{van2013business} thus, treating data as second-class citizens.

In this chapter which serves as a state of the art, we present some key notions related to BPM in general and to the artifact-centric paradigm in particular, in order to make it easier to better apprehend the concepts handled in this manuscript. 
In Section \ref{chap1:sec:bpm-def-key-principles}, we define some basic concepts such as \textit{business process}, \textit{workflow}, \textit{workflow management}, etc., then we present some workflow classification approaches as well as an overview of their standardised automation using BPM. 
In section \ref{chap1:sec:p2p-bpm}, we conduct a review of Peer to Peer (P2P) approaches to BPM. 
%we briefly present business process modelling languages such as BPMN, WF-Net (Workflow Net), YAWL, etc., and centralized WfMS such as ActionWorkflow \cite{actionWorkflow}, FlowMark \cite{flowmark}, Staffware \cite{staffware}, InConcert \cite{inConcert} etc. 
In section \ref{chap1:sec:data-aware-bpm}, we focus on the artifact-centric paradigm to BPM; we present the two general approaches to its implementation (\textit{orchestration} and \textit{choreography}) as well as some frameworks implementing it from the literature. 
Section \ref{chap1:sec:summary} is dedicated to a summary of the explored concepts and a smooth transition to the next chapter.

\mySection{Key Principles of Business Process Management}{}
\label{chap1:sec:bpm-def-key-principles}
%\subsection*{Contexte et définitions}
%\label{sec:contexte}

\mySubSection{Some Business Process Management Basic Concepts}{}
\label{chap1:sec:bpm-def-history}
Research in the CSCW field focuses on the role of computers in collaborative work \cite{schimdt1992taking}. These have given rise to numerous softwares called \textit{CSCW systems} or \textit{groupware}. CSCW systems communicate through networks and provide functionalities facilitating exchanges, coordination, collaboration and co-decision between the actors of a given collaborative work; they thus defy the space and time constraints to which collaborative work is subjected. Indeed, with the help of such systems, actors can either operate on the same site and thus manipulate the same objects (\textit{centralised approach}), or they can operate on geographically distant sites (\textit{distributed approach}); in this case, the objects they manipulate are replicated on the different sites and synchronised at the appropriate time \cite{johansen1988groupware, grudin1994computer, penichet2007classification}. In the same vein, they can act at the same time (\textit{synchronous approach}) or at completely different times and sometimes independently of the actions carried out by others (\textit{asynchronous approach}) \cite{johansen1988groupware, grudin1994computer, penichet2007classification}. 
CSCW systems are often referred to as \textit{workflow systems}. However, it should be noted that workflow is an extension and a generalisation of CSCW to business processes' automation.

\mySubSubSection{Some Definitions}{}
\label{chap1:sec:bpm-basic-concepts-def}
A \textit{business process} can be informally defined as a set of tasks ordered following a specific pattern and whose execution produces a service or a particular business goal \cite{workflow95}. When such a process is managed electronically, it is called \textit{workflow}. The purpose of workflow is to streamline, coordinate and control business processes in an organised, distributed and computerised environment. The peer-review validation of an article in a scientific journal is a common example of business process. Descriptions of it can be found in \cite{peerReview02, van2001proclets, badouel14}. As in most literature works, most of the time, we will use the terms "business process" and "workflow" as synonyms in the rest of this manuscript.

The \textit{Workflow Management Coalition}\footnote{The growing reputation of workflow led to the creation, in 1993, of the \textit{Workflow Management Coalition} (WfMC) as the organisation responsible for developing standards in this field. Official website of the WfMC: \url{https://www.wfmc.org/}.} (WfMC) \cite{workflowModel} defines \textit{Workflow Management} (WfM) as the modelling and computer management of all the tasks and different actors involved in executing a business process. WfM is achieved using \textit{Workflow Management Systems} (WfMS): these are complex systems with the aim of automating at best workflows by providing an appropriate framework to facilitate collaboration between actors involved in business processes' execution \cite{workflow95, van2013business, van2015business, dumas2018fundamental}. WfMS are composed of logically orchestrated tools to specificy, to optimise, to automate and to monitor business processes \cite{workflowModel, dumas2018fundamental}. Technically, the management of a process according to WfM is done in two phases \cite{divitini2001inter}:
\begin{enumerate}
	\item the \textit{process modelling phase}: the process is studied and then specified using a language (usually graphical) called \textit{workflow language}. The resulting specification is called \textit{workflow model};
	\item the \textit{process instantiation and execution phase}: the workflow model is introduced into a WfMS which then instantiates and orchestrates the execution of the underlying process.
\end{enumerate}
Since WfMS are pre-engineered standalone systems, WfM simplifies business processes' automation to their specifications in \textit{workflow languages}.

WfM primarily focuses on business processes' automation. It is not fundamentally concerned with other issues such as the analysis, the verification and the management (maintenance) of workflow (models) unlike BPM, which made these its foundation \cite{van2016don}. BPM is the discipline that combines knowledge from information technology and knowledge from management sciences and applies this to operational business processes \cite{van2013business}. BPM can be seen as an extension of WfM as it primarily supports WfM and provide additional tools to improve business processes. For this, we have chosen to use the expression BPM rather than WfM (which tends to disappear) in the context of this work. It should be noted however that our contributions (chapter \ref{chap2:structured-editing-artifact-type} and \ref{chap3:choreography-workflow-design-execution}) could be perfectly presented as part of restricted WfM (we are not interested in the management and improvement of workflow models).

\mySubSubSection{An Introductive Example of Business Process}{}
\label{chap1:sec:running-example}
BPM is an important technology because it simplifies the automation of business processes which are the foundation of how companies and organisations operate. Business processes can be found everywhere. The examples are diverse and include the following:
\begin{itemize}
\item The design and development of a software by a team (especially when members are geographically dispersed) \cite{theseImine};
\item The simultaneous writing of a scientific paper or the documentation of a product by several researchers (cooperative editing) \cite{theseImine};
\item The follow-up of a medical file \cite{workflow07};
\item The student registration process in a faculty;
\item The withdrawal of a large sum of money from a bank teller;
\item The procedure for taking holidays in a government institution;
\item The procedure for claiming damages from an insurance company;
\item The peer-review process \cite{peerReview02, van2001proclets, badouel14}.
\end{itemize}

The peer-review process presents all the characteristics of the type of processes (\textit{administrative processes}) studied in this manuscript. Then, we will use it as an illustrative example along the whole manuscript (running example). Our description of this process is inspired by those made in \cite{peerReview02, van2001proclets, badouel14}: 
\begin{example}
	\textbf{The peer-review process (running example)}:\\
	The process is triggered when the editor in chief receives a paper for validation submitted by one of the authors who participated in its drafting.
	\begin{itemize}
		\item After receipt, the editor in chief performs a pre-validation after which, he can accept or reject the submission for various reasons (subject of minor interest, submission not within the journal scope, non-compliant format, etc.);
		\item If the submission is rejected, he writes a report then notifies the corresponding author and the process ends; in the other case, he chooses an associated editor and sends him the paper for the continuation of the validation; 
		\item The associated editor prepares the manuscript, forms a referees committee (two members in our case) and then triggers the peer-review evaluation process;
		\item Each referee reads, seriously evaluates the paper and sends back a message and a report to the associated editor;
		\item After receiving reports from all the referees, the associated editor takes a decision and informs the editor in chief who sends the final decision to the corresponding author.
	\end{itemize}
\end{example}

From this description, it is easy to identify all the tasks to be executed, their sequencing, actors involved and the tasks assigned to them. For this case, four actors are involved: an editor in chief ($EC$) who is responsible for initiating the process, an associated editor ($AE$) and two referees ($R1$ and $R2$).
A summary of tasks assignment is presented in table \ref{tableau:tachesExecutant}. We have associated symbols with tasks so that we can easily manipulate them in diagrams. 
\begin{table}[ht]
	\caption{Exhaustive tasks list of a paper validation process in a scientific journal and their respective performers.}
	\label{tableau:tachesExecutant}
	\begin{tabular}[t]{|m{8.4cm}|m{2.7cm}|m{2.63cm}|}
		\hline
		\textbf{Tasks} & \textbf{Associated Symbols}  & \textbf{Executors} \\
		\hline
		Receipt, pre-validation of a submitted paper and possible choice of an associated editor to lead peer-review evaluation & $A$  & $EC$\\
		\hline
		Drafting of a pre-validation report informing on the reasons for the immediate rejection of the paper & $B$ & $EC$ \\
		\hline
		Sending the final decision (acceptance or rejection of the paper) to the author & $D$ & $EC$ \\
		\hline
		Study, eventually formatting of the paper for the examination by a committee & $C$ & $AE$ \\
		\hline
		Constitution of the reading committee (selection of referees) and triggering the peer-review evaluation & $E$ & $AE$ \\
		\hline
		Decision making (paper accepted or rejected) from referees evaluations & $F$ & $AE$ \\
		\hline
		Evaluation of the manuscript by the first (resp. second) referee & $G1$ (resp. $G2$) & $R1$ (resp. $R2$) \\
		\hline
		Drafting of the after evaluation report by the first (resp. second) referee & $H1$ (resp. $H2$) & $R1$ (resp. $R2$) \\
		\hline
		Writing the message according to evaluation by the first (resp. second) referee & $I1$ (resp. $I2$) & $R1$ (resp. $R2$) \\
		\hline
	\end{tabular}
\end{table}

\mySubSubSection{Workflow Typology}{}
\label{chap1:sec:workflow-typology}
The authors of \cite{workflow95} conduct a very interesting study on the classification of workflows in which, they report the lack of a commonly accepted approach to categorising workflows. There are therefore several approaches to workflow classification in the literature.

The classification of workflows according to the nature and behaviour of automated processes is one of the most commonly found in the literature. According to it, workflows are divided into three groups: \textit{production} workflows, \textit{administrative} workflows and \textit{ad-hoc} workflows \cite{mcCready, van1998application}. Production workflows are those that automate highly structured processes that undergo very little (or no) change over time: all the scenarios are known in advance and most of the tasks are carried out by systems. This is the case for processes in industrial production lines. 
Administrative workflows apply to variable processes for which all cases are known; that is, tasks are predictable and their sequencing rules are simple and clearly defined. In these, changes are more frequent than with production workflows and human actors are more involved in the execution of tasks. In particular, this type of workflow brings considerable added value to public administration organisations whose business is focused on administrative routines \cite{boukhedouma2015adaptation}. Our running example, the peer-review process, is an administrative process. In the work presented in this manuscript, we are interested in this type of workflows.
These are opposite of ad-hoc workflows, which automate occasional processes for which it is not always possible to define the set of rules in advance. Processes are therefore only partially specified and may undergo many updates over time.

The workflows' classification made in \cite{workflow95} is orthogonal to the above-mentioned one (they can be used together); it is more concerned with tasks' automation degree. The authors classify workflows based on a measurement system, represented by a continuum ranging from \textit{human-oriented} workflows to \textit{system-oriented} workflows as shown in figure \ref{chap1:fig:dimitrios-classification}.
\begin{figure}[ht!]
	\noindent
	\makebox[\textwidth]{\includegraphics[scale=0.94]{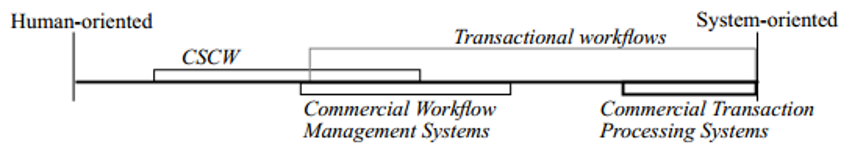}}
	\caption{Classification of workflows according to whether they are human-oriented or system-oriented (source \cite{workflow95}).}
	\label{chap1:fig:dimitrios-classification}
\end{figure}
The first type (human-oriented workflows) includes workflows in which humans collaborate to perform tasks and to coordinate themselves; in these, humans are responsible for ensuring the validity and consistency of the exchanged data and of the workflow's results. The second type of workflows (system-oriented workflows) refers to those in which the use of computer systems to perform tasks is unavoidable because, they involve complex data and computationally-intensive operations. According to this classification system, human-oriented workflows are the ones we are interested in.

In \cite{dumas2005process}, the authors refine the two above classification frameworks. Concerning the refinement of the one that classifies workflows according to the nature and behaviour of automated processes \cite{mcCready}, their classification framework distinguishes \textit{unframed}, \textit{ad hoc framed}, \textit{loosely framed}, and \textit{tightly framed} workflows. 
A workflow is said to be unframed if there is no explicit workflow model associated with it; its execution is strongly conducted by its actors. When actors play a crucial role (no longer limited to the simple execution of tasks, but also including the explicit choice of the control flow, the adjustment of control and data flows, etc.) in the execution of a workflow, it is said to be \textit{user-centric} \cite{badouel2015active}. This is the case for workflows being automated by groupware\footnote{Groupware systems are computer-based systems that support groups of people engaged in a common task (or goal) and that provide an interface to a shared environment \cite{ellis1991groupware}.}. 
In the case of ad hoc framed workflows, workflow models are defined a priori but, they frequently change. 
A workflow is said to be loosely framed when it is defined by a workflow model describing the "right way of doing things", while allowing its actual executions to deviate from this way; this is the preferred type of workflow handled by \textit{Case Management Systems} \cite{van2013business} (see sec. \ref{chap1:sec:gag}).
Finally, a tightly framed process is one which consistently follows a defined process model.

Concerning the classification framework of \cite{workflow95}, authors of \cite{dumas2005process} refine it and consider three types of workflows: \textit{Person-to-Person}, \textit{Person-to-Application}, and \textit{Application-to-Application} workflows. 
Person-to-Person workflows are those for whom all the tasks require human intervention. Application-to-Application workflows are their opposite; in these, all the tasks are executed by software systems. Person-to-Application workflows are in the middle; they involve both human-oriented tasks and system-oriented tasks. Pratically, most of workflows are of this category.

Nowadays, some scientific works require increasingly complex and data-intensive simulations and analysis. Scientific data management is therefore a major challenge \cite{bell2009beyond} with a high level of complexity. Workflow technologies are increasingly used to manage this complexity \cite{juveGideon}. These are responsible for scheduling computational tasks on distributed resources, managing dependencies between tasks and staging data sets in and out of runtime sites. The resulting workflows are called \textit{scientific workflows} and are usually based on a middleware infrastructure (\textit{Grid} or \textit{Cloud}). Ideally, the scientist should be able to integrate almost any scientific data resource into such a workflow during analysis, inspect and visualise the data on-the-fly as it is computed, make parameter changes as needed and re-run only the affected components, and capture sufficient metadata in the final products so that, scientific workflow executions help to explain the results and make them reproducible. Thus, a scientific workflow system becomes a scientific problem-solving environment, adapted to an increasingly distributed and service-oriented infrastructure (Grid or Cloud) \cite{ludascher2006scientific}.

There are many other types of workflows in the literature. We can mention on the fly, \textit{service-oriented} workflows \cite{piccinelli2003service, yongyi2009research}, \textit{structured} workflows \cite{kiepuszewski2000structured, eder2002meta, liu2005analysis}, etc. We do not present them here because they are not of great interest to the work we are doing for this thesis. We invite the interested reader to take a look at the few works mentioned above.

\mySubSection{Business Process Management Lifecycle and Key Activities}{}
\label{chap1:sec:bpm-key-activities-concerns}
A high-level view of the BPM discipline reveals that, its lifecycle consists of three phases on which it is possible to iterate indefinitely: the \textit{(re)design}, \textit{implement/configure}, and \textit{run \& adjust} phases \cite{van2013business} (see fig. \ref{chap1:fig:bpm-lifecycle}). 
During its lifecycle, four key activities namely \textit{model}, \textit{enact}, \textit{analyse}, and \textit{manage} (see fig. \ref{chap1:fig:bpm-key-concerns}) are carried out \cite{van2013business}. 
In this section, we examine what is done during these different activities; we mainly focus on the \textit{"model"} and the \textit{"enact"} activities: they are the only ones common to BPM and WfM and thus, they are of relevant interest for the work presented in this manuscript. 

\mySubSubSection{Business Process Management Lifecycle}{}
\label{chap1:sec:bpm-lifecycle}
\begin{figure}[ht!]
	\noindent
	\makebox[\textwidth]{\includegraphics[scale=0.5]{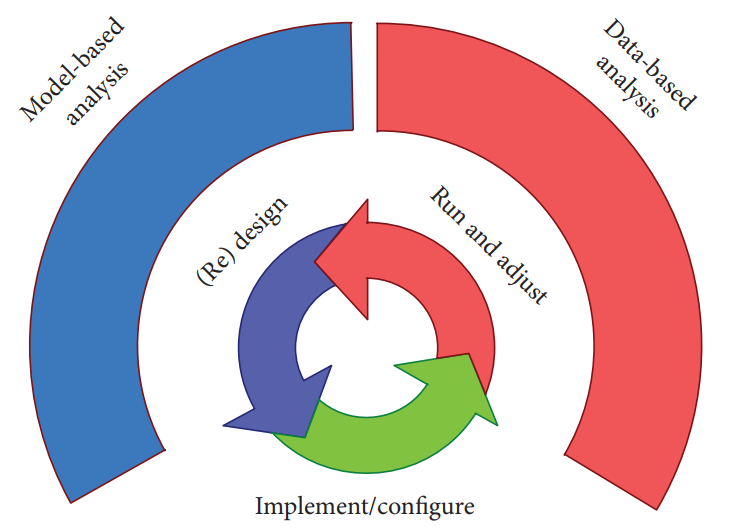}}
	\caption{The three phases of BPM's lifecycle (source \cite{van2013business}).}
	\label{chap1:fig:bpm-lifecycle}
\end{figure}

The automation of a given process using BPM starts with its modelling using one or more workflow languages \cite{dumas2018fundamental}. This \textit{"model" activity} is initiated during the \textit{"(re)design" phase} of the BPM lifecycle. The workflow models obtained during this activity can be analysed (the \textit{"analyse" activity}) either by simulations or by using model checking\footnote{Model checking is an automated technique that, given a finite-state model of a system and a formal property, systematically checks whether this property holds for (a given state in) that model \cite{baier2008principles}.} algorithms (to verify models' soundness): this type of analysis is said to be \textit{model-based}. 
As shown in figure \ref{chap1:fig:bpm-lifecycle}, the (re)design phase is followed by the \textit{"implement/configure" phase} in which, the workflow models obtained in the previous phase are converted, if necessary, into executable workflow models and then, used to configure the process execution environment (the WfMS): this is where the \textit{"model" activity} ends. 
After the \textit{"implement/configure" phase}, comes the \textit{run \& adjust phase}. During this last phase, the workflow is instantiated, executed and managed (adjusted) according to the scenarios foreseen when modelling the underlying process and when designing the host WfMS: these are the purposes of the \textit{"enact"} and \textit{"manage"} activities. Moreover, when a workflow instance is running, produced and logged data can be analysed (to discover possible bottlenecks, waste, and deviations) for possible improvement of its corresponding workflow model: this other type of analysis/monitoring is said to be \textit{data-based}; during the last decade, \textit{process mining} \cite{van2011process} has emerged as one of the leading techniques conducting \textit{data-based analysis}. If enough possible improvements to the workflow model are detected, the cycle can restart to apply them.
\begin{figure}[ht!]
	\noindent
	\makebox[\textwidth]{\includegraphics[scale=0.54]{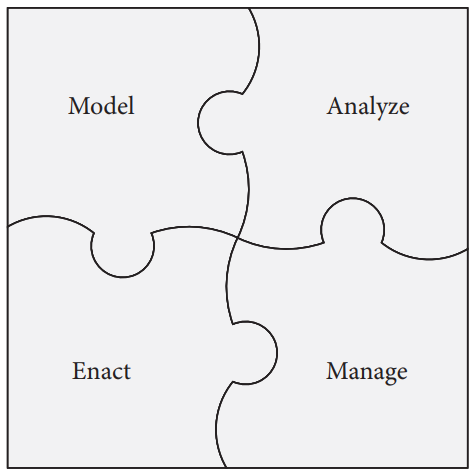}}
	\caption{The four key activities of BPM (source \cite{van2013business}).}
	\label{chap1:fig:bpm-key-concerns}
\end{figure}

\mySubSubSection{The "Model" Activity}{}
\label{chap1:sec:bpm-model-activity}
\noindent\textbf{\textit{Basic concepts}}

Process modelling is a crucial activity in WfM/BPM. As mentioned above (sec. \ref{chap1:sec:bpm-lifecycle}), it is done using dedicated languages called \textit{workflow languages}. Several workflow languages have already been developed. Among the most well-known are the BPMN standard \cite{BPMN} based on statecharts, the UML activity diagrams language \cite{booch2000guide}, the WF-Net (\textit{Workflow Net}) language \cite{wil2003business} which uses a formalism derived from that of Petri nets, the YAWL language \cite{van2005yawl} which is an extension of WF-Net and so forth. 
Some of these languages (BPM, UML activity diagrams) are \textit{informal} (i.e. they do not have a well-defined semantics and do not allow for analysis \cite{zur2013much, van2013business}) while others (WF-Net, YAWL) are based on powerful mathematical (\textit{formal}) tools (Petri nets). Nevertheless, they all allow to express in a diagram (called a \textit{worklow model}), the tasks that make up a given process and the control flow between them. More precisely, workflow languages allow to describe the behaviour of processes through the representation (among others) \cite{grigori2001elements} of :
\begin{itemize}
	\item Tasks that make up the main part of the process;
	\item Information and resources relating to the various tasks;
	\item Sequences or dependencies between those tasks;
	\item Trigger and termination events for the tasks.
\end{itemize}

Tasks are the base of any workflow; a \textit{task} is the smallest unit of hierarchical decomposition of a process. A task represents any work that is performed within a process. It consumes time, one or more resources, requires one or more input objects and produces one or more output objects. You can find examples of tasks in our running example (sec. \ref{chap1:sec:running-example}). 

From a workflow point of view, the term \textit{resource} refers to a system or a human who can execute a task. It is also known as \textit{actor}, \textit{participant}, \textit{stakeholder}, \textit{agent} or \textit{user} depending on the context. Resources can be grouped according to various characteristics, to form either a \textit{role} or an \textit{organisational unit} \cite{grigori2001elements}. A \textit{role} is a group of resources with the same functional capabilities, while an \textit{organisational unit} is a set of resources (or class of resources) that belong to the same structure (department, team, service, cell, etc.).

~

\noindent\textbf{\textit{Routing patterns}}

To achieve its objectives, any workflow language must, for a given process, allow to express at least its tasks and their \textit{routing} (\textit{control flow}). The task control flow is commonly referred to as the \textit{lifecycle (process) model} of the process under study \cite{divitini2001inter, hull2009facilitating}. There are a number of routing patterns identified in the literature as basic ones: these are \textit{sequential}, \textit{parallel}, \textit{alternative} or \textit{conditional} and \textit{iterative} routings (see fig. \ref{chap1:fig:basic-routing}) \cite{van1998application}.
\begin{figure}[ht!]
	\noindent
	\makebox[\textwidth]{\includegraphics[scale=4.5]{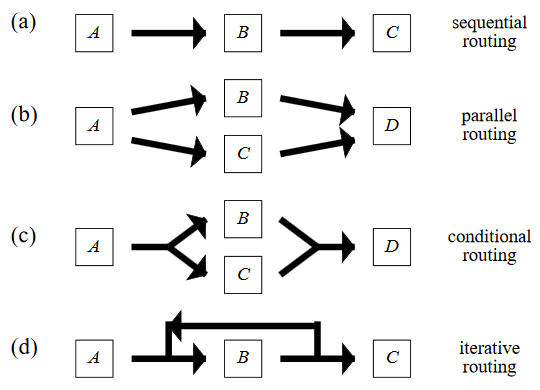}}
	\caption{Four basic routing constructs (source \cite{van1998application}).}
	\label{chap1:fig:basic-routing}
\end{figure}
\begin{itemize}
	\item Sequential routing expresses the fact that tasks must be executed one after the other (task $A$ before tasks $B$ and $C$ in figure \ref{chap1:fig:basic-routing}(a));
	\item Parallel routing is used to specify the potentially concurrent execution of certain tasks. Tasks $B$ and $C$ in figure \ref{chap1:fig:basic-routing}(b) can be executed at the same time; in this case, tasks $A$ and $D$ are considered as \textit{gateways}: $A$ is said to be an \textit{AND-Split} gateway while $D$ is an \textit{AND-Join} gateway;
	\item With alternative routing, one can model a \textit{decision}: i.e. the choice to execute one task rather than another at a given time. In figure \ref{chap1:fig:basic-routing}(c), tasks $B$ and $C$ cannot be both executed; for this case, $A$ is called and \textit{OR-Split} gateway and $D$ is an \textit{OR-Join} gateway;
	\item In some cases, it is necessary to execute a task multiple times.  In figure \ref{chap1:fig:basic-routing}(d) task $B$ is executed one or more times.
\end{itemize}

The search for more advanced and expressive routing patterns has been the subject of many studies \cite{van2012workflow, borger2012approaches}. The interested reader is invited to consult the few references mentioned above to find out more. 

When a given workflow language only allows to specify the routing of the processes' tasks, when it is not interested in modelling the consumed and produced data during tasks execution, and when it only gives a secondary role to the processes' users, it is said to be \textit{process-centric}. This is the case for all the previously mentioned languages (BPMN, WF-Net, UML activity diagrams and YAWL). This type of workflow language is often referred to as "\textit{traditional workflow language}".

~

\noindent\textbf{\textit{Examples of workflow models}}

Figure \ref{chap1:fig:comparing-workflow-languages} shows the orchestration diagrams corresponding to the graphical description of the peer-review process (see its textual description in sec. \ref{chap1:sec:running-example}) using the process-centric notations BPMN and WF-Net. The graphical notations equivalent to sequential flow, \{And, Or\}-Splits and \{And, Or\}-Joins are well represented. Each diagram resumes the \textit{main scenarios} of this process.
\begin{figure}[ht!]
	\noindent
	\makebox[\textwidth]{\includegraphics[scale=0.2]{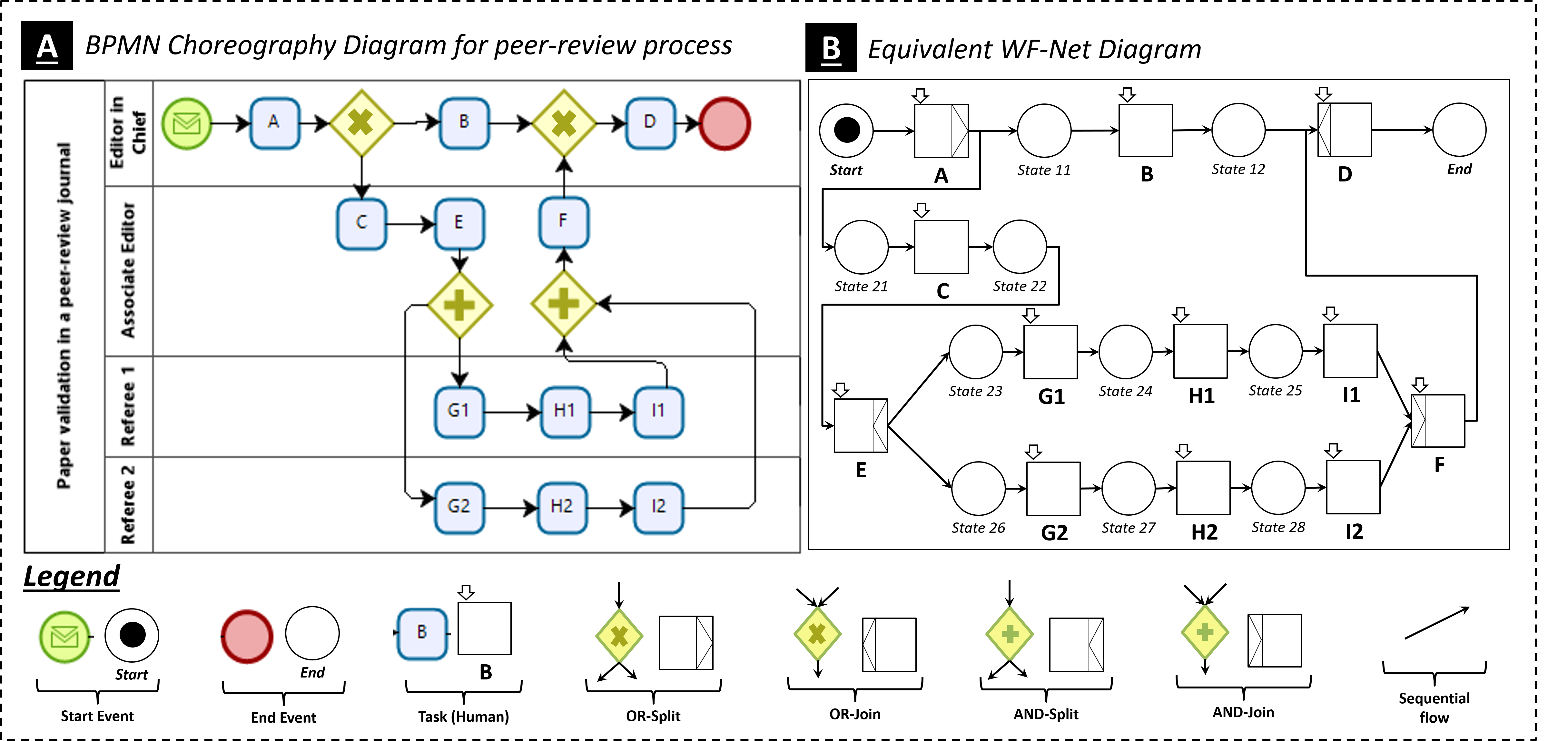}}
	\caption{Orchestration diagrams of the peer-review process.}
	\label{chap1:fig:comparing-workflow-languages}
\end{figure}

\mySubSubSection{The "Enact" Activity}{}
\label{chap1:sec:bpm-enact-activity}
\noindent\textbf{\textit{Overview}}

The "enact" activity takes as input, the workflow models (specifications) obtained during the model activity. If these models are executable (i.e. they have been coded in more technical languages taking into account implementation details) then they are directly introduced into a WfMS suitably installed at the different workflow execution sites; otherwise, they are first converted into executable models then, they are introduced into the WfMS. There are several languages for producing executable workflow models. These are usually proprietary and provided by WfMS designers. Of these languages, \textit{(Web Services) Business Process Execution Language} ((WS-)BPEL) is the standard\footnote{BPEL is standardised by the OASIS consortium. OASIS website: \url{https://www.oasis-open.org/}. BPEL Specification (PDF version): \url{https://docs.oasis-open.org/wsbpel/2.0/OS/wsbpel-v2.0-OS.pdf}.} and is well compatible with BPMN \cite{white2005using, ouyang2006bpmn, leymann2010bpel}.

Once the WfMS is properly configured using workflow models, it can create workflow instances and properly orchestrate their execution. To do this, WfMS must coordinate (according to workflow models) the execution of a set of tools and applications offering various services. In the 1990s, the WfMC developed and proposed an architectural \textit{reference model} for the implementation of WfMS \cite{workflowModel} (see fig. \ref{chap1:fig:wfms-reference-model}). The latter structures and describes precisely, the expected functionalities of a WfMS.

~

\noindent\textbf{\textit{The reference model}}

The WfMC reference model is a centralised architectural model in which the main component is called \textit{workflow enactment service}. The workflow enactment service is responsible for controlling the executions of workflow instances. It is composed of several \textit{workflow engines}. A given workflow engine handles some parts of workflows and also manages some of their resources \cite{van2013business, dumas2018fundamental}. 
\begin{figure}[ht!]
	\noindent
	\makebox[\textwidth]{\includegraphics[scale=0.6]{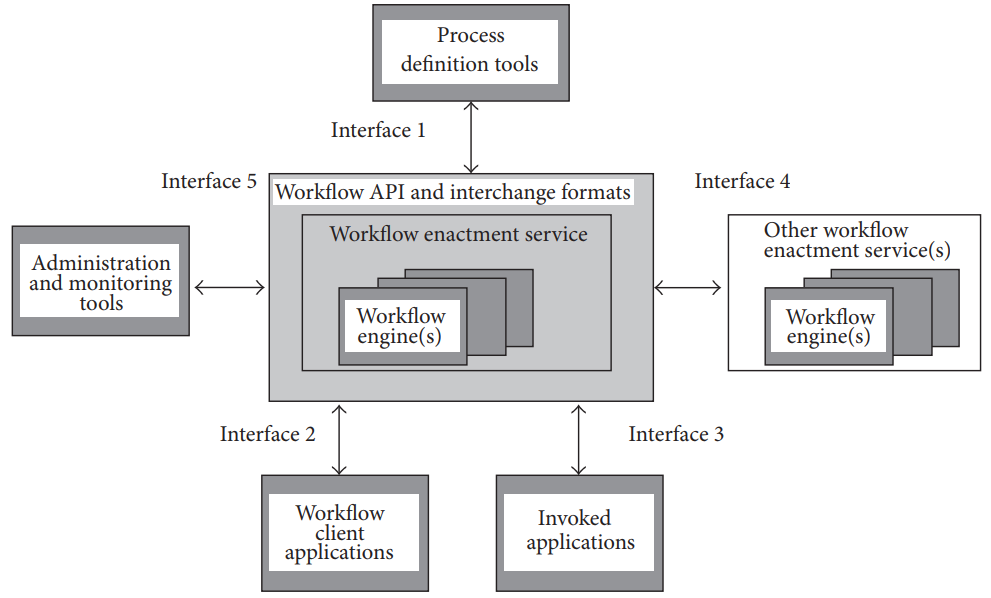}}
	\caption{Reference model of the Workflow Management Coalition (source \cite{workflowModel}).}
	\label{chap1:fig:wfms-reference-model}
\end{figure}
According to the reference model, WfMS must provide tools to facilitate their configurations using workflow models: therein, these tools are referred to as \textit{process definition tools}. Process definition tools are connected to the WfMS core (the workflow enactment service) via \textit{Interface 1}. In order to execute tasks, users use \textit{workflow client applications} that communicate with the WfMS via \textit{Interface 2}. When necessary, a given workflow engine invokes other applications via \textit{Interface 3}. The \textit{administration and monitoring tools} connected via \textit{Interface 5}, are used to monitor and control the workflows. Finally, the WfMS can be connected to other WfMS using \textit{Interface 4}. A considerable effort has been made to standardise the five interfaces shown in figure \ref{chap1:fig:wfms-reference-model}. These efforts led to the production of languages (exchange formats) such as \textit{Workflow Process Definition Language} (WPDL), \textit{XML\footnote{XML: eXtensible Markup Language.} Process Definition Language} (XPDL) and BPEL.

The reference model has been very successful. Firstly, because to this day, it perfectly orchestrates the different tools used for the design and execution of workflows. Secondly, because it has served as the basic model for a very large number of WfMS in the industry. Examples include ActionWorkflow \cite{actionWorkflow}, FlowMark \cite{flowmark}, Staffware \cite{staffware}, InConcert \cite{inConcert}, etc. Because the reference model is a centralised approach (client-server architecture), it has the advantage of facilitating a good mastery of the technologies used in the production of WfMS. Also, the implementation of (generally lightweight) client applications and the overall maintenance of WfMS (which is limited to the maintenance of the central server) are much simpler \cite{theseKanzow}. However, systems based on a client-server architecture show some limitations because of the centralisation of workflow management. Their main weaknesses are: \textit{the (non) fault tolerance}, \textit{the (difficult) scalability} and the strong dependency of the system vis a vis the central server, which stores data, controls and thus, represents \textit{a possible point of congestion} \cite{junYan06, fakas04}. Concretely \cite{junYan06},
\begin{enumerate}
	\item The client-server architecture allows centralised coordination of workflows with little use of the computing potential on the client side. Workflow systems based on such an architecture are very cumbersome. In application areas where several workflow instances need to be executed in parallel, the centralised server can be overloaded with heavy computations and intensive communications when the system load increases, thus becoming a potential bottleneck. 
	\item Client-server systems are vulnerable to server failures. The centralised server is commonly viewed as a single point of congestion in the system. Its malfunction can cause the entire system to shut down. 
	\item The limited scalability of the client-server architecture prevents the WfMS based on it, from dealing with the ever-changing work environment. This also raises difficulties in system configuration, as any changes to the system, such as the admission of new actors, may require changes and updates to the centralised workflow server, which is very impractical and inefficient. Therefore, these WfMS are particularly unsuitable for application areas where workflow actors are required to join and leave frequently.
	\item An important and crucial element of any workflow system is to allow actors to maintain their autonomy and control. However, workflow actors in a client-server-based WfMS are exclusively controlled by centralised servers. A serious problem is that, a large number of actors working on the "lightweight client side" may not be able to exercise their control, decision-making and problem-solving capabilities.
\end{enumerate}

Knowing that various actors involved in a given business process are very often spread over remote sites, the reference model does not seem to be very suitable for efficiently implementing cooperation among them, as would systems based on a distributed architectural model be.
In order to meet the shortcomings of the reference model, several works \cite{theseKanzow, junYan06, fakas04, theseImine, SON} have focused on the production of distributed WfMS built on top of peer-to-peer (P2P) architectures. This approach has also been successful since, systems such as ADEPT \cite{adept} and METEOR$_{2}$ \cite{meteor} have been designed over years \cite{theseKanzow}.

\mySection{Peer to Peer Business Process Management}{}
\label{chap1:sec:p2p-bpm}
Knowing that workflows are naturally distributed, they can sometimes involve resources from different organisations. Within each organisation, WfMS must therefore be built with a strong emphasis on (sometimes inter-organisational) cooperation; this differs from the idea in which classical information systems have often been built. However, even if WfMS must be \textit{interoperable} to facilitate cooperation, they must also ensure the \textit{autonomy} and the \textit{confidentiality} of actors and organisations involved in workflow execution; because, though organisations are aware of the need and necessity to participate in cooperation, they wish to protect their expertise in order to ensure sufficient confidentiality on their local data and local processes \cite{boukhedouma2015adaptation}. 
The main challenge for WfMS designers over the last two decades, has therefore been to build WfMS capable of both ensuring the agility of organisations and fostering the interconnection of business processes, while preserving their autonomy and the confidentiality of their local processes and data.

The production of fully distributed WfMS proved to be an effective solution to this challenge \cite{meilin1998workflow}. This has been made more feasible with the advent of new concepts such as the \textit{Multiagent} paradigm and the \textit{Service-Oriented Architecture} (SOA). In this section, we take a look at how the distributed workflow management approach works, and some of the decentralised WfMS that have been developed for this purpose

\mySubSection{The Advent of the Multiagent and Service-Oriented Concepts}{}
\label{chap1:sec:agent-soa-soc-concepts}

\mySubSubSection{The Multiagent Concept}{}
\label{chap1:sec:agent-concept}
The \textit{agent} and \textit{multiagents systems} concepts emerged in the 1980s. These concepts have generated lots of excitement in different research communities mainly because, they form the basis of a new paradigm for designing and implementing software systems that operates in distributed and open environments\footnote{\textit{"An open system is one in which the structure of the system itself is capable of dynamically changing. The characteristics of such a system are that its components are not known in advance; can change over time; and can consist of highly heterogeneous agents implemented by different people, at different times, with different software tools and techniques. Perhaps the best-known example of a highly open software environment is the internet. The internet can be viewed as a large, distributed information resource, with nodes on the network designed and implemented by different organisations and individuals."} (Katia P. Sycara, \citeyearpar{sycara1998multiagent})}, such as the internet \cite{sycara1998multiagent}. One of the best-known and most famous definitions of the agent concept was formulated by Jacques Ferber \citeyearpar{ferber1997systemes} and states that : \textit{an agent is a physical or logical entity capable of acting upon itself and its environment, which has a partial representation of that environment, which, in a multiagent system, can communicate with other agents, and whose behaviour is the consequence of its observations, knowledge and interactions with other agents}.

Multiagent systems have brought a new way to look at distributed systems and have provided a path to more robust intelligent applications \cite{deloach2001multiagent}. The challenge of the multiagent concept is to build distributed systems in which the nodes (agents), endowed with great autonomy, high reactivity and communicating using an asynchronous messaging system (they therefore possess cooperation and deliberation/decision capabilities \cite{theseKanzow}), can appear and disappear at any time without paralysing the system. A multiagent system is characterised as follows \cite{sycara1998multiagent}:
\begin{enumerate}
	\item Each agent has incomplete information or capabilities for solving problems and thus, has a limited viewpoint;
	\item There is no system global control;
	\item Data are decentralised;
	\item Computation is asynchronous.
\end{enumerate}
Such properties for a multiagent system provide it with several capabilities that have mainly attracted researchers and professionals to the multiagent paradigm. Among these capabilities we can distinguish the following \cite{sycara1998multiagent}: 
\begin{itemize}
	\item The capability to solve problems that are too large and difficult to handle by a centralised agent/server, because of resource limitations, or the sheer risk of having one centralised entity that could be a performance bottleneck or could fail at critical times;
	\item The capability to allow for the interconnection and interoperation of multiple existing legacy systems;
	\item The capability to provide solutions to problems that can naturally be regarded as a society of autonomous interacting components-agents;
	\item The capability to provide solutions that efficiently use information sources that are spatially distributed;
	\item The capability to provide solutions in situations	where expertise is distributed.
\end{itemize}

The above mentioned capabilities of multiagent systems are in line with the desired capabilities of a distributed WfMS. This has motivated the growing use of concepts developed for multiagent systems, in the production of such WfMS. However, the vocabulary used by the designers of distributed WfMS is not always identical to the multiagent jargon, and it is sometimes necessary to abstract the proposed systems to exhibit the multiagent concepts involved in their design and implementation.

\mySubSubSection{The Service-Oriented Architecture}{}
\label{chap1:sec:soa-concept}
\noindent\textbf{\textit{Basic concepts}} 

\textit{Service-Oriented Architecture} (SOA) has spread rapidly as a result of its growing success, and has been widely accepted as a supporting architecture for information systems because of its pivotal concept of \textit{service} \cite{boukhedouma2015adaptation}. Service is the essential concept of SOA and can be defined as \textit{"self-describing, platform-agnostic computational elements that support rapid, low-cost composition of distributed applications. Services perform functions, which can be anything from simple requests to complicated business  processes. Services allow organisations to expose their core competencies programmatically over a network using standard languages and protocols, and be implemented via a self-describing interface based on open standards"} (Mike P. Papazoglou, \citeyearpar{papazoglou2003service}). For MacKenzie et al. \citeyearpar{mackenzie2006reference}, the term service combines the following related ideas :
\begin{itemize}
	\item The offer to perform work for another;
	\item The capability to perform work for another;
	\item The specification of the offered work.
\end{itemize}

From a purely technological point of view, a service is a software component represented by two separate elements: its interface, which allows defining the access modalities to the service (name of the service and the parameters of the public operations defining the signatures of the operations) and its implementation \cite{boukhedouma2015adaptation}. Services are offered by \textit{service providers} (see fig. \ref{chap1:fig:basic-soa}); these are organisations that procure the service implementations, supply their service interfaces and provide related technical and business support \cite{papazoglou2003service}. Service interfaces are available for their searching, their binding, and their invocation by \textit{service consumers} (see fig. \ref{chap1:fig:basic-soa}) \cite{mackenzie2006reference}. Clients of services (service consumers) can be other solutions or applications within an enterprise or clients outside the enterprise. Service providers must therefore provide a distributed computing infrastructure for both intra and cross-enterprise application integration and collaboration. To satisfy these requirements, provided services should be \cite{papazoglou2003service}:
\begin{itemize}
	\item \textit{Technology neutral}: they must be able to be invoked by clients coded with various technologies having a few standards as a common denominator;
	\item \textit{Loosely coupled}: they must not require neither knowledge nor internal structures or conventions (context) at the service consumer or service provider side;
	\item \textit{Transparent from the point of view of their location}: one should be able to locate and invoke the services irrespective of their real location. To do so, the use of a \textit{service registry} where services interfaces and location information are stored, is required (see fig. \ref{chap1:fig:basic-soa}).
\end{itemize}

SOA is an architectural style, a logical way of designing a software system to provide services to either end-user applications or other services distributed in a network through published and discoverable interfaces. Basically, SOA defines an interaction between software agents as an exchange of messages between service consumers (clients) and service providers (see fig. \ref{chap1:fig:basic-soa}).
\begin{figure}[ht!]
	\noindent
	\makebox[\textwidth]{\includegraphics[scale=0.6]{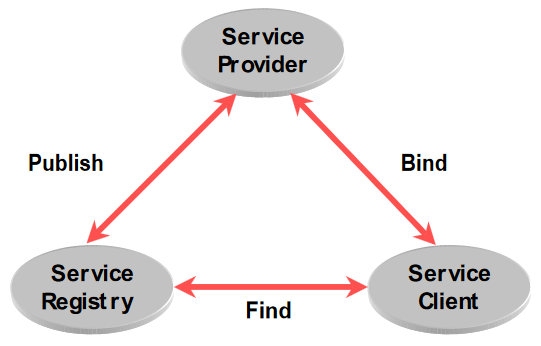}}
	\caption{The basic Service-Oriented Architecture (source \cite{papazoglou2003service}).}
	\label{chap1:fig:basic-soa}
\end{figure}
In SOA, the exchange of messages between agents can be \textit{synchronous} or \textit{asynchronous}.

In the synchronous model, the service consumer invokes a service and expects a result. The invoked service is then designed to immediately return a result and is the only service involved. This model operates similarly to remote procedure call technologies such as Remote Method Invocation (RMI) but with a much loosely coupling between the service consumer and its provider.

In the asynchronous model (which is generally a particular form of \textit{publish/subscribe}\footnote{Publish/Subscribe is a communication paradigm well adapted to the loosely coupled nature of distributed interaction in large-scale applications; with systems based on its interaction scheme, subscribers register their interest in an event, or a pattern of events, and are subsequently asynchronously notified of events generated by publishers \cite{eugster2003many}.}), a given service consumer $A$ expresses its desire to be aware of the execution state of a given service $b$, published (provided) by a service provider $B$, by subscribing to it. During this subscription, it provides information about one or more services $(a_i)$ that it also provides and that must be invoked when the execution state of $b$ has changed. Several services are thus involved, and each agent is generally both a service provider and a service consumer.

SOA has been designed to facilitate the implementation of distributed applications based on Peer-to-Peer architectures (nodes/agents communicate directly without going through a central server) and in which, the skills of each agent are exposed, discoverable and invocable by the others but, the technique and technology used by each agent is confidential. This setting completes the concept of agents to answer correctly to the challenges of distributed WfMS: hence the very increasing use of the concept of service in workflow systems. Actually, some currents of thought claim that SOA was designed to facilitate the automation of business processes and thus, the design of distributed WfMS. This is the case of Hurwitz et al. \citeyearpar{hurwitz2009service} who define SOA as : \textit{"a software architecture for building applications that implement business processes or services by using a set of loosely coupled black-box components orchestrated to deliver a well-defined level of service"}.

~

\noindent\textbf{\textit{Shared-data Overlay Network}} 

As the use of SOA in P2P applications escalates, there is a proliferation of tools to facilitate the design and implementation of these new applications \cite{kaur2013design}. \textit{Shared-data Overlay Network} (SON) \cite{SON} is one of those tools. SON is a middleware offering several \textit{Domain Specific Languages} (DSL) to facilitate the implementation of P2P systems whose components communicate in an asynchronous manner by services invocations. SON combines the powerful concepts of \textit{Component-Based Software Engineering}\footnote{Component-Based Software Engineering is an emerging paradigm of software development whose goal is, composing applications with plug \& play software components on the frameworks; so, to realise software reuse by changing both software architecture and software process \cite{aoyama1998new}.}, \textit{Service-Oriented Computing} and \textit{P2P Computing} in its engineering. 

By using SON middleware, the P2P application designer (software developer) is able not only to specify applications in component-based service model, but also to perform an effective code generation. In fact (see fig. \ref{chap1:fig:son-model}), the software developer defines using a dedicated DSL called \textit{Component Description Meta Language} (CDML), for each component, a set of services (input, internal and output). Then, he only implements the code of the components, i.e., the methods that implement the defined services. Afterwards, a code generation tool, called \textit{Component Generator}, generates a set of Java source files that implement the so-called \textit{container of the component}. These Java files are compiled together with the implementation code to generate a standalone and ready-to-use component. 
Thus, software developers are assisted and have greater ease in developing component and service-based P2P applications. These facilities allow them to focus more on the business logic and to defer to SON, the management of the runtime requirements (e.g., communication mechanisms, instantiation and connection of components, service discovery, etc.).
\begin{figure}[ht!]
	\noindent
	\makebox[\textwidth]{\includegraphics[scale=0.7]{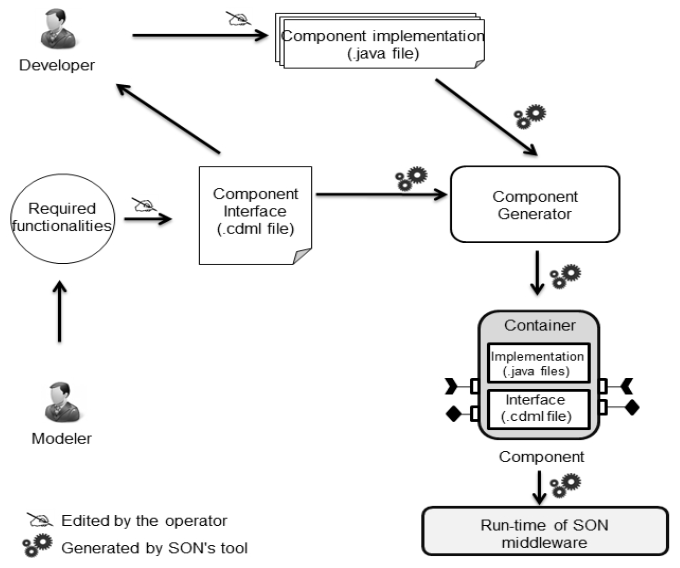}}
	\caption{Overview of a P2P application development process with SON (source \cite{SON}).}
	\label{chap1:fig:son-model}
\end{figure}
We use SON to implement prototypes of some of the models presented in this manuscript.

\mySubSection{Some Existing Distributed WfMS}{}
\label{chap1:sec:some-p2p-wfms}
In this section, we briefly present some existing approaches to distributed workflow management. As mentioned in \cite{theseKanzow}, these approaches have the following characteristics :
\begin{itemize}
	\item They are based on distributed entities that can communicate;
	\item These entities act autonomously, locally and thus, influence the further execution of the process (through their local actions, they choose the next actions to be executed);
	\item Each entity has a confidential local state;
	\item Each entity has only a partial view of the system's overall state at a given time;
	\item Workflow execution results from the automated interaction between the different entities.
\end{itemize}
The approaches to distributed workflow management presented here, can be divided into two categories :
\begin{enumerate}
	\item The first category contains those in which, data and controls are partially distributed and the WfMS is based on a client-server architecture;
	\item The second one is concerned with those in which WfMS, data and controls are fully distributed.
\end{enumerate}

\mySubSubSection{Some Partially Distributed WfMS}{}
\label{chap1:sec:partially-distributed-wfms}
\noindent\textbf{\textit{ADEPT (Advanced Decision Environment for Process Tasks) \cite{adept}}}

The ADEPT project was designed to automate flexible workflows at British Telecom\footnote{Nowadays British Telecom is renamed BT Group and remains the leader in the fixed telephony sector (source Wikipedia - \url{https://fr.wikipedia.org/wiki/BT_Group} - visited the 07/03/2020).}. Its main goal is to allocate resources to business processes using agents. According to its logic, workflow tasks are executed by agents acting as cooperating actors in a system supervised by one or more statically or dynamically assigned servers.
\begin{figure}[ht!]
	\noindent
	\makebox[\textwidth]{\includegraphics[scale=0.7]{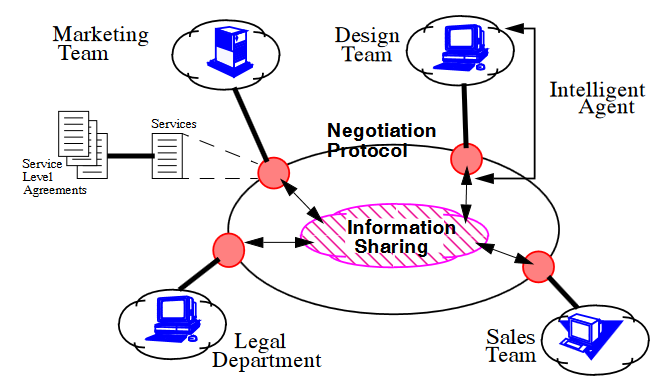}}
	\caption{An ADEPT environment (source \cite{adept}).}
	\label{chap1:fig:adept}
\end{figure}

Each agent is capable of providing one or more services. Services can be atomic (reduced to the execution of a single task) or composite (resulting from the combination of several other services, using operators that define the execution constraints - parallel, sequential, etc. -). If an agent needs the services of another agent, they must enter into an agreement called \textit{Service Level Agreement}\footnote{A Service Level Agreement is a formal contract used to guarantee that consumers' service quality expectation can be achieved \cite{wu2012service}.}. To facilitate the negotiation of agreements between agents, ADEPT provides a negotiation protocol and a service description language. Technically, the service description language allows agents to expose their services so that, they can be discovered by other agents which can then initiate negotiations for the use of those services, through the negotiation protocol.

Figure \ref{chap1:fig:adept} shows the architecture of ADEPT on an example of a workflow in which, four agents (marketing team, design team, sales team and legal department) collaborate to achieve business goals. Other publications on the ADEPT project may be useful for its understanding \cite{reichert1998adept, dadam2000clinical, reichert2003adept}.

~

\noindent\textbf{\textit{EvE (an Event-driven Distributed Workflow Execution Engine) \cite{eve}}}

According to the EvE approach, the distributed execution of workflows is done by event communication between agents (called \textit{brokers}) in charge of executing tasks. These agents perform tasks and create events in response to the occurrence of other events. EvE is based on a multi-server architecture in which, each server manages an entire cluster (a local network). However, this multi-server architecture is made transparent for the different agents thanks to \textit{adapters}; they can communicate independently of their respective domains. EvE provides amongst others the following services: 
\begin{itemize}
	\item Agents managed by servers and distributed across the network, capable of detecting events and executing tasks assigned to them; thereby, generating new events that are notified to other agents thanks to an inter-server communication mechanism that has been set up;
	\item A data warehouse in which information about agents, runtime data and \textit{Event - Condition - Action}\footnote{ECA is a paradigm that specifies the desired behaviour for reactive systems (i.e. systems that maintain ongoing interactions with their environments \cite{manna2012temporal}). In such a system centered around the ECA paradigm, when an event occurs, a condition is evaluated (by a querying mechanism) and the system takes corresponding action \cite{almeida2005modular}.} (ECA) event handling rules are stored. The information stored in the warehouse can be updated dynamically without the need to restart the system;
	\item Logging services for failure analysis and recovery. EvE supports exception notification, alerts and has the ability to resume execution after temporary disconnections;
\end{itemize}
\begin{figure}[ht!]
	\noindent
	\makebox[\textwidth]{\includegraphics[scale=0.6]{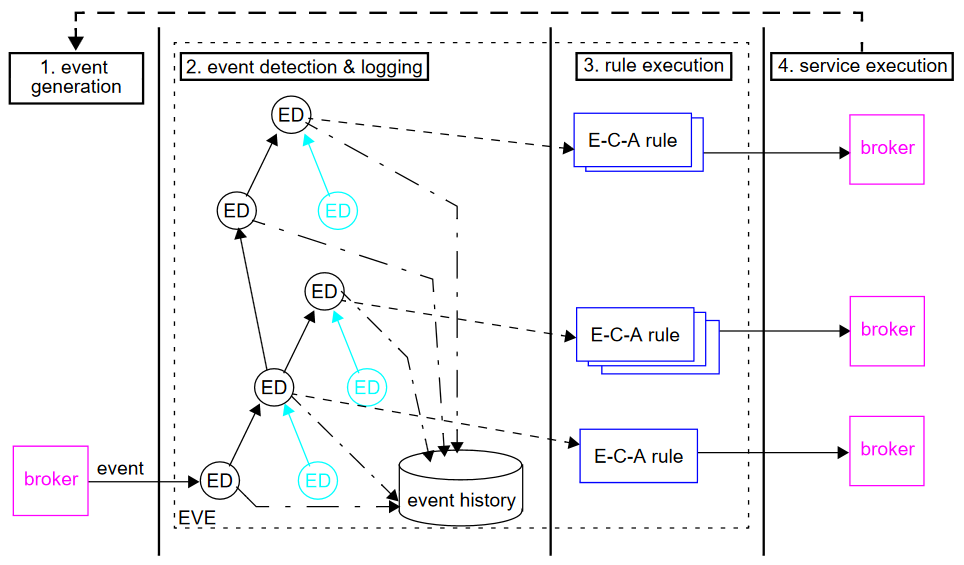}}
	\caption{The workflow execution process in EvE (source \cite{eve}).}
	\label{chap1:fig:eve}
\end{figure}

The execution of a workflow starts as soon as an event is generated by a broker. The local EVE-server (its manager) then performs event detection and ECA-rule execution. Within the execution of each rule, task assignment determines responsible brokers, which are then notified and subsequently react as defined by their ECA-rules. Particularly, brokers can generate new events, which again are handled by EVE-servers, and so on transitively (see fig. \ref{chap1:fig:eve}).

\mySubSubSection{Some Fully Distributed WfMS}{}
\label{chap1:sec:fully-distributed-wfms}
\noindent\textbf{\textit{METEOR$_2$ (Managing End-To-End OpeRations 2) \cite{das1997orbwork, meteor}}}

The METEOR$_2$ project is a continuation of the METEOR \cite{krishnakumar1995managing} effort. It is intended to reliably support coordination of users and automated tasks in real-world multi-enterprise heterogeneous computing environments. Key capabilities of the METEOR$_2$ WfMS include a comprehensive toolkit for building workflows and supporting high-level process modelling, detailed workflow specification and automatic code generation for its workflow enactment systems. 
\begin{figure}[ht!]
	\noindent
	\makebox[\textwidth]{\includegraphics[scale=0.8]{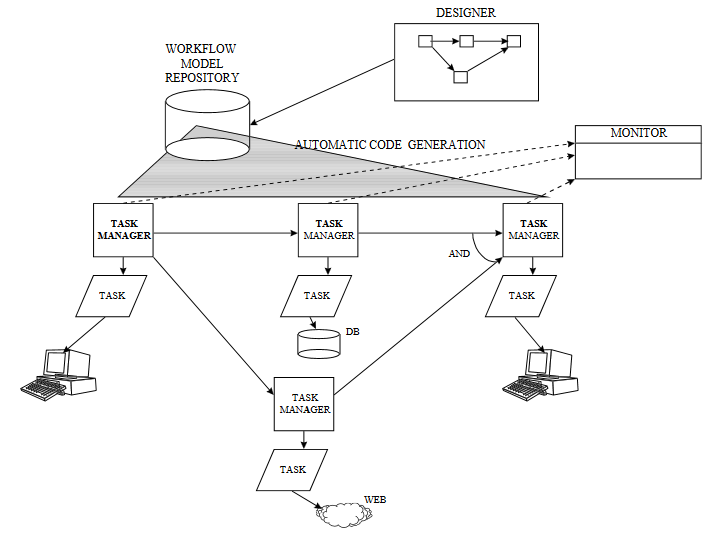}}
	\caption{The METEOR$_2$ architecture (source \cite{das1997orbwork}).}
	\label{chap1:fig:meteor2}
\end{figure}

METEOR$_2$ introduces concepts to represent each workflow as a set of tasks, task managers, processing entities and interfaces, in order to execute them in a completely distributed manner. Figure \ref{chap1:fig:meteor2} shows the various modules in METEOR$_2$ and their interaction. As can be seen in the picture, METEOR$_2$ includes a workflow designer that is used to create workflow models in a dedicated language. Once created, workflow models are stored in a workflow model repository. METEOR$_2$ also includes a workflow code generator that can read a stored workflow model and generate a convenient specific distributed workflow application. The generated application called the \textit{runtime system}, consists of a set of communicating agents called \textit{task managers} and their associated tasks, web-based user interfaces, a distributed recovery mechanism, a distributed scheduler and various monitoring components. All these workflow component are \textit{Common Object Request Broker Architecture}\footnote{CORBA is a standard middleware for distributed object systems. In its paradigm, a client application wishing to perform an operation on a server object, sends a request. The request is received by an Object Request Broker (ORB), responsible for all of the mechanisms required to find the object implementation for the request, to prepare the object implementation to receive the request, and to communicate the data making up the request to the server object. A server object accessible by CORBA is referred to as a CORBA object \cite{houlding2004system}.} (CORBA) objects and thus, they possess communication capabilities.

~

\noindent\textbf{\textit{The "Web Workflow Peer" Approach \cite{fakas04}}}

The approach proposed by Fakas and Karakostas is based on the concepts of \textit{Web Workflow Peer Directory} (WWPD) and \textit{Web Workflow Peer} (WWP). WWPD is an active directory system that maintains a list of all peers (WWP) that are available to participate in web workflow processes. It allows peers to register with the system and offer their services and resources to other peers. During the execution of workflow processes, the WWPD assists WWP to locate other WWP and use their services and resources. In their setting, key functionalities and data are distributed among WWP. The architecture is completely decentralised as no central workflow engine is employed to coordinate the process execution. The WWP encapsulates the necessary knowledge to perform the activities that are assigned to it and also to delegate some of the process execution to other WWP. The only centralised feature is the WWPD.
\begin{figure}[ht!]
	\noindent
	\makebox[\textwidth]{\includegraphics[scale=0.8]{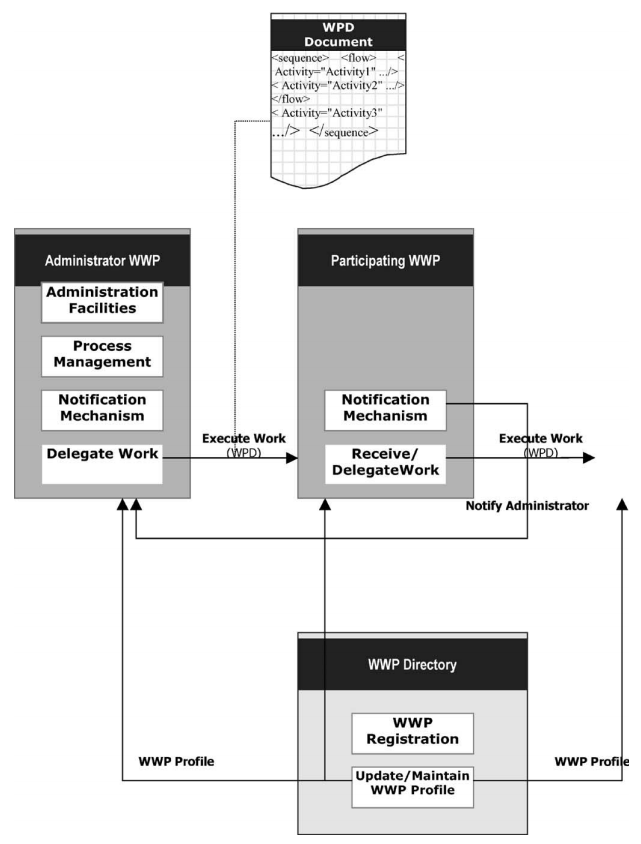}}
	\caption{A P2P workflow architecture (source \cite{fakas04}).}
	\label{chap1:fig:wwp}
\end{figure}

A WWP is a processing unit with an interface that is exposed on the Web and which can be accessed using Internet protocols. Its interface describes different types of processing capabilities, each corresponding to a workflow activity. When combined, such activities form a workflow process. A WWP that initiates and administers the process is called the \textit{Administrator Peer}. Other WWP delegated to carry out workflow activities are called the \textit{Participating Peers} (see fig. \ref{chap1:fig:wwp}). In practice, all peers are capable of becoming administrators in different workflow process instances. WWP use mobile documents called \textit{Workflow Process Description} as communication medium. Segments of those documents move from site to site and conveys structural information about the running workflow instance.

Workflow process administration is achieved by employing a notification mechanism. For instance, at the completion of an activity the WWP notifies the Administrator Peer so that, an updated status of the process instance is maintained. Similarly, upon expiration of an activity deadline, the  Administrator Peer notifies the WWP responsible for the expired activity. As far as we know, there is still no real workflow system based on this promising architecture.

~

\noindent\textbf{\textit{SwinDeW (Swinburne Decentralised Workflow) \cite{junYan06}}}

Combining workflow and P2P concepts, SwinDeW \cite{junYan06} has been designed as a special P2P system, which provides workflow management support in a truly decentralised way. SwinDeW adopts a flat, flexible and loosely coupled structure with an intentional absence of both a centralised device for data storage, and a centralised control engine for coordination. SwinDeW offers several distributed protocols, especially for the definition, instantiation and execution of processes.
\begin{figure}[ht!]
	\noindent
	\makebox[\textwidth]{\includegraphics[scale=0.5]{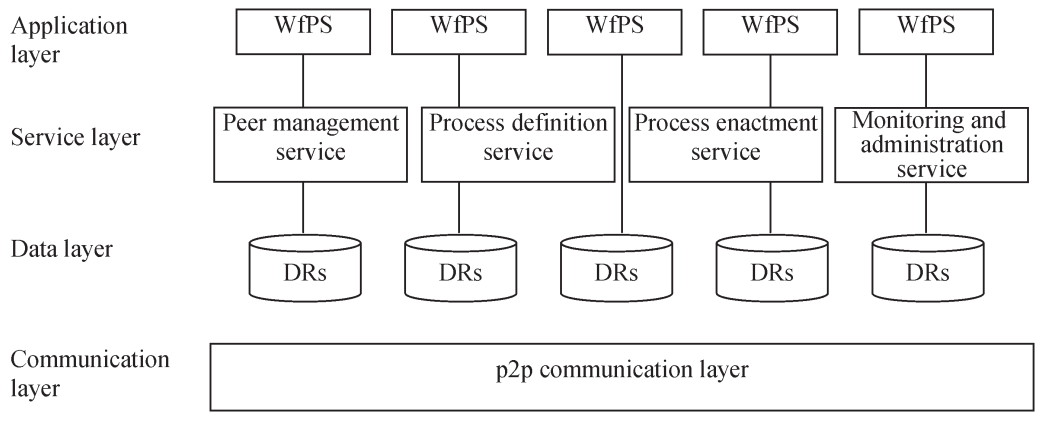}}
	\caption{A high-level view of SwinDeW's architecture (source \cite{junYan06}).}
	\label{chap1:fig:swindew}
\end{figure}

The SwinDeW system is defined as four layers (see fig. \ref{chap1:fig:swindew}). The top layer is the application layer; it defines application-related functions to fulfil workflows. \textit{Workflow Participant Software} (WfPS) is an application that provides interfaces to interact with a workflow participant and other WfPS, requesting services and responding to requests. Core services of the workflow system are provided at the service layer, which include the peer management service, the process definition service, the process enactment service, and the monitoring and administration service. The data layer consists of distributed Data Repositories (DR) that store workflow-related information. Finally, the monitoring and administration service provides supervisory capabilities and status monitoring.

In SwinDeW, a peer is given by a WfPS and a set of DR (see fig. \ref{chap1:fig:swindew-wfps}). Each peer resides on a physical machine, enabling direct communication with other peers in order to carry out the workflow. A peer is a self-managing entity that is associated with and operates on behalf of a workflow participant. From the functional perspective, the WfPS of a peer consists of three software components :
\begin{enumerate}
	\item A user component which serves as  a "bridge" between the associated workflow participant and the workflow environment;
	\item A task component that is in charge of the execution of tasks conducted by the associated participant;
	\item A flow component which helps to fit an individual task into the workflow. It deals with data dependency and control dependency among tasks by handling incoming and outgoing messages.
\end{enumerate}
A peer (agent) consists also in a set of four DR : the peer repository, the resource and tool repository, the task repository, and the process repository.
\begin{enumerate}
	\item A given peer repository stores an organisational model that represents organisational entities and their relationships. This repository represents a user's view of the completely defined organisational model;
	\item A resource and tool repository stores part of the resource model, which represents non human resources such as machines, external hardware, tools, etc.
	\item A task repository stores a set of active task instances, which represent the work allocated to the associated workflow participant in the context of process instances;
	\item A process repository stores a partial process definition distributed to the considered peer.
\end{enumerate}
\begin{figure}[ht!]
	\noindent
	\makebox[\textwidth]{\includegraphics[scale=0.5]{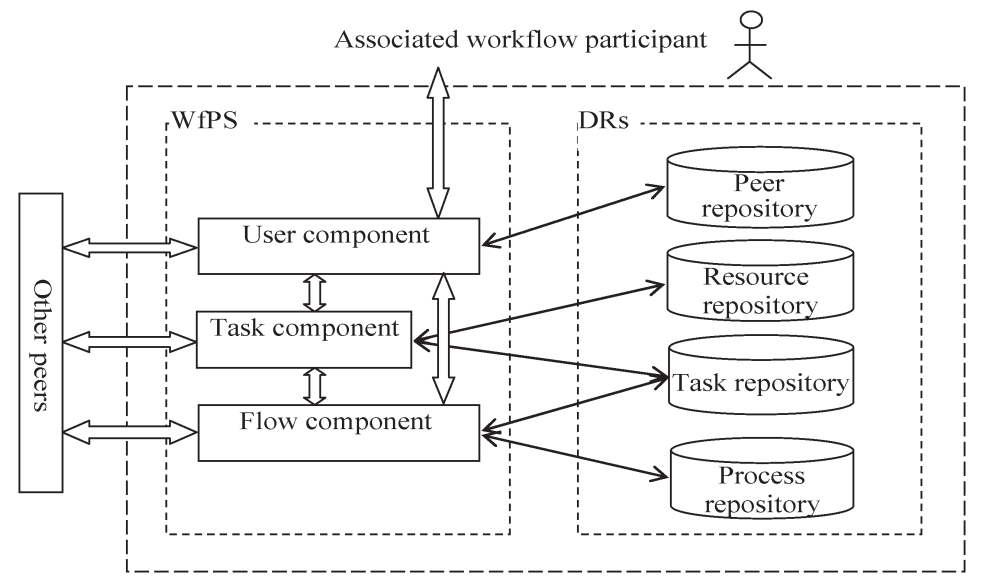}}
	\caption{Structure of a peer in SwinDeW (source \cite{junYan06}).}
	\label{chap1:fig:swindew-wfps}
\end{figure}

Workflow processes in SwinDeW are defined by a definition peer, which is associated with an authorised participant such as a process engineer. The resulting workflow models are stored in a distributed manner, in the process repositories of various peers. To avoid the distribution of too large workflow models, SwinDeW uses a "\textit{know what you should know}" policy to partition these models and thus, to configure each peer only with the partitions of the models that are of interest to it.

In SwinDeW, a workflow instance is executed under the management of the workflow system. Once such an instance is created, a peer network is also constructed for carrying out this process instance. Various task instances are scheduled to enact at different sites, step by step. The execution of a task depends on the satisfaction of two conditions : the \textit{information condition}, which defines the start condition of a task from the data dependency
perspective (a task can be executed only after essential input data are available), and the \textit{control condition}, which indicates the start condition of a task from the control dependency perspective (a task can be executed only after some relevant work has been logically completed). Peers collaborate with one another through direct message exchange, to properly schedule the execution of various task instances. There are two kinds of messages flowing between peers: \textit{information messages} and \textit{control messages}, which are structured in XML format. When a peer receives messages from its predecessor peers directly, it evaluates the information and control conditions of the task instance independently, starts working when both the conditions are satisfied, and notifies its successor peers directly by delivering information messages and control messages after the task instance is completed. The successor peers repeat the same procedure until the completion of the whole process instance. This approach is then fully distributed; moreover, it has an implementation.

\mySection{Artifact-Centric Business Process Management}{}
\label{chap1:sec:data-aware-bpm}
%\subsection*{Contexte et définitions}
%\label{sec:contexte}
Emerged in the early 2000s, the \textit{artifact-centric} paradigm of BPM is one of those that has been much studied over the last two decades. This paradigm has been pioneered by IBM \cite{nigam2003business} and revisited in several works such as \cite{abi2016towards, deutsch2014automatic, hull2009facilitating, lohmann2010artifact, assaf2017continuous, assaf2018generating, boaz2013bizartifact, lohmann2011artifact, estanol2012artifact}; it proposes a new approach to BPM by focusing on both automated processes (tasks and their sequencing) and data manipulated through the concept of "\textit{business artifact}" (\textit{artifact-centric modelling}). In this section we present the key concepts of the artifact-centric paradigm as well as some artifact-centric frameworks from the literature.

\mySubSection{Artifact-Centric BPM Basic Concepts}{}
\label{chap1:sec:artifact-centric-bpm-key-concept}

\mySubSubSection{The Aim of Artifact-Centric BPM}{}
\label{chap1:sec:aim-artifact-centric-bpm}
In order to be able to better model workflows, process modelling should include a specification of the order in which tasks are executed (control flow), the way data are processed (data flow), and how different branches in distributed and inter-organisational business processes and services are invoked and coordinated (message flow) \cite{lohmann2011artifact}. These three conceptual models of workflows are also known as the \textit{process}, the \textit{informational} and the \textit{organisational} models \cite{divitini2001inter}. Traditional approaches (BPMN, YAWL, BPEL, etc.) to BPM are process-centric (they are also said to be \textit{imperative}): they generally offer two different views on business processes: 
\begin{enumerate}
	\item Collaboration diagrams (sometimes called interconnected models) that emphasise the local control flow of each participant of the process;
	\item Choreography diagrams (interaction models) that describe the process from the point of view of the messages that are exchanged among the participants.
\end{enumerate}
Traditional approaches thus express workflow models by means of diagrams which define how a workflow is supposed to operate, but give little importance (or none at all) to the information produced as a consequence of the process execution: data are treated as second-class citizens.

To precisely remedy this, researchers have developed the artifact-centric \cite{nigam2003business} approach to the design and execution of business processes. Artifact-centric models do not specify processes as a sequence of tasks to be executed or messages to be exchanged (i.e. imperatively), but from the point of view of the data objects (called \textit{business artifacts} or simply \textit{artifacts}) that are manipulated throughout the course of the process (i.e. declaratively) \cite{lohmann2011artifact}. They rely on the assumption that any business needs to record details of what it produces in terms of concrete information. Artifacts are proposed as a means to record this information. They model key business-relevant entities which are updated by a set of services (specified by pre and postconditions) that implement business process tasks. This approach has been successfully applied in practice and it provides a simple and robust structure for workflow modelling \cite{estanol2012artifact}.

\mySubSubSection{How the Artifact-Centric Approach to BPM Works}{}
\label{chap1:sec:artifact-centric-bpm-approach}
According to the artifact-centric paradigm, BPM takes place in two main phases guided by the concept of artifact. In order to automate a given process, the designer must first of all focus on artifact modelling; i.e. he must provide data structures capable of storing and logically conveying the information produced during the execution of workflows. Then, these artifacts models will be introduced into a WfMS supporting artifact-centric execution of workflows for the enactment.

~

\noindent\textbf{\textit{What is an artifact ?}}

\cite{nigam2003business} define an artifact as \textit{"a concrete, identifiable, self-describing chunk of information that can be used by a business person to actually run a business"}. Artifacts are business-relevant objects that are created, evolve, and (typically) archived as they pass through the workflow \cite{hull2009facilitating}; they represent key conceptual objects of workflow that evolve as they move through an enterprise. Artifacts are modelled through \textit{artifacts types} (or \textit{models}). An expected characteristic of artifacts is that they should be self-describing: this requirement allows a business person to be able to look at an artifact and determine if he or she can work on it. Toward this end, an artifact type should include both:
\begin{enumerate}
	\item an \textit{information model} (or "data schema"), for holding information about the artifact as it moves through the process, from creation to archival storage; and
	\item a \textit{lifecycle model} (or "lifecycle schema"), which describes how and when tasks (activities or services) might be invoked on the artifacts as they move through the process.
\end{enumerate}
A prototypical example of an artifact type is the "air courier package delivery", whose information model can hold data about a package including sender, receiver, steps occurring in transport and billing activity; and whose lifecycle model would specify the possible ways that the delivery service and billing might be carried out \cite{hull2013data, hull2009facilitating}.

Several approaches of modelling the lifecycle of artifacts have been studied in the literature. The most commonly used approach is that in which, some form of \textit{finite state machines} (automata) \cite{hull2009facilitating} are used to specify lifecycles. Other variants presenting the lifecycle of an artifact by a Petri net \cite{lohmann2010artifact}, logical formulae depicting legal successors of a state \cite{damaggio2012artifact} have also been proposed.

~

\noindent\textbf{\textit{The artifact-centric execution}}

Artifact-centric models can be executed by artifact-centric WfMS. This new type of WfMS put stress on how artifacts are created, updated and exchanged between various actors. In these, artifacts are considered as \textit{adaptive documents} that conveys all the information concerning a particular execution case of a given process, from its inception in the system to its termination. In particular, this information provides details on the case's execution status as well as on its lifecycle (a representation of the possible evolutions of this status). To do this, during the execution of a given process, the actions carried out by each of the actors (agents) have the effect of updating (\textit{editing}) the artifacts involved in that execution. If the process is cooperative, the artifact representing it will be updated by several agents: it is said to be cooperatively edited and thus, the execution of a given business process according to the artifact-centric approach, can be assimilated to the \textit{cooperative editing} of documents.

Two major trends in the artifact-centric modelling approach have been developed: \textit{orchestration} and \textit{choreography} \cite{hull2009facilitating}. 
\begin{enumerate}
	\item Orchestration suggests the creation of centralised systems (usually called \textit{artifact hubs}), coordinated by an \textit{orchestrator} whose role is to facilitate interaction between actors while ensuring that business goals are met.
	\item Choreography-oriented approaches get rid of the orchestrator, and model actors as autonomous agents coordinating with artifacts and communicating in a P2P manner, to accomplish business goals. In these, each agent focuses on achieving a local business goal and the achievement of the global business goal is the result of aggregating results from different local business goals.
\end{enumerate}
Compared to choreography, orchestration reduces the agents' autonomy by making the orchestrator the main controller of interactions. Also, the orchestrated approach does not scale well. A limitation of the choreography-oriented approach is the lack of a single synchronisation point from which, it is possible to know the process's (actual) global execution state. Despite this, we agree with \cite{lohmann2010artifact} that, this completely decentralised approach is the one that best fits the modelling of the intrinsically distributed nature of business processes.

\mySubSection{Some Existing Artifact-Centric BPM Frameworks}{}
\label{chap1:sec:existing-artifact-centric-bpm}
There are several frameworks in the literature, that implement artifact-centric concepts. We briefly present some of them in this section. We begin by presenting purely artifact-centric approaches; then we look at an even more flexible model, recently developed as a \textit{data-centric} solution for case management.

\mySubSubSection{Some Purely Artifact-Centric BPM Frameworks}{}
\label{chap1:sec:purely-artifact-centric-bpm}
\noindent\textbf{\textit{Proclets \cite{van2001proclets}}}

The concept of \textit{proclets} was introduced to specify business processes in which, objects' lifecycles can be modelled at different levels of granularity and cardinality. A proclet can be seen as a lightweight workflow process equipped with a knowledge base that contains information on previous interactions; it is thus equipped with an explicit lifecycle or active documents (i.e., documents aware of tasks and processes): in this setting then, a proclet is both an agent and an artifact. Proclets can find each other using a \textit{naming service}, and communicate with each other to exchange messages through \textit{channels} (see fig. \ref{chap1:fig:proclet}). 
\begin{figure}[ht!]
	\noindent
	\makebox[\textwidth]{\includegraphics[scale=0.7]{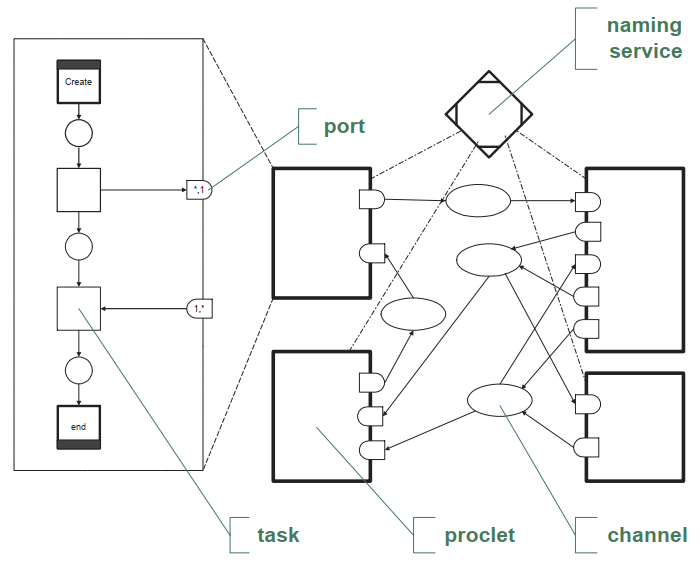}}
	\caption{Graphical representation of the proclet-based framework (source \cite{van2001proclets}).}
	\label{chap1:fig:proclet}
\end{figure}

In the proclet-based framework, the lifecycle of proclet instance is described by a \textit{proclet class} used as artifact type. Like an ordinary workflow model, a proclet class describes the order in which tasks can/need to be executed for individual instances of the class. Proclet classes are specified using a graphical language based on a sub-class of Petri nets so-called \textit{class of sound WF-nets}.

Proclets are well-suited to deal with settings in which several instances of data objects are involved. Proclets are considered to be distributed and autonomous enough to decide how to interact with the other proclets: thus, the proclet-based framework does not model proclets' locations. Moreover, the execution model of proclets is similar to choreography; the interoperation of proclets is not managed or facilitated by a centralised hub.

~

\noindent\textbf{\textit{Artifact hosting \cite{hull2009facilitating}}}

Hull et al. extend the artifact-centric model proposed by Nigam and Caswell \citeyearpar{nigam2003business}, to provide an interoperation framework in which data (artifacts) are hosted on central infrastructures named \textit{artifact-centric hubs}. Data hosted in artifact-centric hubs can be read and written by agents. This model is between choreography and orchestration because, agents are all connected to the hub but are not coordinated by a particular orchestrator. Unlike traditional orchestration schemes, the hub enables the participating agent to be pro-active, and serves primarily as a shared resource for coordinating activities. Participating agents can access information about the running artifact instances, can progress those instances along their lifecycles, and can subscribe to events in order to be alerted about significant steps in the progress of artifacts through their lifecycles. Security mechanisms are proposed for controlling access to data hosted in the hub.
\begin{figure}[ht!]
	\noindent
	\makebox[\textwidth]{\includegraphics[scale=0.5]{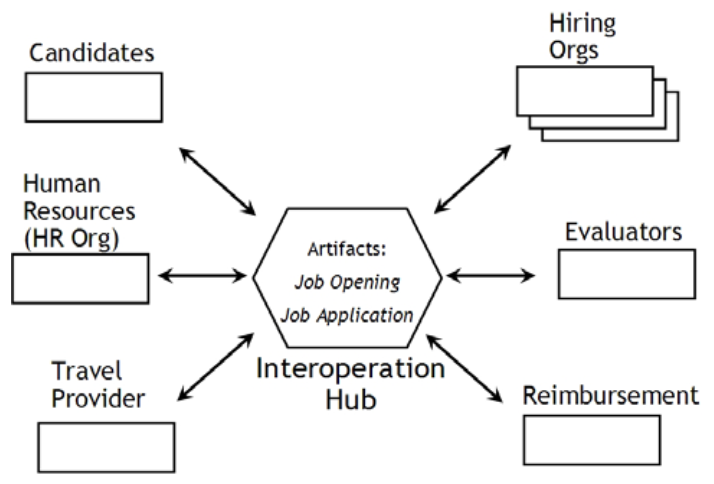}}
	\caption{An example of interoperation using an artifact-centric hub (source \cite{hull2009facilitating}).}
	\label{chap1:fig:artifact-hub}
\end{figure}

Figure \ref{chap1:fig:artifact-hub} illustrates an example of six groups of agents (potentially organisations) coordinating using artifacts that are managed in a centralised hub. These are the agents related to the resources (\textit{Candidates}, the \textit{Human Resources Organisation}, the \textit{Hiring Organisations}, the \textit{Evaluators}, the \textit{Travel Provider} and the \textit{Reimbursement}) that carry out the employee hiring process in a given enterprise.

~

\noindent\textbf{\textit{Artifact-centric choreographies \cite{lohmann2010artifact}}}

Lohmann and Wolf \citeyearpar{lohmann2010artifact} provide a more choreography-like framework for artifact-centric interoperation. They abandon the fact of having a single artifact hub \cite{hull2009facilitating} and they introduce the idea of having several agents which operate on artifacts. Some of those artifacts are \textit{mobile} (their location may change over time); thus, the authors provide a systematic approach for modelling artifact location and its impact on the accessibility of actions using a Petri net. Their model was designed with the conviction that by making explicit who is accessing an artifact and where the artifact is located, one will be able to automatically generate an interaction model that can serve as a contract between agents, and which make sure that global goal states specified on artifacts are reached. They thus propose an approach to automatically derive such an interaction model.

~

\noindent\textbf{\textit{Declarative choreographies \cite{sun2012declarative}}}

In \cite{sun2012declarative}, the authors are also interested in choreographies. More precisely, they develop a language allowing to model (in a declarative manner) the collaboration between several actors (the choreography) and a distributed algorithm allowing the execution of the choreographies specified in their language. Their choreography language has four distinct features :
\begin{enumerate}
	\item Each type of actor is an artifact schema with a selected sub-part of its information model visible to choreography specification.
	\item Correlations between actor types and instances are explicitly specified, along with cardinality constraints on correlated instances (e.g. each Order instance may correlate with exactly one Payment instance and multiple Vendor instances).
	\item Messages can include data; data in both messages and artifacts can be used in choreography constraints.
	\item The language is declarative and uses logic rules based on a mix of first-order logic and a set of binary temporal operators from DecSerFlow\footnote{DecSerFlow: \textit{a Declarative Service Flow Language} is a graphical, extendible language for expressing process models in a declarative way; it captures what is the high-level process behaviour without expressing how it is procedurally executed, hence giving a concise and easily interpretable feedback to the business manager \cite{lamma2007learning}.} \cite{van2006decserflow}.
\end{enumerate}
In particular, Skolem\footnote{Thoralf Albert Skolem (1887-1963) is a Norwegian mathematician and logician. He is particularly known for his work in mathematical logic and set theory which now bears his name, such as the L\"{o}wenheim-Skolem theorem or the notion of skolemisation (source, Wikipedia: \url{https://en.wikipedia.org/wiki/Thoralf_Skolem}, visited the 02/04/2020).} notations are used to both reference correlated actor instances and to manipulate dependencies among messages.

\mySubSubSection{A Guarded Attribute Grammars Based Framework to Data-Centric Case Management}{}
\label{chap1:sec:gag}
\noindent\textbf{\textit{What is case management ?}}

Highly important processes in organisations that have a tremendous impact on the success and add the most value, involve a high degree of knowledge work: they are driven by users' decisions (\textit{user-centric}) making it difficult to specify them into a set of activities with precedence relations at design-time (they are said to be \textit{knowledge-intensive}). Because knowledge-intensive processes are subject of frequent exceptions, traditional BPM solutions are not able to support them sufficiently \cite{hauder2014research}. \textit{Adaptive Case Management} (ACM) is gaining interest among researchers and practitioners as an  emerging paradigm to master situations in which adaptions have to be made at run-time (unpredictable situations) by so called knowledge workers. In contrast to traditional BPM, the ACM paradigm is not dictating knowledge workers a predefined course of action, but provides them with the required information at the right time (they are \textit{data-centric}) and authorises them to make decisions on their own \cite{hauder2014research}. 

The notion of \textit{case} in the ACM context, is closely related to the concept of artifact. Both involve the notion of a conceptual entity that progresses through time, according to some set of guidelines or lifecycle schema, and both taking advantage of a growing set of data accumulated over the case instance lifecycle \cite{hull2013data}. ACM can be seen as an extension of the artifact-centric paradigm in which, the flexibility of workflow models (types of artifacts) is highly valued; therefore, both users and data are treated as first-class citizens.

There is a growing research interest in the ACM paradigm and several models have already emerged. Guard-Stage-Milestone (GSM) \cite{hull2011business, damaggio2013equivalence}, a declarative model of the lifecycle of artifacts, was recently introduced and has been adopted as a basis of \textit{Case Management Model and Notation} (CMMN), the OMG standard for ACM. The GSM model defines Guards, Stages and Milestones to control the enabling, enactment and completion of (possibly hierarchical) activities; it then allows for dynamic creation of subtasks (the stages), and handles data attributes. However, interaction with users are modelled as incoming messages from the environment, or as events from low-level (atomic) stages. In this way, users do not explicitly contribute to the choice of a flow for a process.

Recently, Badouel et al. \citeyearpar{badouel14, badouel2015active} have proposed a more user-centric and data-driven ACM model (AWGAG) based on the concepts of Active Workspaces (AW) and Guarded Attribute Grammars (GAG). We are particularly interested in this model because, it incorporates concepts that we will manipulate throughout this manuscript. These include the concepts of \textit{grammars} as artifact types, \textit{structured documents} (trees) as artifacts, \textit{artifact editing}, etc.

~

\Needspace{5\baselineskip}
\noindent\textbf{\textit{AWGAG \cite{badouel2015active}}}

The AWGAG model of collaborative systems is centered on the notion of user's workspace. It assume that the workspace of a user is given by a map. It is a tree used to visualise and organise tasks in which, the user is involved together with information used for their resolution. The workspace of a given user may, in fact, consist of several maps where each map is associated with a particular service offered by the user. In short, one can assume that a user offers a unique service so that any workspace can be identified with its graphical representation as a map.
\begin{figure}[ht!]
	\noindent
	\makebox[\textwidth]{\includegraphics[scale=0.7]{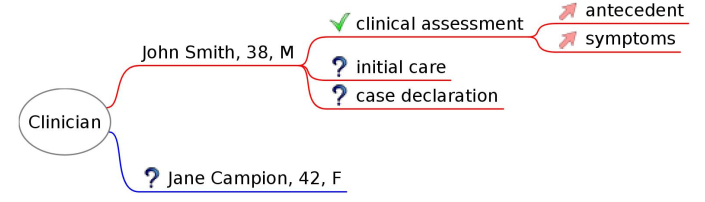}}
	\caption{Active workspace of a clinician (source \cite{badouel2015active}).}
	\label{chap1:fig:aw-clinician}
\end{figure}

As an example, figure \ref{chap1:fig:aw-clinician} shows a map that represents the workspace of a clinician acting in the context of a disease surveillance system. The service provided by the clinician is identifying the symptoms of influenza in a patient, clinically examining the patient, eventually placing him under therapeutic care, declaring the suspect cases to the disease surveillance center, and monitoring the patient based on subsequent requests from the epidemiologist or the biologist.

Each call to this service, namely when a new patient comes to the clinician, creates a new tree rooted at the central node of the map. This tree is an artifact that represents a structured document for recording information about the patient all along being taken over in the system. Initially, the artifact is reduced to a single (open) node that bears information about the name, age and sex of the patient. An open node, graphically identified by a question mark, represents a pending task that requires the clinician's attention. In this example the initial task of a given artifact is to clinically examine the patient. This task is refined into three subtasks: clinical assessment, initial care, and case declaration.

In the AWGAG model, a task is interpreted as a problem to be solved, that can be completed by refining it into sub-tasks using business rules. A business rule is modelled by a production $P: s_0 \rightarrow s_1 \ldots s_n$ expressing that task $s_0$ can be reduced to subtasks $s_1$ to $s_n$. For instance, the production 
\begin{itemize}
	\item[] $patient \rightarrow clinicalAssessment \ initialCare \ caseDeclaration$
\end{itemize}
states that, a task of sort $patient$, the axiom of the grammar associated with the service provided by the clinician, can be refined by three subtasks whose sorts are respectively $clinicalAssessment$, $initialCare$, and $caseDeclaration$. If several productions with the same left-hand side $s_0$ exist, then the choice of a particular production corresponds to a decision made by the user. In the example, the clinician has to decide whether the case under investigation has to be declared to the disease surveillance center or not. This decision can be reflected by the following two productions:
\begin{itemize}
	\item[] $suspectCase: caseDeclaration \rightarrow followUp$
	\item[] $benignCase: caseDeclaration \rightarrow$
\end{itemize}
If the case is reported as suspect, then the clinician will have to follow up the case according to further requests of the biologist or the epidemiologist. On the contrary (i.e. the clinician has described the case as benign), the case is closed with no follow up actions.

AWGAG  model considers artifacts as trees whose nodes are sorted and whose productions are taken into a grammar (GAG). The lifecycle of an artifact is implicitly given by the set of productions of the underlying GAG:
\begin{enumerate}
	\item The artifact initially associated with a case, is reduced to a single open node.
	\item An open node $X$ of sort $s$ can be refined by choosing a production $P: s \rightarrow s_1 \ldots s_n$ that fits its sort. The open node $X$ becomes a closed node under the decision of applying production $P$ to it. In doing so, task $s$ associated with $X$ is replaced by $n$ subtasks $s_1$ to $s_n$, and new open nodes $X_1$ of sort $s_1$ to $X_n$ of sort $s_n$, are created accordingly: the artifact is said to be edited (see fig. \ref{chap1:fig:aw-artifact-edition}).
	\item The case reaches completion when its associated artifact is closed, i.e. it no longer contains open nodes.
\end{enumerate}
\begin{figure}[ht!]
	\noindent
	\makebox[\textwidth]{\includegraphics[scale=0.5]{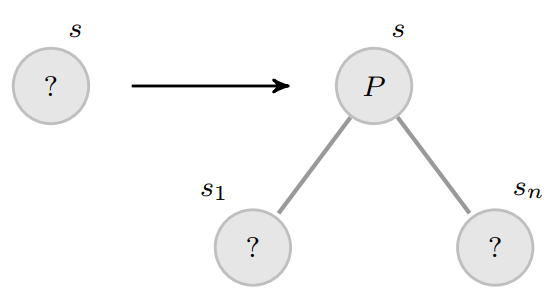}}
	\caption{Artifact edition in AWGAG (source \cite{badouel2015active}).}
	\label{chap1:fig:aw-artifact-edition}
\end{figure}

Additional information are attached to open nodes using \textit{attributes}, to model the interactions and data exchanged between the various tasks associated with them. For that, each sort comes equipped with a set of \textit{inherited} attributes and a set of \textit{synthesised} attributes where: inherited attributes represents input data, i.e. necessary data for the associated task to be executed, while synthesised attributes represents output data, i.e. data that are produced after the task being executed. This formalism puts emphasis on a declarative (logical) decomposition of tasks to avoid over-constrained schedules. Indeed, business rules do not prescribe any ordering on task executions. Ordering of tasks depend on the exchanged data and are therefore determined at runtime. In this way, the AWGAG model allows as much concurrency as possible in the execution of the current pending tasks.

Furthermore, a given AWGAG model is flexible and can incrementally be designed: one can initially let the designer manually develop large parts of the map, and progressively improve the automation of the process by refining the classification of the nodes, and introducing new business rules when recurrent patterns of activities are detected. The AWGAG model presents great properties such as \textit{distribution} and \textit{soundness}; these properties and more details are discussed in \cite{badouel2015active}. Several implementations and extensions of the AWGAG model are currently being carried out.

\mySection{Summary}{}
\label{chap1:sec:summary}
%\subsection*{Contexte et définitions}
%\label{sec:contexte}
In this chapter, we have provided an overview of the basic concepts related to BPM. To this end, we have provided clear and concise definitions of numerous notions. We presented the concept of WfMS, how it works and the current challenges in the production of such systems. Knowing that the main current challenges in the production of WfMS are to provide them with the flexibility to better cope with the distributed nature of the workflows they manage, and extend their expressiveness so that they can address as first-class citizens, other perspectives of workflows such as data and users, we presented new paradigms to BPM and a few approaches to P2P BPM.

In addition, we focused on the artifact-centric paradigm for BPM and established that, according to it, the execution of a given workflow can be seen as the cooperative editing of one or more documents called artifacts. In these cases, it is preferable that the manipulated documents are structured; they can be exchanged between the different actors in the workflows' execution (they are said to be mobile) to serve as a support to help their coordination, but also, indirectly, to be the fruit of their cooperation: this is the research axis that we followed in this thesis.

To make our task easier, it would be a good idea to look at structured documents' cooperative editing workflows: it is the subject of the next chapter. Since Badouel and Tchoup\'e have theorised an asynchronous cooperative editing model for structured documents of which, some of the concepts were brilliantly taken up in the development of the AWGAG model, it is their editing model that will be the main subject of our study.

	\mathversion{normal2}
	\myChapter{A Workflow for Structured Documents' Cooperative Editing : Key Principles and Algorithms}{}
\label{chap2:structured-editing-artifact-type}
\myMiniToc{section}{Contents}
% If no minitoc then
% \startcontents[chapters]
\mySection{Introduction}{}
\label{chap2:sec:introduction}
The purpose of this chapter is to introduce the main concepts related to asynchronous cooperative editing of structured documents. In an asynchronous cooperative editing workflow, several authors located on geographically distant sites coordinate to edit asynchronously the same structured document. In such editing workflows (see fig. \ref{chap2:fig:badouel-tchoupe-workflow}), the desynchronised editing phases in which each co-author edits on his site, his copy of the document, alternate with the synchronisation-redistribution phases in which, the different contributions (local replicas) are merged (on a dedicated site) into a single document, which is then redistributed to the various co-authors for the continuation of the edition. This pattern is repeated until the document is completely edited.

Badouel and Tchoup\'e \citeyearpar{badouelTchoupeCmcs} have theorised an asynchronous cooperative editing workflow in which, stakeholders (several subsystems - sites - distributed across a network) work by editing and exchanging (partial) replicas of documents representing their perceptions (views) at any given time. Therefore, each subsystem (actor) has a partial view of the edited document at any given time, and the current (global) document is given by the merging of different (partial) documents from the various subsystems. In their model, collaborations between actors can be divided into three sequential phases (see fig. \ref{chap2:fig:badouel-tchoupe-workflow}): 
\begin{itemize}
	\item The \textit{distribution phase} where global structured document (an artifact) is replicated to each subsystem;
	\item The \textit{editing phase} in which local processes of subsystems are executed, inducing an update of the local replica of the global document; 
	\item The \textit{synchronisation phase} in which the various local documents updated are merged into a global document.
\end{itemize}

In this chapter, we are mainly interested in the model proposed by Badouel and Tchoup\'e. However, we are not just doing a systematic review of literature of their model. We subtly introduce three contributions which further validate their model, and polish the path towards our goal of producing a completely decentralised model for the automation of administrative workflows :
\begin{enumerate}
	\item First of all, we extend the merge algorithm proposed by them so that, it can be applied in the more general case where conflicts might appear. To this end, we propose a consensus reconciliation algorithm that generates conflict-free maximum prefixes of the documents resulting from the merging of several conflicting replicates.
	\item Second, we propose a generic system architecture that can be used to produce workflow systems for the cooperative editing of structured documents based on their model.
	\item Finally, we propose a cooperative editing system prototype based on the proposed architecture and coded by cross-fertilisation of Java and Haskell.
\end{enumerate}

In the rest of this chapter, in section \ref{chap2:sec:cooperative-editing-concepts}, we will present some basic concepts related to cooperative editing of documents, as well as some existing cooperative editing systems. In section \ref{chap2:sec:grammatical-cooperative-editing}, we will present the main concepts of Badouel and Tchoup\'e's model as well as the reconciliation algorithm that we propose. In section \ref{chap2:sec:architecture-cooperative-editing}, will be presented, the generic architecture of workflow systems that we propose, as well as a prototype system built according to it. We will conclude this chapter in section \ref{chap2:sec:conclusion}.

\mySection{Basic Concepts on Cooperative Editing Workflows}{}
\label{chap2:sec:cooperative-editing-concepts}
Cooperative editing is a work of (hierarchically) organised groups, that operate according to a schedule involving delays and a division of labor (coordination). 
Like any CSCW, cooperative editing is subject to spatial and temporal constraints. Thus, one distinguishes distributed or not, and synchronous or asynchronous cooperative editing. When distributed, the various editing sites are geographically dispersed and each of them has a local copy of the document to be edited; systems that support such an edition should offer algorithms for data replication \cite{Yasushi2005} and for the fusion of updates. When asynchronous, various co-authors get involved at different times to bring their different contributions.

A cooperative editing workflow goes generally, from the creation of the document to edit, to the production of the final document through the alternation and repetition of distribution, editing and synchronisation phases. The literature is full of several cooperative editing workflows and of their management systems. We present a few in this section.

\mySubSection{Real-Time Cooperative Editing Workflows}{}
\label{chap2:sec:real-time-cooperative-editing}
In these generally centralised systems (Etherpad\footnote{Official website of Etherpad: \url{http://www.etherpad.org/}, visited the 04/04/2020.} \cite{epad}, Google Docs\footnote{Google Docs is accessible online at \url{https://www.docs.google.com/}, visited the 04/04/2020.}, Framapad\footnote{Get more information on Framapad at \url{http://www.framasoft.org/}, visited the 04/04/2020.}, Fidus Writer\footnote{Official website of Fidus Writer: \url{https://www.fiduswriter.org/}, visited the 04/04/2020.} \cite{fiduswriter}, etc.), the original document is created by a co-author on the central server. The latter then invites his colleagues to join him for the editing; they therefore connect to the editing session usually identified by a URL (distribution phase, although the document is generally not really duplicated). During an editing session (synchronous editing phase), all connected co-authors work on a single copy of the document but in different contexts. When the integration is automatic, changes performed by one of them are immediately (automatically) propagated to be incorporated into the basic document (synchronisation phase), and the latter is then redistributed to others. The changes are saved progressively and the server usually keeps multiple versions of the document.

The majority of real-time editors use the model of operational transformations \cite{theseOster, theseMounir}. Their architectures are therefore based on the one defined by this model. Meaning that, they distinguish two main components: an \textit{integration algorithm}, responsible for the receipt, dissemination and execution of operations and a \textit{set of processing functions} that are responsible for "merging" updates by serialising two concurrent operations.

\mySubSection{Asynchronous Cooperative Editing Workflows}{}
\label{chap2:sec:async-cooperative-editing}
This edit mode is distinguished by real distribution phases in which, the document to be edited is replicated on different sites, using appropriate algorithms \cite{Yasushi2005}. A co-author may then contribute at any time, by editing his local copy of the document. Here, we focus on a few asynchronous cooperative editors operating in client-server mode.

~

\noindent\textbf{\textit{Wikiwikiweb (Wikis)}}

Wikis \cite{wikiwikiweb} are a family of collaborative editors for editing web pages from a browser. To edit a page on a Wiki, one must duplicate it and contribute. After editing, he just have to save it and to publish a new version of that page. In a competing editing case, it is the last published version which will be visible. Even though it is still possible to access the previously published versions, there is no guarantee that a new version of the page preserves intentions (incorporates aspects) of previous versions. For this aspect, a Wiki can be seen much more as a web page version manager.

~

\noindent\textbf{\textit{CVS (Concurrent Versions System)}}

Under CVS \cite{cvs}, versions of a document are managed in a space called repository, and each user has a personal workspace.
To edit a document, the user must create a replica in his workspace. He will amend this replica, then will release a new version of the document in the repository. In case the document is concurrently edited by several authors and at least one update has already been published, the author wishing to publish a new update, will be forced to consult and integrate all previous updates through  dedicated tools, integrated in CVS.

~

\Needspace{5\baselineskip}
\noindent\textbf{\textit{SVN (Subversion)}}

SVN\footnote{Check more about SVN at \url{http://www.subversion.apache.org/}, visited the 04/04/2020.} \cite{svn} was created to replace CVS. Its main goal was to propose a better implementation of CVS. So as CVS, SVN relies on an optimistic protocol of concurrent access management: the \textit{copy-edit-merge} paradigm. SVN provides many technical changes like a new commit algorithm, the management of metadata versions, new user commands and many others features.

~

\noindent\textbf{\textit{Git}}

The main purpose of Git\footnote{Official website of Git: \url{https://www.git-scm.com/}, visited the 04/04/2020.} is the management of various files in a content tree considered as a deposit (all files of a source code for example). To edit a deposit, a given user connects to it and clones (forks). He obtains a copy of that deposit, modifies it locally through a set of commands provided by Git. Then he offers his contribution to primary maintainer which can validate it and thus, merges it with the original deposit. During this operation, new versions of modified files are created in the main repository. It is therefore possible under Git, to access any revision of a given file.

\mySubSection{Badouel and Tchoup\'e's Cooperative Editing Workflow}{}
\label{chap2:sec:badouel-tchoupe-cooperative-editing}
Badouel and Tchoup\'e \citeyearpar{badouelTchoupeCmcs} proposed a workflow for cooperative editing of structured documents (those with regular structures defined by grammatical models such as DTD, XML schema \cite{xml2000}, etc.), based on the concept of "view". The authors use context free grammars as documents models. A document is thus, a derivation tree for a given grammar.

The lifecycle of a document in their workflow can be sketched as follows: initially, the document to edit ($t$) is in a specific state (initial state); various co-authors who are potentially located in distant geographical sites, get a copy of $t$ which they locally edit. For several reasons (confidentiality, security, efficiency, etc. \cite{tchoupeAtemkeng2}), a given co-author "$i$" does not necessarily have access to all the grammatical symbols that appear in the tree (document);  only a subset of them can be considered relevant for him: that is his \textit{view} ($\mathcal{V}_i$). The locally edited document, is therefore a \textit{partial replica} (denoted $t_{\mathcal{V}_i}$) of the original document. This one is obtained by \textit{projection} ($\pi$) of the original document with regard to the view of the considered co-author ($t_{\mathcal{V}_i}=\pi_{\mathcal{V}_i}(t)$). The edition is asynchronous and local documents obtained are called updated partial replicas denoted by $t_{\mathcal{V}_i}^{maj}$.

Badouel and Tchoup\'e focus only on the positive edition: edited documents are only increasing; thus, the co-authors cannot remove portions of the document when a synchronisation has already been performed. To both ensure that property, and be able to tell a co-author where he shall contribute, the documents being edited are represented by trees with \textit{buds} that indicate the only places where editions are possible.
Buds are typed; a \textit{bud of sort $X$} is a leaf node labelled $X_\omega$: it can only be edited (extended in a subtree) by using a \textit{$X$-production} (production with $X$ as left hand side).

When a synchronisation point is reached, all contributions $t_{\mathcal{V}_i}^{maj}$ of different co-authors are merged in a single global document $t_f$. To ensure that the merging is always possible (convergence), Badouel and Tchoup\'e assume that on each site, the editions are controlled by a local grammar. These local grammars are obtained from the global one, by projection along the corresponding views \cite{badouelTchoupeCmcs, tchoupeAtemkeng2}.
\begin{figure}[ht!]
	\noindent
	\makebox[\textwidth]{\includegraphics[scale=0.65]{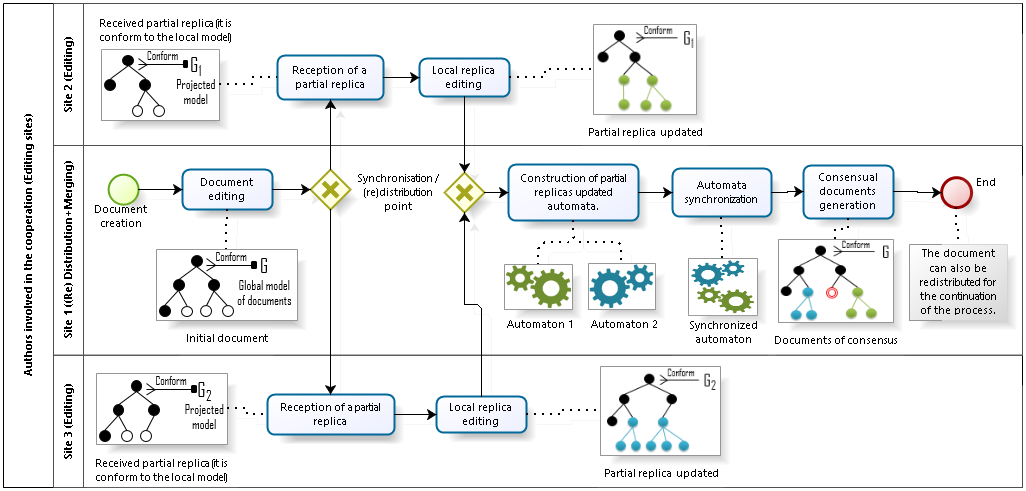}}
	\caption{A BPMN orchestration diagram sketching a cooperative editing workflow of a structured document according to Badouel and Tchoup\'e.}
	\label{chap2:fig:badouel-tchoupe-workflow}
\end{figure}

Figure \ref{chap2:fig:badouel-tchoupe-workflow} gives an overview, with a BPMN orchestration diagram, of the structured documents' cooperative editing workflow according to Badouel and Tchoup\'e's proposal; at site 1, operations of (re)distribution and merging of the document in accordance with a (global) model $G$, are realised; at sites 2 and 3, edition of partial replicas in accordance with (local) models $G_1$ and $G_2$, derived by projecting the global documents model $G$, are done.

In summary, the workflow of Badouel and Tchoup\'e is different from the others with its concept of "view" and by the fact that, it exclusively manipulates (partial) structured documents.

\mySection{Tree Automata for Extracting Consensus from Partial Replicas of a Structured Document}{}
\label{chap2:sec:grammatical-cooperative-editing}

In this section, we will better study Badouel and Tchoup\'e work on structured editing. We will adopt and adapt the different mathematical tools they proposed, to produce a more general algorithm for merging partial replicates, by taking into account the cases where these would be in conflict.

\mySubSection{Structured Cooperative Editing and Notion of Partial Replication}{} 
\label{chap2:sec:structured-cooperative-editing-partial-rep}

\mySubSubSection{Structured Document, Editing and Conformity}{}
\label{chap2:sec:structured-document-edition-conformity}

In the XML community, the document model is typically specified using a DTD or a XML Schema \cite{xml2000}. It is shown that these DTD are equivalent to (regular) grammars with special characteristics called   \textit{XML grammars} \cite{XMLG}. The (context free) grammars are therefore a generalisation of the DTD and, on the basis of the studies they have undergone, mainly as formal models for the specification of programming languages, they provide an ideal framework for the formal study of the transformations involved in XML technologies. That's why we use them in our work as tools for specifying the structure of documents.

We are only interested in the structure of the documents regardless of their contents and the attributes they may contain.  
We will therefore represent the abstract structure of a structured document by a tree, and its model by an abstract context free grammar; a valid structured document will then be a derivation tree for this grammar. 
A context free grammar defines the structure of its instances (the documents that are conform to it) by means of productions. 
A production, generally denoted $p: X_0 \rightarrow X_1 \ldots X_n$, is comparable in this context, to a structuring rule which shows how the symbol $X_0$, located in the left part of the production, is divided into a sequence of other symbols $X_1 \ldots X_n$, located on its right side. More formally, 

\begin{definition}
An \textbf{abstract context free grammar} is given by $\mathbb{G}=\left(\mathcal{S},\mathcal{P},A\right)$
composed of a finite set $\mathcal{S}$ of \textbf{grammatical symbols} or \textbf{sorts} corresponding to the different \textbf{syntactic categories} involved, a particular grammatical symbol $A\in\mathcal{S}$ called \textbf{axiom}, and a finite set $\mathcal{P}\subseteq\mathcal{S}\times\mathcal{S}^{*}$ of \textbf{productions}. 
A production $P=\left(X_{P(0)},X_{P(1)}\cdots X_{P(n)}\right)$ is denoted $P:X_{P(0)}\rightarrow X_{P(1)}\cdots X_{P(n)}$ and $\left|P\right|$ 
denotes the length of the right hand side of $P$. A production with the symbol $X$ as left part is called a \textit{X-production}.
\end{definition}

For certain treatments on trees (documents), it is necessary to designate precisely a particular node. Several indexing techniques exist, among them, the so-called \textit{Dynamic Level Numbering} \cite{boe04} based on identifiers with variable lengths, inspired by the \textit{Dewey} decimal classification (see fig. \ref{chap2:fig:indexed-tree}). According to this indexing system, a tree can be defined as follows:

\begin{definition}
A \textbf{tree} whose nodes are labelled in an alphabet $\mathcal{S}$, is a partial map $t:\mathbb{N}^*\rightarrow \mathcal{S}$, whose domain $\mathit{Dom}(t)\subseteq \mathbb{N}^*$ is a prefix closed set such that, for all $u\in \mathit{Dom}(t)$, the set $\{i\in\mathbb{N}~|~u\cdot i\in\mathit{Dom}(t)\}$ is a non-empty interval of integers $[1,\cdots,n]\cap\mathbb{N}$ ($\epsilon \in \mathit{Dom}(t) \mathit{~is~ the~ root~ label}$); the integer $n$ is the \textbf{arity} of the node whose address is $u$. 
$t(u)$ is the value (label) of the node in $t$, whose address is $u$.
%Un arbre $t_{w}$ est \textbf{un sous-arbre} de $t$ de racine le noeud d'adresse $w \in \mathit{Dom}(t)$ a pour domaine $\mathit{Dom}(t_{w})=\{u ~|~ w.u \in \mathit{Dom}(t)\} \mbox{ et }  t_{w}(u) = t(w.u)$
If $t_1,\cdots,t_n$ are trees and $a\in \mathcal{S}$, we denote $t = a [t_1,\ldots,t_n]$, the tree $t$ of domain $\mathit{Dom}(t)=\{\varepsilon\}\cup\{i\cdot u~|~1\leq i\leq n\, ,\; u\in \mathit{Dom}(t_i)\}$ with $t(\varepsilon)=a$, and $t(i\cdot u)=t_i(u)$. 

%L'arbre vide sera noté \textit{nil} et $\mathit{next} ~t= [t_1, \cdots, t_n]$ est la liste des sous-arbres de l'arbre $t = a[t_1,\ldots,t_n]$. 
\end{definition}
\begin{figure}[ht!]
	\noindent
	\makebox[\textwidth]{\includegraphics[scale=0.45]{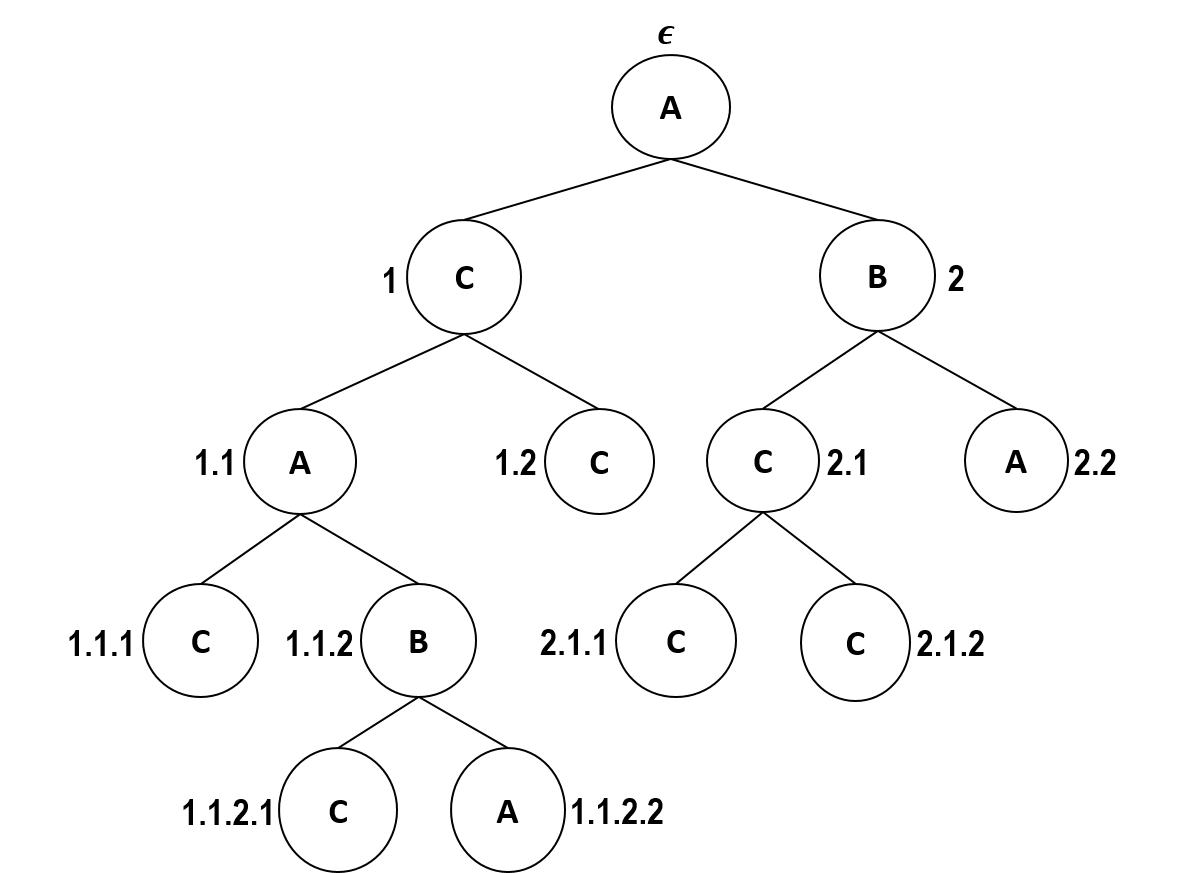}}
	\caption{Example of an indexed tree.}
	\label{chap2:fig:indexed-tree}
\end{figure}

Let $t$ be a document and $\mathbb{G}=\left(\mathcal{S},\mathcal{P},A\right)$ a grammar. $t$ is a derivation tree for $\mathbb{G}$ if its root node is labelled by the axiom $A$ of $\mathbb{G}$, and if for all internal node $n_0$ labelled by the sort $X_0$, and whose sons $n_1, \ldots n_n$, are respectively labelled by the sorts $X_1,\ldots, X_n$, there is one production $P \in \mathcal{P}$ such that, $P:X_0\rightarrow X_1\cdots X_n$ and $\left|P\right|=n$. 
It is also said in this case, that $t$ belongs to the language generated by $\mathbb{G}$ from the symbol $A$, and it is denoted $t \in \mathscr{L}{\left( \mathbb{G}, ~A \right)}$ or  $t\therefore \mathbb{G}$.

%\begin{comment}
There is a bijective correspondence between the set of derivation trees of one grammar and all its \textit{Abstract Syntax Tree} (\textit{AST}). In an AST, nodes are labelled by the names of the productions. 

\begin{definition}
The set $AST(\mathbb{G},X)$ of \textbf{abstract syntax trees} according to the grammar $\mathbb{G}$ associated with grammatical symbol $X$, consists of trees in the form $P[t_1,\ldots,t_n]$ where $P$ is a production such that, $X=X_{P(0)}$, $n=|P|$ and $t_i\in AST(\mathbb{G},X_i), ~X_{i} = X_{P(i)}$ for all $1\leq i\leq n$. 
%Les arbres de syntaxe abstraite sont donc les termes pour la signature dont les sortes sont les symboles grammaticaux et dont les opérateurs sont les productions de la grammaire où la production $P:X_{P(0)}\rightarrow X_{P(1)}\cdots X_{P(n)}$  est vue comme un opérateur d'arité $X_{P(1)}\times\cdots\times X_{P(n)}\rightarrow X_{P(0)}$.
\end{definition}
AST are used to show that a given tree, labelled with grammatical symbols, is an instance of a given grammar.
%\end{comment}

A structured document being edited, is represented by a tree containing \textit{buds} (or \textit{open nodes}) which indicate in a tree, the only places where editions (i.e updates) are possible\footnote{Note that, we are interested only in the \textit{positive edition} based on a partial optimistic replication \cite{Yasushi2005} of edited documents. In fact, the edited documents are only increasing: there is no possible erasure as soon as a synchronisation has been performed.}.
Buds are typed; a \textit{bud of sort $X$} is a leaf node labelled by $X_\omega$ (see fig. \ref{chap2:fig:tree-with-bud}): it can only be edited (i.e extended to a subtree) by using an \textit{X-production}. Thus, a structured document being edited and that have the grammar $\mathbb {G} = (\mathcal {S}, \mathcal {P}, A) $ as model, is a derivation tree for the extended grammar 
$\mathbb{G}_{\Omega}=(\mathcal{S}\cup\mathcal{S}_{\omega},\mathcal{P}\cup\mathcal{S}_{\Omega},A)$, obtained from $\mathbb {G} $ as follows: for all sort $X$, we not only add in the set $\mathcal{S}$ of sorts a new sort $X_{\omega}$, but we also add  a new $\varepsilon$-production $X_{\Omega} : X_{\omega} \rightarrow \varepsilon$ in the set $\mathcal {P}$ of productions; so we have: $\mathcal{S}_{\omega}=\{X_{\omega},~ X\in\mathcal {S}\}$ and $\mathcal{S}_{\Omega} = \{X_{\Omega} : X_{\omega} \rightarrow \varepsilon,~ X_{\omega} \in \mathcal{S}_{\omega}\}$.
\begin{figure}[ht!]
	\noindent
	\makebox[\textwidth]{\includegraphics[scale=0.45]{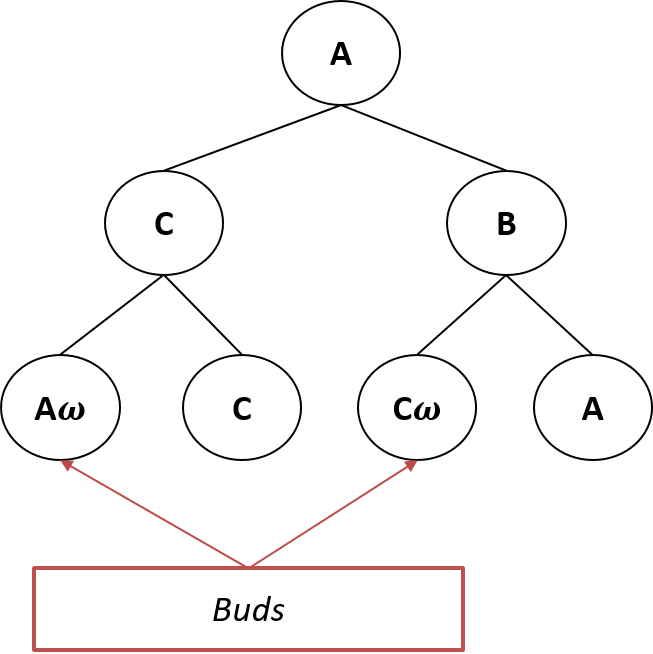}}
	\caption{Example of a tree that contains buds.}
	\label{chap2:fig:tree-with-bud}
\end{figure}

When we look at the productions of a grammar, we can notice that each sort is associated with a set of productions. From this point of view therefore, we can consider a grammar as an application
\[
	 gram : symb \rightarrow [(prod,~[symb])]
\]
which associates to each sort, a list of pairs formed by a production name and the list of sorts in the right hand side of this production. Such an observation suggests that a grammar can be interpreted as a (descending) tree automaton that can be used for recognition or for the generation of its instances.

\begin{definition}
\label{chap2:def:tree-automaton}
A (descending) \textbf{tree automaton} defined on $\Sigma$, is a quadruplet $\mathcal{A}=(\Sigma,Q,R,q_0)$ of a 
 set $\Sigma$ of symbols %avec arité  (signature)
; its elements are the nodes' labels of the trees to be generated (or recognised), a set $Q$ of states, a particular state $q_0 \in Q$ called initial state, and a finite set $R\subseteq  Q \times \Sigma \times Q^*$ of transitions.
\begin{itemize}
	\item An element of $R$ is denoted $q\rightarrow \left( \sigma, [ q_1,\cdots,q_n]\right)$ or in an equivalent way  $q\stackrel{\sigma}{\rightarrow}(q_{1},\ldots,q_{n})$: intuitively, $[ q_1,\cdots,q_n]$ is the list of states accessible from $q$ by crossing a transition labelled $\sigma$.
	\item If $q\stackrel{\sigma_1}{\rightarrow}\left( q_{1}^1, \cdots, q_{n_1}^1\right), \cdots, q\stackrel{\sigma_k}{\rightarrow}\left( q_{1}^k, \cdots, q_{n_k}^k\right)$ denotes the set of transitions associated to the state $q$, we denote \textbf{$next~q=[\left( \sigma_1,[ q_{1}^1, \cdots, q_{n_1}^1]\right),\cdots, \left( \sigma_k,[ q_{1}^k, \cdots, q_{n_k}^k]\right)]$}, the list that consists of pairs $\left( \sigma_i,[ q_{1}^i, \cdots, q_{n_i}^i]\right)$. 
	A transition of the form $q\rightarrow(\sigma,[\;])$, is called \textit{final transition} and a state possessing this transition is a \textit{final state}.
\end{itemize}
\end{definition}

One can interpret a grammar $\mathbb{G}=\left(\mathcal{S},\mathcal{P},A\right)$ as a (descending) tree automaton \cite{Comon} $\mathcal{A}=(\Sigma,Q,R,q_0)$ considering that:
\begin{itemize}
	\item[(1)] $\Sigma=\mathcal{P}$ is the type of labels of the nodes forming the tree to recognise; 
	\item[(2)] $Q=\mathcal{S}$ is the type of states and, 
	\item[(3)] $q\rightarrow \left( \sigma,[ q_1,\cdots,q_n]\right)$ is a transition of the automaton when the pair $\left(\sigma,[q_1,\cdots,q_n] \right)$ appears in the list $(gram~~ q)$\footnote{Reminder: $gram$ is the application obtained by abstraction of $\mathbb{G}$ and have as type : $gram : symb \rightarrow [(prod,~[symb])]$.}.
\end{itemize}
We note $\mathcal{A}_{\mathbb{G}} $ the tree automaton derived from $\mathbb{G}$.

To obtain the set $AST_\mathcal{A}$ of \textit{AST} generated by a tree automaton $\mathcal{A}$ from an initial state $q_0$, one must:
\begin{itemize}
	\item[(1)] Create a root node $r$, associate the initial state $q_0$ and add it to the set $AST_\mathcal{A}$ initially empty;
	\item[(2)] Remove from $AST_\mathcal{A}$ an AST $t$ under construction, i.e. with at least one leaf node $node$ unlabelled. Let $q$ be the state associated to $node$.
	For all transition $q\stackrel{\sigma}{\rightarrow} \left(q_1,\cdots,q_n\right)$   of $\mathcal{A}$, add in $AST_\mathcal{A}$ the trees $t'$ which are replicas of $t$ in which, the node $node$ has been substituted by a node $node'$ labelled $\sigma$ and possessing $n$ (unlabelled) sons, each associated to a (distinct) state $q_i$ of $[ q_1,\cdots,q_n]$; 
	\item[(3)] Iterate step (2) until he obtains trees with all the leaf nodes labelled (they are consequently associated to the final states of $\mathcal{A}$): these are \textit{AST}.
\end{itemize}
We note $\mathcal{A} \models t \triangleright q$ the fact that the tree automaton $\mathcal{A}$ accepts the tree $t$ from the initial state $q$, and
 $\mathscr{L}(\mathcal{A}, q)$ (tree language) the set of trees generated by the automaton $\mathcal{A}$ from the initial state $q$. 
Thus, $ \left(\mathcal{A}  \models t \triangleright q \right) \Leftrightarrow \left( t \in \mathscr{L}(\mathcal{A}, q) \right)$.

As for automata on words, one can define a synchronous product on tree automata to obtain the automaton recognising the intersection, the union, etc., of regular tree languages \cite{Comon}. We introduce below the definition of the synchronous product of $k$ tree automata whose adaptation will be used in the next section for the derivation of the consensual automaton. 

\begin{definition}
\label{chap2:def:synchronous-product}
\textbf{Synchronous product of $k$ automata:}\\
Let $\mathcal{A}_1=\left(\Sigma,Q^{(1)},R^{(1)},q_{0}^{(1)}\right), \ldots , \mathcal{A}_k=\left(\Sigma,Q^{(k)},R^{(k)},q_{0}^{(k)}\right) $ be $k$ tree automata. The synchronous product of these $k$ automata $\mathcal{A}_1 \otimes \cdots \otimes \mathcal{A}_k$ denoted $\otimes_{i=1}^{k} \mathcal{A}^{(i)}$, is the automaton $\mathcal{A}_{(sc)}=(\Sigma,Q,R,q_{0})$ defined as follows: 
\begin{itemize}
	\item[\textbf{(a)}] Its states are vectors of states : $Q =Q^{(1)}\times\cdots\times Q^{(k)}$; 
	\item[\textbf{(b)}] Its initial state is the vector formed by the initial states of the different automata : $q_{0}=\left(q_{0}^{(1)},\cdots, q_{0}^{(k)}\right)$; 
	\item[\textbf{(c)}] Its transitions are given by :\\
				$\left(q^{(1)}, \ldots, q^{(k)}\right)$ $\stackrel{a}{\rightarrow}\left(\left(q^{(1)}_1,\ldots,q^{(k)}_1\right),\ldots,\left(q^{(1)}_n,\ldots,q^{(k)}_n\right)\right)$ $\Leftrightarrow$\\ 
				   $\left( q^{(i)}\stackrel{a}{\rightarrow}\left(q^{(i)}_1,\ldots,q^{(i)}_n\right) \quad \forall i,~ 1\leq i\leq k \right)$
\end{itemize}
\end {definition}

\mySubSubSection{Notions of View, Projection, Reverse Projection and Merging}{}
\label{chap2:sec:view-projection-expansion-merging}
\noindent\textbf{\textit{View, associated projection and merging}}

The derivation tree giving the (global) representation of a structured document edited in a cooperative way, makes visible the set of grammatical symbols of the grammar that participated in its construction. As we mentioned in section \ref{chap2:sec:badouel-tchoupe-cooperative-editing} above, for reasons of confidentiality (accreditation degree), a co-author manipulating such a document will not necessarily have access to all of these grammatical symbols; only a subset of them can be considered relevant for him: it is his \textit{view}. A view $\mathcal{V}$ is then a subset of grammatical symbols ($\mathcal{V} \subseteq \mathcal{S}$). 
%Intuitivement,  il s'agit des sortes visibles par un co-auteur dans la représentation globale (arbre de dérivation) du document considéré. 

A partial replica of $t$ according to the view $\mathcal{V}$, is a partial copy of $t$ obtained by deleting in $t$, all the nodes labelled by symbols that are not in $\mathcal{V}$. 
Figure \ref{chap2:fig:partial-view} shows a document $t$ (center) and two partial replicas $t_{v_1}$ (left) and $t_{v_2}$ (right) obtained respectively by projections from the views $\mathcal{V}_1 = \{A,B\}$ and $\mathcal{V}_2 = \{A,C\}$. 
\begin{figure}[ht!]
	\noindent
	\makebox[\textwidth]{\includegraphics[scale=0.5]{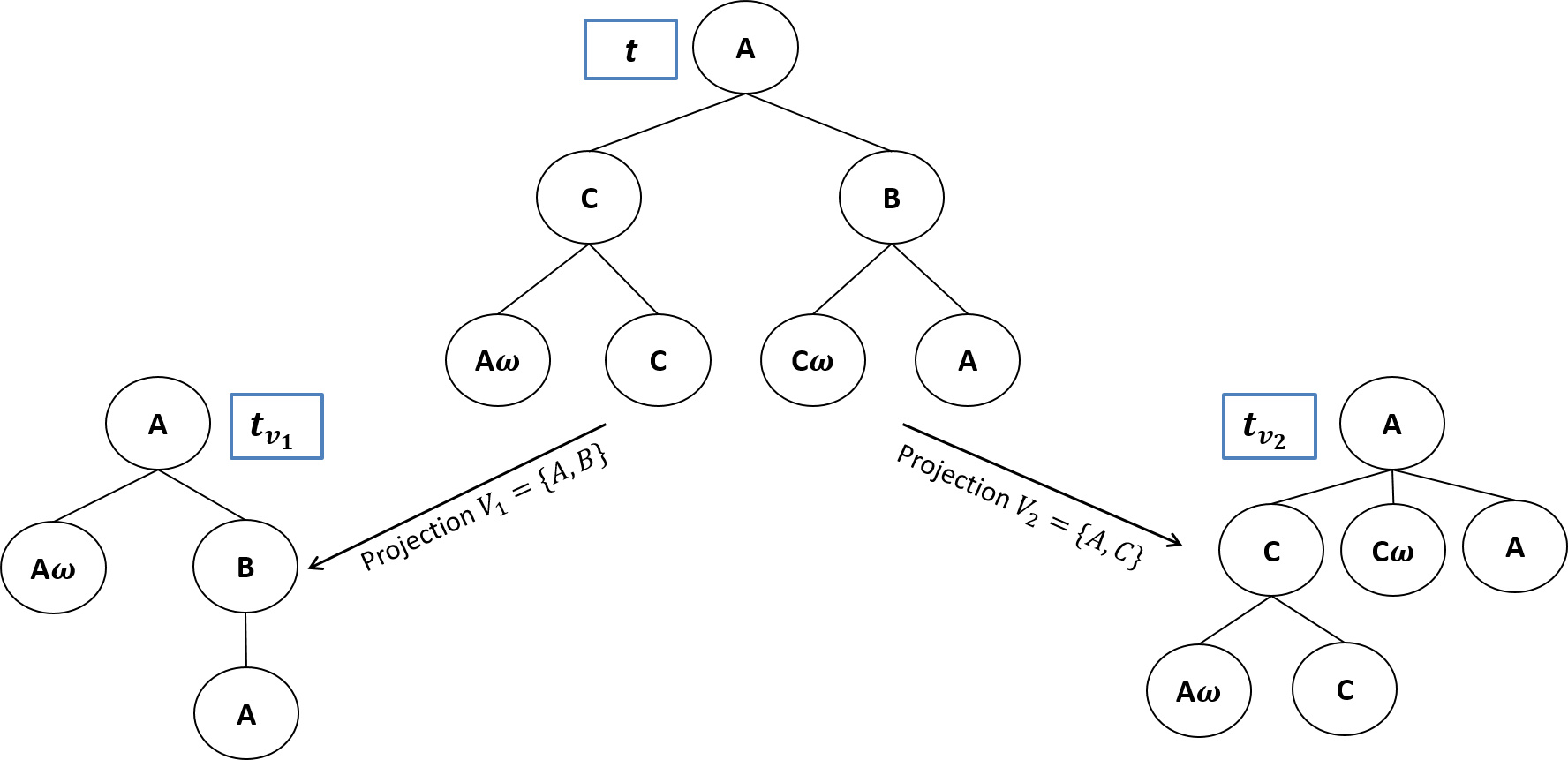}}
	\caption{Example of projections made on a document and partial replicas obtained.}
	\label{chap2:fig:partial-view}
\end{figure}

Practically, a partial replica is obtained via a \textit{projection} operation denoted $\pi$. We therefore denote $\pi_{\mathcal{V}}(t)= t_{\mathcal{V}}$ the fact that $t_{\mathcal{V}}$ is a partial replica obtained by projection of $t$ according to the view $\mathcal{V}$.

Note $t_{\mathcal{V}_i} \leq t_{\mathcal{V}_i}^{maj}$ the fact that, the document $t_{\mathcal{V}_i}^{maj}$ is an update of the document $t_{\mathcal{V}_i}$, i.e. $t_{\mathcal{V}_i}^{maj}$ is obtained from $t_{\mathcal{V}_i}$ by replacing some of its buds by trees.
In an asynchronous cooperative editing process, there are synchronisation points\footnote{Recall that a synchronisation point can be defined statically or triggered by a co-author as soon as certain properties are satisfied.} in which, one tries to merge all contributions $t_{\mathcal{V}_i}^{maj}$ of the various co-authors to obtain a single comprehensive document $t_f$\footnote{It may happen that the edition must be continued after the merging (this is the case if there are still buds in the merged document): one must redistribute to each of the $n$ co-authors a (partial) replica $t_{\mathcal{V}_i}$ of $t_f$ such that  $t_{\mathcal{V}_i} = \pi_{\mathcal{V}_i}(t_f)$ for the continuation of the editing process.}. A merging algorithm that does not incorporate conflict management and relies on a solution to the \textit{reverse projection} problem was given by Badouel and Tchoup\'e.

~

\noindent\textbf{\textit{Partial replica and reverse projection (expansion)}}

The \textit{reverse projection} (also called the \textit{expansion}) of an updated partial replica $t_{\mathcal{V}_i}^{maj}$ relatively to a given grammar $\mathbb{G}=\left(\mathcal{S},\mathcal{P},A\right)$, is the set $T_{i\mathcal{S}}^{maj}$ of documents conform to $\mathbb{G}$, that admit $t_{\mathcal{V}_i}^{maj}$ as partial replica according to ${\mathcal{V}_i}$:
\begin{itemize}
	\item[] $ T_{i\mathcal{S}}^{maj} = \left\{t_{i\mathcal{S}}^{maj} \therefore \mathbb{G}~ | ~ \pi_{\mathcal{V}_i}\left(t_{i\mathcal{S}}^{maj}\right) = t_{\mathcal{V}_i}^{maj} \right\}
	$
\end{itemize} 

A solution to the problem of evaluating the expansion of a given partial replica using tree automata, was proposed by Badouel and Tchoup\'e; in that solution, productions of the grammar $\mathbb{G}$ are used, to bind to a view $\mathcal{V}_i \subseteq \mathcal{S}$ a tree automaton $\mathcal{A}^{(i)}$ such as, the trees it recognises from an initial state built from $t_{\mathcal{V}_i}^{maj}$, are exactly those having this partial replica as projection according to the view $\mathcal{V}_i$:
$ T_{i\mathcal{S}}^{maj} = \mathscr{L}\left(\mathcal{A}^{(i)},~q_{t_{\mathcal{V}_i}}\right) $. 
Practically, they have considered that a state $q$ of the automaton $\mathcal{A}^{(i)}$ is a pair $\left(Tag ~X, \;ts\right)$ where $X$ is a grammatical symbol, $ts$ is a forest (tree set), and \textit{Tag} is a label that is either \textit{Open} or \textit{Close}, and indicates whether the concerned state $q$ can be used to generate a \textit{closed} node or a \textit{bud}. The states of $\mathcal{A}^{(i)}$ are typed: a state of the form $\left(Tag ~X, \;ts\right)$ is of type $X$. We also have a function named \textit{typeState} which, when applied to a state, returns its type\footnote{ $typeState :: state\rightarrow symb$ \\
 $.~~~typeState ~\left(Open ~X, \;ts\right) = X$\\
 $.~~~typeState ~\left(Close ~X, \;ts\right) = X$
	}.
A transition from one state $q$, is of one of the forms $\left(Close~X, \;ts \right) \rightarrow \left(p, \;[q_1, \ldots, q_n]\right)$ or $\left(Open ~X, \;[\;]\right) \rightarrow \left(X_\omega, \;[\;]\right)$. 
A transition of the form $\left(Close~X, \;ts \right) \rightarrow \left(p, \;[q_1, \ldots, q_n]\right)$ is used to generate AST of type $X$ (whose root is labelled by a \textit{X-production}) admitting ''$ts$'' as projection according to the view ${\mathcal{V}_i}$ if $X$ does not belong to ${\mathcal{V}_i}$, and ''$x [ ts ]$'' otherwise.
Similarly, a transition of the form $\left(Open ~X, \;[\;]\right) \rightarrow \left(X_\omega, \;[\;]\right)$ is used to generate a single AST reduced to a bud of type $X$. 

The interested reader may consult \cite{badouelTchoupeCmcs} for a more detailed description of the process of associating a tree automaton with a view and the section \ref{chap2:sec:consensus-illustration} for an illustration.

\mySubSection{Reconciliation by Consensus}{}
\label{chap2:sec:consensus-reconciliation}

\mySubSubSection{Issue and Principle of the Solution of Reconciliation by Consensus}{}
\label{chap2:sec:consensus-reconciliation-issue-principle}
There are generally two distinct phases when synchronising replicas of a document \cite{balasubramaniam98what}: the \textit{updates detection} phase, which consists of recognising the different replica nodes (locations) where updates have been made since the last synchronisation, and the  \textit{propagation} phase, which consists in combining the updates made on the various replicas to produce a new synchronised state (document) for each replica.
In an asynchronous cooperative editing workflow of several partial replicas of a document, when you reach a synchronisation point, you can end up with unmergeable replicas in their entirety as they contain not compatible updates\footnote{This is particularly the case if there is at least one node of the global document accessible by more than one co-author and edited by at least two of them using different productions.}: they must be reconciled. This can be done by questioning (cancelling) some local editing actions in order to resolve conflicts and result in a coherent global version said of consensus.

Studies on reconciling a document versions are based on heuristics \cite{tomMens} as there is no general solution to this problem.
In our case, since all editing actions are reversible\footnote{Reminder: the editing actions made on a partial replica may be cancelled as long as they do not have been incorporated into the global document.} and it is easy to locate conflicts when trying to merge the partial replicas (see section \ref{chap2:sec:consensus-calculation}), we have a canonical method to eliminate conflicts: when merging, we replace any node of the global document whose replicas are in conflict, by a bud. Thus, we prune at the nodes where a conflict appears, replacing the corresponding subtree with a bud of the appropriate type, indicating that this part of the document is not yet edited: the documents obtained are called consensus. These are the maximum prefixes without conflicts of the fusion of the documents resulting from the different expansions of the various updated partial replicas. For example, in figure \ref{chap2:fig:consensus-workflow}, the parts highlighted (blue backgrounds) in trees (f) and (g) are in conflict;  they are replaced in the consensus tree (h) by a bud of type $C$ (node labelled $C_{\omega}$).

The problem of the consensual merging of $k$ updated partial replica whose global model is given by a grammar $\mathbb{G}=\left(\mathcal{S},\mathcal{P},A\right)$ can therefore be stated as follows :\\
\textbf{Problem of the consensual merging}: given $k$ views $(\mathcal{V}_i)_{1 \leq i\leq k}$ and $k$ partial replicas $(t_{\mathcal{V}_i}^{maj})_{1 \leq i\leq k}$, merge consensually the family $(t_{\mathcal{V}_i}^{maj})_{1 \leq i\leq k}$ is to find the largest documents $t^{maj}_{\mathcal{S}}$ conforming to $\mathbb{G}$ such that, for any document $t$ conforming to $\mathbb{G}$ and admitting $t_{\mathcal{V}_i}^{maj}$ as projection along the view ${\mathcal{V}_i}$, $t^{maj}_{\mathcal{S}}$ and $t$ are eventually updates each for other. i.e.:
\begin{center}
$ 
\left( t^{maj}_{\mathcal{S}} \in \otimes_{i=1}^{k} t_i,  ~t_i \in T_{i\mathcal{S}}^{maj} \right) \Leftrightarrow 
\left\{ \begin{array}{l}
	i) ~~	\forall i, \, 1\leq i \leq k, ~ \forall t \therefore \mathbb{G} \mbox{ such that }~ \pi_{\mathcal{V}_i}(t) =   t_{\mathcal{V}_i}^{maj} ,~t^{maj}_{\mathcal{S}} \cong t. \\
	ii) ~~\nexists t' \leq t^{maj}_{\mathcal{S}} \mbox{ such that } t' \in  \otimes_{i=1}^{k} t_i,  ~t_i \in T_{i\mathcal{S}}^{maj} 
\end{array}\right.
$\footnote{The binary relation $\cong$ when it exists between two trees $t_1$ and $t_2$ expresses the fact that they are possibly updates each for other. This relationship is more explicitly explained in definition \ref{chap2:def:relation-cong}.}
\end{center}

The solution that we propose to this problem stems from an instrumentation of that proposed for the expansion (section \ref{chap2:sec:view-projection-expansion-merging}). 
Indeed, we use an associative and commutative operator noted $\otimes$, to synchronise the tree automata $\mathcal{A}^{(i)}$ constructed to carry out the various expansions, in order to generate the tree automaton of consensual merging. Noting $\mathcal{A}_{(sc)}$ this automaton, the documents of the consensus are the trees of the language generated by the automaton $\mathcal{A}_{(sc)}$ from an initial state built from the vector made of the initial states of the automata $(\mathcal{A}^{(i)})$: 
$\mathscr{L}(\mathcal{A}_{(sc)},~(q_{t_{\mathcal{V}_i}^{maj}})) = consensus  \{\mathscr{L}(\mathcal{A}^{(i)},~q_{t_{\mathcal{V}_i}^{maj}})\}$.	
$\mathcal{A}_{(sc)}$ is obtained by proceeding as follows:
\begin{itemize}
	\item[(1)] For each view $\mathcal{V}_i$, build the automaton $\mathcal{A}^{(i)}$ which will carry out the expansion of the partial replica $t_{\mathcal{V}_i}^{maj}$ as previously indicated (sec. \ref{chap2:sec:view-projection-expansion-merging}): $  \mathscr{L}\left(\mathcal{A}^{(i)},~q_{t_{\mathcal{V}_i}^{maj}}\right) = T_{i\mathcal{S}}^{maj}$;
	\item[(2)] Using the operator $\otimes$, compute the automaton generating the consensus language $\mathcal{A}_{(sc)}=\otimes_{i=1}^{k}\mathcal{A}^{(i)}$. 
\end{itemize}

\mySubSubSection{Consensus Calculation}{}
\label{chap2:sec:consensus-calculation}
Before presenting the consensus calculation algorithm, let us specify using the concepts introduced in section \ref{chap2:sec:structured-document-edition-conformity}, the notion of (two) documents in conflict.
Let $t_1,t_2:\mathbb{N}^*\rightarrow \mathbf{A}$ be two trees (documents)  with respectively $\mathit{Dom}(t_1)$ and $\mathit{Dom}(t_2)$ their domains. We say that $t_1$ and $t_2$ admit a consensus, and we note $t_1\curlywedgeuparrow t_2$, if their roots are of the same type\footnote{Trees we handle are AST and therefore, the nodes are labelled by productions names. Any node labelled by an \textit{X-production} is said of \textit{type} $X$. Furthermore, there is a function \textit{typeNode} such that \textit{typeNode(t(w))} returns the type of the node located at the address $w$ in $t$.}, i.e. $\left( t_1\curlywedgeuparrow t_2 \right) \Leftrightarrow \left(typeNode(t_{1}(\epsilon)) = typeNode(t_{2}(\epsilon)) \right)$ \footnote{It may then be noted that two documents (AST) admit no consensus if their roots are of different types. However, for applications that interest us, namely structured editing, since the editions are done from the root (which is always of the type of the axiom) to the leaves using productions, the documents we manipulate always admit at least a consensus.}. It is then say that, $t_1$ and $t_2$ are in conflict, and it is noted $t_1 \curlyveeuparrow t_2$, when they admit a consensus but are not mergeable in their entirety. 
Intuitively, two documents $t_1$ and $t_2$ (not reduced to buds) are not fully mergeable (see fig. \ref{chap2:fig:tree-in-conflict}), if there exists an address $w\in \mathit{Dom}(t_1) \cap \mathit{Dom}(t_2)$ such that, if we note $n_1$ (resp. $n_2$) the node located to address $w$ in $t_1$ (resp. in $t_2$), then, $n_1$ and $n_2$ which are not buds, are of the same type but have different labels. i.e. 
\begin{center}
$
\left( t_1\curlyveeuparrow t_2 ~with~t_1(\epsilon)\neq X_\omega,\, t_2(\epsilon)\neq X_\omega\right) \Leftrightarrow \left\{\begin{array}{l}
													\left(t_1\curlywedgeuparrow t_2\right)\\
													\mbox{and}\\
													\left(\begin{array}{l}\exists w \in \mathit{Dom}(t_1)\cup \mathit{Dom}(t_2), ~t_{1}(w) \neq t_{2}(w) \neq X_\omega ~ and\\
													typeNode(t_{1}(w)) = typeNode(t_{2}(w))
													\end{array}
													\right) 
											 \end{array}
								\right.
$
\end{center}
\begin{figure}[ht!]
	\noindent
	\makebox[\textwidth]{\includegraphics[scale=0.37]{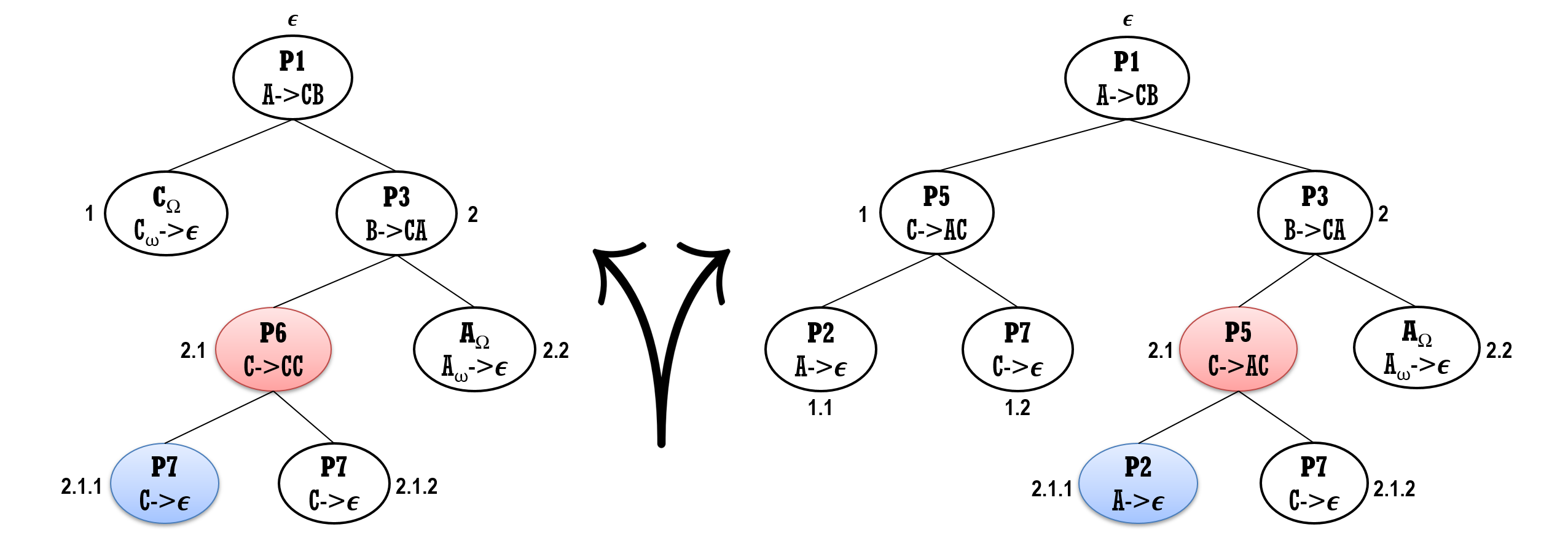}}
	\caption{Example of documents in conflict.}
	\label{chap2:fig:tree-in-conflict}
\end{figure}

Figure \ref{chap2:fig:tree-in-conflict} shows two conflicting documents. In fact, at address 2.1 we have two nodes of the same type ("C") but edited with different C-productions:  production $C \rightarrow C ~C $ in the first document, and production $C \rightarrow A~C $ in the second one.

~

\noindent\textbf{\textit{Consensus among multiple (two) documents}}

Let $t_1,t_2:\mathbb{N}^*\rightarrow \mathbf{A}$ be two trees (documents) in conflict with respectively $\mathit{Dom}(t_1)$ and $\mathit{Dom}(t_2)$ their domains.
The consensual tree $t_c:\mathbb{N}^*\rightarrow \mathbf{A}$ derived from $t_1$ and $t_2$ ($t_c = t_1 \otimes t_2$) has as domain the union of domains of the two trees in which, we subtract elements belonging to domains of subtrees derived from the conflicting nodes. In fact, we prune at the nodes in conflict and they appear in the consensus tree as a (unique) bud.
So, 
\begin{center}
$\forall w \in \mathit{Dom}(t_c)$, 
			$ t_c(w) =   
		    \left\{ \begin{array}{lll}
										t_1(w) & \mathit{if} & typeNode(t_1(w)) = typeNode(t_2(w)) ~and ~ t_1(w)=t_2(w)\\
										t_1(w) & \mathit{if} & typeNode(t_1(w)) = typeNode(t_2(w)) ~and ~ t_2 (w) = X_\omega \\
										t_2 (w) & \mathit{if} & typeNode(t_1(w)) = typeNode(t_2(w)) ~and ~ t_1(w) = X_\omega \\
										t_1(w) & \mathit{if} & w \notin Dom(t_2) ~and~ \exists u,\,v \in \mathbb{N}^* ~ tq ~ w=u.v, \\
														&							& ~ t_2 (u) = X_\omega ~and ~typeNode(t_1(u)) = typeNode(t_2(u))  \\
										t_2(w) & \mathit{if} & w \notin Dom(t_1) ~and~ \exists u,\,v \in \mathbb{N}^* ~ tq ~ w=u.v, \\
														&							&~ t_1(u) = X_\omega ~and~typeNode(t_1(u)) = typeNode(t_2(u))  \\
										X_\omega & \mathit{if} & typeNode(t_1(w)) = typeNode(t_2(w)) ~and ~ t_1(w) \neq X_\omega\\
															&							&~and~ t_2 (w) \neq X_\omega ~and~ t_1(w) \neq t_2(w)  										
									\end{array}
								\right.
  $  
\end{center}

Figure \ref{chap2:fig:consensus-documents} presents the document resulting from the consensual merging of the documents in figure \ref{chap2:fig:tree-in-conflict}. We have prune at the level of nodes 2.1  in both documents which are in conflict.
\begin{figure}[ht!]
	\noindent
	\makebox[\textwidth]{\includegraphics[scale=0.4]{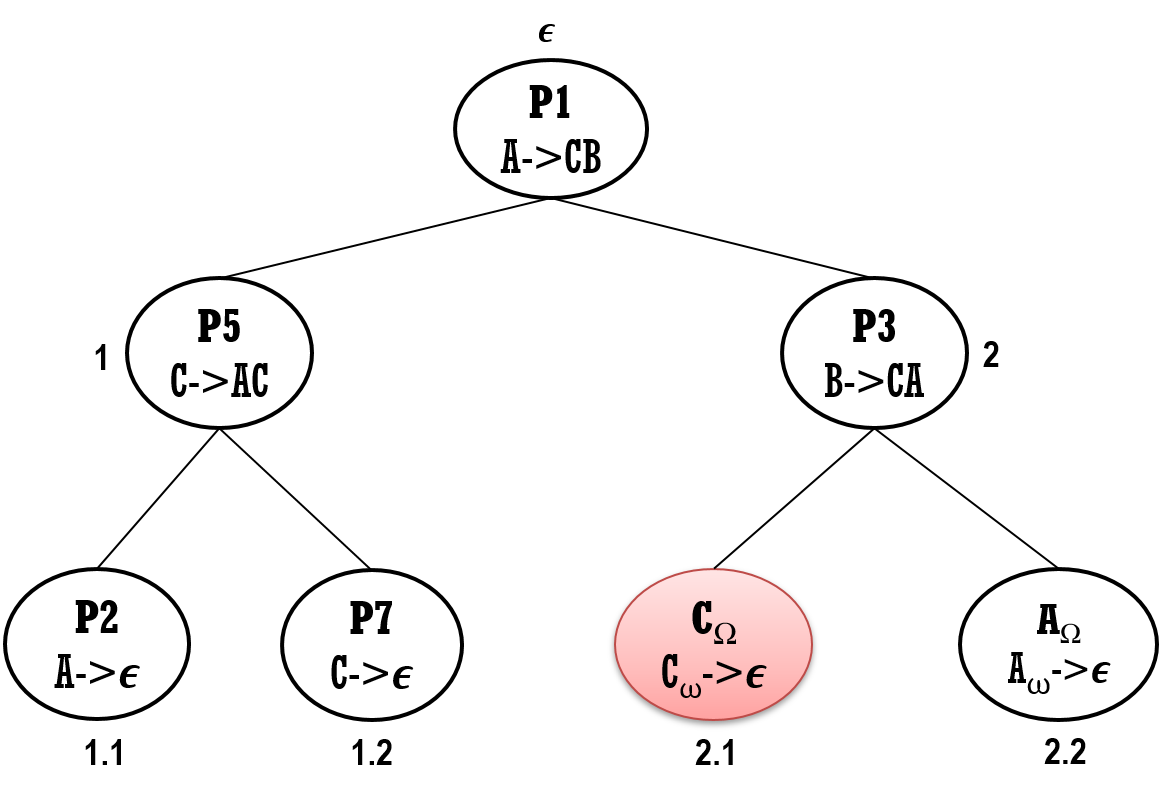}}
	\caption{Document resulting from the consensual merging of the documents in figure \ref{chap2:fig:tree-in-conflict}.}
	\label{chap2:fig:consensus-documents}
\end{figure}

When $t_c = t_1 \otimes t_2 $, there may be nodes of $t_1$ and those of $t_2$ which are updates of the nodes of $t_c$: it is said in this case that $t_1$ (resp. $t_2$) and $t_c$ are \textit{updates each for other}. 

\begin{definition}
\label{chap2:def:relation-cong}
Let $t_1, ~t_2$ two documents that are not in conflict. It will be said that \textit{they are updates each for other} and it is noted $t_1 \cong t_2$, if there exists at least two addresses $w, ~w'$ of their respective domains such that $t_1(w)$ (resp. $t_2(w')$) is a bud and $t_2(w)$ (resp. $t_1(w')$) is not.
\end{definition}

~

\noindent\textbf{\textit{Construction of the consensus automaton}}

Consideration of documents with buds requires the readjustment of some models. For example, in the following, we will handle \textit{tree automata with exit states} instead of tree automata introduced in definition \ref{chap2:def:tree-automaton}. Intuitively, a state $q$ of an automaton is called an \textit{exit state} if there is a unique transition $ q \rightarrow (X_\omega, [\, ])$ associated to it for generating a tree reduced to a bud of type $X \in \Sigma$: $q$ is then of the form \textit{(Open X, [])}.

\begin{definition}
A \textbf{tree automata with exit states} $\mathcal{A}$ is a quintuplet $(\Sigma,Q,R,q_0, \mathit{exit})$ where $\Sigma,Q,R,q_0$ designate the same objects as those introduced in definition \ref{chap2:def:tree-automaton}, and $\mathit{exit}$ is a predicate defined on the states ($\mathit{exit}: Q \rightarrow Bool$). Any state $q$ of $Q$ for which $\mathit{exit}~q$ is $\mathit{True}$ is an exit state.
\end{definition}

A type for automata with exit states can be defined in Haskell \cite{Antony} by the listing of algorithm \ref{chap2:algo:automata-type} in which, \textit{state} and \textit{prod} are type variables respectively representing the type of the automaton states and the type of labels of the AST to generate.
\begin{algorithm}
\small
\caption{A Haskell type for automata with exit states.}
\label{chap2:algo:automata-type}
\begin{Verbatim} [numbers=left]
data Automata prod state = Auto{
          exit:: state -> Bool, 
          next :: state -> [(prod,[state])]
          }
\end{Verbatim}
\end{algorithm}  

In section \ref{chap2:sec:consensus-calculation} (\textit{consensus among multiple (two) documents}) above, we said that, \textit{when two nodes are in conflict, "they appear in the consensus tree as a (unique) bud}". From the point of view of automata synchronisation, the concept of "nodes in conflict" is the counterpart of the concept of "states in conflict" (as we specify below), and the above extract is reflected in the automata context by:  "\textit{when two \textit{state} are in conflict, they appear in the \textit{consensus automaton} in the form of a (single) exit state}". 
Thus, if we consider two states of the same type $q_0^1$ and $q^2_0$ (which are not exit states) of two automata $auto_1$ and $auto_2$ with associated transitions families respectively $q_0^1 \rightarrow [(a^1_{1}, qs_1), \dots, (a^1_{n_1}, qs_{n_1})]$ and $q^2_0 \rightarrow [(a^2_{1}, qs'_1), \dots, (a^1_{n_2}, qs'_{n_2})]$, we say that the states $q^1_0$ and $q^2_0$ are in conflict (and we note $q^1_0 \curlyveeuparrow q^2_0$) if there is no transition starting from each of them and with the same label, i.e.  
\[\left(q^1_0 \curlyveeuparrow q^2_0\right) \Leftrightarrow \left(\nexists a^3, ~~~\left(a^3, qs\right) \in \left\{next~ q_o^1\right\}, ~\left(a^3, qs'\right) \in \left\{next~ q_o^2\right\}, ~|qs|=|qs'|\right)\]

This can be coded in Haskell by the function \Verb|isInConflicts| of algorithm \ref{chap2:algo:automata-state-in-conflict}.

\begin{algorithm}
\small
\caption{A Haskell function to check if two states of a given automaton are in conflict.}
\label{chap2:algo:automata-state-in-conflict}
\begin{Verbatim} [numbers=left]
isInConflicts state1@(tagsymb1, ts1) state2@(tagsymb2, ts2) = 
            null [a1 | (a1,states1) <- next auto1 state1,
                       (a2,states2) <- next auto2 state2,
                        a1==a2,
                       (length states1)==(length states2)]
\end{Verbatim}
\end{algorithm}

If $X$ is the type of two states $q$ and $q'$ in conflict, they admit a single consensual state $q_{\omega}=(Open X, [\,])$ such as $next~ q_{\omega} = [(X_{\omega}, [\,])]$. It is therefore obvious that two given automata admit a consensual automaton when their initials states are of the same type. The function defined in algorithm \ref{chap2:algo:automata-admit-consensus} performs this test.
\begin{algorithm}
\small
\caption{A Haskell function used to check if two given automata admit a consensus.}
\label{chap2:algo:automata-admit-consensus}
\begin{Verbatim} [numbers=left]
haveConsensus q0 q0' = (typeState q0) == (typeState q0')
\end{Verbatim}
\end{algorithm}

The operator $\otimes$ used to calculate the synchronised consensual automaton $\mathcal{A}_{(sc)}=\otimes_{i}^{k}\mathcal{A}^{(i)}$ is a relaxation of the operator used for calculating the automata product presented in definition \ref{chap2:def:synchronous-product}. $\mathcal{A}_{(sc)} = (\Sigma, ~Q, ~R, ~q_0, ~exit)$  is an automaton with exit states and is constructed as follows:

\begin{itemize}
	\item Its states are vectors of states : $Q =Q^{(1)}\times\cdots\times Q^{(k)}$; 
	\item its initial state is formed by the vector of initial states of different automata : $q_{0}=\left(q_{0}^{(1)},\cdots, q_{0}^{(k)}\right)$;
	\item For the \textit{exit} function, it is considered that when a given automaton $\mathcal{A}^{(j)}$ reaches an \textit{exit state}\footnote{The corresponding node in the reverse projection of the document is a bud and reflects the fact that, the corresponding author did not edit it. In the case that this node is shared with another co-author who edited it  in its (partial) replica, it is the edition made by the latter that will be retained when merging.}, it no longer contributes to the behaviour, but is not opposed to the synchronisation of the other automata: it is said to be \textit{asleep} (see algorithm \ref{listingConsensus}, lines 17, 19 and 24). So, a state $q = (q^1, \cdots, q^k)$ is an exit state if: 
	\begin{itemize}
		\item[(a)] all composite states $q^i$ are asleep (see algorithm \ref{listingConsensus}, line 6) or 
		\item[(b)] there exist any two states $q^{i}$ and $q^{j}, ~i \neq j$, components of $q$ that are in conflict (see algorithm \ref{listingConsensus}, line 12)  \\
		$\left( exit~\left(q^{(1)},\ldots,q^{(k)}\right)\right)$ $\Leftrightarrow$\\
		$\left( \left( exit~q^{(i)}, \forall i \in \{1\ldots k\} \right) \mbox{or} \left( \exists i, j, ~~i\neq j,  ~ q^{(i)} \curlyveeuparrow  q^{(j)}\right) \right)$;
	\end{itemize}
	
	\item Its transitions are given by:
	\begin{itemize} 
		\item[\textbf{(a)}] If $exit~ q$  then $q \rightarrow (X_\omega, [\,])$ is the unique transition of $q$; $X$ is the type of $q$.
		\item[\textbf{(b)}] Else  $\left(q^{(1)}, \ldots, q^{(k)}\right) \stackrel{a}{\rightarrow}\left(\left(q^{(1)}_1,\ldots,q^{(k)}_1\right),\ldots,\left(q^{(1)}_n,\ldots,q^{(k)}_n\right)\right) \Leftrightarrow$  +				$\forall i,~ 1\leq i\leq k$ 
		\begin{itemize} 
			\item[\textbf{(b1)}] $exit~q^{(i)}$ and $\left( q^{(i)}_j=\left(Open~X,~ [\;]\right), \forall j, ~1\leq j\leq n \right)~ $     \textit{ /*} $q^{(i)}$ \textit{is asleep */}, else\\
			\item[\textbf{(b2)}] ~  $q^{(i)}\stackrel{a}{\rightarrow}\left(q^{(i)}_1,\ldots,q^{(i)}_n \right)$.
		\end{itemize}
	\end{itemize}		
\end{itemize}
\textbf{(a)} reflects the fact that, if a state $q$ is an exit one, we associate a single transition for generating a tree reduced to a bud of the type of $q$ (see algorithm \ref{listingConsensus}, line 12).\\ 
With \textbf{(b1)} we say that, if the component $q^{(i)}$ of $q$ is an exit state, then for all composite state $\left(q^{(1)}_j,\ldots,q^{(k)}_j\right)$, ($1\leq j\leq n$) appearing in the right hand side of the transition \textbf{(b)}, the $i^{th}$ component should be asleep. Since it must not prevent other non-asleep states to synchronise, it must be of the form $(Open~X,~ [\;])$ where $X$ is the type of the other states $q^{(l)}_j$ (yet to be synchronised) belonging to $(q^{(1)}_j,\ldots,q^{(k)}_j)$ (see function \textit{fwdSlpState} defined in algorithm \ref{listingConsensus} line 24, and used in lines 17 and 19). Finally, with \textbf{(b2)} we stipulate that if  $q^{(i)}\stackrel{a}{\rightarrow}\left(q^{(i)}_1,\ldots,q^{(i)}_n \right)$ is a transition of the automaton $\mathcal{A}^{i} $, then for all composite state $\left(q^{(1)}_j,\ldots,q^{(k)}_j\right)$, ($1\leq j\leq n$) appearing in the right part of the transition \textbf{(b)} above, the $i^{th}$ component is $q^{(i)}_j$ (see algorithm \ref{listingConsensus}, lines 13 to 16).

\begin{algorithm}
\small
\caption{Consensus Listing.}
\label{listingConsensus}
\begin{Verbatim} [numbers=left]
autoconsens::(Eq p, Eq x) =>(x -> p) -> Automata p (Tag x, [st1])
     -> Automata p (Tag x, [st2]) 
     -> Automata p ((Tag x, [st1]), (Tag x, [st2]))
autoconsens symb2prod auto1  auto2 = Auto exit_ next_ where 
   exit_ (state1, state2) = case haveConsensus state1 state2 of
      True  -> (exit auto1 state1) && (exit auto2 state2) 
      False -> True
   next_ (state1, state2) = case haveConsensus state1 state2 of
      False -> []
      True  -> case (exit auto1 state1, exit auto2 state2) of
         (False, False) -> case (isInConflicts state1 state2) of
            True  -> [(symb2prod (typeState state1), [])]   
            False -> [(a1, zip states1 states2) | 
                      (a1, states1) <- next auto1 state1,
                      (a2, states2) <- next auto2 state2,
                      a1 == a2, (length states1) == (length states2)]
         (False, True)  -> [(a, zip states1 (fwdSlpState  states1)) |
                            (a, states1) <- next auto1 state1]
         (True, False)  -> [(a, zip (fwdSlpState  states2) states2) |
                            (a, states2) <- next auto2 state2]
         (True, True)   -> [(a1, []) | (a1,[]) <- next auto1 state1,
                            (a2, []) <- next auto2 state2, a1 == a2]
   where 
      fwdSlpState states = map (\state -> (Open (typeState state), [])) states
\end{Verbatim}
\end{algorithm}

\begin{proposition}
The tree automaton $\mathcal{A}=\otimes_{i=1}^{k}\mathcal{A}^{(i)}$  recognises/generates from the initial state $q_0 = (q_{01}, \ldots , q_{0k})$, all the trees from the consensual merging of trees recognised/generated by each automaton $\mathcal{A}^{(i)}$ from the initial state $q_{0i}$. Moreover, these trees are the biggest prefixes without conflicts of merged trees.
\[
\left( \otimes_{i=1}^{k}\mathcal{A}^{(i)} \models t \triangleright q_0\right) \Leftrightarrow 
  \left\{\begin{array}{l}
	i) ~~ \forall i,\quad\exists t_i\qquad  \mathcal{A}^{(i)}\models t_i \triangleright q_{0i}~\mbox{ and } ~t_i\cong t\\
	ii) ~~\forall t' \mbox{ prefix of } t, \quad \neg \left( \otimes_{i=1}^{k}\mathcal{A}^{(i)} \models t' \triangleright q_0\right)
	\end{array}\right.
\]
\end{proposition}

\begin{proof}
A tree $t$ is recognised by the synchronised automaton $\otimes_{i=1}^{k}\mathcal{A}^{(i)}$ if and only if, one can label each of its nodes by a state of the automaton in accordance with what is specified by the transitions of the automaton. 
Moreover, all the leaf nodes of $t$ must be labelled by using \textit{final transitions}; in our case, they are of the form $q\rightarrow (p, [\;])$.
This means that, if a node whose initial label is $a$, is labelled by the state $q$, and if it admits $n$ successors respectively labelled by $q_1,\ldots,q_n$, then $q\stackrel{a}{\rightarrow}(q_1,\ldots,q_n)$ must be a transition of the automaton. As the automaton is deterministic\footnote{Automata $\mathcal{A}^{(i)}$ being deterministic (see proposition 3.3.3 of \cite{theseTchoupe}), $\otimes_{i=1}^{k}\mathcal{A}^{(i)}$ is deterministic as synchronous product of deterministic automata.}, this labelling is unique (including the initial state attached to the root of the tree). 
By focusing our attention both on the state $q$ labelling a node, and its $i^{th}$ component $q_i$, on each of the branches of $t$, 
\begin{itemize}
	\item[(1)] we cut as soon as we reach an exit state in relation to the automaton $\mathcal{A}^{(i)} $ (i.e. $q_i$ is an exit state), or,
	\item[(2)] if $q$ is an exit state (in this case we are handling a leaf) and $q_i$ is not an exit state relatively to $\mathcal{A}^{(i)}$ (in this case, $q_i$ was in conflict with at least one other component $q_j$ of $q$), we replace that node with any subtree $t'_i$ that can be generated by $\mathcal{A}^{(i)}$ from the state $q_i$.
\end{itemize}
So,
\[
 \left(\otimes_{i=1}^{k}\mathcal{A}^{(i)} \models t \triangleright q_0 \right) \Rightarrow 
\left(\forall i,\;\exists t_i\quad \mathcal{A}^{(i)}\models t_i\triangleright q_{0i}\mbox{ and }t_i\cong t\right)
\]
since a state of $\mathcal{A}$ is an exit one if and only if, each of its components is an exit state (in the $\mathcal{A}^{i}$) or, if at least two of its components are in conflict. \\
Conversely, suppose $ \mathcal{A}^{(i)}\models t_i \triangleright q_{0i}$, by definition of the synchronised consensual automaton, we have $ \otimes_{i=1}^{k}\mathcal{A}^{(i)}\models \otimes_{i=1}^{k}t_i\triangleright (q_{01}, \ldots , q_{0k})$. So overall,
\[
L\left(\otimes_{i=1}^{k}\mathcal{A}^{(i)},~q_0\right) =\left\{ \otimes_{i=1}^{k}t_i | \quad  \mathcal{A}^{(i)}\models t_i\triangleright q_{0i} \right\}
\]

Suppose that $t$ is recognised by $\otimes_{i=1}^{k}\mathcal{A}^{(i)}$; thus, there is a labelling of its nodes with the states of $\otimes_{i=1}^{k}\mathcal{A}^{(i)}$, and as such the transitions used for the labelling of its leaves are final. 
Let $t_p$ be a prefix of $t$. Let us show that $t_p$ is not recognised by $\otimes_{i=1}^{k}\mathcal{A}^{(i)}$ using the fact that, any labelling of $t_p$ has at least one leaf node labelled by a state that is not associated to a final transition.
The labels associated to the nodes of $t_p$ are the same as those associated to the nodes of same addresses in $t$ because, $t_p$ is a prefix of $t$ and $\otimes_{i=1}^{k}\mathcal{A}^{(i)}$ is deterministic. $t_p$ is obtained from $t$ by pruning some subtrees of $t$; so, naturally, it has a (non-zero) number of leaf nodes that can be developed to obtain $t$. Let us choose such a node and call it $n_f$. Suppose that it is labelled $p$ and was associated with a state $q_f= (q_{1}, \ldots q_{k})$ when labelling $t$.
The \textit{p-transition} that allowed to recognise $n_f$ is not a final transition. Indeed, $n_f$ has in $t$, |$p$| sons whose labels can be supposed to be the states $q_{f_1}, \ldots, q_{f_{|p|}}$. This means that, according to the labelling process and considering the fact that $\otimes_{i=1}^{k}\mathcal{A}^{(i)}$ is deterministic, the single transition used for the labelling of $n_f$ and of its |$p$| sons is $q_f\stackrel{p}{\rightarrow}\left(q_{f_1},\ldots,q_{f_{|p|}}\right)$ which, is not a final transition. Therefore, $t_p$ is not recognised by $\otimes_{i=1}^{k}\mathcal{A}^{(i)}$.
\end{proof}

\mySubSubSection{Illustration}{}
\label{chap2:sec:consensus-illustration}
\begin{figure}[ht!]
	\noindent
	\makebox[\textwidth]{\includegraphics[scale=0.3]{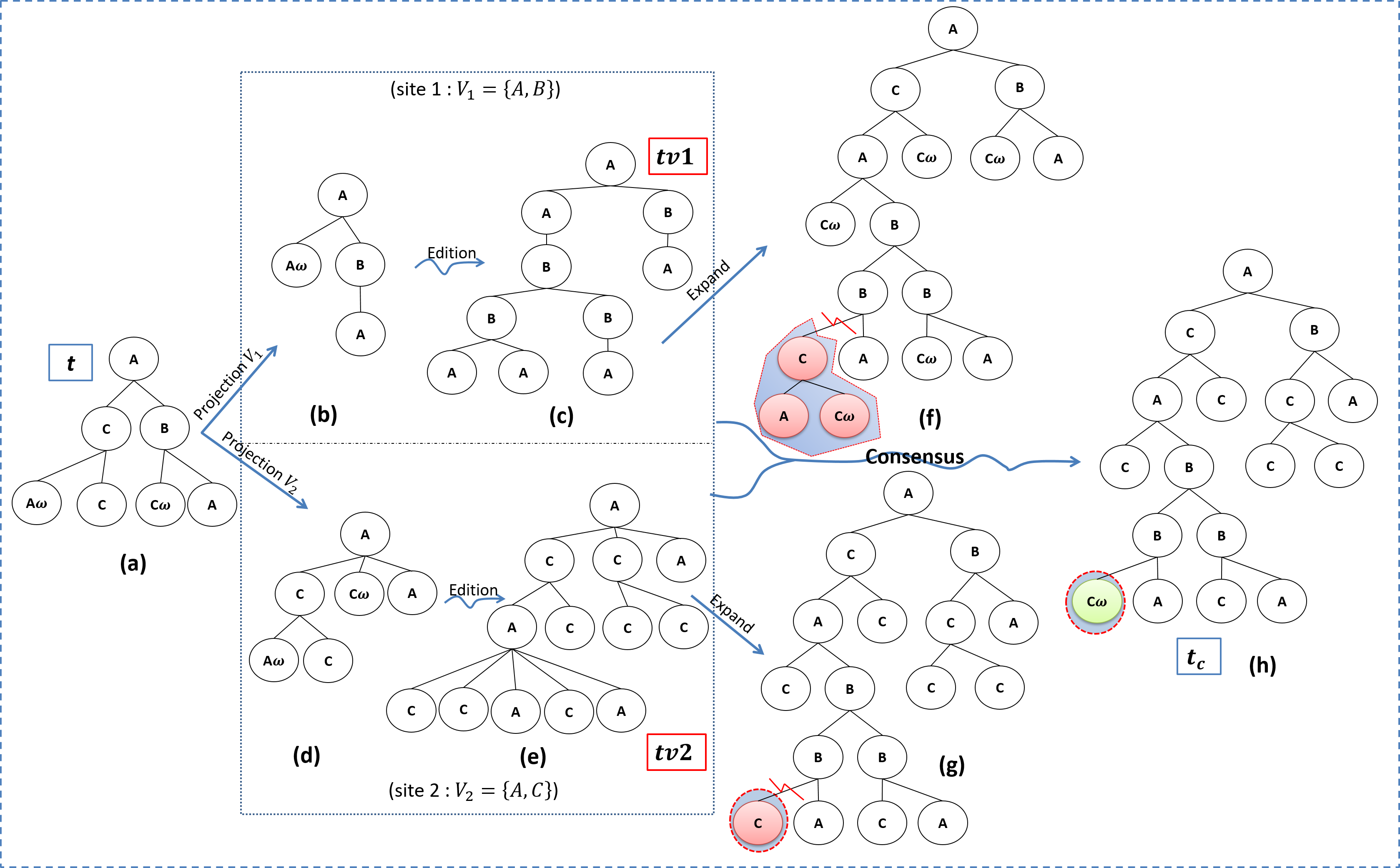}}
	\caption{An edition with conflicts and corresponding consensus.}
	\label{chap2:fig:consensus-workflow}
\end{figure}

Figure \ref{chap2:fig:consensus-workflow} %et \ref{Workflowast} 
is an illustration of an asynchronous cooperative editing process generating partial replicas (fig. \ref{chap2:fig:consensus-workflow}(c) and fig. \ref{chap2:fig:consensus-workflow}(e)) in conflict\footnote{By realising expansions of each of the replicas, we respectively obtain among others, the documents presented by figure \ref{chap2:fig:consensus-workflow}(f) and figure \ref{chap2:fig:consensus-workflow}(g) on which, one can easily observe a conflict highlighted by areas having a blue background.} from the grammar having as productions: 
\begin{comment}
$P_{1}:\; A\rightarrow C\; B \qquad P_{2}:\; A\rightarrow \varepsilon \qquad  \; P_{3}:\; B\rightarrow C\; A\; \; \\ ~~~~~~~~~~~~~~~~~~~~~~~~~\qquad P_{4}:\; B\rightarrow B\; B\;    \; \qquad P_{5}:\; C\rightarrow A\; C  \qquad 
				 \; P_{6}:\; C\rightarrow C\; C\qquad  	P_{7}:\;C\rightarrow \varepsilon$
\end{comment}
\[ 
\begin{array}{lll}
				P_{1}:\; A\rightarrow C\; B & \; P_{3}:\; B\rightarrow C\; A\; & \; P_{5}:\; C\rightarrow A\; C  \\
				P_{2}:\; A\rightarrow \varepsilon & \; P_{4}:\; B\rightarrow B\; B\; & \; P_{6}:\; C\rightarrow C\; C\\
				& & P_{7}:\;C\rightarrow \varepsilon
\end{array}
 \]

Initially in the process, two partial replicas (fig. \ref{chap2:fig:consensus-workflow}(b) and fig. \ref{chap2:fig:consensus-workflow}(d)) are obtained by projections of the global document (fig. \ref{chap2:fig:consensus-workflow}(a)). After their update (fig. \ref{chap2:fig:consensus-workflow}(c) and fig. \ref{chap2:fig:consensus-workflow}(e)) a synchronisation point is reached and, by applying the approach described in section \ref{chap2:sec:consensus-reconciliation-issue-principle}, a consensus document is found (fig. \ref{chap2:fig:consensus-workflow}(h)). 
More precisely, as detailed below, we associate the automata $\mathcal{A}^{(1)}$ and $\mathcal{A}^{(2)}$ respectively to the updated partial replicas $tv1$ and $tv2$ (fig. \ref{chap2:fig:consensus-workflow}(c) and fig. \ref{chap2:fig:consensus-workflow}(e)), then we build the automaton of consensus $\mathcal{A}_{(sc)}=\mathcal{A}^{(1)}\otimes\mathcal{A}^{(2)}$ by applying the approach described in section \ref{chap2:sec:consensus-calculation} (\textit{construction of the consensus automaton}) and finally, we generate the simplest documents of the consensus (fig. \ref{chap2:fig:consensus-example-trees}) among which is the document in figure \ref{chap2:fig:consensus-workflow}(h).

~

\noindent\textbf{\textit{Linearisation of a structured document}}

To simplify the presentation, we represent in the following, trees by their linearisation in the form of a Dyck word. To do this, we associate a (various) pair of brackets to each grammatical symbol and the linearisation of a tree is obtained by performing a Depth First Search (DFS) of the resulting tree (see fig. \ref{chap2:fig:tree-linearisation}).
\begin{figure}[ht!]
	\noindent
	\makebox[\textwidth]{\includegraphics[scale=0.6]{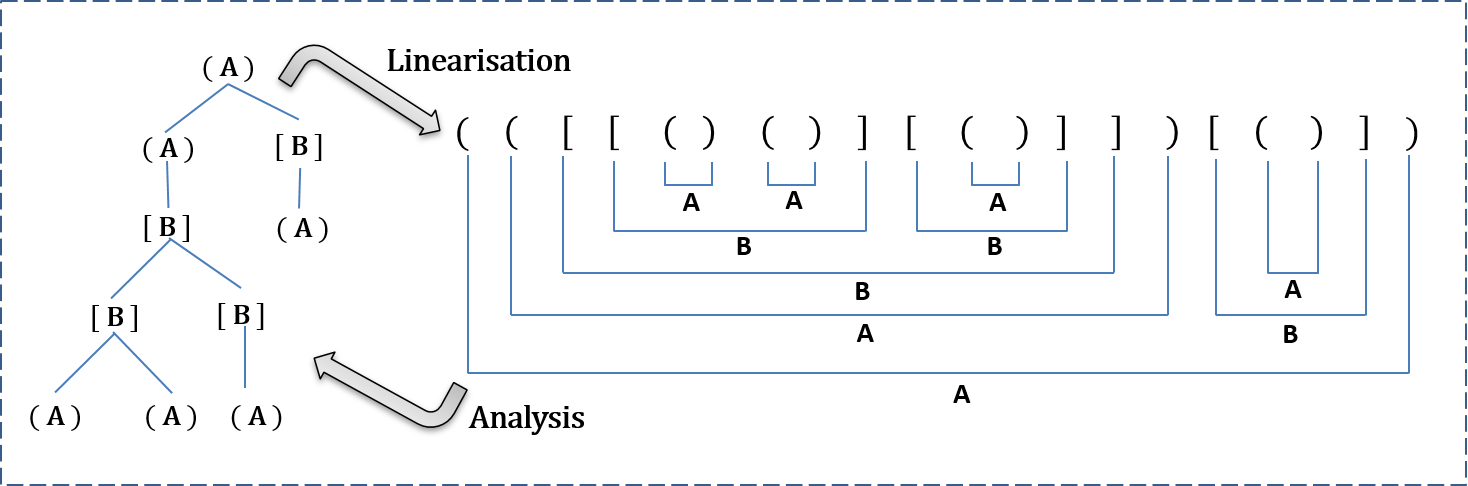}}
	\caption{Linearisation of a tree ($tv1$): the Dyck symbols '$($' and '$)$' (resp. '$[$' and '$]$') have been associated with the grammatical symbol $A$ (resp. $B$).}
	\label{chap2:fig:tree-linearisation}
\end{figure}

~

\noindent\textbf{\textit{The transition schemas for the view $\{A, B\}$}}

A list of trees (forest) is represented by the concatenation of their linearisation. We use the opening parenthesis '(' and the closing one ')' to represent Dyck symbols associated with the visible symbol $A$, and the opening bracket '[' and the closing one ']' to represent those associated with the visible symbol $B$. Each transition of the automata associated to partial replicas according to the view $\{A, B\}$ is conform to one of the transition schemas\footnote{We do not represent the whole set of transition schemas in this example; only the useful subset for reconciliation of closed documents is shown here because the documents to reconcile in this example are all closed (has no buds). To consider buds, one should, for each visible sort $X$, associate a new pair of Dyck symbols to the bud of type $X$ then, derive the new schemas.} in table \ref{chap2:table:trans-schem-a-b}.
\begin{table}[ht]
	\centering
	\caption{The transition schemas for the view $\{A, B\}$.}
	\label{chap2:table:trans-schem-a-b}
	\begin{tabular}[t]{lcll}
	$\langle A,w_{1} \rangle$ & $\longrightarrow$ & $(P_{1}, [\langle C,u \rangle, \langle B,v \rangle])$ & if $w_{1}=u[v]$ \\
	
	$\langle A,w_{2} \rangle$ & $\longrightarrow$ & $(P_{2}, [\textit{ }])$ & if $w_{2}=\epsilon$\\
	
	$\langle B,w_{3} \rangle$ & $\longrightarrow$ & $(P_{3}, [\langle C,u \rangle, \langle A,v \rangle])$ & if $w_{3}=u(v)$\\
	
	$\langle B,w_{4} \rangle$ & $\longrightarrow$ & $(P_{4}, [\langle B,u \rangle, \langle B,v \rangle])$ & if $w_{4}=[u][v]$\\
	
	$\langle C,w_{5} \rangle$ & $\longrightarrow$ & $(P_{5}, [\langle A,u \rangle, \langle C,v \rangle])$ & if $w_{5}=(u)v$\\
	
	$\langle C,w_{6} \rangle$ & $\longrightarrow$ & $(P_{6}, [\langle C,u \rangle, \langle C,v \rangle])$ & if $w_{6}=uv\neq\epsilon$\\
	
	$\langle C,w_{7} \rangle$ & $\longrightarrow$ & $(C_\omega,[\,])$ & if $w_{7}=\epsilon$ \\
	\end{tabular}
\end{table}
These schemas are obtained from the grammar productions \cite{badouelTchoupeCmcs}, and the pairs $\langle X,w_{i} \rangle$ are states where $X$ is a grammatical symbol and $w_{i}$ a forest encoded in the Dyck language. The first schema for example, states that the AST generated from the state $\langle A,w_{1} \rangle$ are those obtained using the production $P_{1}$ to create a tree of the form $P_{1}[t_{1}, t_{2}]$; $t_{1}$ and $t_{2}$ being generated respectively from the states $\langle C,u \rangle$ and $\langle B,v \rangle$ such that $w_{1}=u[v]$. The state $\langle C,w_{7} \rangle$ with $w_{7}=\epsilon$ being an exit state, the rule $\langle C,w_{7} \rangle$ $\longrightarrow$ $(C_{\omega}, [\,])$ linked to the production $P_{7}$ states that, the AST generated from the state $\langle C,w_{7} \rangle$ is reduced to a bud of type $C$ ($C$ is the symbol located at the left hand side of $P_{7}$).

~

\noindent\textbf{\textit{Construction of the automaton $\mathcal{A}^{(1)}$ associated to $tv1$}}

Having associated Dyck symbols '(' and ')' (resp. '[' and ']') to the grammatical symbol $A$ (resp. $B$), the linearisation of the partial replica $tv1$ (fig. \ref{chap2:fig:consensus-workflow}(c)) gives $(([[()()][()]])[()])$. As $A$ is the axiom of the grammar, the initial state of the automaton $\mathcal{A}^{(1)}$ is $q_{0}^{1}=\langle A,([[()()][()]])[()] \rangle$. When considering only the  states accessible from $q_{0}^{1}$ and by applying the previous schema of transition, we obtain the automaton in table \ref{chap2:table:automaton-a-b}, for the replica $tv1$ (fig. \ref{chap2:fig:consensus-workflow}(c)).
\begin{table}[ht]
	\caption{The tree automaton associated to the replica $tv1$.}
	\label{chap2:table:automaton-a-b}
	\begin{tabular}[t]{|lcp{5.7cm}|lcp{5cm}|}
	\hline
	$q_{0}^{1}$ & $\longrightarrow$ & $(P_{1}, [q_{1}^{1}, q_{2}^{1}])$ & & with & $q_{1}^{1}=\langle C,([[()()][()]]) \rangle$ and $q_{2}^{1}=\langle B,() \rangle$\\
	
	%\end{tabular}
	%\begin{tabular}[t]{|lcp{5.cm}|lcp{5cm}|}
	\hline
	$q_{1}^{1}$ & $\longrightarrow$ & $(P_{5}, [q_{3}^{1}, q_{4}^{1}])$ & & with & $q_{3}^{1}=\langle A,[[()()][()]] \rangle$ and $q_{4}^{1}=\langle C,\epsilon \rangle$\\
	\hline
	%\end{tabular}
	%\begin{tabular}[t]{lcp{5.3cm}lcp{5cm}}
	
	$q_{1}^{1}$ & $\longrightarrow$ & $(P_{6}, [q_{4}^{1}, q_{1}^{1}])$ | $(P_{6}, [q_{1}^{1}, q_{4}^{1}])$ & & &\\
	\hline
	%\end{tabular}
	%\begin{tabular}[t]{lcp{5.3cm}lcp{5cm}}
	
	$q_{2}^{1}$ & $\longrightarrow$ & $(P_{3}, [q_{4}^{1}, q_{5}^{1}])$ & & with & $q_{5}^{1}=\langle A,\epsilon \rangle$\\
	\hline
	%\end{tabular}
	%\begin{tabular}[t]{lcp{5.3cm}lcp{5cm}}
	
	$q_{3}^{1}$ & $\longrightarrow$ & $(P_{1}, [q_{4}^{1}, q_{6}^{1}])$ & & with & $q_{6}^{1}=\langle B,[()()][()] \rangle$\\
	\hline
	%\end{tabular}
	%\begin{tabular}[t]{lcp{5.3cm}lcp{5cm}}
	
	$q_{4}^{1}$ & $\longrightarrow$ & $(C_\omega,[\,])$ & & &\\
	\hline
	%\end{tabular}
	%\begin{tabular}[t]{lcp{5.3cm}lcp{5cm}}
	
	$q_{5}^{1}$ & $\longrightarrow$ & $(P_{2}, [\,])$ & & &\\
	\hline
	%\end{tabular}
	%\begin{tabular}[t]{lcp{5.3cm}lcp{5cm}}
	
	$q_{6}^{1}$ & $\longrightarrow$ & $(P_{4}, [q_{7}^{1}, q_{2}^{1}])$ & & with & $q_{7}^{1}=\langle B,()() \rangle$\\
	\hline
	%\end{tabular}
	%\begin{tabular}[t]{lcp{5.3cm}lcp{5cm}}
	
	$q_{7}^{1}$ & $\longrightarrow$ & $(P_{3}, [q_{8}^{1}, q_{5}^{1}])$ & & with & $q_{8}^{1}=\langle C,() \rangle$\\
	\hline
	%\end{tabular}
	%\begin{tabular}[t]{lcp{5.3cm}lcp{5cm}}
	
	$q_{8}^{1}$ & $\longrightarrow$ & $(P_{5}, [q_{5}^{1}, q_{4}^{1}])$ & & &\\
	\hline
	%\end{tabular}
	%\begin{tabular}[t]{lcp{5.3cm}lcp{5cm}}
	
	$q_{8}^{1}$ & $\longrightarrow$ & $(P_{6}, [q_{8}^{1}, q_{4}^{1}])$ | $(P_{6}, [q_{4}^{1}, q_{8}^{1}])$ & & &\\
	\hline
	\end{tabular}
\end{table}
The state $q_{4}^{1}=\langle C,\epsilon \rangle$ is the only exit state of  $\mathcal{A}^{(1)}$.
It is easy to verify that the document of figure \ref{chap2:fig:consensus-workflow}(f) resulting from the reverse projection of $tv1$, belongs to the language accepted by the automaton $\mathcal{A}^{(1)}$.

~

\noindent\textbf{\textit{Construction of the automaton $\mathcal{A}^{(2)}$ associated to $tv2$}}

As before, by associating to the grammatical symbol $C$ (resp. $A$) the Dyck symbols '[' and ']' (resp. '(' and ')'), we obtain the transition schemas for the automata associated to the partial replicas according to the view $\{A, C\}$ (see table \ref{chap2:table:trans-schem-a-c}).
\begin{table}[ht]
	\centering
	\caption{The transition schemas for the view $\{A, C\}$.}
	\label{chap2:table:trans-schem-a-c}
	\begin{tabular}[t]{lcll}
	
	$\langle A,w_{1} \rangle$ & $\longrightarrow$ & $(P_{1}, [\langle C,u \rangle, \langle B,v \rangle])$ & if $w_{1}=[u]v$\\
	
	$\langle A,w_{2} \rangle$ & $\longrightarrow$ & $(P_{2}, [\textit{ }])$ & if $w_{2}=\epsilon$\\
	$\langle B,w_{3} \rangle$ & $\longrightarrow$ & $(P_{3}, [\langle C,u \rangle, \langle A,v \rangle])$ & if $w_{3}=[u](v)$\\
	
	$\langle B,w_{4} \rangle$ & $\longrightarrow$ & $(P_{4}, [\langle B,u \rangle, \langle B,v \rangle])$ & if $w_{4}=uv\neq\epsilon$\\
	$\langle B,w_{5} \rangle$ & $\longrightarrow$ & $(B_\omega,[\,])$ & if $w_{5}=\epsilon$\\
	
	$\langle C,w_{6} \rangle$ & $\longrightarrow$ & $(P_{5}, [\langle A,u \rangle, \langle C,v \rangle])$ & if $w_{6}=(u)[v]$\\

	$\langle C,w_{7} \rangle$ & $\longrightarrow$ & $(P_{6}, [\langle C,u \rangle, \langle C,v \rangle])$ & if $w_{7}=[u][v]$\\
	
	$\langle C,w_{8} \rangle$ & $\longrightarrow$ & $(P_{7},[\textit{ }])$ & if $w_{8}=\epsilon$\\
	\end{tabular}
\end{table}

The linearisation of the partial replica $tv2$ (fig. \ref{chap2:fig:consensus-workflow}(e)) is $([([\textit{ }][\textit{ }]()[\textit{ }]())[\textit{ }]][[\textit{ }][\textit{ }]]())$. The automaton $\mathcal{A}^{(2)}$ associated to this replica has as initial state $q_{0}^{2}=\langle A,[([\textit{ }][\textit{ }]()[\textit{ }]())[\textit{ }]][[\textit{ }][\textit{ }]]() \rangle$ and its transitions are the ones in table \ref{chap2:automaton-a-c}. 
\begin{table}[ht]
	\caption{The tree automaton associated to the replica $tv2$.}
	\label{chap2:automaton-a-c}
	\begin{tabular}[t]{|lcp{5.7cm}|lcp{5cm}|}
	\hline
	$q_{0}^{2}$ & $\longrightarrow$ & $(P_{1}, [q_{1}^{2}, q_{2}^{2}])$ & & with & $q_{1}^{2}=\langle C,([\,][\,]()[\,]())[\,] \rangle$ and $q_{2}^{2}=\langle B,[[\,][\,]]() \rangle$\\
	\hline
	%\end{tabular}
	%\begin{tabular}[t]{lcp{5.3cm}lcp{5cm}}
	
	$q_{1}^{2}$ & $\longrightarrow$ & $(P_{5}, [q_{3}^{2}, q_{4}^{2}])$ & & with & $q_{3}^{2}=\langle A,[\,][\,]()[\,]() \rangle$ and $q_{4}^{2}=\langle C,\epsilon \rangle$\\
	\hline
	%\end{tabular}
	%\begin{tabular}[t]{lcp{5.3cm}lcp{5cm}}
	
	$q_{2}^{2}$ & $\longrightarrow$ & $(P_{3}, [q_{5}^{2}, q_{6}^{2}])$ & & with & $q_{5}^{2}=\langle C,[\,][\,] \rangle$ and $q_{6}^{2}=\langle A,\epsilon \rangle$\\
	\hline
	%\end{tabular}
	%\begin{tabular}[t]{lcp{5.3cm}lcp{5cm}}
	
	$q_{3}^{2}$ & $\longrightarrow$ & $(P_{1}, [q_{4}^{2}, q_{7}^{2}])$ & & with & $q_{7}^{2}=\langle B,[\,]()[\,]() \rangle$\\
	\hline
	%\end{tabular}
	%\begin{tabular}[t]{lcp{5.3cm}lcp{5cm}}
	
	$q_{4}^{2}$ & $\longrightarrow$ & $(P_{7}, [\,])$ & & &\\
	
	\hline
	%\end{tabular}
	%\begin{tabular}[t]{lcp{5.3cm}lcp{5cm}}
	
	$q_{5}^{2}$ & $\longrightarrow$ & $(P_{6}, [q_{4}^{2}, q_{4}^{2}])$ & & &\\
	
	\hline
	%\end{tabular}
	%\begin{tabular}[t]{lcp{5.3cm}lcp{5cm}}
	
	$q_{6}^{2}$ & $\longrightarrow$ & $(P_{2}, [\,])$ & & &\\
	
	\hline
	%\end{tabular}
	%\begin{tabular}[t]{lcp{5.3cm}lcp{5cm}}
	
	$q_{7}^{2}$ & $\longrightarrow$ & $(P_{4}, [q_{8}^{2}, q_{7}^{2}])$ | $(P_{4}, [q_{9}^{2}, q_{10}^{2}])$ | $(P_{4}, [q_{11}^{2}, q_{11}^{2}])$ | $(P_{4}, [q_{12}^{2}, q_{13}^{2}])$ | $(P_{4}, [q_{7}^{2}, q_{8}^{2}])$ & & with & $q_{8}^{2}=\langle B,\epsilon \rangle$, $q_{9}^{2}=\langle B,[\,] \rangle$, $q_{10}^{2}=\langle B,()[\,]() \rangle$, $q_{11}^{2}=\langle B,[\,]() \rangle$, $q_{12}^{2}=\langle B,[\,]()[\,] \rangle$ and $q_{13}^{2}=\langle B,() \rangle$\\
	
	\hline
	%\end{tabular}
	%\begin{tabular}[t]{lcp{5.3cm}lcp{5cm}}
	
	$q_{8}^{2}$ & $\longrightarrow$ & $(B_\omega,[\,])$ & & &\\
	
	\hline
	%\end{tabular}
	%\begin{tabular}[t]{lcp{5.3cm}lcp{5cm}}
	
	$q_{9}^{2}$ & $\longrightarrow$ & $(P_{4}, [q_{8}^{2}, q_{9}^{2}])$ | $(P_{4}, [q_{9}^{2}, q_{8}^{2}])$ & & &\\
	
	\hline
	%\end{tabular}
	%\begin{tabular}[t]{lcp{5.3cm}lcp{5cm}}
	
	$q_{10}^{2}$ & $\longrightarrow$ & $(P_{4}, [q_{8}^{2}, q_{10}^{2}])$ | $(P_{4}, [q_{13}^{2}, q_{11}^{2}])$ | $(P_{4}, [q_{14}^{2}, q_{13}^{2}])$ | $(P_{4}, [q_{10}^{2}, q_{8}^{2}])$ & & with & $q_{14}^{2}=\langle B,()[\,]    \rangle$\\
	
	\hline
	%\end{tabular}
	%\begin{tabular}[t]{lcp{5.3cm}lcp{5cm}}
	
	$q_{11}^{2}$ & $\longrightarrow$ & $(P_{3}, [q_{4}^{2}, q_{6}^{2}])$ & & &\\
	
	\hline
	%\end{tabular}
	%\begin{tabular}[t]{lcp{5.3cm}lcp{5cm}}
	
	$q_{12}^{2}$ & $\longrightarrow$ & $(P_{4}, [q_{8}^{2}, q_{12}^{2}])$ | $(P_{4}, [q_{9}^{2}, q_{14}^{2}])$ | $(P_{4}, [q_{11}^{2}, q_{9}^{2}])$ | $(P_{4}, [q_{12}^{2}, q_{8}^{2}])$ & & &\\
	
	\hline
	%\end{tabular}
	%\begin{tabular}[t]{lcp{5.3cm}lcp{5cm}}
	
	$q_{13}^{2}$ & $\longrightarrow$ & $(P_{4}, [q_{8}^{2}, q_{13}^{2}])$ | $(P_{4}, [q_{13}^{2}, q_{8}^{2}])$ & & &\\
	
	\hline
	%\end{tabular}
	%\begin{tabular}[t]{lcp{5.3cm}lcp{5cm}}
	
	$q_{14}^{2}$ & $\longrightarrow$ & $(P_{4}, [q_{8}^{2}, q_{14}^{2}])$ | $(P_{4}, [q_{13}^{2}, q_{9}^{2}])$ | $(P_{4}, [q_{14}^{2}, q_{8}^{2}])$ & & &\\
	\hline
	\end{tabular}
\end{table}
The state $q_{8}^{2}=\langle B,\epsilon \rangle$ is the only exit state of the automaton $\mathcal{A}^{(2)}$.

~

\noindent\textbf{\textit{Construction of the consensus automaton $\mathcal{A}_{(sc)}$}}

By application of synchronous product of several tree automata described in section \ref{chap2:sec:consensus-calculation} (\textit{construction of the consensus automaton}) to the automata $\mathcal{A}^{(1)}$ and $\mathcal{A}^{(2)}$, the consensual automaton $\mathcal{A}_{(sc)}=\mathcal{A}^{(1)}\otimes\mathcal{A}^{(2)}$ has $q_{0}=(q_{0}^{1}, q_{0}^{2})$ as initial state. $\mathcal{A}^{(1)}$ has a transition from $q_{0}^{1}$ to $[q_{1}^{1}, q_{2}^{1}]$ labelled $P_{1}$. Similarly, $\mathcal{A}^{(2)}$ has a transition from $q_{0}^{2}$ to $[q_{1}^{2}, q_{2}^{2}]$ labelled $P_{1}$. So, we have in $\mathcal{A}_{(sc)}$, a transition labelled $P_{1}$ for accessing states $[q_{1}=(q_{1}^{1}, q_{1}^{2}), q_{2}=(q_{2}^{1}, q_{2}^{2})]$ from the initial state $q_{0}=(q_{0}^{1}, q_{0}^{2})$. Following this principle, we construct the consensual automaton in table \ref{chap2:table:automaton-a-b-c}. 
\begin{table}[ht]
	\caption{The consensual tree automaton.}
	\label{chap2:table:automaton-a-b-c}
	\begin{tabular}[t]{|lcp{5.7cm}|lcp{5cm}|}
	\hline
	& & & & & $q_{0}=(q_{0}^{1}, q_{0}^{2})$\\
	\hline
	$q_{0}$ & $\longrightarrow$ & $(P_{1}, [q_{1}, q_{2}])$ & & with & $q_{1}=(q_{1}^{1}, q_{1}^{2})$ and $q_{2}=(q_{2}^{1}, q_{2}^{2})$\\
	\hline
	%\end{tabular}
	%\begin{tabular}[t]{lcp{5.3cm}lcp{5cm}}
	
	$q_{1}$ & $\longrightarrow$ & $(P_{5}, [q_{3}, q_{4}])$ & & with & $q_{3}=(q_{3}^{1}, q_{3}^{2})$ and $q_{4}=(q_{4}^{1}, q_{4}^{2})$\\
	
	\hline
	%\end{tabular}
	%\begin{tabular}[t]{lcp{5.3cm}lcp{5cm}}
	
	$q_{2}$ & $\longrightarrow$ & $(P_{3}, [q_{5}, q_{6}])$ & & with & $q_{5}=(q_{4}^{1}, q_{5}^{2})$ and $q_{6}=(q_{5}^{1}, q_{6}^{2})$\\
	
	\hline
	%\end{tabular}
	%\begin{tabular}[t]{lcp{5.3cm}lcp{5cm}}
	
	$q_{3}$ & $\longrightarrow$ & $(P_{1}, [q_{4}, q_{7}])$ & & with & $q_{7}=(q_{6}^{1}, q_{7}^{2})$\\
	
	\hline
	%\end{tabular}
	%\begin{tabular}[t]{lcp{5.3cm}lcp{5cm}}
	
	$q_{4}$ & $\longrightarrow$ & $(P_{7}, [\,])$ & & &\\
	
	\hline
	%\end{tabular}
	%\begin{tabular}[t]{lcp{5.3cm}lcp{5cm}}
	
	$q_{5}$ & $\longrightarrow$ & $(P_{6}, [q_{8}, q_{8}])$ & & with & $q_8=(q_{s1}, q^2_4)$ and $q_{s1}=\langle Open~C,[\,] \rangle $\\
	
	\hline
	%\end{tabular}
	%\begin{tabular}[t]{lcp{5.3cm}lcp{5cm}}
	
	$q_{6}$ & $\longrightarrow$ & $(P_{2}, [\,])$ & & &\\
	
	\hline
	%\end{tabular}
	%\begin{tabular}[t]{lcp{5.3cm}lcp{5cm}}
	
	$q_{7}$ & $\longrightarrow$ & $(P_{4}, [q_{9}, q_{10}])$ | $(P_{4}, [q_{11}, q_{12}])$ | $(P_{4}, [q_{13}, q_{14}])$ | $(P_{4}, [q_{15}, q_{16}])$ | $(P_{4}, [q_{17}, q_{18}])$ & & with & $q_{9}=(q_{7}^{1}, q_{8}^{2})$, $q_{10}=(q_{2}^{1}, q_{7}^{2})$, $q_{11}=(q_{7}^{1}, q_{9}^{2})$, $q_{12}=(q_{2}^{1}, q_{10}^{2})$, $q_{13}=(q_{7}^{1}, q_{11}^{2})$, $q_{14}=(q_{2}^{1}, q_{11}^{2})$, $q_{15}=(q_{7}^{1}, q_{12}^{2})$, $q_{16}=(q_{2}^{1}, q_{13}^{2})$, $q_{17}=(q_{7}^{1}, q_{7}^{2})$ and $q_{18}=(q_{2}^{1}, q_{8}^{2})$\\
	\hline
	%\end{tabular}
	
	%\begin{tabular}[t]{|lcp{6.56cm}|lcp{4.99cm}|}
	%\hline
	$q_{8}$ & $\longrightarrow$ & $(P_{7}, [\,])$ & & &\\
	
	\hline
	%\end{tabular}
	%\begin{tabular}[t]{lcp{5.3cm}lcp{5cm}}
	
	$q_{9}$ & $\longrightarrow$ & $(P_{3}, [q_{19}, q_{20}])$ & & with & $q_{19}=(q_{8}^{1}, q_{s1})$ and $q_{20}=(q_{5}^{1}, q_{s2})$, $q_{s2}=\langle Open~A,[\,] \rangle $\\
	
	\hline
	%\end{tabular}
	%\begin{tabular}[t]{lcp{5.3cm}lcp{5cm}}
	
	$q_{13}$ & $\longrightarrow$ & $(P_{3}, [q_{21}, q_{6}])$ & & with & $q_{21}=(q_{8}^{1}, q_{4}^{2})$\\
	
	\hline
	%\end{tabular}
	%\begin{tabular}[t]{lcp{5.3cm}lcp{5cm}}
	
	$q_{14}$ & $\longrightarrow$ & $(P_{3}, [q_{4}, q_{6}])$ & & &\\
	
	\hline
	%\end{tabular}
	%\begin{tabular}[t]{lcp{5.3cm}lcp{5cm}}
	
	$q_{18}$ & $\longrightarrow$ & $(P_{3}, [q_{22}, q_{20}])$ & & with & $q_{22}=(q_{4}^{1}, q_{s1})$\\
	
	\hline
	%\end{tabular}
	%\begin{tabular}[t]{lcp{5.3cm}lcp{5cm}}
	
	$q_{19}$ & $\longrightarrow$ & $(P_{5}, [q_{20}, q_{22}])$ | $(P_{6}, [q_{19}, q_{22}])$ | $(P_{6}, [q_{22}, q_{19}])$ & & &\\
	\hline
	%\end{tabular}
	%\begin{tabular}[t]{lcp{5.3cm}lcp{5cm}}
	
	$q_{20}$ & $\longrightarrow$ & $(P_{2}, [\,])$ & & &\\
	\hline
	%\end{tabular}\\
	
	%\begin{tabular}[t]{|p{8.43cm}| p{6.5cm}|}
	%\hline
	$q_{10}$ & $\longrightarrow$ & $(B_\omega, [\,])$ & & &\\ 
	\hline
	$q_{11}$ & $\longrightarrow$ & $(B_\omega, [\,])$ & & &\\
	\hline
	$q_{12}$ & $\longrightarrow$ & $(B_\omega, [\,])$ & & &\\
	\hline
	$q_{15}$ & $\longrightarrow$ & $(B_\omega, [\,])$ & & &\\
	\hline
	$q_{16}$ & $\longrightarrow$ & $(B_\omega, [\,])$ & & &\\
	\hline
	$q_{17}$ & $\longrightarrow$ & $(B_\omega, [\,])$ & & &\\
	\hline
	$q_{21}$ & $\longrightarrow$ & $(C_\omega, [\,])$ & & &\\
	\hline
	$q_{22}$ & $\longrightarrow$ & $(C_\omega, [\,])$ & & &\\
	\hline
	\end{tabular}
\end{table}

The states $\{q_{10}, q_{11}, q_{12}, q_{15}, q_{16}, q_{17}, q_{21}, q_{22}\}$ are the exit states of the automaton $\mathcal{A}_{(sc)}$. They are states whose composite states are either in conflict (for example $q_{10}=(q_{2}^{1}, q_{7}^{2})$ et $q_{2}^{1} \curlyveeuparrow q_{7}^{2}$), or are all exit states (for example $q_{22}=(q_{4}^{1}, q_{s1}))$.

The use of the function that generates the simplest AST (with buds) of a tree language from its automaton \cite{badouelTchoupeCmcs} on $\mathcal{A}_{(sc)}$, produces \textit{four} AST whose derivation trees (the consensus) are shown in figure \ref{chap2:fig:consensus-example-trees}.
\begin{figure}[ht!]
	\noindent
	\makebox[\textwidth]{\includegraphics[scale=0.3]{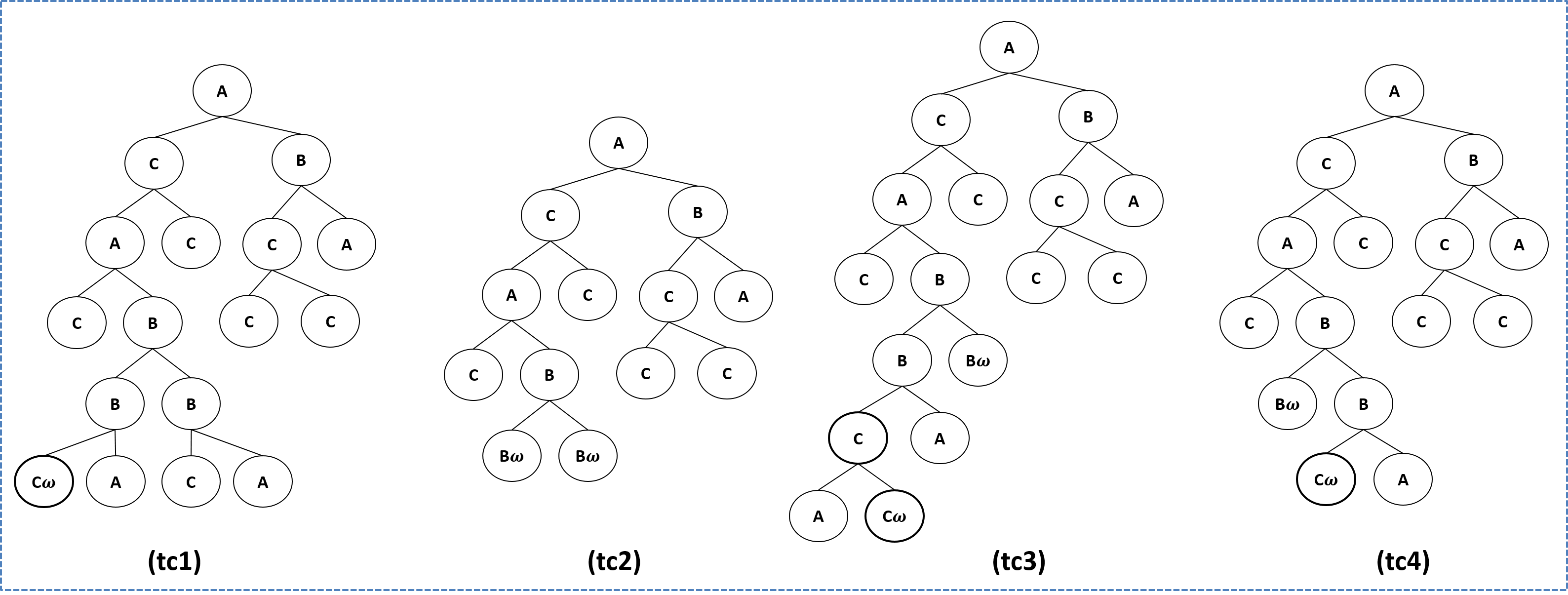}}
	\caption{Consensual trees generated from the automaton $\mathcal{A}_{(sc)}$}
	\label{chap2:fig:consensus-example-trees}
\end{figure}

\mySection{A Software Architecture for Centralised Management of Structured Documents in a Cooperative Editing Workflow}{}
\label{chap2:sec:architecture-cooperative-editing}

In this section, we will focus on the implementation of a system that can support cooperative editing as perceived by Badouel and Tchoup\'e. This effort is motivated by the fact that: 
\begin{enumerate}
\item \textit{This type of editing workflow applies to structured documents}: 
this leads to the fact that, one can locally perform validations in accordance with a local model derived from the global one;
\item \textit{This type of editing workflow is particularly compatible with administrative workflows}: concepts of "view" and partial replica introduced by Badouel and Tchoupé, make that the type of workflow they offer is particularly adapted for the specification of many administrative processes. Consider, for example, the process "tracking a medical record in a health center with the reception and consultation services": the aforesaid record can be modelled as a structured document in which the members of the host service (reception) cannot view and/or modify certain information contained therein; those information, requesting the expertise of the consulting staff for example. Therefore, one can associate views to each of these services. It is left only to specify the medical record's circuit and an editing workflow of the type described in the previous section is obtained;
\item \textit{A generic architectural model describing precisely an approach for the implementation of this type of workflow does not exist}: the only prototype \cite{artTinyCE} which was designed around the concepts handled (view, partial replica, merging, etc.) for this type of workflow, was more of a graphic tool (editor) for the experimentation of concepts and algorithms presented in \cite{badouelTchoupeCmcs}; workflow management is not addressed in it: this tool cannot be used to specify an editing workflow, it does not support routing or storage of artifacts, nothing is done concerning monitoring, etc., yet these concerns are among the most important to be taken care of by a workflow management infrastructure \cite{ima}.
\end{enumerate}

\mySubSection{The Proposed Architecture}{}
\label{chap2:sec:archi-proposed-architecture}

\mySubSubSection{Overall Operations}{}
\label{chap2:sec:archi-overall-operations}
The architecture that we propose is composed of three tiers: some \textit{clients}, a \textit{central server} and several \textit{administration tools}. We consider that, each participant in a given workflow has a client. Initially, the workflow owner (comparable to a deposit owner in Git) connects to the server from his client. He creates his workflow by specifying all necessary informations (the workflow name, the overall grammar, different participants, their rights and their views, the basic document and the workflow's circuit), then triggers the process. Next, participants concerned by the newly created workflow receive an alert message from the system, inviting them to participate. Each participant must therefore connect himself to the server to obtain a partial replica of the workflow model (encoded in a specification file written in a dedicated DSL) and state (his local document model, a partial replica of the initial document, etc.) according to his rights and his view on the given workflow. A given participant performs his duties and submits his local (partial) replica to the central server which performs synchronisations as soon as possible and the process continues (see fig. \ref{chap2:fig:badouel-tchoupe-workflow}) until the end.  For specific needs (authentication, access to corporate data, etc.), clients and server may require the intervention of an administration tool (database, paperwork and many others). These three tiers are interconnected around a middleware as presented in figure \ref{chap2:fig:architecture}.
\begin{figure}[ht!]
	\noindent
	\makebox[\textwidth]{\includegraphics[scale=0.3]{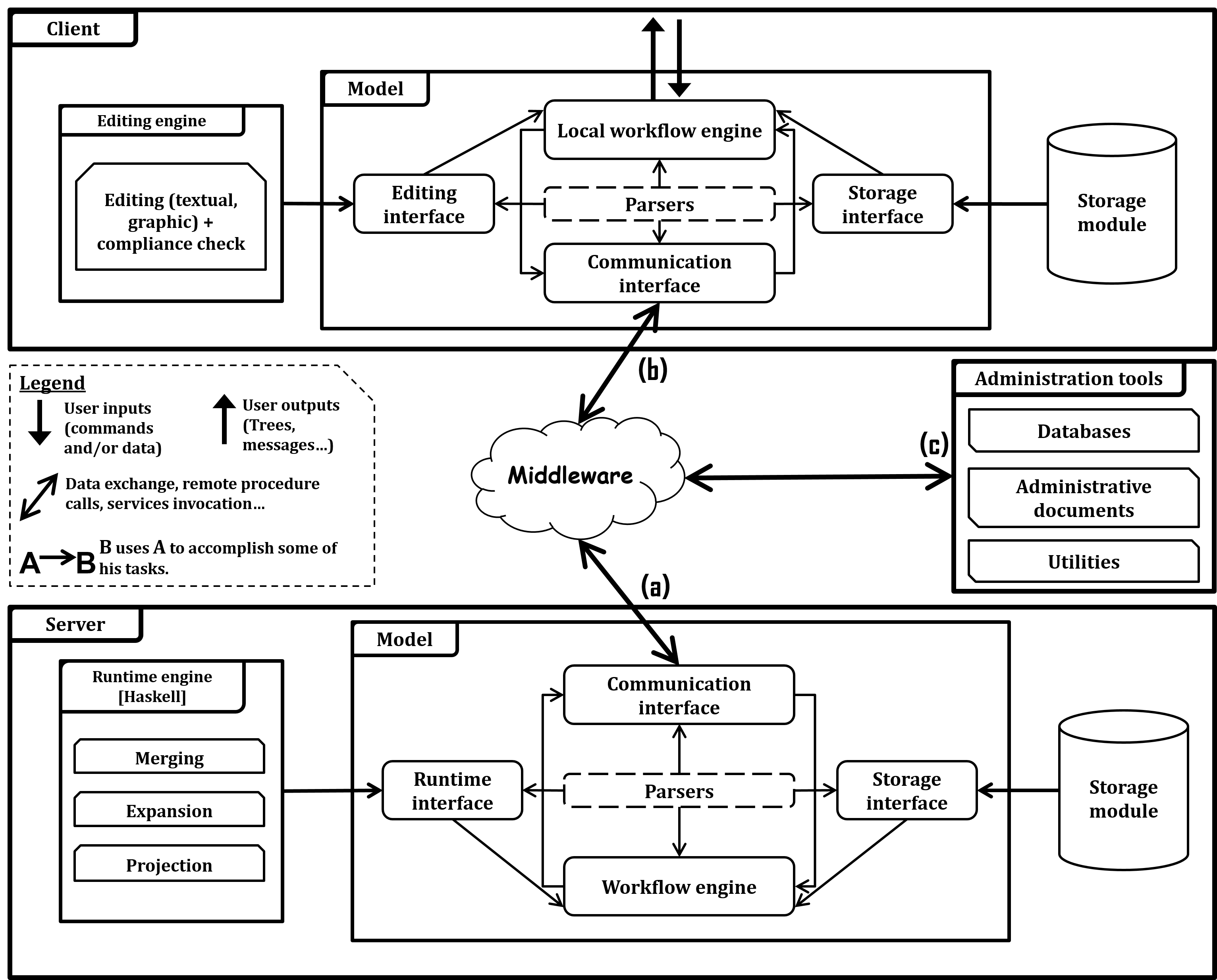}}
	\caption{A software architecture (three-tiers) for centralised management of structured documents' cooperative editing workflows.}
	\label{chap2:fig:architecture}
\end{figure}

\mySubSubSection{Server Architecture}{}
\label{chap2:sec:archi-server-architecture}
The server is responsible for the storage, restoration, execution and monitoring of workflows. Its architecture is based on three basic elements as shown in figure \ref{chap2:fig:architecture}(a) : its \textit{model}, \textit{storage module} and its \textit{runtime engine}.
\begin{enumerate}
\item \textit{The model}: it is the one orchestrating all the tasks supported by the server. It consists of a workflow engine, a set of parsers and three communication interfaces (the interface with the middleware, that with the storage module and the one with the runtime engine).
\item \textit{The storage module}: it is responsible for the storage of workflows. Like CVS, it maintains a main repository for each workflow. The repository space of a given workflow includes its specification file written in a DSL. There are also (global) document versions showing the state of the workflow at given times. These versions of the underlying documents, facilitate the control and monitoring of workflows.
\item \textit{The runtime engine}: it consists of implementations of projection, expansion and consensual merging algorithms. These implementations are used by the workflow engine in the realisation of these tasks. A runtime engine written entirely in Haskell, was proposed in \cite{artTinyCE}. However, it is quite rigid and almost impossible to adapt to the architecture presented here. To this end, we present in section \ref{chap2:sec:archi-tinyce-v2-cross-fertilisation}, a more flexible version of the latter.
\end{enumerate}

\mySubSubSection{Client Architecture}{}
\label{chap2:sec:archi-client-architecture}
The client (figure \ref{chap2:fig:architecture}(b)) is also based on three entities: a \textit{model}, an \textit{editing engine} and a \textit{storage module}. The model is responsible for organising and controlling the execution of tasks and user commands. For each new local workflow, the model generates an editing environment which is used by the editing engine to provide conventional facilities of structured document editors (compliance check, syntax highlighting, graphical editing of documents presentations, etc.). Each workflow is locally represented by a specification file and by one structured document representing the current perception of the overall workflow from the current local site. When reaching synchronisation phases, the local structured document is forwarded to the server site, where it is merged with others, in one structured document representing the current state of the overall workflow : it is therefore, a coordination support between the workflow engines of the client and of the server.

\mySubSubSection{The Middleware}{}
\label{chap2:sec:archi-middleware}
The middleware is responsible for the interaction between different tiers of our architecture. It must be designed so that, the coupling between these tiers is as weak as possible. One can for this purpose, consider a SOA in which:
\begin{itemize}
\item Our clients are service clients;
\item The server is a service provider for clients and a client of services offered by the administration tools;
\item The administration tools are service providers.
\end{itemize}
With such an architecture, we can guarantee the independence of each tier and thus, an easier maintenance.

\mySubSection{TinyCE v2}{}
\label{chap2:sec:archi-tinyce-v2}

\mySubSubSection{Presentation of TinyCE v2}{}
\label{chap2:sec:archi-tinyce-v2-presentation}
Due to its technical nature and to the number of technologies it needs for its instantiation, the architecture presented above has not yet been fully implemented. However, many of its components have already been implemented and tested in a test project called TinyCE v2\footnote{TinyCE v2 is a more advanced version of TinyCE \cite{artTinyCE}.} (a Tiny Cooperative Editor version 2).

TinyCE v2 is an editor prototype providing graphic and cooperative editing of the abstract structure of structured documents. It is used following a networked client-server model. Its user interface offers to the user, facilities for the creation of workflows (documents, grammars, actors and views (see fig. \ref{chap2:fig:workflow-creation})), edition and validation of partial replicas (see fig. \ref{chap2:fig:user-connexion}). 
Moreover, this interface also offers the functionality to experiment the concepts of projection, expansion and consensual merging (see fig. \ref{chap2:fig:workflow-merging}). TinyCE v2 is designed using Java and Haskell languages. It offers several implementations of our architecture concepts namely: parsers, storage modules, server's runtime engine, workflow engines and communication interfaces.
\begin{figure}[ht!]
	\noindent
	\makebox[\textwidth]{\includegraphics[scale=0.85]{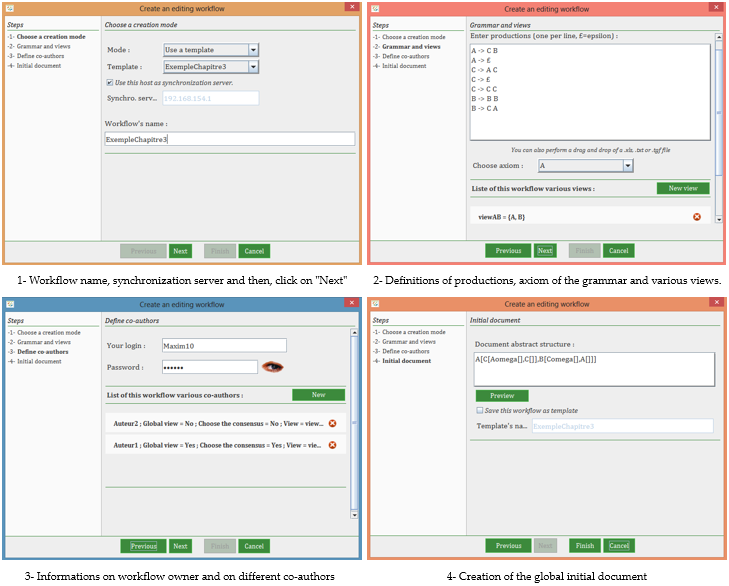}}
	\caption{Some screenshots showing the creation process of a cooperative editing workflow in TinyCE v2.}
	\label{chap2:fig:workflow-creation}
\end{figure}

\begin{figure}[ht!]
	\noindent
	\makebox[\textwidth]{\includegraphics[scale=0.78]{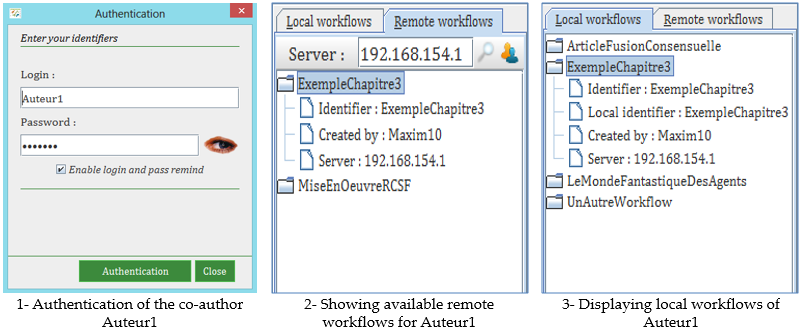}}
	\caption{Some screenshots of TinyCE v2 showing the authentication window of a co-author (Auteur1) as well as those displaying the various local and remote workflows in which he is implicated.}
	\label{chap2:fig:user-connexion}
\end{figure}

\begin{figure}[ht!]
	\noindent
	\makebox[\textwidth]{\includegraphics[scale=0.55]{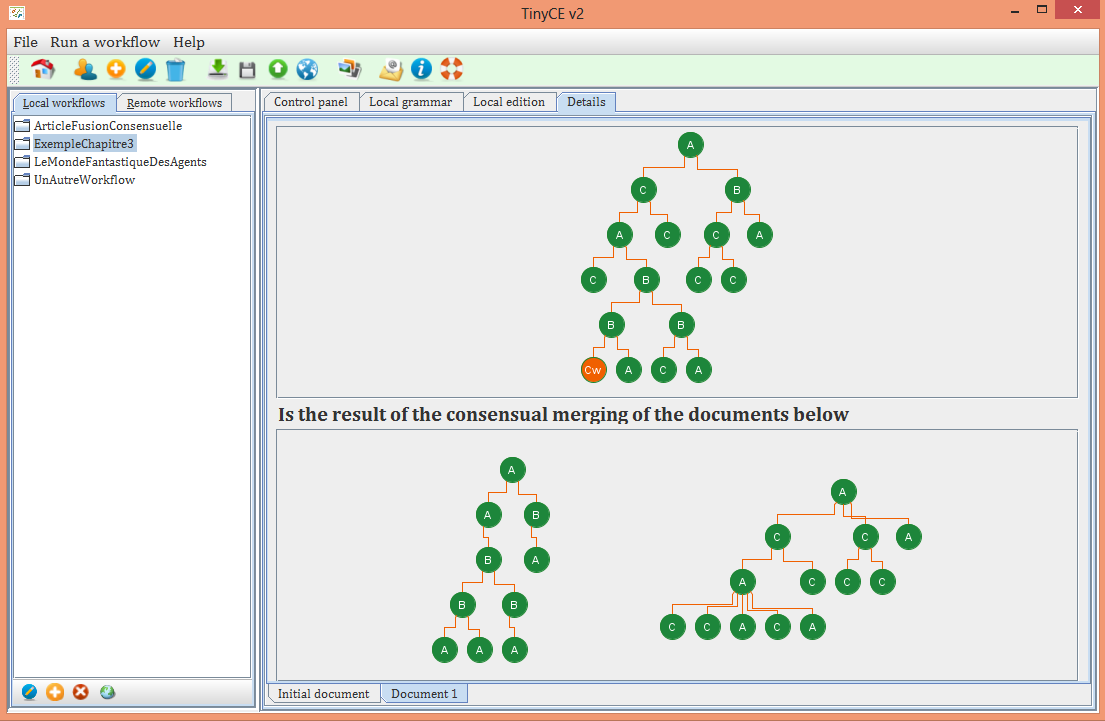}}
	\caption{An illustration of consensual merging in TinyCE v2.}
	\label{chap2:fig:workflow-merging}
\end{figure}

\mySubSubSection{Java-Haskell Cross-Fertilisation in TinyCE v2}{}
\label{chap2:sec:archi-tinyce-v2-cross-fertilisation}
As in \cite{artTinyCE}, the runtime engine of TinyCE v2 exploits the possibility offered by Java, to run an "external program".  Indeed, we designed an interface of TinyCE v2 (runtime interface) capable of launching a Haskell interpreter (GHCi - Glasgow Haskell Compiler interactive\footnote{Official website of GHC: \url{http://www.haskell.org/ghc/}, visited the 04/04/2020.} - in this case) and make it execute various commands. When creating a workflow, TinyCE v2 generates a Haskell program file (.hs), containing data types and functions necessary to achieve the operations of projection, expansion and consensual merging on the structured document representing the state of that workflow. In this way, we considerably reduce the use frequency of parsers presented in \cite{artTinyCE}. The functions are more open to changes as they are contained in a text file and not in a compiled program as in \cite{artTinyCE}. In fact, the main differences between our Java-Haskell cross-fertilisation approach and the one of \cite{artTinyCE} are almost the same that drive the debates on interpreted and compiled languages; our approach is likened to interpreted languages and that of \cite{artTinyCE}, to compiled languages. So, even though our approach can present security risks (that can be addressed using PKI (Public Key Infrastructure) and standard encryption systems like AES (Advanced Encryption Standard), RSA (Rivest Shamir Adleman), etc.), it has the advantage of being portable and easier to maintain. 

\mySection{Summary}{}
\label{chap2:sec:conclusion}
This chapter was devoted to the study of asynchronous cooperative editing concepts in general, and to the study of notions related to the model introduced by Badouel and Tchoup\'e for structured documents cooperative editing \cite{badouelTchoupeCmcs}. In order to be more efficient, we have chosen to study these models with new contributions, including: an algorithm for reconciling potentially conflicting partial replicas of a structured document, and a generic architecture for designing workflow systems that can be modelled as structured cooperative editing systems in the sense of Badouel and Tchoup\'e. The correction of the proposed algorithms has been demonstrated. Implementations of these have been made, notably in TinyCE v2, the cooperative editor prototype implemented in Java and Haskell according to the new architecture proposed in this chapter.

The aim of these studies was to familiarise us with some key mathematical tools, in particular: grammars, views, projection/replication algorithms, fusion/reconciliation algorithms, etc. These mathematical tools are proving to be effective in the implementation of artifact-centric BPM models as demonstrated by Badouel et al. with the AWGAG model \cite{badouel14, badouel2015active}. For this purpose, they will form the foundation of the new artifact-centric model for the completely decentralised design and execution of administrative workflows (assimilated to the asynchronous cooperative edition of mobile structured documents) on a SOA that we propose in the next chapter.

	\mathversion{normal}
	\mathversion{normal2}
	\myChapter{A Choreography-like Workflow Design and Distributed Execution Framework Based on Structured Mobile Artifacts' Cooperative Editing}{}
\label{chap3:choreography-workflow-design-execution}
\myMiniToc{section}{Contents}
% If no minitoc then
% \startcontents[chapters]
\mySection{Introduction}{}
\label{chap3:sec:introduction}
%\subsection*{Contexte et définitions}
%\label{sec:contexte}

As outlined in chapter \ref{chap1:artifact-centric-bpm}, section \ref{chap1:sec:artifact-centric-bpm-approach}, the execution of a given business process according to the artifact-centric approach can be assimilated to the cooperative editing of documents. Indeed, in IBM's work \cite{nigam2003business}, an artifact (also called "\textit{adaptive document}") is considered as a document that conveys all the information concerning a particular case of execution of a given process, from its inception into the system to its termination. In particular, this information provides details on the execution status of the case as well as on its lifecycle (a representation of the possible evolutions of this status). 
To do this, during the execution of a given process, the actions carried out by each of the stakeholders (agents) have the effect of updating (\textit{editing}) the artifacts involved in that execution. If the process is cooperative, the artifact representing it will be updated by several agents: it is said to be cooperatively edited (\textit{cooperative editing}). 
%In choreography-oriented approaches (see chapter \ref{chap1:artifact-centric-bpm}, Section \ref{chap1:sec:artifact-centric-bpm-approach}), in order to increase parallelism of execution, this cooperative editing is distributed, asynchronous and deals with copies (replicas) of the artifact that are subsequently merged at the appropriate time.

In this chapter, we propose a new artifact-centric approach to BPM. In this one, artifacts are seen as structured documents (annotated trees) that can be exchanged between the different agents involved in the execution of a given business process particular case (it is in this sense that they are said to be mobile); during their life, they are edited appropriately to make the system converge towards the achievement of one of the considered process's business goals. The approach presented in this chapter is based on the asynchronous structured cooperative editing techniques proposed in the work of Badouel et al. \cite{badouelTchoupeCmcs, theseTchoupe, tchoupeAtemkeng2} and extended in chapter \ref{chap2:structured-editing-artifact-type} \cite{tchoupeZekeng2016, tchoupeZekeng2017, zekengTchoupe2018} of this manuscript.  

\noindent The major contributions of this chapter are as follows:	
\begin{enumerate}
\item The proposal of another tree-based model of "business artifact", which makes it possible to better assimilate them to structured documents edited cooperatively;

\item The proposal of a choreography-oriented artifact-centric execution model in which agents execute the same and unique update (editing of artifact upon receipt) and diffusion (dissemination of updates) protocol;

\item The proposal of a prototype of a distributed system allowing to fully experiment the approach investigated in this chapter.
\end{enumerate}
  
In the rest of this chapter, in section \ref{chap3:sec:model-overview}, we present an overview of the studied artifact-centric model and the distributed execution of the peer-review process used as a running example in this manuscript. In section \ref{chap3:sec:modelling-artifacts}, we introduce and formally define the concepts of artifact and artifact-type (GMWf). We then present in section \ref{chap3:sec:agents-and-choreography}, the internal structure (architecture and features) of an agent, the notion of accreditation as well as the new artifact-centric and completely decentralised execution model of administrative processes that we propose. Illustrations of section \ref{chap3:sec:agents-and-choreography}'s algorithms are given in section \ref{chap3:sec:choreograpghy-illustration} in order to facilitate their understanding. In the same vein, a prototype system allowing to fully experiment the approach investigated in this chapter is presented in section \ref{chap3:sec:p2ptinywfms}. In section \ref{chap3:sec:discussion}, we discuss the obtained results as well as a positioning of these results in relation to those in the literature. The section \ref{chap3:sec:conclusion} is devoted to the conclusion.

\mySection{Overview of the Artifact-Centric Model Presented in this Thesis}{}
\label{chap3:sec:model-overview}
In this section, a brief description of the artifact-centric model studied in this chapter is given. Furthermore, an overview of the distributed execution of the peer-review process using this model is presented.

\mySubSection{A Description of the Artifact-Centric Model Presented in this Thesis}{}
\label{chap3:sec:model-description}
We outlined in this chapter's introduction that the presented artifact-centric model is based on the asynchronous structured cooperative editing techniques proposed in the work of Badouel et al. As in these works, an artifact is represented by a tree containing "\textit{open nodes}" on some of its leaves, materialising the tasks to be executed or being executed and, an attributed grammar called the \textit{Grammatical Model of Workflow} (GMWf) is used as \textit{artifact type}. The symbols of a given GMWf represent the process tasks and each of its productions represents a scheduling of a subset of these tasks; intuitively, a production given by its left and right hand sides, specifies how the task on the left hand side precedes (must be executed before) those on the right hand side (see sec. \ref{chap3:sec:artifacts-structure}).
When a task is executed on a given site, the corresponding open node in the artifact is closed accordingly (it is said to be \textit{closed}) and the data produced during execution are filled in its attributes; then, one of the GMWf's production having the considered task as left hand side is chosen by the local actor to expand the open node into a subtree highlighting in the form of new open nodes, the new tasks to be executed: this is what editing an artifact consists of.
\begin{figure}[ht!]
	\noindent
	\makebox[\textwidth]{\includegraphics[scale=0.27]{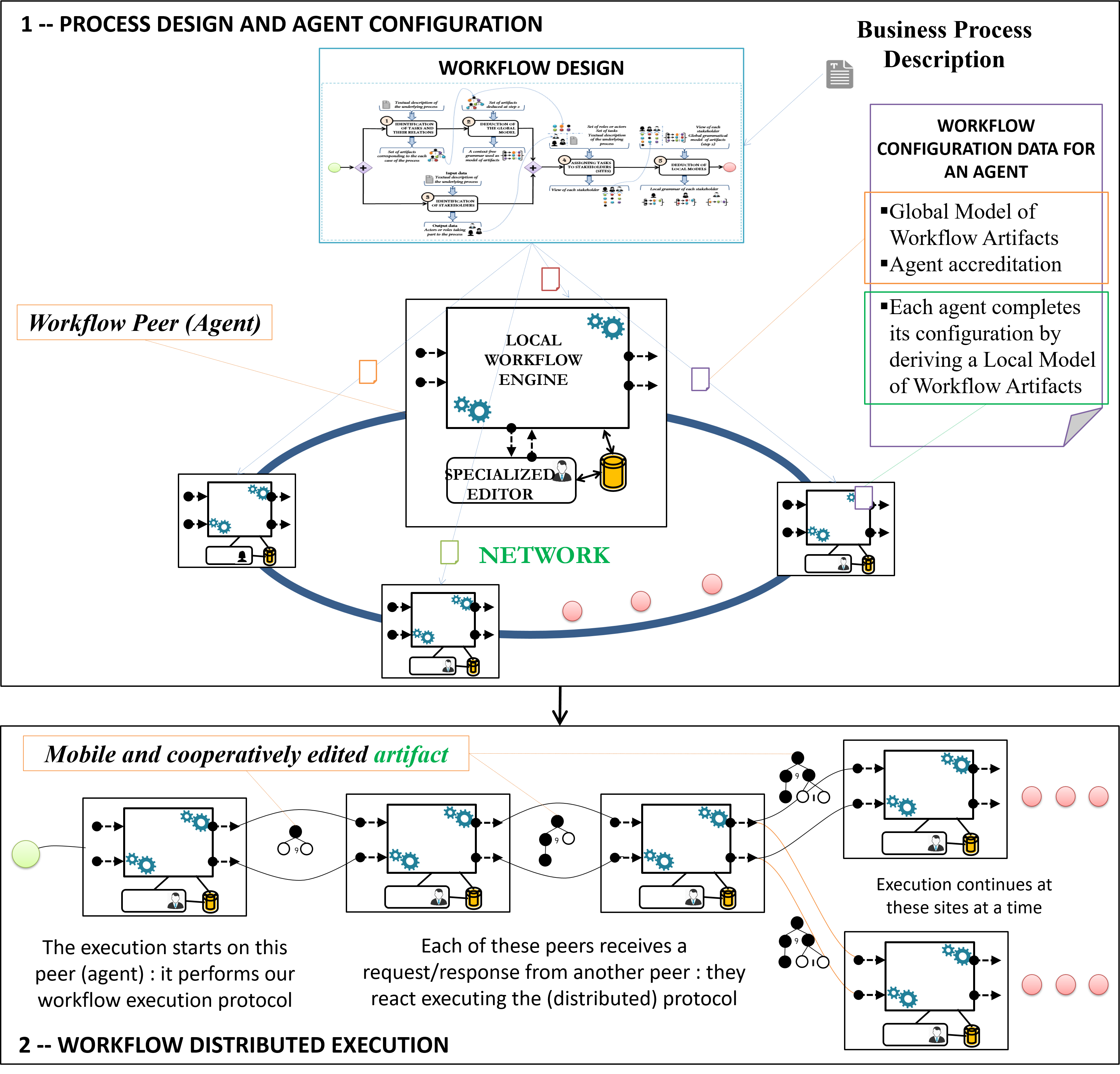}}
	\caption{An overview of the artifact-centric BPM model presented in this chapter.}
	\label{chap3:fig:overview-protocol}
\end{figure}

We are especially interested on administrative processes and the approach we propose for their automation is declined in two steps: derivation of different models (target artifacts and their model, accreditations, etc.) from a textual description of the process and, implementation of a choreography between the agents communicating by asynchronous exchange of artifacts for its execution. More precisely, from the observation that one can analyse the textual description of an administrative business process to exhibit all the possible execution scenarios leading to its business goals, we propose to model each of these scenarios by an annotated tree in which, each node corresponds to a task of the process assigned to a given agent, and each hierarchical decomposition (a node and its sons) represents a scheduling of these tasks: these annotated trees are called \textit{target artifacts}. From these target artifacts, are derived a GMWf (\textit{artifact type}) which contains both the \textit{information model} (modelled by its attributes) and the \textit{lifecycle model} (thanks to the set of its productions) which are two essential notions of the artifact-centric modelling paradigm \cite{hull2013data}. Once the GMWf is obtained, we propose to add organisational information called \textit{accreditations} in this chapter; they aim, as in \cite{badouelTchoupeCmcs, theseTchoupe, tchoupeAtemkeng2, tchoupeZekeng2016, tchoupeZekeng2017, zekengTchoupe2018}, is to enrich the notion of access to different parts of artifacts, by offering a simple mechanism for modelling the generally different perceptions that actors have on processes and their data. With the couple (GMWf, accreditations), each autonomous agent is configured (see fig. \ref{chap3:fig:overview-protocol} (1)) and is ready to proceed to the decentralised execution of the studied process.

The actual execution is a choreography in which the agents are reactive autonomous software components, communicating in P2P mode and are  driven by human agents (actors) in charge of executing tasks. An agent's reaction to the reception of a message (an artifact) consists in the execution of a five-step protocol clearly described in this chapter (see sec. \ref{chap3:sec:the-protocols}). 
This protocol allows it to: (1) \textit{merge} the received artifact with the one it hosts locally in order to consider all updates, (2) \textit{project} the artifact resulting from the merger in order to hide the parts to which the local actor may not have access and highlight the tasks to be locally executed, (3) make the local actor \textit{execute} the revealed tasks and thus edit the potentially partial replica of the artifact obtained after the projection, (4) integrate the new updates into the artifact through an operation called \textit{expansion-pruning} and finally, (5) \textit{diffuse} the updated artifact to other sites for further execution of the process if necessary. 
The agents' operational capabilities allow that, for the execution of a given process, an artifact created by one of them (initially reduced to an open node), moves from site to site to indicate tasks that are ready to be executed at the appropriate time and to provide necessary data (created by other agents) for that execution; the mobile artifact, cooperatively edited by agents, thus "grows" as it transits through the distributed system so formed (see fig. \ref{chap3:fig:overview-protocol} (2)).

\mySubSection{The Running Example: the Peer-Review Process}{}
\label{chap3:sec:running-example}

\mySubSubSection{Description of the Peer-Review Process}{}
\label{chap3:sec:peer-review-description}
The peer-review process \cite{peerReview02} is a common example of administrative business process. We presented a brief description of it inspired by those made in \cite{peerReview02, van2001proclets, badouel14}, in chapter \ref{chap1:artifact-centric-bpm}, section \ref{chap1:sec:running-example}. Described in this way, we will use the peer-review process as an illustrative example in this chapter.

Lets recall that from that description, we have identified all the tasks to be executed, their sequencing, actors involved and the tasks assigned to them. Precisely, four actors are involved: an editor in chief ($EC$) who is responsible for initiating the process, an associated editor ($AE$) and two referees ($R1$ and $R2$).
A summary of tasks assignment was presented in table \ref{tableau:tachesExecutant} (page \pageref{tableau:tachesExecutant}), and orchestration diagrams using BPMN and WF-Net were also presented in figure \ref{chap1:fig:comparing-workflow-languages} (page \pageref{chap1:fig:comparing-workflow-languages}).

\mySubSubSection{Overview of the Peer-Review Process Artifact-Centric Execution using the Model Presented in this Thesis}{}
\label{chap3:sec:peer-review-overview}
To run the peer-review process described above according to the artifact-centric model presented in this chapter, four agents controlled by four actors (\textit{the editor in chief}, \textit{the associated editor} and the \textit{two referees}) will be deployed. 
Each of them will be pre-configured using a global Grammatical Model of Workflow (GMWf) and a set of accreditations. As we will see later, the global GMWf is used as model of artifacts and formally describes all the process tasks to be executed as well as their execution order (see fig. \ref{chap1:fig:comparing-workflow-languages}), and the accreditations set specifies the permissions (\textit{reading}, \textit{writing} and \textit{execution}) of each of the four actors relative to these tasks. 
After the pre-configuration of agents, each of them will derive (by projection \cite{tchoupeAtemkeng2}) a local GMWf which will locally guide the execution of the tasks to guarantee the confidentiality of some workflow data (contained in a mobile artifact) and the consistency of local updates with the global GMWf.

The artifact-centric execution of a scientific paper validation workflow will be triggered on the editor in chief's site, by introducing (in this site) an artifact (an annotated tree) reduced to its root node. Each node of the artifact represents a task and encapsulates an attribute containing its execution status. Therefore at a given time, the whole artifact contains information on already executed tasks and on data produced during their execution; it also exhibits tasks that are ready to be executed.
The analysis of this artifact by the local agent will highlight the expected contributions from the editor in chief. Guided by the local GMWf, he (here tasks are executed by a human) will perform tasks resulting in the consistent updating of the artifact's local copy; meaning that new nodes will be added to the artifact and some of its existing nodes will be updated: this is what we call editing an artifact. Then, this (updated) copy will be immediately analysed by the local agent to determine whether the currently managed process scenario is complete (this is the case when the artifact local copy structure matches one of the target artifacts: all causally dependent tasks have been executed) or not: in this case, all sites on which execution must continue are determined and an execution request is addressed to each of them (the artifact is sent to them).
Figure \ref{chap3:fig:overview-example} sketches an overview of exchanges that can take place between the four agents of the peer-review process when validating a scientific paper using the model presented in the present chapter. The scenario presented there, corresponds to the one in which the paper is pre-validated by the editor in chief and therefore, is analysed by a peer review committee. Note that there may be situations where multiple copies of the artifact are updated in parallel; this is notably the case when they are present on site 3 (first referee) and 4 (second referee).
\begin{figure}[ht!]
	\noindent
	\makebox[\textwidth]{\includegraphics[scale=0.28]{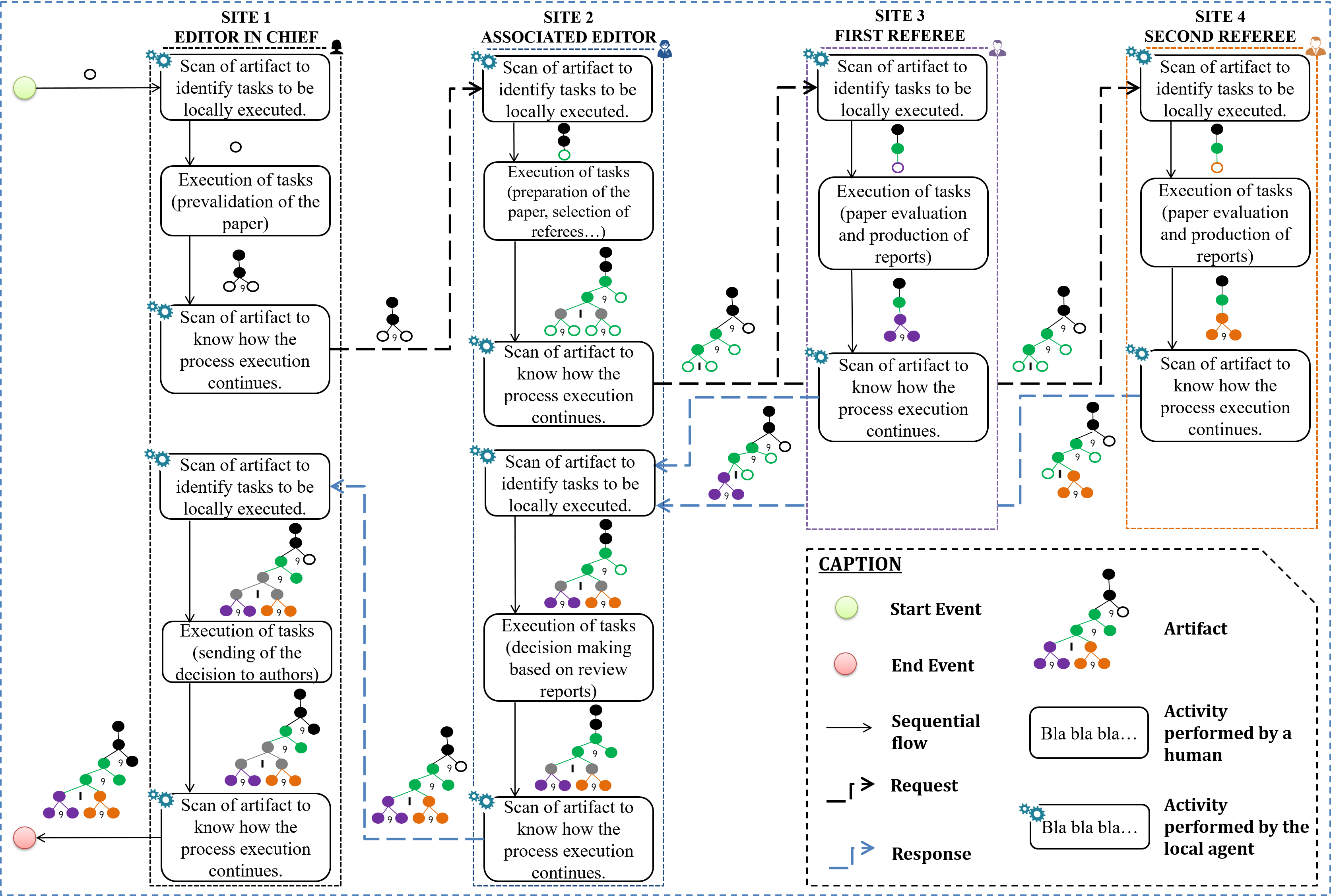}}
	\caption{An overview of the artifact-centric execution of the peer-review process using the model presented in this chapter.}
	\label{chap3:fig:overview-example}
\end{figure}

\mySection{Modelling Artifacts}{}
\label{chap3:sec:modelling-artifacts}

\mySubSection{Artifacts' Structure}{}
\label{chap3:sec:artifacts-structure}
Let's consider an administrative process $\mathcal{P}_{op}$ to be automated. The set $\left\{ \mathcal{S}_{op}^1,\ldots,\mathcal{S}_{op}^k \right\}$ of $\mathcal{P}_{op}$ execution scenarios is known in advance and so, $\mathcal{P}_{op}$ can be specified as any oriented graph with tools like BPMN or as a petri net with tools like YAWL. 
Moreover, each execution scenario of $\mathcal{P}_{op}$ can be modelled using an annotated tree $t_i$. Indeed, starting from the fact that a given scenario $\mathcal{S}_{op}^i$ consists of a set $\mathbb{T}_n = \{X_1, \ldots, X_n\}$ of $n$ (non-recursive) tasks to be executed in a specific order (in parallel or in sequence), one can represent $\mathcal{S}_{op}^i$ as a tree $t_i$ in which each node (labelled $X_i$) potentially corresponds to a task $X_i$ of $\mathcal{S}_{op}^i$, and each hierarchical decomposition (a node and its sons) corresponds to a scheduling: the task associated with the parent node must be executed before those associated with the son nodes; the latter must be executed according to an order - parallel or sequential - that can be specified by particular annotations. Indeed, it is enough to have two annotations "$\fatsemi$" (is sequential to) and "$\parallel$" (is parallel to) to be applied to each hierarchical decomposition. The annotation "$\fatsemi$" (resp. "$\parallel$") reflects the fact that the tasks associated with the son nodes of the decomposition must (resp. can) be executed in sequence (resp. in parallel).

Considering the running example (the peer-review process), the two scenarios that make it up can be modelled using the two annotated trees in figure \ref{chap3:fig:global-artefacts}. In particular, we can see that the tree $art_1$ shows how the task "Receipt and pre-validation of a submitted paper" assigned to the editor in chief ($EC$), and associated with the symbol $A$ (see table \ref{tableau:tachesExecutant}, page \pageref{tableau:tachesExecutant}), must be executed before tasks associated with the symbols $B$ and $D$, that are to be executed in sequence. This annotated tree represents the scenario where the paper received by the editor in chief, is immediately rejected.
\begin{figure}[ht!]
	\noindent
	\makebox[\textwidth]{\includegraphics[scale=0.6]{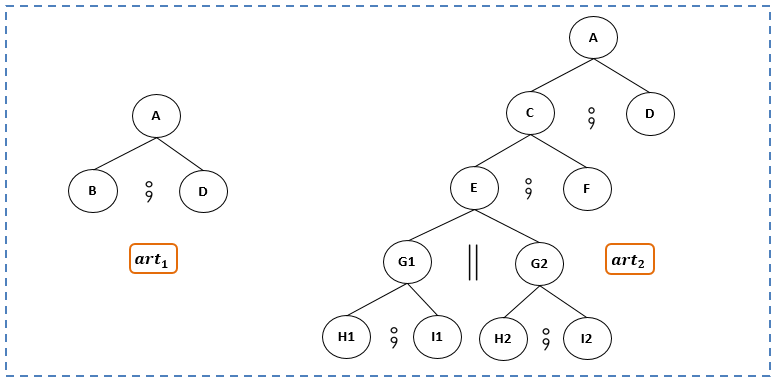}}
	\caption{Target artifacts of a peer-review process.}
	\label{chap3:fig:global-artefacts}
\end{figure}

\mySubSection{Target Artifacts and Grammatical Model of Workflow}{}
\label{chap3:sec:target-artifacts-and-gmwf}
In this chapter, we use the expression \textit{target artifact} to designate the annotated tree $t_i$ modelling a given scenario $\mathcal{S}_{op}^i$ of a given administrative process $\mathcal{P}_{op}$. From the set of target artifacts of a given process, it is possible to derive an abstract grammar\footnote{It is enough to consider the set of target artifacts as a regular tree language: there is therefore a (abstract) grammar to generate them.} that can be enriched to serve as a \textit{artifact type} as defined in \cite{hull2009facilitating}: it is this grammar that we designate by the expression \textit{Grammatical Model of Workflow (GMWf)}.

Let's consider the set $\left\{t_1,\ldots,t_k\right\}$ of target artifacts modelling the $k$ execution scenarios of a given process $\mathcal{P}_{op}$ of $n$ tasks ($\mathbb{T}_n = \{X_1, \ldots, X_n\}$). Each $t_i$ is a derivation tree for an abstract grammar (a GMWf) $\mathbb{G}=\left(\mathcal{S},\mathcal{P},\mathcal{A}\right)$ whose set of symbols is $\mathcal{S}=\mathbb{T}_n$ (all process tasks) and each production $p \in \mathcal{P}$ reflects a hierarchical decomposition contained in at least one of the target artifacts. Each production is therefore exclusively of one of the following two forms: $p: X_0 \rightarrow X_1 \fatsemi \ldots \fatsemi X_n$ or  $p: X_0 \rightarrow X_1 \parallel \ldots \parallel X_n$. The first form $p: X_0 \rightarrow X_1 \fatsemi \ldots \fatsemi X_n$ (resp. the second form $p: X_0 \rightarrow X_1 \parallel \ldots \parallel X_n$) means that task $X_0$ must be executed before tasks $\left\{X_1,\ldots,X_n\right\}$, and these must be (resp. these can be) executed in sequence (resp. in parallel). A GMWf can therefore be formally defined as follows:
\begin{definition}
	\label{defGMWf1}
	A \textbf{Grammatical Model of Workflow} (GMWf) is defined by $\mathbb{G}=\left(\mathcal{S},\mathcal{P},\mathcal{A}\right)$
	where:
	\begin{itemize}
	\item $\mathcal{S}$ is a finite set of \textbf{grammatical symbols} or \textbf{sorts} corresponding to various \textbf{tasks} to be executed in the studied business process; 
	\item $\mathcal{A}\subseteq\mathcal{S}$ is a finite set of particular symbols called \textbf{axioms}, representing tasks that can start an execution scenario (roots of target artifacts), and 
	\item $\mathcal{P}\subseteq\mathcal{S}\times\mathcal{S}^{*}$ is a finite set of \textbf{productions} decorated by the annotations "$\fatsemi$" (is sequential to) and "$\parallel$" (is parallel to): they are \textbf{precedence rules}. 
	A production $P=\left(X_{P(0)},X_{P(1)},\cdots, X_{P(|P|)}\right)$ is either of the form $P: X_0 \rightarrow X_1 \fatsemi \ldots \fatsemi X_{|P|}$, or of the form $P: X_0 \rightarrow X_1 \parallel \ldots \parallel X_{|P|}$ and $\left|P\right|$ 
	designates the length of $P$'s right-hand side.
	A production with the symbol $X$ as left-hand side is called a \textit{X-production}.
	\end{itemize}
\end{definition}

Let's illustrate the notion of GMWf by considering the one generated from an analysis of the target artifacts obtained in the case of the peer-review process (see fig. \ref{chap3:fig:global-artefacts}). The derived GMWf is  $\mathbb{G}=\left(\mathcal{S},\mathcal{P},\mathcal{A}\right)$ in which, the set $\mathcal{S}$ of grammatical symbols is
$\mathcal{S}=\{A, B, C, D, E, F, G1, G2, H1, H2, I1, I2\}$ (see table \ref{tableau:tachesExecutant});
the only initial task (axiom) is $A$ (then $\mathcal{A}=\{A\}$) and the set $\mathcal{P}$ of productions is:
\[ 
\begin{array}{l|l|l|l}
P_{1}:\; A\rightarrow B\fatsemi D & \; P_{2}:\; A\rightarrow C\fatsemi D\; & \; P_{3}:\; C\rightarrow E\fatsemi F\; & \; P_{4}:\; E\rightarrow G1\parallel G2    \\
P_{5}:\; G1\rightarrow H1 \fatsemi I1 & \; P_{6}:\; G2\rightarrow H2 \fatsemi I2\; & \; P_{7}:\; B\rightarrow \varepsilon\; & \; P_{8}:\; D\rightarrow \varepsilon  \\
P_{9}:\; F\rightarrow \varepsilon & \; P_{10}:\; H1\rightarrow \varepsilon & \; P_{11}:\; I1\rightarrow \varepsilon\; & \; P_{12}:\; H2\rightarrow \varepsilon  \\
P_{13}:\; I2\rightarrow \varepsilon &  &  &   \\
\end{array}
\]

For some administrative business processes, there may be special cases where it is not possible to strictly schedule the tasks of a scenario using the two (only) forms of productions selected for GMWf. For example, this is the case for the scenario of a four-task process with tasks $A, B, C$ and $D$, where the task $A$ precedes all others, the tasks $B$ and $C$ can be executed in parallel and precede $D$. 
In these cases, the introduction of a certain number of new symbols known as \textit{(re)structuring symbols} (not associated with tasks) can make it possible to produce a correct scheduling that respects the form imposed on productions. For the previous example, the introduction of a new symbol $S$ allows us to obtain the following productions: $p_1: A \rightarrow S \fatsemi D$, $p_2: S \rightarrow B \parallel C$, $p_3:B \rightarrow \epsilon$, $p_4:C \rightarrow \epsilon$ and $p_5:D \rightarrow \epsilon$, which properly model the required scheduling. 
To deal with such cases, the previously given GMWf definition (definition \ref{defGMWf1}) is slightly adapted by integrating the (re)structuring symbols; the resulting definition is as follows:
\begin{definition} 
	\label{defGMWf2}
	A \textbf{Grammatical Model of Workflow} (GMWf) is defined by $\mathbb{G}=\left(\mathcal{S},\mathcal{P},\mathcal{A}\right)$
	wherein, 
	$\mathcal{P}$ and $\mathcal{A}$ refer to the same purpose as in definition \ref{defGMWf1}, 
	$\mathcal{S}=\mathcal{T} \cup \mathcal{T}_{Struc}$ 
	is a finite set of \textbf{grammatical symbols} or \textbf{sorts} in which, those of $\mathcal{T}$ correspond to \textbf{tasks} of the studied business process, while those of $\mathcal{T}_{Struc}$ are (re)structuring symbols.
\end{definition}

\mySubSection{Artifact Type and Artifact Edition}{}
\label{chap3:sec:gmwf-as-artifact-type}
As formalised in definition \ref{defGMWf2}, a GMWf perfectly models the tasks and control flow of administrative processes (lifecycle model). To remain faithful to the artifact-centric philosophy, the GMWf definition must be adjusted to be able to use it as an artifact type. In particular, it is necessary to equip it with tools allowing to represent the information model (the data) of processes as well as the dynamic (evolutionary) character of artifacts.

\mySubSubSection{Modelling the Information Model of Processes with GMWf}{}
\label{chap3:sec:gmwf-information-model}
The structure of the consumed and produced data by business processes differs from one process to another. It is therefore not easy to model them using a general type, although several techniques to do so have emerged in recent years \cite{badouel14}. For the work presented in this chapter, tackling the data structure of automated processes has no proven interest because, it does not bring any added value to the proposed model: a representation of these data using a set of variables is largely sufficient.

To represent the potential consumed and produced data by the tasks of a process modelled using GMWf, we use the notion of \textit{attribute} embedded in the nodes associated with tasks. To take them into account, we adjust for the last time, the definition of GMWf. We thus attach to each symbol, an attribute named $status$, allowing to store all the data of the associated task; its precise type is left to the discretion of the process designer. However, for the purposes of this work, we will consider it a string. The new definition of GMWf is thus the following one:
\begin{definition} 
	\label{defGMWf3}
	A \textbf{Grammatical Model of Workflow} (GMWf) is defined by $\mathbb{G}=\left(\mathcal{S},\mathcal{P},\mathcal{A}\right)$
	wherein, 
	$\mathcal{S}$, $\mathcal{P}$ and $\mathcal{A}$ refer to the same purpose as in definition \ref{defGMWf2}.
	Each grammatical symbol $X\in\mathcal{S}$ is associated with an attribute named \textbf{\textit{status}} of type string, that can be updated when task $X$ is executed; $\textbf{X.status}$ provides access (read and write) to its content.
\end{definition}

A GMWf is therefore ultimately an attributed grammar whose instances represent the different execution scenarios of the underlying business process. In artifact-centric models, the artifact used as a communication medium between agents executing the tasks, must represent at each moment, the execution state of the underlying process. As defined up to now, GMWf models do not satisfy this second concern: they cannot therefore be used as artifact types. We will now equip them with tools to allow them to endow their instances (the artifacts) with the ability to report about the execution state of the process they represent.

\mySubSubSection{Artifact Type}{}
\label{chap3:sec:artifact-type}
For each task, it is important to know whether or not it has already been executed; if not, it is also important to know whether or not it is ready to be executed. Recall also that, we model the execution of processes as the desynchronised cooperative editing of mobile artifacts (which are exchanged by agents). This implies that the artifact-centric model of this chapter considers that, an artifact is a structured document that is initially empty, and which is completed as it circulates between the agents. Contrary to the models in the literature, at each moment, the artifact thus contains only a (potentially empty) part of the lifecycle model of the process. This is why we have chosen not to represent it as a (tree) state machine but rather as an annotated tree that is incrementally built in accordance with an attributed grammar.

Concretely, an \textit{artifact} is an annotated tree that potentially contains buds (this is the equivalent of the notion of structured document being edited as presented in chapter \ref{chap2:structured-editing-artifact-type}). A \textit{bud} or \textit{open node} is a typed leaf node indicating in an artifact, a place where an edition is possible; i.e. a node associated with a task that has not yet been executed. A bud can be unlocked (\textit{unlocked bud}) or locked (\textit{locked bud}) depending on whether the task associated with it is ready to be executed\footnote{A task is ready to be executed if all the tasks that precede it according to the precedence constraint set have already been executed and the agent that currently holds the mobile artifact, have the necessary accreditation to trigger its execution.} or not. More formally, a \textit{bud of type $X \in \mathcal{S}$} is a leaf node labelled either by $X_{\overline{\omega}}$ or by $X_\omega$ depending on its state (\emph{locked} or \emph{unlocked}). An artifact containing no buds is said to be \textit{closed}. Such an artifact, symbolises the end of tasks execution with respect to the agent hosting the artifact. An example of an artifact related to the peer-review process and containing buds is shown in figure \ref{chap3:sec:artifact-with-buds}. In this one, we can see that the tasks associated with symbols $A$ and $C$ have already been executed. Task $E$ is ready to be executed while tasks $F$ and $D$ are not ready to be executed yet.
\begin{figure}[ht!]
	\noindent
	\makebox[\textwidth]{\includegraphics[scale=0.4]{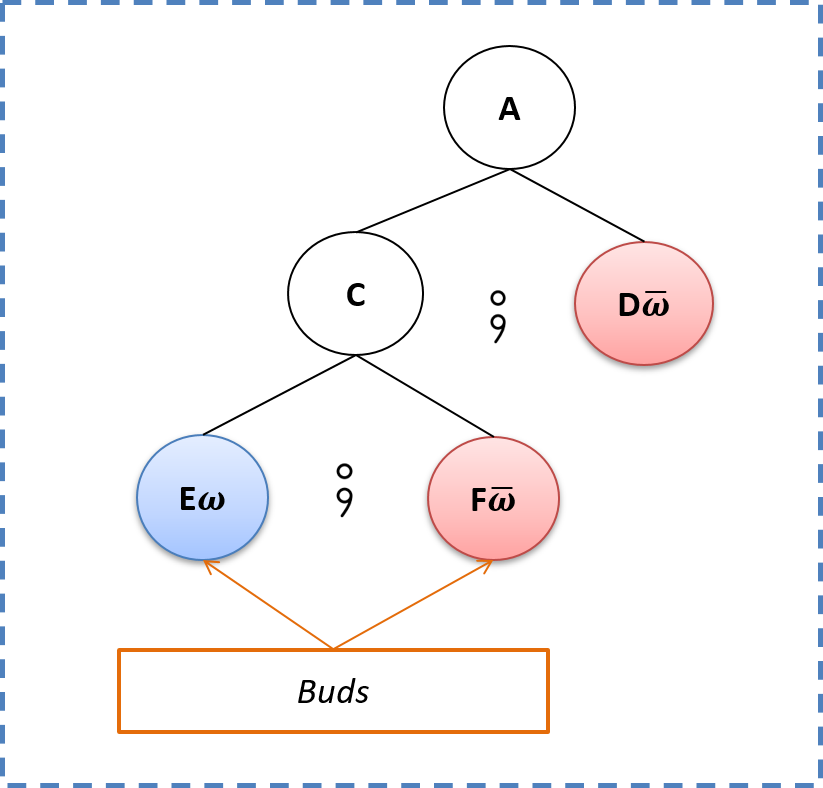}}
	\caption{An intentional representation of an annotated artifact containing buds.}
	\label{chap3:sec:artifact-with-buds}
\end{figure}

From the thus given definition of (mobile) artifact, it is clear that an artifact is updated only at the level of its leaves and therefore, it only "grows" (positive editing). Knowing moreover that, the correct and complete execution of a given administrative process corresponds to the execution of one of its scenarios, we deduce that: for a process $\mathcal{P}_{op}$ whose GMWf is $\mathbb{G}=\left(\mathcal{S},\mathcal{P},\mathcal{A}\right)$, a given mobile artifact, is a prefix to one of its target artifacts. Thus, the type (model) of this artifact is a grammar $\mathbb{G}_{\Omega}=(\mathcal{S}\cup\mathcal{S}_{\omega},\mathcal{P}\cup\mathcal{S}_{\Omega},\mathcal{A}\cup\mathcal{A}_{\omega})$ obtained by extending $\mathbb{G}$ (for bud recognition and recognition of all possible prefixes of target artifacts) as follows:
\begin{enumerate}
\item For all sort $X$, add in the set $\mathcal{S}$ of sorts, two new sorts $X_{\overline{\omega}}$ and $X_{\omega}$;

\item For all new sort $X_{\omega}$ added to $\mathcal{S}$, add in the set $\mathcal {P}$ of productions two new $\varepsilon$-productions $X_{\Omega} : X_{\omega} \rightarrow \varepsilon$ and $X_{\overline{\Omega}} : X_{\overline{\omega}} \rightarrow \varepsilon$; 
we then have: $\mathcal{S}_{\omega}=\{X_{\overline{\omega}}, ~X_{\omega},~ X\in\mathcal {S}\}$, $ \mathcal{A}_{\omega}=\{X_{\overline{\omega}}, ~X_{\omega},~ X\in\mathcal{A}\} $ $ and $ $\mathcal{S}_{\Omega} = \{X_{\Omega} : X_{\omega} \rightarrow \varepsilon, ~X_{\overline{\Omega}} : X_{\overline{\omega}} \rightarrow \varepsilon,~ X_{\overline{\omega}}~and~X_{\omega} \in \mathcal{S}_{\omega}\}$.
\end{enumerate}

\mySubSubSection{Artifact Edition}{}
\label{chap3:sec:artifact-edition}
If we still consider a running process $\mathcal{P}_{op}$ whose GMWf is $\mathbb{G}=\left(\mathcal{S},\mathcal{P},\mathcal{A}\right)$, then, the \textit{editing} of an artifact $t$ circulating between agents consists of developing one or more of its buds into a subtree while updating their \textit{status} attributes. Concretely, for a bud ${X_\omega}$ of the said artifact one can:
\begin{enumerate}
\item Execute the task associated with $X$; 

\item Choose an $X$-production $P \in \mathcal{P}$ to be used for the development of ${X_\omega}$;

\item If $P$ is of the form $P:~X \rightarrow X_1 \parallel X_2 \parallel \ldots \parallel X_{|P|}$ (resp. $P:~X \rightarrow X_1 \fatsemi X_2 \fatsemi \ldots \fatsemi X_{|P|}$) then, create $|P|$ buds $X_{1\omega}, X_{2\omega}, \ldots, X_{|P|\omega}$ (resp. $X_{1\omega}, X_{2\overline{\omega}}, \ldots, X_{|P|\overline{\omega}}$) respectively of type $X_1, X_2, \ldots , X_{|P|}$, and replace in the artifact, the considered bud ${X_\omega}$ by the parallel (resp. sequential) subtree $X[X_{1\omega}, X_{2\omega}, \ldots, X_{|P|\omega}]$\footnote{The tree coded by $X[X_{1\omega}, X_{2\omega}, \ldots, X_{|P|\omega}]$ is the one whose root is labelled $X$ and has as sons, $|P|$ nodes labelled by $X_{1\omega}, X_{2\omega}, \ldots, X_{|P|\omega}$ respectively.} (resp. $X[X_{1\omega}, X_{2\overline{\omega}}, \ldots, X_{|P|\overline{\omega}}]$);
%4) if $P$ is of the form $P:~X \rightarrow X_1 \fatsemi X_2 \fatsemi \ldots \fatsemi X_{|P|}$ then create $|P|$ buds $X_{1\omega}, X_{2\overline{\omega}}, \ldots, X_{|P|\overline{\omega}}$ respectively of type $X_1, X_2, \ldots , X_{|P|}$ and replace in the artifact, the considered bud ${X_\omega}$ by the sequential subtree $X[X_{1\omega}, X_{2\overline{\omega}}, \ldots, X_{|P|\overline{\omega}}]$;

\item Update the execution status of the task associated with $X$: $X.status$ $=$ $"bla~bla~\ldots"$.
Updating $t$ results in an artifact $t^{maj}$ and we note $t \leq t^{maj}$.
\end{enumerate}

Although it is obvious, it seems important to clarify that the editing of an artifact is only a consequence of process tasks' execution by actors located on agents. We can therefore imagine this scenario for the peer-review process (see fig. \ref{chap3:fig:artifact-edition}): the associated editor who received a request from the editor in chief to peer-review a given article, has executed task $C$ (i.e. he has appraised the paper and formatted it to prepare the peer-review). Through a dedicated tool (a specialised editor), he has been invited to submit a report on the execution of the said task (via the filling of a form for example). The submission of that report, in which he may have provided a copy of the formatted paper as well as comments for referees, will cause (in background) the update of the artifact. The retrieved data will thus be stored in the $status$ attribute of task $C$ and the bud $C_\omega$ will be extended into a subtree as described above (the only production available for this purpose is $P_{3}: C\rightarrow E\fatsemi F$).
\begin{figure}[ht!]
	\noindent
	\makebox[\textwidth]{\includegraphics[scale=0.4]{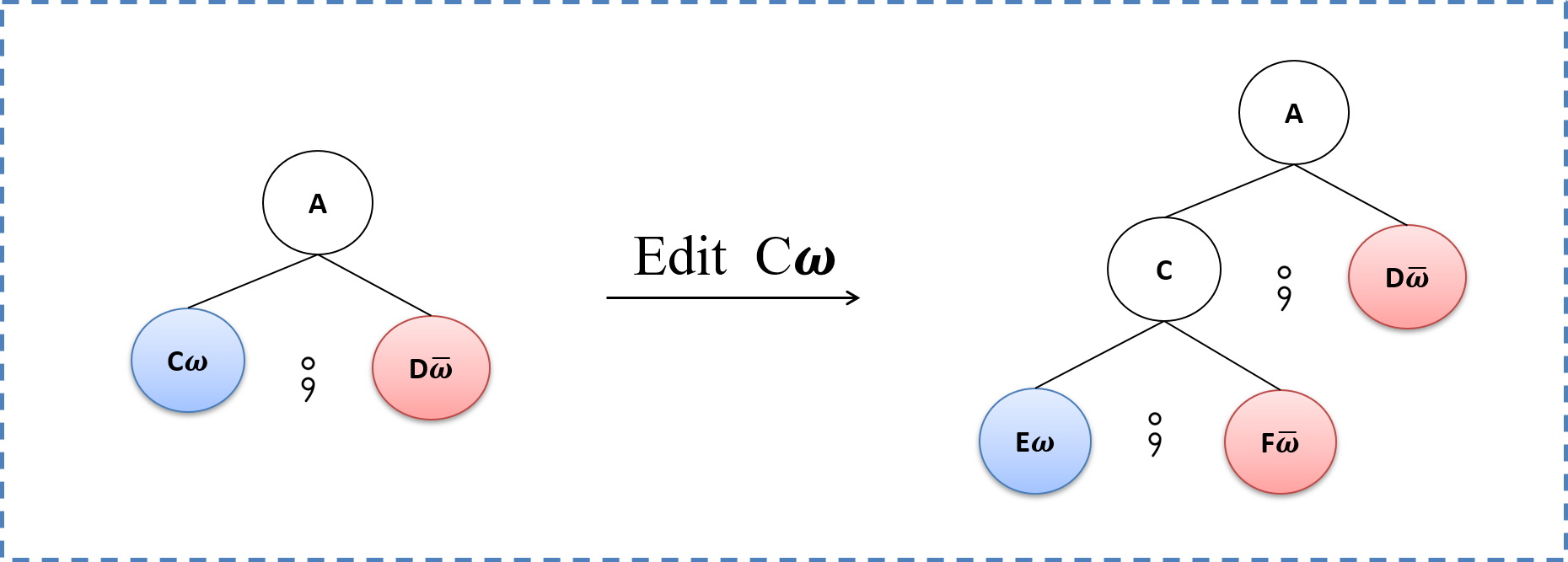}}
	\caption{An example of artifact edition: the bud $C_\omega$ is extended in a subtree.}
	\label{chap3:fig:artifact-edition}
\end{figure}

\mySection{Agent and choreography}{}
\label{chap3:sec:agents-and-choreography}

Now that we have formally defined the structure and the editing model of artifacts, let's focus on the structure of agents that oversee the execution of tasks and update the artifact accordingly, as well as on the artifact-centric choreography implemented between them. 

\mySubSection{Relations between Agent, Actor and Choreography}{}
\label{chap3:sec:agent-definition}
We borrowed the term \textit{agent} from \cite{lohmann2010artifact}. In the case of our study, an \textit{agent} (which we also call a \textit{peer}) is a software component, installed at a given site, piloted by a human agent called \textit{actor} (the tasks of the processes we handle are executed by humans) and capable of interacting with other agents by service invocation (message exchange). An agent is completely \textit{autonomous}: i.e. it encapsulates all the data and functions necessary for the execution of the tasks assigned to it, or precisely, tasks assigned to the actor piloting it. The agent is \textit{reactive}: it reacts in the same way to each message it receives by executing a well-defined protocol that goes from the analysis of the received message (artifact) to the possible transmission of other messages. As announced in the introduction, each message contains a collectively edited \textit{mobile artifact}. For the execution of a given process, the choreography is therefore a result of the messages (artifact replicas) exchanges between the agents involved and of the reaction of the latter to the reception of messages.

\mySubSection{Structure of an Agent}{}
\label{chap3:sec:agent-structure}

An agent is built to be able to fully manage the lifecycle (creation, storage, edition/execution) of a given business process' artifacts. Thus, an agent is made up of three major software components: a local workflow engine (LWfE), a specialised graphical editor and a storage device (see fig. \ref{chap3:fig:simplify-architecture-peer}).
\begin{figure}[ht!]
	\noindent
	\makebox[\textwidth]{\includegraphics[scale=0.45]{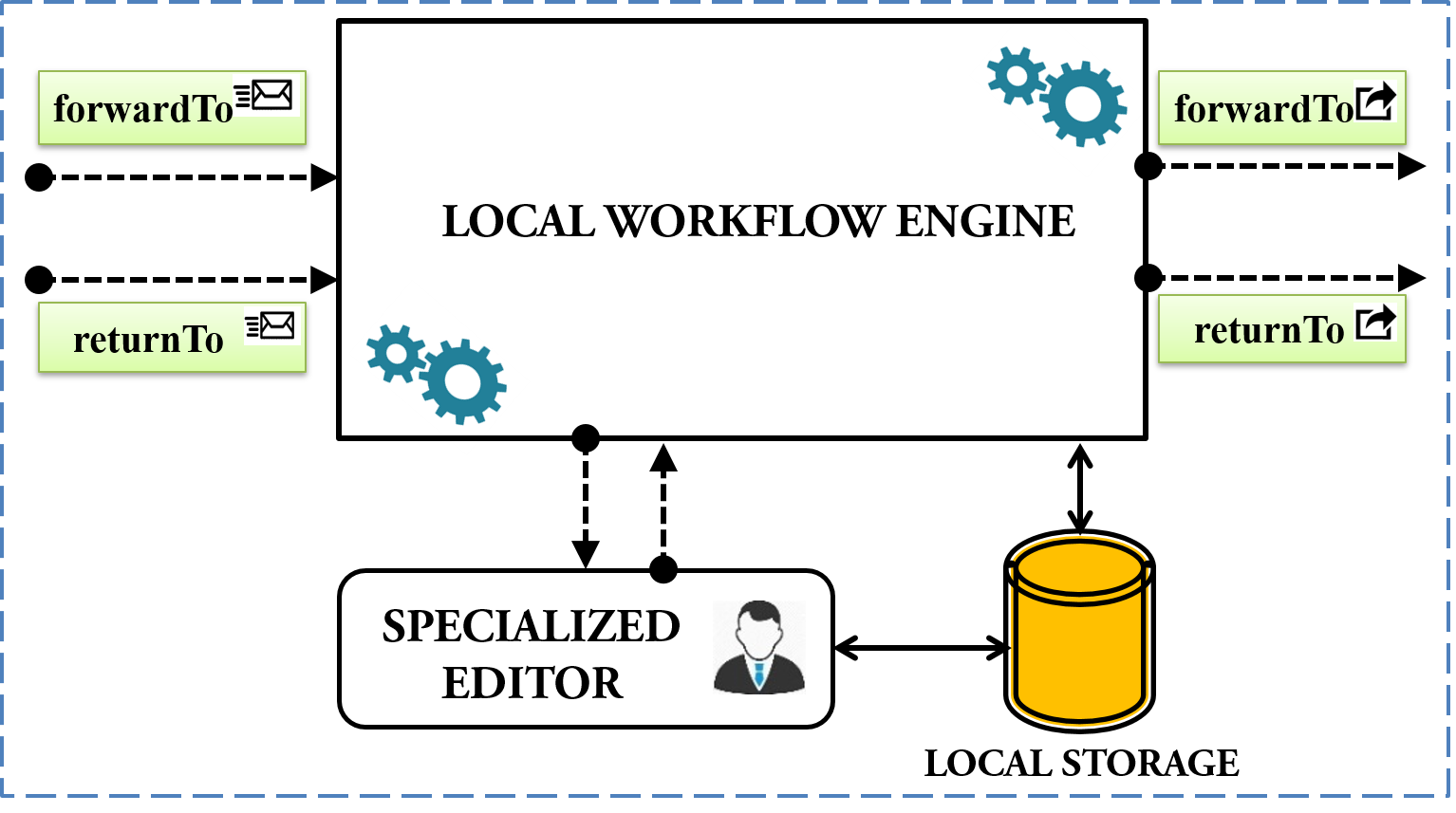}}
	\caption{Simplified architecture of an agent.}
	\label{chap3:fig:simplify-architecture-peer}
\end{figure}

\mySubSubSection{The Local Workflow Engine}{}
\label{chap3:sec:the-local-workflow-engine}

The local workflow engine (LWfE) is the main component of an agent. 
It receives messages from other agents and reacts by executing a well defined protocol (see sec. \ref{chap3:sec:the-protocols}). It communicates with the other agent's engines via its communication interface that exposes four services : two input services or \textit{provided services} (\textit{returnTo} and \textit{forwardTo}) connected to two corresponding output services or \textit{required services} (\textit{returnTo} and \textit{forwardTo}) so that: 

\begin{itemize}
	\item The invocation by an agent $j$ of the service \textit{forwardTo} offered by an agent $i$, causes on $i$, the execution of its corresponding input service \textit{forwardTo}. This service makes it possible to send a \textbf{request} from agent $j$ to agent $i$. The request contains the replica of the mobile artifact located on agent $j$. This artifact must contain buds to be completed (executed) by actor $A_i$ (the human agent piloting agent $i$). 
	\item The invocation by an agent $i$ of the service \textit{returnTo} offered by an agent $j$, causes the execution on $j$, of its corresponding input service \textit{returnTo}. This service allows agent $i$ to return the \textbf{response} to a request previously received from agent $j$. As the request, the response contains the replica of the mobile artifact located on agent $i$.
\end{itemize}

\mySubSubSection{The Storage Device}{}
\label{chap3:sec:the-storage-device}

A database (DB) of documents (a JSON\footnote{JavaScript Object Notation, \url{http://www.json.org}, \url{https://www.mongodb.com}, visited the 04/04/2020.} DB for example) is used by the LWfE to store an agent's configuration and data (especially artifacts) that it handles. 	
	
\mySubSubSection{The Specialised Editor}{}
\label{chap3:sec:the-specialised-editor}
	
Each agent provides a specialised editor (preferably WYSIWYG\footnote{What You See Is What You Get.}) that allows its actor (the pilot) to execute tasks. More precisely, the specialised editor allows the actor to view the tasks that are assigned to him, those ready to be executed, and when he has executed a task, it gives him the means to record an execution report. Any (editing) action carried out by the local actor via the specialised editor, causes the consistent update (as presented in section \ref{chap3:sec:artifact-edition}) of the mobile artifact local replica.

The specialised editor is particularly important as it guarantees controlled access to the artifact. Indeed, as announced in the introduction and following the steps of \cite{hull2009facilitating}, for reasons of confidentiality/security, actors do not necessarily have the right to access all information relating to the execution of a process in which they are involved. It is therefore important to provide a mechanism for regulating access to this information (stored in the artifact). In our case, we define this mechanism under the name \textit{accreditation} and we include it in the configuration of an agent in the same way as the GMWf of the studied process.

\mySubSection{Concepts of Accreditation, Partial Replica of an Artifact and Local GMWf}{}
\label{chap3:sec:accreditation-partial-replica}

\mySubSubSection{Concept of Accreditation}{}
\label{chap3:sec:accreditation}

Let's consider a process $\mathcal{P}_{op}$ and its GMWf $\mathbb{G}=\left(\mathcal{S},\mathcal{P},\mathcal{A}\right)$. The accreditation of an agent provides information on the rights (permissions) its actor has on each sort (task) of $\mathbb{G}$. 
To simplify, the nomenclature of rights manipulated here is inspired by the one used in Unix-like operating systems. Three types of accreditation are then defined: accreditation in reading \textit{(r)}, in writing \textit{(w)} and in execution \textit{(x)}. 
\begin{enumerate}
	\item \textit{Accreditation in reading \textit{(r)}}: when an agent is accredited in reading on a sort $X$, its actor has the right to know if the associated task is executed. Moreover, he can access its execution status.
	We call an agent's (actor's) \textbf{\textit{view}} the set of sorts on which it (he) is accredited in reading.
	\item \textit{Accreditation in writing \textit{(w)}}: when an agent is accredited in writing on a sort $X$, its actor can execute the associated task. 
	Note that a task can be executed only by exactly one actor: for a given sort, a single agent is accredited in writing; this is an important point of the model which guarantees the absence of execution conflicts. 
	Since the dedicated editors for "updating artifacts" are of type WYSIWYG (see sec. \ref{chap3:sec:the-specialised-editor}), any agent accredited in writing on a symbol must therefore be accredited in reading on it. 
	\item \textit{Accreditation in execution \textit{(x)}}: an agent accredited in execution on a sort $X$ is authorised to ask the agent which is accredited in writing on it, to execute it. Note that this request can be made without the agent being accredited in reading\footnote{In fact, as we will see later (sec. \ref{chap3:sec:the-protocols} (\textit{the diffusion protocol}, page \pageref{chap3:sec:execution-protocol-diffusion})), it is an automatic task (of the agent) that sends the execution request and not the actor.} on the considered sort.
\end{enumerate}

\noindent More formally, an accreditation is defined as follows:

\begin{definition} \label{defSyllabaire}
	An \textbf{accreditation} $\mathcal{A}_{A_i}$ defined on the set $\mathcal{S}$ of grammatical symbols for an agent $i$ piloted by an actor $A_i$, is a triplet $\mathcal{A}_{A_i}=\left(\mathcal{A}_{A_i(r)},\mathcal{A}_{A_i(w)},\mathcal{A}_{A_i(x)}\right)$ such that, 
	$\mathcal{A}_{A_i(r)} \subseteq \mathcal{S}$ also called \textbf{view} of $i$ (or \textbf{view} of $A_i$), is the set of symbols on which $i$ is accredited in reading, 
	$\mathcal{A}_{A_i(w)} \subseteq \mathcal{A}_{A_i(r)}$ is the set of symbols on which $i$ is accredited in writing and  
	$\mathcal{A}_{A_i(x)} \subseteq \mathcal{S}$ is the set of symbols on which $i$ is accredited in execution.
\end{definition}

The accreditations of various agents must be produced by the workflow designer just after modelling the scenarios in the form of target artifacts. From the task assignment for the peer-review process in the running example (see table \ref{tableau:tachesExecutant}), it follows that the accreditation in writing of the editor in chief is $\mathcal{A}_{EC(w)}=\{A, B, D\}$, that of the associated editor is $\mathcal{A}_{AE(w)}=\{C, E, F\}$ and that of the first (resp. the second) referee is $\mathcal{A}_{R_1(w)}=\{G1, H1, I1\}$ (resp. $\mathcal{A}_{R_2(w)}=\{G2, H2, I2\}$).
Even more, since the editor in chief can only perform the task $D$ if the task $C$ is already executed (see artifacts $art_1$ and $art_2$, fig. \ref{chap3:fig:global-artefacts}), in order for the editor in chief to be able to ask the associated editor to perform this task, it (the agent) must be accredited in execution on it; so we have $\mathcal{A}_{EC(x)}=\{C\}$.
Moreover, in order to be able to access all the information on the peer-review evaluation of a paper (task $C$) and to summarise the right decision to send to the author, the editor in chief must be able to consult the reports (tasks $I1$ and $I2$) and the messages (tasks $H1$ and $H2$) of the different referees, as well as the final decision taken by the associated editor (task $F$). These tasks, added to $\mathcal{A}_{EC(w)}$\footnote{Recall that in our case, we use WYSIWYG tools and therefore, one can only execute what he see.} constitute the set $\mathcal{A}_{EC(r)}=\mathcal{V}_{EC}=\{A, B, C, D, H1, H2, I1, I2, F\}$ of tasks on which, it is accredited in reading. By doing so for each of other agents, we deduce the accreditations represented in table \ref{tableau:vuesActeurs}.
\begin{table}[ht]
	\centering
	\caption{Accreditations of the different agents taking part in the peer-review process.}
	\label{tableau:vuesActeurs}
	\begin{tabular}[t]{|m{3.5cm}|m{10.3cm}|}
		\hline
		\textbf{Agent} & \textbf{Accreditation} \\
		\hline
		Editor in Chief ($EC$) & $\mathcal{A}_{EC}=\left(\{A, B, C, D, H1, H2, I1, I2, F\}, \{A, B, D\}, \{C\}\right)$ \\
		\hline
		Associated Editor ($AE$) & $\mathcal{A}_{AE}=\left(\{A, C, E, F, H1, H2, I1, I2\}, \{C, E, F\}, \{G1, G2\}\right)$ \\
		\hline
		First referee ($R1$) & $\mathcal{A}_{R1}=\left(\{C, G1, H1, I1\}, \{G1, H1, I1\}, \emptyset\right)$ \\
		\hline
		Second referee ($R2$) & $\mathcal{A}_{R2}=\left(\{C, G2, H2, I2\}, \{G2, H2, I2\}, \emptyset\right)$ \\
		\hline
	\end{tabular}
\end{table}

~

\noindent\textbf{\textit{In summary, what is the workflow model ?}}

To summarise, we state that in the artifact-centric model presented in this chapter, an administrative process $\mathcal{P}_{op}$ is completely specified using a triplet $\mathbb{W}_f=\left(\mathbb{G}, \mathcal{L}_{P_k}, \mathcal{L}_{\mathcal{A}_k} \right)$ called \textit{a Grammatical Model of Administrative Workflow Process} (GMAWfP) and composed of: a GMWf, a list of actors (agents) and a list of their accreditations. 
The GMWf is used to describe all the tasks of the studied process and their scheduling, while the list of accreditations provides information on the role played by each actor involved in the process execution.  
A GMAWfP can then be formally defined as follows:
\begin{definition}
	\label{defMGSPWfA}
	A \textbf{Grammatical Model of Administrative Workflow Process} (GMAWfP) $\mathbb{W}_f$ for a given business process, is a triplet $\mathbb{W}_f=\left(\mathbb{G}, \mathcal{L}_{P_k}, \mathcal{L}_{\mathcal{A}_k} \right)$
	wherein $\mathbb{G}$ is the studied process (global) GMWf, $\mathcal{L}_{P_k}$ is the set of $k$ agents taking part in its execution and $\mathcal{L}_{\mathcal{A}_k}$ represents the set of these agents' accreditations. 
\end{definition}

\mySubSubSection{Concept of Partial Replica of an Artifact}{}
\label{chap3:sec:partial-replica}

To effectively ensure that actors only have access to information of proven interest to them, each agent let them access only to a potentially partial replica $t_{\mathcal{V}_i}$ of the mobile artifact $t$. 
The $t$'s partial replicas are obtained by projections according to the views of each actor. 
A partial replica $t_{\mathcal{V}_i}$ of $t$ according to the \textit{view} $\mathcal{V}_{A_i} =  \mathcal{A}_{A_i(r)} $, is a partial copy of $t$ obtained by means of the so-called \textit{projection operator} denoted $\pi$ as presented below. 

Technically, the projection $t_{\mathcal{V}_i}$ of an artifact $t$ according to the view $\mathcal{V}_{i} = \mathcal{A}_{A_i(r)}$ is obtained by deleting in $t$ all nodes whose types do not belong to $\mathcal{V}_{i}$. In our case, the main challenges in this operation are:
\begin{enumerate}
	\item[\textbf{(1)}] nodes of $t_{\mathcal{V}_i}$ must preserve the previously existing execution order between them in $t$,
	\item[\textbf{(2)}] $t_{\mathcal{V}_i}$ must be build by using exclusively the only two forms of production retained for GMWf and
	\item[\textbf{(3)}] $t_{\mathcal{V}_i}$ must be unique in order to ensure the continuation of process execution (see sec. \ref{chap3:sec:the-protocols}).
\end{enumerate}

The projection operation is noted $\pi$. Inspired by the one proposed in \cite{badouelTchoupeCmcs}, it projects an artifact by preserving the hierarchy (father-son relationship) between nodes of the artifact (it thus meets challenge \textbf{(1)}); but in addition, it inserts into the projected artifact when necessary, new additional \textit{(re)structuring symbols } (accessible in reading and writing by the agent for whom the projection is made). This enables it to meet challenge \textbf{(2)}. The details of how to accomplish the challenge \textbf{(3)} are outlined immediately after the algorithm (algorithm \ref{chap3:algo:artifact-projection}) is presented.
\begin{figure}[ht!]
	\noindent
	\makebox[\textwidth]{\includegraphics[scale=0.35]{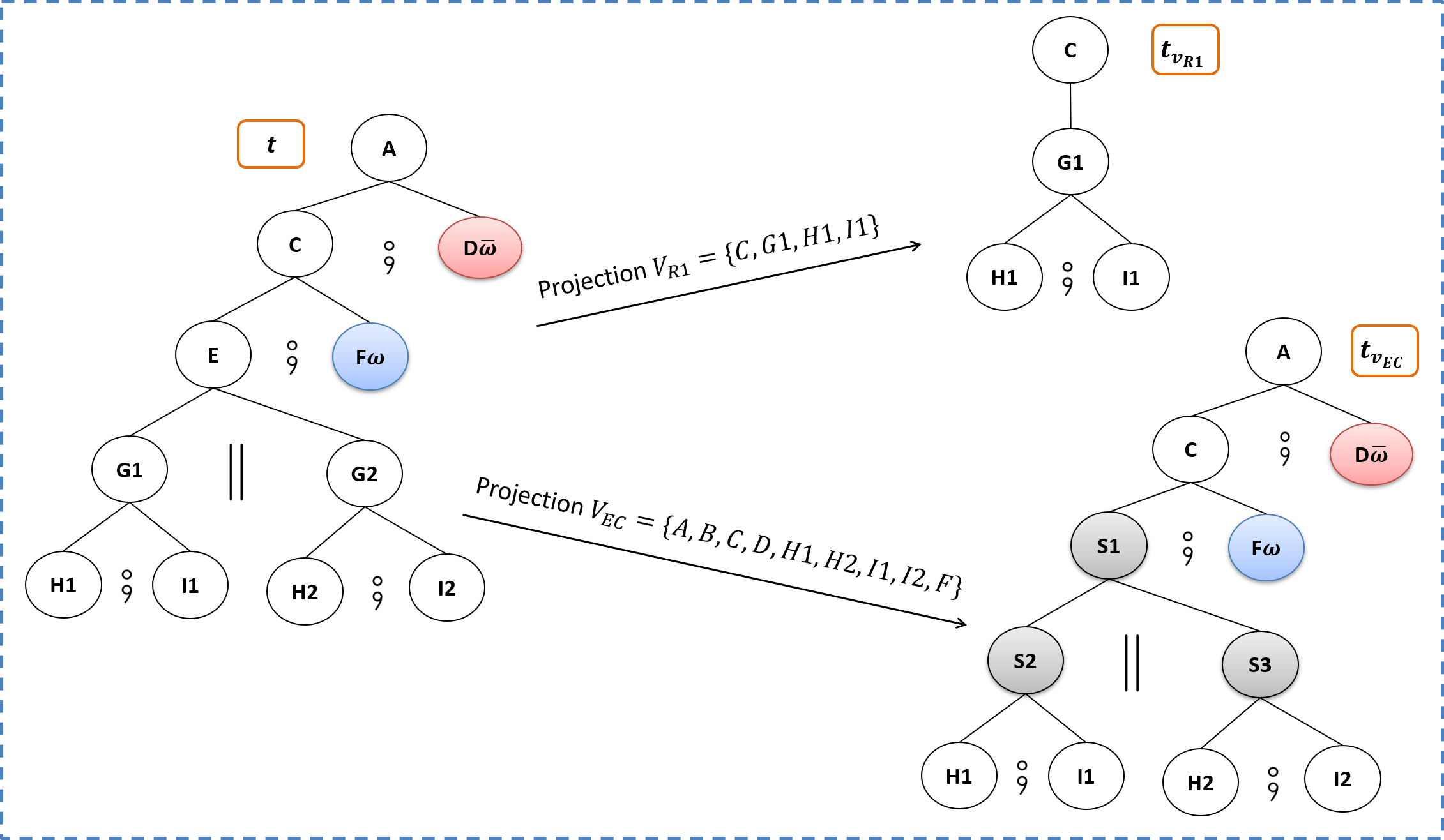}}
	\caption{Example of projections made on an artifact and partial replicas obtained.}
	\label{chap3:fig:partial-replicas}
\end{figure}

Figure \ref{chap3:fig:partial-replicas} illustrates the projection of an artifact of the peer-review process relatively to the $R1$ (first referee) and $EC$ (Editor in Chief) agent views. Note the presence in $t_{\mathcal{V}_{EC}}$ of new (re)structuring symbols (in gray). These last ones make it possible to avoid introducing in $t_{\mathcal{V}_{EC}}$, the production $p: C \rightarrow H1 \fatsemi I1 \parallel H2 \fatsemi I2 \fatsemi F$ whose form does not correspond to the two forms of production retained for the GMWf writing\footnote{Note that this production specifies in its right-hand side that we must have parallel and sequential treatments.
Inserting $S1$, $S2$ and $S3$ allows to rewrite $p$ in four productions $p1: C \rightarrow S1 \fatsemi F$, $p2: S1 \rightarrow S2 \parallel S3$, $p3: S2 \rightarrow H1 \fatsemi I1$ and $p4: S3 \rightarrow H2 \fatsemi I2$.}.

~

\noindent\textbf{\textit{The algorithm}}

Let's consider an artifact $t$ and note by $n=X\left[t_1,\ldots,t_m\right]$ a node of $t$ labelled with the symbol $X$ and having $m$ sub-artifacts $t_1,\ldots,t_m$. Note also by $p_n$, the production of the GMWf that was used to extend node $n$; the type of $p_n$ is either \textit{sequential} (i.e. $p_n: X \rightarrow X_1 \fatsemi \ldots \fatsemi X_m$ where $X_1,\ldots,X_m$ are the roots of the sub-artifacts $t_1,\ldots,t_m$) or \textit{parallel} ($p_n: X \rightarrow X_1 \parallel \ldots \parallel X_m$). 
Concretely, to project $t$ according to a given view $\mathcal{V}$ (i.e to find $\mathit{projs_t}=\pi_{\mathcal{V}}\left(t \right)$), the recursive processing presented in algorithm \ref{chap3:algo:artifact-projection}, is applied to the root node $n=X\left[t_1,\ldots,t_m\right]$ of $t$.

\begin{algorithm}
\small
\caption{Algorithm to project a given artifact according to a given view.}
\label{chap3:algo:artifact-projection}
\begin{mdframed}[style=MyFrame]
	%\begin{itemize}[leftmargin=*]
	{\large\textbullet} $~$ \textbf{If symbol $X$ is visible ($X \in \mathcal{V}$)} then :
	%\begin{enumerate}[leftmargin=*]
	
	\textbf{1.}$~$ $n$ is kept in the artifact;
	
	\textbf{2.}$~$ For each sub-artifact $t_i$ of $n$, having node $n_i=X_i\left[t_{i_1},\ldots,t_{i_k}\right]$ as root (of which $p_{n_i}$ is the production that was used to extend it), the following processing is applied :
	%\begin{itemize}[leftmargin=*]
	
	$~~$\textbf{a.}$~$ The projection of $t_i$ according to $\mathcal{V}$ is done. We obtain the list $\mathit{projs_{t_i}} = \pi_{\mathcal{V}}\left(t_i \right) = \left\{t_{i_{\mathcal{V}_1}},\ldots,t_{i_{\mathcal{V}_l}}\right\}$;
	
	$~~~~$\textbf{b.}$~$ If the type of $p_{n_i}$ is the same as the type of $p_{n}$ or the projection of $t_i$ has produced no more than one artifact ($\left|\mathit{projs_{t_i}}\right| \leq 1$), we just replace $t_i$ by artifacts $t_{i_{\mathcal{V}_1}},\ldots,t_{i_{\mathcal{V}_l}}$ of the list $\mathit{projs_{t_i}}$;
	
	$~~~~$ Otherwise, a new (re)structuring symbol $S_i$ is introduced and we replace the sub-artifact $t_i$ with a new artifact $new\_t_i$ whose root node is $n_{t_i}=S_i\left[t_{i_{\mathcal{V}_1}},\ldots,t_{i_{\mathcal{V}_l}}\right]$;
	%\end{itemize}
	
	\textbf{3.}$~$ If the list of new sub-artifacts of $n$ contains only one element $t_1$ having $n_1=S_1\left[t_{1_{\mathcal{V}_1}},\ldots,t_{1_{\mathcal{V}_l}}\right]$ (with $S_1$ a newly created (re)structuring symbol) as root node, we replace in this one, $t_1$ by the sub-artifacts $t_{1_{\mathcal{V}_1}},\ldots,t_{1_{\mathcal{V}_l}}$ of $n_1$. This removes a non-important (re)structuring symbol $S_1$.
	%\end{enumerate}
	
	\noindent{\large\textbullet} $~$ \textbf{Else}, $n$ is deleted and the result of the projection ($\mathit{projs_t}$) is the union of the projections of each of its sub-artifacts: $\mathit{projs_t} = \pi_{\mathcal{V}}\left(t \right) = \bigcup^{m}_{i=1} \pi_{\mathcal{V}}\left(t_i \right)$
	%\end{itemize}
\end{mdframed}
\end{algorithm}

Note that the algorithm described here can return several artifacts (a forest). To avoid that it produces a forest in some cases and thus meet challenge \textbf{(3)}, we make the following assumption: 
\begin{displayquote}
\textit{GMAWfP manipulated in this work are such that all agents are accredited in reading on the GMWf axioms (\textbf{axioms' visibility assumption}).}
\end{displayquote}
The designer must therefore ensure that all agents are accredited in reading on all GMWf axioms. To do this, after modelling a process $\mathcal{P}_{op}$ and obtaining its GMWf $\mathbb{G}=\left(\mathcal{S},\mathcal{P},\mathcal{A}\right)$, it is sufficient (if necessary) to create a new axiom $A_{\mathbb{G}}$ on which, all actors will be accredited in reading, and to associate it with new unit productions\footnote{A production of a context free grammar is a \textit{unit production} if it is on the form $A \rightarrow B$, where $A$ and $B$ are non-terminal symbols.} $pa : A_{\mathbb{G}} \rightarrow X_a$ where, $X_a \in \mathcal{A}$ is a symbol labelling the root of a target artifact.
Moreover, the designer of the GMWf must statically choose the agent responsible for initiating the process. This agent will therefore be the only one to possess an accreditation in writing on the new axiom $A_{\mathbb{G}}$.

\begin{proposition}
	\label{propositionStabiliteProjArt}
	For all GMAWfP $\mathbb{W}_f=\left(\mathbb{G}, \mathcal{L}_{P_k}, \mathcal{L}_{\mathcal{A}_k} \right)$ verifying the axioms' visibility assumption, the projection of an artifact $t$ which is conform to its GMWf ($t \therefore \mathbb{G}$) according to a given view $\mathcal{V}$, results in a single artifact $t_{\mathcal{V}}=\pi_{\mathcal{V}} \left(t\right)$ (stability property of $\pi$).
\end{proposition}

\begin{proof}
	Let's show that $\pi_{\mathcal{V}} \left(t\right)$ produces a single tree $t_{\mathcal{V}}$ which is an artifact. 
	Note that the only case in which the projection of an artifact $t$ according to a view $\mathcal{V}$ produces a forest, is when the root node of $t$ is associated with an invisible symbol $X$ ($X \notin \mathcal{V}$). Knowing that $t \therefore \mathbb{G}$ and that $\mathbb{W}_f$ validates the axioms' visibility assumption, it is deduced that the root node of $t$ is labelled by one of the axioms $A_{\mathbb{G}}$ of $\mathbb{G}$ and that $A_{\mathbb{G}} \in \mathcal{V}$ (hence the uniqueness of the produced tree). Since the projection operation preserves the form of productions, it is concluded that $t_{\mathcal{V}}=\pi_{\mathcal{V}} \left(t\right)$ is an artifact.
\end{proof}

A Haskell implementation of this projection algorithm is introduced in appendix \ref{appendice1:algorithms-implementations} of this manuscript. Another implementation in Java has also been proposed and integrated into the prototype that we will present in section \ref{chap3:sec:p2ptinywfms} of this chapter.

\mySubSubSection{The Need of a Local GMWf}{}
\label{chap3:sec:local-gmwf}
Since the artifact copy manipulated at a specific site is a potentially partial replica of the mobile (global) artifact, and since its editing depends on the agent's perception (view) of the process, it becomes crucial to provide each agent with a local GMWf. The latter will serve in addition to preserve the possible confidentiality of certain tasks and data, to guide the local actions of updating the artifact in order to ensure the convergence of the system to a coherent business goal state.
The local GMWf of an agent can be derived by projecting the global GMWf $\mathbb{G}$ according to the view $\mathcal{V}_i$ of its pilot (\textbf{\textit{GMWf projection}}). This projection is carried out using $\Pi$ operator and the GMWf obtained is noted $\mathbb{G}_{\mathcal{V}_i}=\Pi_{\mathcal{V}_i}\left(\mathbb{G}\right)$.

~

\noindent\textbf{\textit{A naive algorithm for non-recursive GMWf projection}}

The goal of this algorithm is to derive by projection of a given GMWf $\mathbb{G}=\left(\mathcal{S},\mathcal{P},\mathcal{A}\right)$ according to a view $\mathcal{V}$, a local GMWf $\mathbb{G}_{\mathcal{V}} = \left(\mathcal{S}_{\mathcal{V}},\mathcal{P}_{\mathcal{V}}, \mathcal{A}_{\mathcal{V}}\right)$ (we note $\mathbb{G}_{\mathcal{V}} = \Pi_{\mathcal{V}}\left(\mathbb{G} \right)$). The proposed algorithm is algorithm \ref{chap3:algo:gmwf-projection}.

\begin{algorithm}
\small
\caption{Algorithm to project a given GMWf according to a given view.}
\label{chap3:algo:gmwf-projection}
\begin{mdframed}[style=MyFrame]
	
	\noindent\textbf{1.}$~$ First of all, it is necessary to generate all the target artifacts denoted by $\mathbb{G}$ (see note (1) below); 
	we thus obtain a set $arts_{\mathbb{G}}=\left\{t_1,\ldots,t_n\right\}$;
	
	\noindent\textbf{2.}$~$ Then, each of the target artifacts must be projected according to $\mathcal{V}$. We thus obtain a set $arts_{\mathbb{G}_{\mathcal{V}}} = \left\{t_{\mathcal{V}_1},\ldots,t_{\mathcal{V}_m}\right\}$ (with $m \leq n$ because there may be duplicates; %\footnote{Deux artefacts cibles différents peuvent avoir la même projection suivant une vue $\mathcal{V}$ donnée.}; 
	in this case, only one copy is kept) of artifacts partial replicas;
	
	\noindent\textbf{3.}$~$ Then, collect the different (re)structuring symbols appearing in artifacts of $arts_{\mathbb{G}_{\mathcal{V}}}$, making sure to remove duplicates (see note (2) below) 
	and to consequently update the artifacts and the set $arts_{\mathbb{G}_{\mathcal{V}}}$. We thus obtain a set $\mathcal{S}_{\mathcal{V}_{Struc}}$ of symbols and a final set $arts_{\mathbb{G}_{\mathcal{V}}} = \left\{t_{\mathcal{V}_1},\ldots,t_{\mathcal{V}_l}\right\}$ (with $l \leq m$) of artifacts. These are exactly the only ones that must be conform to the searched GMWf $\mathbb{G}_{\mathcal{V}}$. So we call them, \textit{local target artifacts for the view $\mathcal{V}$};
	
	\noindent\textbf{4.}$~$ At this stage, it is time to collect all the productions that made it possible to build each of the \textit{local target artifacts for the view $\mathcal{V}$}. We obtain a set $\mathcal{P}_{\mathcal{V}}$ of distinct productions.\\
	\textbf{The searched local GMWf $\mathbb{G}_{\mathcal{V}} = \left(\mathcal{S}_{\mathcal{V}},\mathcal{P}_{\mathcal{V}}, \mathcal{A}_{\mathcal{V}}\right)$ is such as}:
	%\begin{itemize}
	
	$~~$\textbf{a.}$~$ its set of symbols is $\mathcal{S}_{\mathcal{V}} = \mathcal{V} \cup \mathcal{S}_{\mathcal{V}_{Struc}}$;
	
	$~~$\textbf{b.}$~$ its set of productions is $\mathcal{P}_{\mathcal{V}}$;
	
	$~~$\textbf{c.}$~$ its axioms are in $\mathcal{A}_{\mathcal{V}} = \mathcal{A}$
	%\end{itemize}
	%\end{enumerate}
	
	~
	
	\noindent\textit{\textbf{Note (1):}$~$ To generate all the target artifacts denoted by a GMWf $\mathbb{G}=\left(\mathcal{S},\mathcal{P},\mathcal{A}\right)$, one just has to use the set of productions to generate the set of artifacts having one of the axiom $A_{\mathbb{G}}\in \mathcal{A}$ as root. In fact, for each axiom $A_{\mathbb{G}}$, it should be considered that every $A_{\mathbb{G}}$-production $P=\left(A_{\mathbb{G}},X_1\cdots X_n\right)$ induces artifacts $\left\{t_1, \ldots, t_m\right\}$ such as: the root node of each $t_i$ is labelled $A_{\mathbb{G}}$ and has as its sons, a set of artifacts $\left\{t_{i_1},\ldots,t_{i_n}\right\}$, part of the Cartesian product of the sets of artifacts generated when considering each symbol $X_1,\cdots, X_n$ as root node.}
	
	\noindent\textit{\textbf{Note (2):}$~$ In this case, two (re)structuring symbols are identical if for all their appearances in nodes of the different artifacts of $arts_{\mathbb{G}_{\mathcal{V}}}$, they induce the same local scheduling.}
\end{mdframed}
\end{algorithm}

Figure \ref{chap3:fig:gmwf-projection} illustrates the research of a local model $\mathbb{G}_{\mathcal{V}_{EC}}$ such as $\mathbb{G}_{\mathcal{V}_{EC}} = \Pi_{\mathcal{V}_{EC}}\left(\mathbb{G}\right)$ with $\mathcal{V}_{EC}=\mathcal{A}_{EC(r)}=\{A, B, C, D, H1, H2, I1, I2, F\}$. Target artifacts generated from $\mathbb{G}$ (fig. \ref{chap3:fig:gmwf-projection}(b)) are projected to obtain two \textit{local target artifacts for the view $\mathcal{V}_{EC}$} (fig. \ref{chap3:fig:gmwf-projection}(c)). 
From the local target artifacts thus obtained, the searched GMWf is produced (fig. \ref{chap3:fig:gmwf-projection}(d)).
\begin{figure}[ht!]
	\noindent
	\makebox[\textwidth]{\includegraphics[scale=0.3]{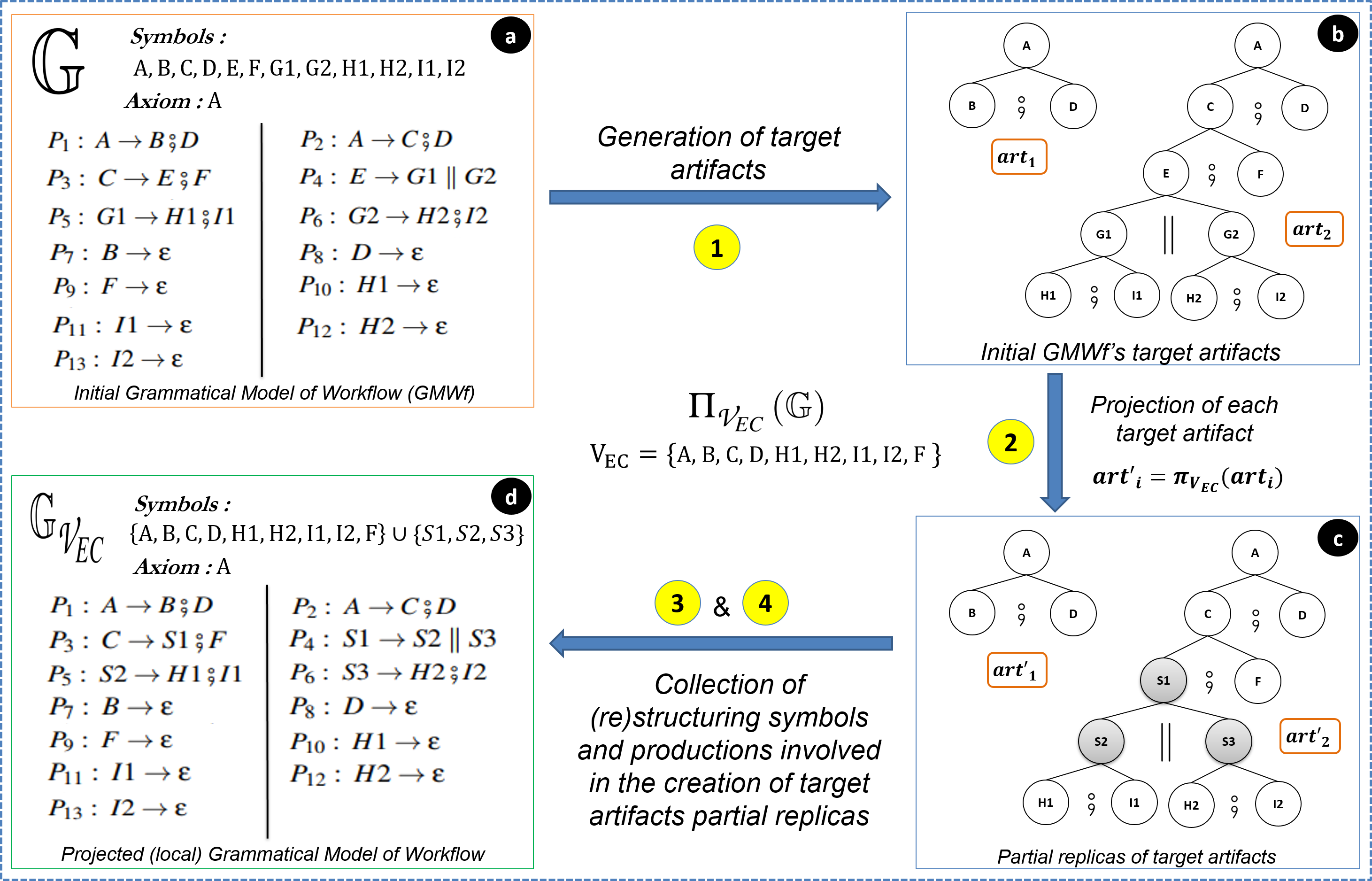}}
	\caption{Example of projection of a GMWf according to a given view.}
	\label{chap3:fig:gmwf-projection}
\end{figure}

The GMWf projection algorithm presented here only works for GMWf that do not allow recursive symbols\footnote{It is only in this context that all the target artifacts can be enumerated.}. We therefore assume that:
\begin{displayquote}
\textit{For the execution model presented in this chapter, the manipulated GMAWfP are those whose GMWf do not contain recursive symbols (\textbf{non-recursive GMWf assumption})}.
\end{displayquote} 
Therefore, it is no longer possible to express iterative routing between process tasks (in the general case); except in cases where the maximum number of iterations is known in advance. This algorithm has some interesting properties and the interested reader will find an introduction to its Haskell implementation in appendix \ref{appendice1:algorithms-implementations}.

\begin{proposition}
	\label{propositionStabiliteProjGMWf}
	For all GMAWfP $\mathbb{W}_f=\left(\mathbb{G}, \mathcal{L}_{P_k}, \mathcal{L}_{\mathcal{A}_k} \right)$ verifying the axioms' visibility and the non-recursivity of GMWf assumptions, the projection of its GMWf $\mathbb{G}=\left(\mathcal{S},\mathcal{P},\mathcal{A}\right)$ according to a given view $\mathcal{V}$, is a GMWf $\mathbb{G}_{\mathcal{V}} = \Pi_{\mathcal{V}}\left(\mathbb{G} \right)$ for a GMAWfP $\mathbb{W}_{f_{\mathcal{V}}}$ verifying the assumptions of axiom visibility and non-recursivity of GMWf (stability property of $\Pi$).
\end{proposition}

\begin{proof}
	Let's show that $\mathbb{G}_{\mathcal{V}} = \Pi_{\mathcal{V}}\left(\mathbb{G} \right)$ is a GMWf for a new GMAWfP $\mathbb{W}_{f_{\mathcal{V}}}=\left(\mathbb{G}_{\mathcal{V}}, \mathcal{L}_{P_k}, \mathcal{L}_{\mathcal{A}_{\mathcal{V}_k}} \right)$ that verifies the assumptions of axioms' visibility and non-recursivity of GMWf. 
	As $\mathbb{W}_f=\left(\mathbb{G}, \mathcal{L}_{P_k}, \mathcal{L}_{\mathcal{A}_k} \right)$ validates the non-recursivity of GMWf assumption, the set of target artifacts ($arts_{\mathbb{G}}=\left\{t_1,\ldots,t_n\right\}$) that it denotes is finite and can therefore be fully enumerated. Knowing further that $\mathbb{W}_f$ validates the axioms' visibility assumption, it is deduced that the set $arts_{\mathbb{G}_{\mathcal{V}}} = \left\{t_{\mathcal{V}_1}=\pi_{\mathcal{V}}\left(t_1\right), \ldots,t_{\mathcal{V}_n}=\pi_{\mathcal{V}}\left(t_n\right)\right\}$ is finite and the root node of each artifact $t_{\mathcal{V}_i}$ is associated with an axiom $A_{\mathbb{G}} \in \mathcal{A}$ (see proposition \ref{propositionStabiliteProjArt}). $\mathbb{G}_{\mathcal{V}}$ being built from the set $arts_{\mathbb{G}_{\mathcal{V}}}$, its axioms $\mathcal{A}_{\mathcal{V}}=\mathcal{A}$ are visible to all actors and its productions are only of the two forms retained for GMWf. In addition, each new (re)structuring symbol ($S \in \mathcal{S}_{\mathcal{V}_{Struc}}$)) is created and used only once to replace a symbol that is not visible and not recursive (by assumption) when projecting artifacts of $arts_{\mathbb{G}}$. The new symbols are therefore not recursive. By replacing in $\mathcal{L}_{\mathcal{A}_k}$ the view $\mathcal{V}$ by $\mathcal{V} \cup \mathcal{S}_{\mathcal{V}_{Struc}}$, one obtains a new set $\mathcal{L}_{\mathcal{A}_{\mathcal{V}_k}}$ of accreditations for a new GMAWfP $\mathbb{W}_{f_{\mathcal{V}}}=\left(\mathbb{G}_{\mathcal{V}}, \mathcal{L}_{P_k}, \mathcal{L}_{\mathcal{A}_{\mathcal{V}_k}} \right)$ verifying the assumptions of axioms' visibility and non-recursivity of GMWf.
\end{proof}

\begin{proposition}
	\label{propositionCoherenceArtefact}
	For all GMAWfP $\mathbb{W}_f=\left(\mathbb{G}, \mathcal{L}_{P_k}, \mathcal{L}_{\mathcal{A}_k} \right)$ verifying the axioms' visibility and the non-recursivity of GMWf assumptions, the projection of an artifact $t$ which is conform to the GMWf $\mathbb{G}$ according to a given view $\mathcal{V}$, is an artifact which is conform to the projection of $\mathbb{G}$ according to $\mathcal{V}$ $\left(\forall t, ~t \therefore \mathbb{G} \Rightarrow \pi_{\mathcal{V}}\left(t\right) \therefore \Pi_{\mathcal{V}}\left(\mathbb{G} \right)\right)$.
\end{proposition}

\begin{proof}
	Knowing that the considered GMAWfP $\mathbb{W}_f=\left(\mathbb{G}, \mathcal{L}_{P_k}, \mathcal{L}_{\mathcal{A}_k} \right)$ verifies the axioms' visibility and the non-recursivity of GMWf assumptions, it is deduced that the set of its target artifacts $arts_{\mathbb{G}}$ (those who helped to build its GMWf $\mathbb{G}$) is finite and any artifact that is conform to its GMWf $\mathbb{G}$ is a target artifact $\left( \forall t, ~t \therefore \mathbb{G} \Leftrightarrow t \in arts_{\mathbb{G}} \right)$. Therefore, considering a given artifact $t$ such that $t$ is conform to $\mathbb{G}$ ($t \therefore \mathbb{G}$), one knows that it is a target artifact ($t \in arts_{\mathbb{G}}$) and its projection according to a given view $\mathcal{V}$ produces a single artifact $t_{\mathcal{V}}=\pi_{\mathcal{V}}\left(t\right)$ (see "stability property of $\pi$", proposition \ref{propositionStabiliteProjArt}) such as $t$ and $t_{\mathcal{V}}$ have the same root (one of the axioms $A_{\mathbb{G}} \in \mathcal{A}$ of $\mathbb{G}$). Since $t$ is a target artifact, its projection $t_{\mathcal{V}}$ (through the renaming of some potential (re)structuring symbols) is part of the set $arts_{\mathbb{G}_{\mathcal{V}}}$ of artifacts that have generated $\mathbb{G}_{\mathcal{V}} = \Pi_{\mathcal{V}}\left(\mathbb{G} \right)$ by applying the projection principle described in the algorithm \ref{chap3:algo:gmwf-projection}. Therefore, the productions involved in the construction of $t_{\mathcal{V}}$ are all included in the set of productions of the GMWf $\mathbb{G}_{\mathcal{V}} = \Pi_{\mathcal{V}}\left(\mathbb{G} \right)$. As the set of axioms of $\mathbb{G}_{\mathcal{V}}$ is $\mathcal{A}_{\mathcal{V}} = \mathcal{A}$, it is deduced that $A_{\mathbb{G}} \in \mathcal{A}_{\mathcal{V}}$ and concluded that $t_{\mathcal{V}} \therefore \mathbb{G}_{\mathcal{V}}$.
\end{proof}

\begin{proposition}
	\label{propositionReciproqueCoherenceArtefact}
	Consider a GMAWfP $\mathbb{W}_f=\left(\mathbb{G}, \mathcal{L}_{P_k}, \mathcal{L}_{\mathcal{A}_k} \right)$ verifying the axioms' visibility and the non-recursivity assumptions. For all artifact $t_{\mathcal{V}}$ which is conform to $\Pi_{\mathcal{V}}\left(\mathbb{G} \right)$, it exists at least one artifact $t$ which is conform to $\mathbb{G}$ such that $t_{\mathcal{V}}=\pi_{\mathcal{V}}\left(t\right)$ $\left(\forall t_{\mathcal{V}}, ~t_{\mathcal{V}} \therefore \Pi_{\mathcal{V}}\left(\mathbb{G} \right) \Rightarrow \exists t, ~t \therefore \mathbb{G} ~and~ t_{\mathcal{V}}=\pi_{\mathcal{V}}\left(t\right) \right)$.
\end{proposition}

\begin{proof}
	With proposition \ref{propositionStabiliteProjGMWf} ("stability property of $\Pi$") it has been shown that the projection $\mathbb{G}_{\mathcal{V}} = \Pi_{\mathcal{V}}\left(\mathbb{G} \right)$ according to the view $\mathcal{V}$ of a GMWf $\mathbb{G}$ verifying the axioms' visibility and the non-recursivity assumptions, is a GMWf verifying the same assumptions. On this basis and using similar reasoning to that used to prove the proposition \ref{propositionCoherenceArtefact}, it's been determined that an artifact $t_{\mathcal{V}}$ that is conform to $\mathbb{G}_{\mathcal{V}}$, is one of its target artifacts (\textit{local target artifact for the view $\mathcal{V}$}): i.e, $t_{\mathcal{V}} \in arts_{\mathbb{G}_{\mathcal{V}}}$. Referring to the projection process which made it possible to obtain $\mathbb{G}_{\mathcal{V}}$, it is determined that the set $arts_{\mathbb{G}_{\mathcal{V}}}$ is exclusively made up of the projections of the set $arts_{\mathbb{G}}=\left\{t_1,\ldots,t_n\right\}$ of $\mathbb{G}$'s target artifacts. $t_{\mathcal{V}}$ is therefore the projection of at least one target artifact $t_i \in arts_{\mathbb{G}}$ of $\mathbb{G}$ $\left(t_{\mathcal{V}}=\pi_{\mathcal{V}}\left(t_i\right)\right)$. Knowing that $\forall t, ~t \therefore \mathbb{G} \Leftrightarrow t \in arts_{\mathbb{G}}$ (see proof of proposition \ref{propositionCoherenceArtefact}), it is deduced that $t_i \therefore \mathbb{G}$ and the proof of this proposition is made.
\end{proof}

By applying the GMWf projection algorithm presented above to the running example, one obtain the productions listed in table \ref{tableau:gramLocales} for the different agents respectively. Let us note that this algorithm simply project each target artifact according to the view of the considered agent, then gather the productions in the obtained partial replicas while removing the duplicates. In the illustrated case here, we have considered an update of the GMWf of the peer-review process so that it validates the axioms' visibility assumption (see sec. \ref{chap3:sec:partial-replica}).
\begin{table}[h]
	\centering
	\caption{Local GMWf productions of all the agents involved in the peer-review process.}
	\label{tableau:gramLocales}
	\begin{tabular}[t]{|m{3.5cm}|m{10.5cm}|}
		\hline
		\textbf{Agent} & \textbf{Productions of local GMWf} \\
		\hline
		Editor in Chief ($EC$) &
		\[ 
		\begin{array}{l|l|l}
		P_{1}:\; A_{\mathbb{G}}\rightarrow A & \; P_{2}:\; A\rightarrow B\fatsemi D\; & \; P_{3}:\; A\rightarrow C\fatsemi D  \\
		P_{4}:\; C\rightarrow S1\fatsemi F & \; P_{5}:\; S1\rightarrow S2\parallel S3\; & \; P_{6}:\; S2\rightarrow H1 \fatsemi I1  \\
		P_{7}:\; S3\rightarrow H2 \fatsemi I2 & \; P_{8}:\; B\rightarrow \varepsilon\; & \; P_{9}:\; D\rightarrow \varepsilon \\
		P_{10}:\; F\rightarrow \varepsilon & \; P_{11}:\; H1\rightarrow \varepsilon\; & \; P_{12}:\; I1\rightarrow \varepsilon  \\
		P_{13}:\; H2\rightarrow \varepsilon & \; P_{14}:\; I2\rightarrow \varepsilon \; &   \\
		\end{array}
		\]
		\\
		\hline
		Associated Editor ($AE$) & 
		\[ 
		\begin{array}{l|l|l}
		P_{1}:\; A_{\mathbb{G}}\rightarrow A & \; P_{2}:\; A\rightarrow C \; & \; P_{3}:\; C\rightarrow E\fatsemi F  \\
		P_{4}:\; E\rightarrow S1\parallel S2 & \; P_{5}:\; S1\rightarrow H1\fatsemi I1 \; & \; P_{6}:\; S2\rightarrow H2\fatsemi I2  \\
		P_{7}:\; H1\rightarrow \varepsilon & \; P_{8}:\; I1\rightarrow \varepsilon \; & \; P_{9}:\; H2\rightarrow \varepsilon \\
		P_{10}:\; I2\rightarrow \varepsilon & \; P_{11}:\; F\rightarrow \varepsilon \; & \; P_{12}:\; A_{\mathbb{G}}\rightarrow \varepsilon \\
		\end{array}
		\]
		\\
		\hline
		First referee ($R1$) & 
		\[ 
		\begin{array}{l|l|l}
		P_{1}:\; A_{\mathbb{G}}\rightarrow C & P_{2}:\; C\rightarrow G1\; & \; P_{3}:\; G1\rightarrow H1\fatsemi I1 \\
		P_{4}:\; H1\rightarrow \varepsilon & P_{5}:\; I1\rightarrow \varepsilon \; & \; P_{6}:\; A_{\mathbb{G}}\rightarrow \varepsilon \\
		\end{array}
		\]
		\\
		\hline
		Second referee ($R2$) & 
		\[ 
		\begin{array}{l|l|l}
		P_{1}:\; A_{\mathbb{G}}\rightarrow C & P_{2}:\; C\rightarrow G2\; & \; P_{3}:\; G2\rightarrow H2\fatsemi I2 \\
		P_{4}:\; H2\rightarrow \varepsilon & P_{5}:\; I2\rightarrow \varepsilon \; & \; P_{6}:\; A_{\mathbb{G}}\rightarrow \varepsilon \\
		\end{array}
		\]
		\\
		\hline
	\end{tabular}
\end{table}

\mySubSection{The Artifact-Centric Choreography}{}
\label{chap3:sec:execution-model}

In this section, we are interested in the actual execution of a process $\mathcal{P}_{op}$ whose GMWf is $\mathbb{G}=\left(\mathcal{S},\mathcal{P},\mathcal{A}\right)$.

\mySubSubSection{Initial Configuration of an Agent}{}
\label{chap3:sec:initial-configuration-of-a-peer}

Each agent $i$ taking part in the choreography, has a single identifier (its ID). For a proper execution, it manages a local copy of the process' global GMWf $\mathbb{G}$, accreditations of various agents involved and its local GMWf $\mathbb{G}_{i}$. In addition, it handles a list $RET_i$ of agents who have made requests and whose answers are yet to be sent, as well as two queues: $REQ_i$ which stores requests waiting to be executed, and $ANS_i$ which temporally stores answers received from agents to which requests were previously made. A local copy $t_i$ of the mobile artifact and its (potentially partial) replica $t_{\mathcal{V}_{i}}$ are also handled by agent $i$.

\mySubSubSection{The Execution Choreography and Agent's Behaviour}{}
\label{chap3:sec:architecture-of-a-peer}
The execution of an instance of the process is triggered when an artifact $t$ is introduced into the system (on the appropriate agent); this artifact is in fact an unlocked bud of the type of one axiom $A_\mathbb{G} \in \mathcal{A}$ (initial task) of the (global) GMWf $\mathbb{G}$.

An artifact that arrives on a given agent is either a request or a response to a request; depending on the case, it is inserted in the appropriate queue ($REQ_i$ or $ANS_i$).
As soon as possible\footnote{For instance at the end of the local replica completion or after a given time interval.}, the artifact is removed from the queue, merged with the local copy (if it exists) and is then completed as needed.
Completing an artifact consists of executing in a coherent way, the various tasks it imposes, i.e. those on which the current agent is accredited in writing.

At the end of the completion on an artifact, if its configuration shows that it must be completed by other agents (this is the case if the artifact contains buds created by the current agent and whose agent accredited in writing, are remote), replicas of the artifact are sent to the said agents by invoking the service \textit{forwardTo}.
Otherwise, the artifact is complete (it contains no more buds), or semi-complete (it contains buds that had been created by other agents and on which, the current agent is not accredited in writing); in which case, a replica is returned to the agent from which the artifact was previously received by invoking the service \textit{returnTo}.

The execution of the process ends when all the tasks constituting a scenario of the process have been executed. In this case, the artifact that is cooperatively edited is complete (closed) on the agent where the process was triggered.

\mySubSubSection{The Protocols}{}
\label{chap3:sec:the-protocols}

The activity that takes place on an agent in relation to the handling of a given artifact, breaks down into five sub-activities (see fig. \ref{chap3:fig:peer-architecture}); each of them is managed by a dedicated protocol. These activities are the following: 

\begin{itemize}
	\item \textit{creation} (initialisation of a new case) or \textit{receipt-merger} of a replica.
	\item \textit{replication}: it consists in the extraction (from the local replica) of the partial replica that the local agent has to complete (manage its execution).
	\item \textit{execution}: it consists in the extension by the local actor (via the specialised editor) of the buds for which he is accredited in writing.
	\item \textit{expansion-pruning}: it consists in the reconstruction by expansion of the local (global) replica from the updated local partial replica.
	\item \textit{diffusion}: it corresponds to the return of the response to a request, or to the sending of requests.
\end{itemize}
\begin{figure}[ht!]
	\noindent
	\makebox[\textwidth]{\includegraphics[scale=0.28]{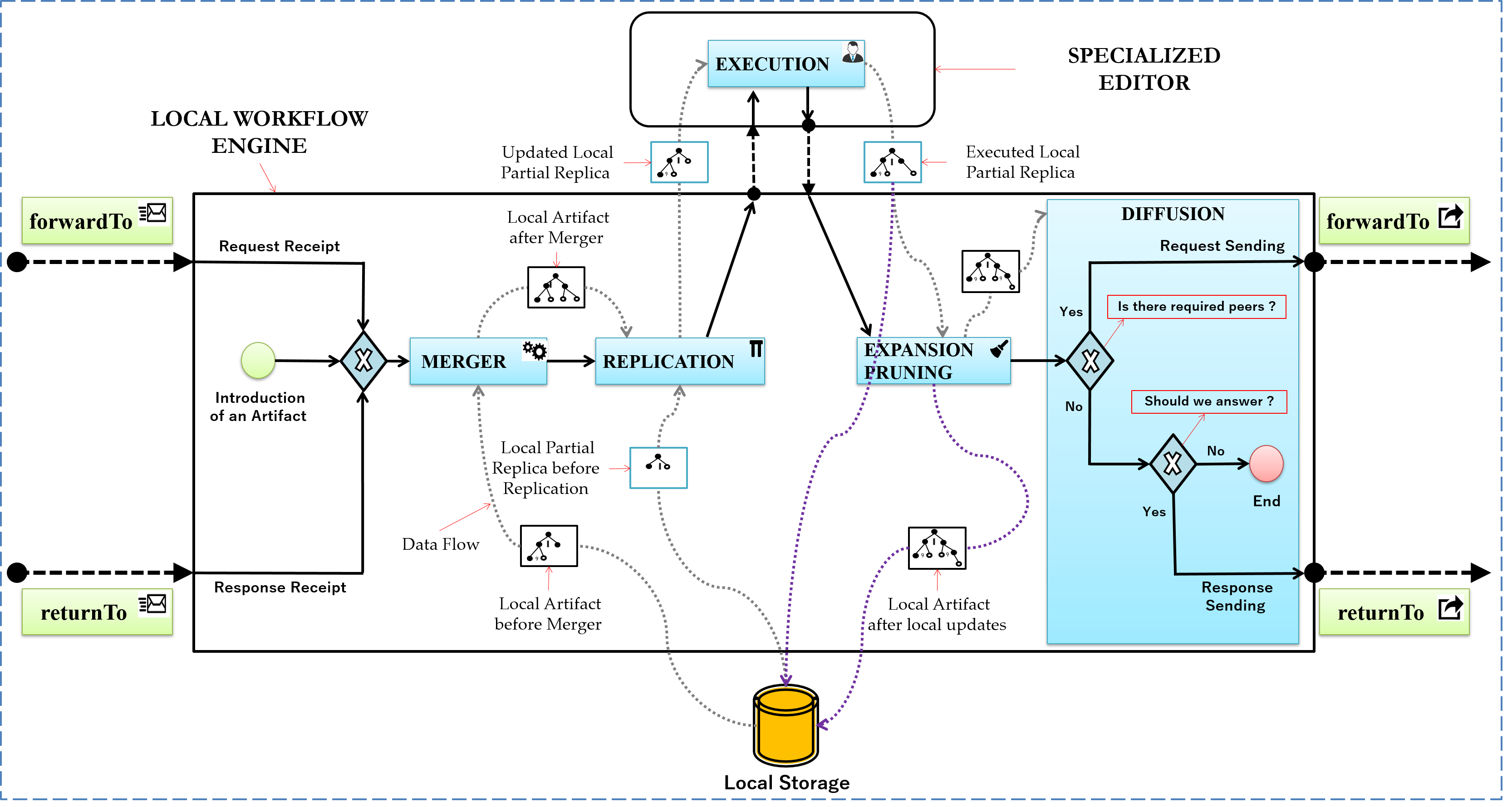}}
	\caption{Activity of an agent in the system.}
	\label{chap3:fig:peer-architecture}
\end{figure}

The management protocols for these different activities are described in the following paragraphs.

~

\noindent\textbf{\textit{The Receipt-Merger Protocol}}

An artifact is received either when a new case is initialised or after a request or a response is delivered. In all cases, a merge using an adaptation of the algorithm in \cite{badouelTchoupeCmcs} is performed. The goal in this step is to update the local copy $t_i$ of the global artifact from those received (the $(t^{maj}_j)_{1 \leq j \neq i \leq k}$ contained in queues $REQ_i$ and $ANS_i$). 
For that purpose (see algorithm \ref{algorithmeFusion}), we merge $t_i$ with each artifact $t^{maj}_j$ from the requests queue (algorithm \ref{algorithmeFusion}, lines 2 to 6) and from the responses queue (algorithm \ref{algorithmeFusion}, lines 7 to 10) until they are empty. 
At each merge, $t_i$ is updated (algorithm \ref{algorithmeFusion}, lines 3 and 8). 
For each received request, the identity of the sender is kept in the list $RET_i$ (algorithm \ref{algorithmeFusion}, line 5) to be able to return a response at the end of the request processing.
Note that during the merge, some previously locked buds can be unlocked: this is the case if all the tasks that precede them have been executed.

\begin{algorithm}
\small
\caption{Merger protocol executed by an agent $i$.\label{algorithmeFusion}}
\begin{algorithmic}[1]
	\Procedure{Merger}{}
		\For{$req : REQ_i$}\Comment{While there is a request}
			\State $t_i \gets merge(t_i,~\textrm{\textit{req.artifact}})$\Comment{We merge the artifact of the request with $t_i$}
			\State delete $req$ from $REQ_i$
			\State $enqueue(req.sender,~RET_i)$\Comment{And we add the request sender in $RET_i$ queue}
		\EndFor
		\For{$ans : ANS_i$}
			\State $t_i \gets merge(t_i,~\textrm{\textit{ans.artifact}})$\Comment{We merge the artifact of the answer with $t_i$}
			\State delete $ans$ from $ANS_i$
		\EndFor
	\EndProcedure
\end{algorithmic}
\end{algorithm}

~

\Needspace{5\baselineskip}
\noindent\textbf{\textit{The Replication Protocol}}

Replication is done just after the merge. The objective here is to update the local partial replica $t_{\mathcal{V}_{i}}$ from the local (global) artifact $t_i$.
To do this (see algorithm \ref{algorithmeReplication}), the local workflow engine proceeds as follows:
\begin{itemize}
	\item It realises the expansion\footnote{It is important to note that the expansion algorithm used here only returns one artifact (see \textit{the expansion-pruning protocol}, page \pageref{chap3:sec:execution-protocol-expansion-pruning}), unlike the one presented in \cite{badouelTchoupeCmcs} which generates a potentially infinite family of artifacts represented by a tree automaton. This uniqueness is guaranteed by the fact that the expansion of $t_{\mathcal{V}_{i}}$ into $t^{maj}_i$ is done using a three-way approach (\textit{three-way merge} \cite{tomMens}). In fact, the expansion is carried out based on the grammatical model $\mathbb{G}$ and on the view $\mathcal{V}_{i}$, but also on the prefix $t_i$ of (the local global artifact replica) $t^{maj}_i$.} of the partial replica $t _{\mathcal{V}_{i}}$ to obtain a global artifact $t^{maj}_i$ which integrates all the updates made during the previous execution (algorithm \ref{algorithmeReplication}, line 2). This operation is necessary, since at the end of the previous expansion, there may have been a pruning (see \textit{the expansion-pruning protocol}, page \pageref{chap3:sec:execution-protocol-expansion-pruning}) which removed from the global artifact local copy $t_i$, some updates contained in $t_{\mathcal{V}_{i}}$;
	\item Then, it merges $t_i$ and $t^{maj}_i$ in one artifact $t_{i_f}$ (algorithm \ref{algorithmeReplication}, line 3);
	\item Finally, it realises the projection of $t_{i_f}$ relatively to the view $\mathcal{V}_{i}$ to obtain the new version of $t_{\mathcal{V}_{i}}=\pi_{\mathcal{V}_i}(t_{i_f})$ (algorithm \ref{algorithmeReplication}, line 4).
\end{itemize}

\begin{algorithm}
\small
\caption{Replication protocol executed by an agent $i$.\label{algorithmeReplication}}
\begin{algorithmic}[1]
	\Procedure{Replication}{}
		\State $t^{maj}_i \gets expand(t_{\mathcal{V}_{i}}, ~t_i, ~\mathcal{V}_{i}, ~\mathbb{G})$
		\State $t_{i_f} \gets merge(t_i,~t^{maj}_i)$
		\State $t_{\mathcal{V}_{i}} \gets \textrm{\textit{projection}}(t_{i_f}, ~\mathcal{V}_{i}, ~\mathbb{G}_{i})$
	\EndProcedure
\end{algorithmic}
\end{algorithm}

~

\noindent\textbf{\textit{The Execution Protocol}}

This protocol (algorithm \ref{algorithmeExecution}) is executed after the production of the local partial replica $t_{\mathcal{V}_{i}}$ by an agent $i$. It is executed by the local actor through the specialised editor, in order to extend the (unlocked) buds of $t_{\mathcal{V}_{i}}$ on which, he is accredited in writing.

The execution of the artifact's local replica by the agent $i$, must be done "as far as possible" by respecting the scheduling (sequential or parallel) of the tasks. Indeed, during the extension of a bud, if there is unlocking or creation of new unlocked buds on which the current agent is accredited in writing, its actor must extend/execute them; this is the purpose of the \textit{while} loop in algorithm \ref{algorithmeExecution}. In addition, the extension of buds whose type $S$ corresponds to a (re)structuring symbol, is automatically done by the local workflow engine when the local GMWf has only one $S$-production.

\begin{algorithm}
\small
\caption{Execution protocol executed by an agent $i$.\label{algorithmeExecution}}
\begin{algorithmic}[1]
	\Procedure{Execution}{}
		\While{$not~isEmpty(buds \gets nextLocalUnlockedBuds(t_{\mathcal{V}_{i}},~\mathcal{A}_{A_i(w)}))$}\Comment{While there are tasks that can be concurrently executed by the actor $A_i$ of agent $i$ in the partial replica $t_{\mathcal{V}_{i}}$}
			\State $bud \gets prompt("Choose~a~task~to~execute",~buds)$\Comment{Actor $A_i$ chooses the task (bud) to execute}
			\State $prods \gets localExecutionPossibilities(bud.type)$\Comment{The specialized editor (agent) generates and activates the set of execution possibilities according to the current (local) configuration}
			\State $choice \gets prompt("Choose~an~execution~possibility",~prods)$\Comment{Actor $A_i$ executes the selected task and provide feedback through the specialized editor}
			\State $t_{\mathcal{V}_{i}} \gets \textrm{\textit{updateArtifact}}(choice,~t_{\mathcal{V}_{i}})$\Comment{Then $t_{\mathcal{V}_{i}}$ is updated accordingly}
		\EndWhile
	\EndProcedure
\end{algorithmic}
\end{algorithm}

~

\noindent\textbf{\textit{The Expansion-Pruning Protocol}}
\label{chap3:sec:execution-protocol-expansion-pruning}

After completion of the partial replica $t_{\mathcal{V}_{i}}$, the updates must be propagated to the local (global) replica $t_i$ of the artifact. This makes it possible to highlight (if they exist) the tasks for which requests must be made, or to determine if answers to requests can be returned.
Algorithm \ref{algorithmeExpansion} allows addressing this concern. For that, the expansion of the local updated partial replica $t_{\mathcal{V}_{i}}$ is made (algorithm \ref{algorithmeExpansion}, line 2) to obtain the global artifact $t^{maj}_i$ which integrates all the contributions made by actor $A_i$ during the previous local execution phase.

Note that the artifact $t^{maj}_i$ may have so called \textit{upstairs buds}\footnote{Intuitively, a node ${n_{X_{\bar{\omega}}}}$ associated with the task $X$ is an \textit{upstair bud} if, $X$ (not already executed) precedes at least one task $Y$ made visible (and naturally not already executed) on the site of an agent $i$, $i$ not having any accreditation in reading on $X$.} (these are the internal nodes of $t^{maj}_i$ that do not belong to $\mathcal{V}_{i}=\mathcal{A}_{A_i (r)}$, and which are not in $t_i$ - see fig. \ref{chap3:fig:execution-figure-2}). 
To prevent and manage this situation, a pruning of $t^{maj}_i$ is performed (algorithm \ref{algorithmeExpansion}, line 3) to ensure compliance with task-related precedence constraints of executions. 
To do this, for every path of $t^{maj}_i$ starting from the root, we prune at the level of the first (upstairs) bud encountered; it must appear unlocked if all the tasks that precede it have already been executed. The artifact obtained after this phase is the new version of $t_i$, and represents the current state of the process execution from the point of view of agent $i$. 

\begin{algorithm}
\small
\caption{Expansion-Pruning protocol executed by an agent $i$.\label{algorithmeExpansion}}
\begin{algorithmic}[1]
	\Procedure{Expansion-Pruning}{}
		\State $t^{maj}_i \gets expand(t_{\mathcal{V}_{i}}, ~t_i, ~\mathcal{V}_{i}, ~\mathbb{G})$
		\State $t_{i} \gets pruning(t^{maj}_i, ~t_i)$%\Comment{We merge the artefact of the request with $t_i$}
	\EndProcedure
\end{algorithmic}
\end{algorithm}

~

\noindent\textbf{\textit{A three-way merging expansion algorithm}}\label{three-way-merge}: consider an (global) artifact under execution $t$, and $t_{\mathcal{V}}=\pi_{\mathcal{V}}\left(t\right)$ its partial replica on the site of an actor $A_i$ whose view is $\mathcal{V}$. Consider the partial replica $t_{\mathcal{V}}^{maj} \geq t_{\mathcal{V}}$ obtained by developing some unlocked buds of $t_{\mathcal{V}}$ as a result of $A_i$'s contribution. The expansion problem consists in finding an (global) artifact under execution $t_f$, which integrates nodes of $t$ and $t_{\mathcal{V}}$. To solve this problem made difficult by the fact that $t$ and $t_{\mathcal{V}}$ are conform to two different models ($\mathbb{G}$ and $\mathbb{G}_{\mathcal{V}} = \Pi_{\mathcal{V}} \left(\mathbb{G} \right)$), we perform a three-way merge {\cite{tomMens}. We merge the artifacts $t$ and $t_{\mathcal{V}}$ using a (global) target artifact $t_g$ such that: 
\begin{enumerate}
	\item[\textbf{(a)}] $t$ is a prefix of $t_g$ ($t \leq t_g$)
	\item[\textbf{(b)}] $t_{\mathcal{V}}^{maj}$ is a prefix of the partial replica of $t_g$ according to $\mathcal{V}$ $\left(t_{\mathcal{V}}^{maj} \leq \pi_{\mathcal{V}}\left(t_g \right)\right)$
\end{enumerate}	
The proposed algorithm proceeds in two steps.
	
~
	
\noindent\textit{Step 1 - Search for the merging guide $t_g$}:
the search of a merging guide is done by the algorithm \ref{chap3:algo:search-guide}.

\begin{algorithm}
\small
\caption{Algorithm to search a merging guide.}
\label{chap3:algo:search-guide}
\begin{mdframed}[style=MyFrame]
	\noindent\textbf{1.}$~$ First of all, we have to generate the set $arts_{\mathbb{G}}=\left\{t_1,\ldots,t_n\right\}$ of target artifacts denoted by $\mathbb{G}$;
	
	\noindent\textbf{2.}$~$ Then, we must filter this set to retain only the artifacts $t_i$ admitting $t$ as a prefix (criterion \textbf{(a)}) and whose projections according to $\mathcal{V}$ ($t_{i_{\mathcal{V}_j}}$) admit $t_{\mathcal{V}}^{maj}$ as a prefix (criterion \textbf{(b)}). It is said that an artifact $t_a$ (whose root node is $n_a=X_a[t_{a_1},\ldots,t_{a_l}]$) is a prefix of a given artifact $t_b$ (whose root node is $n_b=X_b[t_{b_1},\ldots,t_{b_m}]$) if and only if the root nodes $n_a$ and $n_b$ are of the same types (i.e $X_a=X_b$) and:
	%\begin{itemize}[leftmargin=*]
	
		$~~$\textbf{a.}$~$ The node $n_a$ is a bud or,
		
		$~~$\textbf{b.}$~$ The nodes $n_a$ and $n_b$ have the same number of sub-artifacts (i.e $l=m$), the same type of scheduling for the sub-artifacts and each sub-artifact $t_{a_i}$ of $n_a$ is a prefix of the sub-artifact $t_{b_i}$ of $n_b$.
	%\end{itemize}
	
	\noindent We obtain the set $guides=\left\{t_{g_1},\ldots,t_{g_k}\right\}$ of artifacts that can guide the merging;
	
	\noindent\textbf{3.}$~$ Finally, we randomly select an element $t_g$ from the set $guides$.
	%\end{enumerate}
\end{mdframed}
\end{algorithm}

~

\noindent\textit{Step 2 - Merging $t$, $t_{\mathcal{V}}^{maj}$ and $t_g$}:
we want to find an artifact $t_f$ that includes all the contributions already made during the workflow execution. The structure of the searched artifact $t_f$ is the same as that of $t_g$: hence the interest to use $t_g$ as a guide. The merging is carried out by the algorithm \ref{chap3:algo:three-way-merge}.

\begin{algorithm}
\small
\caption{Three-way merging algorithm.}
\label{chap3:algo:three-way-merge}
\begin{mdframed}[style=MyFrame]
	\noindent A prefixed depth path of the three artifacts ($t$, $t_{\mathcal{V}}^{maj}$ and $t_g$) is made simultaneously until there is no longer a node to visit in $t_g$. Let $n_{t_i}$ (resp. $n_{t_{\mathcal{V}_j}^{maj}}$ and $n_{t_{g_k}}$) be the node located at the address $w_i$ (resp. $w_j$ and $w_k$) of $t$ (resp. $t_{\mathcal{V}}^{maj}$ and $t_g$) and currently being visited. If nodes $n_{t_i}$, $n_{t_{\mathcal{V}_j}^{maj}}$ and $n_{t_{g_k}}$ are such that (\textbf{processing}):
	
	%\begin{itemize}[leftmargin=*]
	\noindent\textbf{1.}$~$ $n_{t_{\mathcal{V}_j}^{maj}}$ is associated with a (re)structuring symbol (fig. \ref{chap3:fig:expansion-pattern}(d)) then: we take a step forward in the depth path of $t_{\mathcal{V}}^{maj}$ and we resume processing;
	
	\noindent\textbf{2.}$~$ $n_{t_i}$, $n_{t_{\mathcal{V}_j}^{maj}}$ and $n_{t_{g_k}}$ exist and are all associated with the same symbol $X$ (fig. \ref{chap3:fig:expansion-pattern}(a) and \ref{chap3:fig:expansion-pattern}(b)) then:
	%\begin{enumerate}[leftmargin=*]
	we insert $n_{t_{\mathcal{V}_j}^{maj}}$ (it is the most up-to-date node) into $t_f$ at the address $w_k$; 
	if $n_{t_{\mathcal{V}_j}^{maj}}$ is a bud then we prune (delete sub-artifacts) $t_g$ at the address $w_k$; 
	we take a step forward in the depth path of the three artifacts and we resume processing.
	%\end{enumerate}
	
	\noindent\textbf{3.}$~$ $n_{t_i}$, $n_{t_{\mathcal{V}_j}^{maj}}$ and $n_{t_{g_k}}$ exist and are respectively associated with symbols $X_i$, $X_j$ and $X_k$ such that $X_k \neq X_i$ and $X_k \neq X_j$ (fig. \ref{chap3:fig:expansion-pattern}(e)) then: 
	%\begin{enumerate}[leftmargin=*]
	we add $n_{t_{g_k}}$ in $t_f$ at address $w_k$. This is an upstair bud; 
	we take a step forward in the depth path of $t_g$ and we resume processing.
	%\end{enumerate}
	
	\noindent\textbf{4.}$~$ $n_{t_i}$ (resp. $n_{t_{\mathcal{V}_j}^{maj}}$) and $n_{t_{g_k}}$ exist and are associated with the same symbol $X$ (fig. \ref{chap3:fig:expansion-pattern}(c) and \ref{chap3:fig:expansion-pattern}(f)) then: 
	%\begin{enumerate}[leftmargin=*]
	we insert $n_{t_i}$ (resp. $n_{t_{\mathcal{V}_j}^{maj}}$) into $t_f$ at the address $w_k$;
	if $n_{t_i}$ (resp. $n_{t_{\mathcal{V}_j}^{maj}}$) is a bud, we prune $t_g$ at the address $w_k$; 
	we take a step forward in the depth path of the artifacts $t$ (resp. $t_{\mathcal{V}}^{maj}$) and $t_g$, then we resume processing.
	%\end{enumerate}
	%\end{itemize}
\end{mdframed}
\end{algorithm}

\begin{figure}[ht!]
	\noindent
	\makebox[\textwidth]{\includegraphics[scale=0.27]{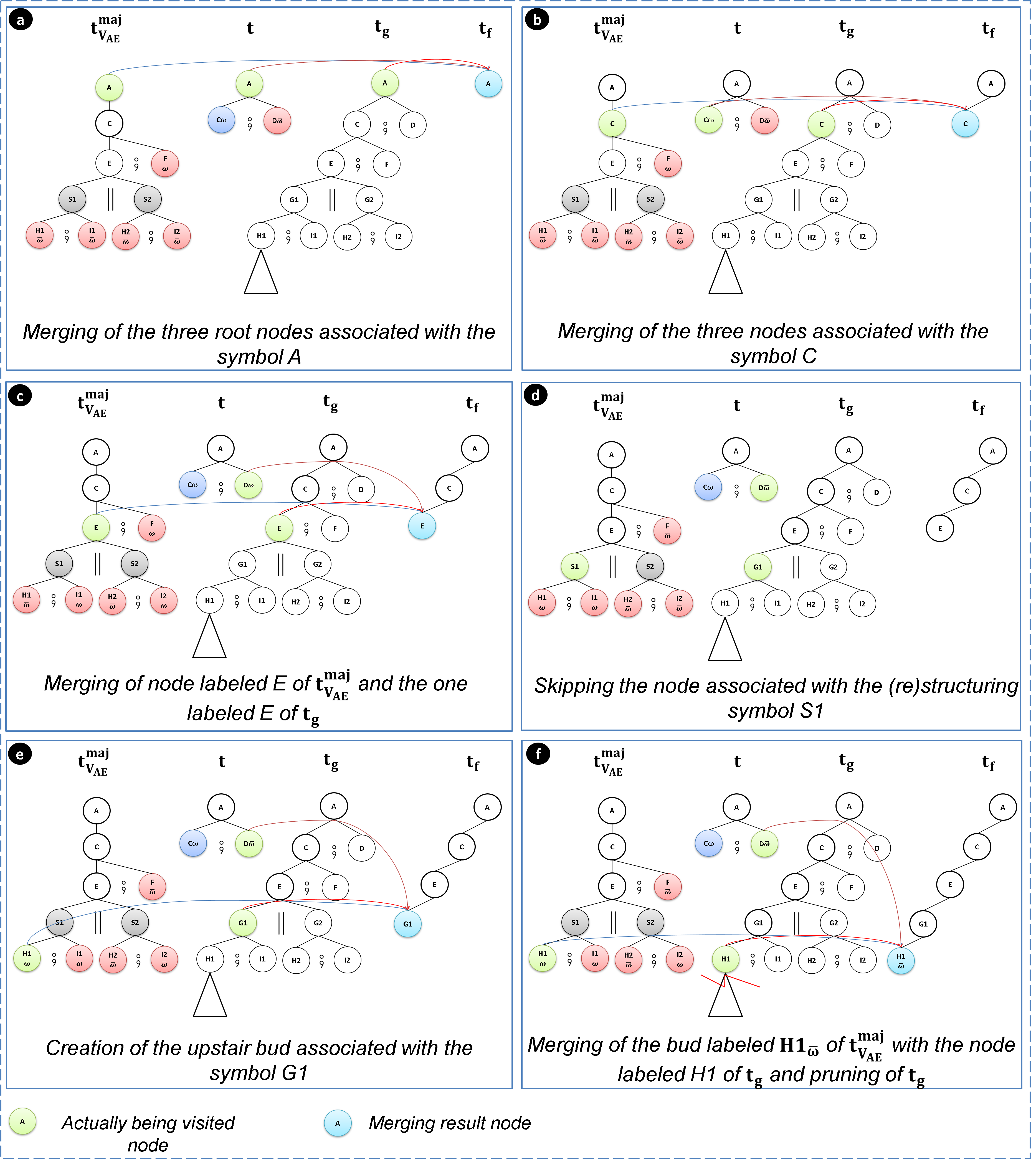}}
	\caption{Some scenarios to be managed during the expansion.}
	\label{chap3:fig:expansion-pattern}
\end{figure}

As for the other key algorithms, a Haskell implementation of this expansion-pruning algorithm is introduced in appendix \ref{appendice1:algorithms-implementations}.

\begin{proposition}
	\label{propositionAtLeastOneGuide}
	For any update $t_{\mathcal{V}}^{maj}$ in accordance with a GMWf $\mathbb{G}_{\mathcal{V}} = \Pi_{\mathcal{V}}\left(\mathbb{G} \right)$, of a partial replica $t_{\mathcal{V}}=\pi_{\mathcal{V}}\left(t\right)$ obtained by projecting (according to the view $\mathcal{V}$) an artifact $t$ being executed in accordance with the GMWf $\mathbb{G}$ of a GMAWfP verifying the axioms' visibility and the non-recursivity assumptions, there is at least one target artifact (the three-way merge guide) $t_g \in arts_{\mathbb{G}}$ of $\mathbb{G}$ such as:
	\begin{enumerate}
		\item[\textbf{(a)}] $t$ is a prefix of $t_g$ ($t \leq t_g$)
		\item[\textbf{(b)}] $t_{\mathcal{V}}^{maj}$ is a prefix of the partial replica of $t_g$ according to $\mathcal{V}$ $\left(t_{\mathcal{V}}^{maj} \leq \pi_{\mathcal{V}}\left(t_g \right)\right)$
	\end{enumerate}	
\end{proposition}	

\begin{proof}
	Thanks to the proposals \ref{propositionStabiliteProjGMWf}, \ref{propositionCoherenceArtefact} and the artifact editing model used here (see sec. \ref{chap3:sec:artifact-edition}) it is established that since the artifact $t$ being executed in accordance with $\mathbb{G}$ is a prefix of a non-empty set of $\mathbb{G}$'s target artifacts $arts_{\mathbb{G}}^{'} = \left\{t_{1}^{'},\ldots, t_{n}^{'}\right\}$ ($\forall 1 \leq i \leq n, ~t \leq t_{i}^{'}$), its projection $t_{\mathcal{V}}$ according to the view $\mathcal{V}$ is a prefix of a non-empty set $arts_{\mathbb{G}_{\mathcal{V}}}^{'} = \left\{t_{{\mathcal{V}}_{1}}^{'},\ldots, t_{{\mathcal{V}}_{m}}^{'}\right\}$ of $\mathbb{G}_{\mathcal{V}} = \Pi_{\mathcal{V}}\left(\mathbb{G} \right)$'s local target artifacts for the said view ($\forall 1 \leq j \leq m, ~t_{\mathcal{V}} \leq t_{{\mathcal{V}}_{j}}^{'}$): elements of $arts_{\mathbb{G}}^{'}$ are potential merging guides candidates that all verify the property \textbf{(a)}. In addition, using the propositions \ref{propositionStabiliteProjGMWf} and \ref{propositionReciproqueCoherenceArtefact}, it is established that each element of $arts_{\mathbb{G}_{\mathcal{V}}}^{'}$ is the projection of at least one element of $arts_{\mathbb{G}}^{'}$ according to the view $\mathcal{V}$ \textbf{(1)}. Given that $t_{\mathcal{V}}^{maj}$ is obtained by developing buds of $t_{\mathcal{V}}$ in accordance with $\mathbb{G}_{\mathcal{V}}$, it is inferred that $t_{\mathcal{V}}^{maj}$ is a prefix of a non-empty subset $arts_{\mathbb{G}_{\mathcal{V}}}^{maj} \subseteq arts_{\mathbb{G}_{\mathcal{V}}}^{'}$ of local target artifacts for the view $\mathcal{V}$ \textbf{(2)}. With the proposition \ref{propositionReciproqueCoherenceArtefact} once again, it is determined that for each artifact $t_{{\mathcal{V}}_{j}}^{'} \in arts_{\mathbb{G}_{\mathcal{V}}}^{maj}$, there is at least one artifact $t_{g_j}$ that is conform to $\mathbb{G}$ such as $t_{{\mathcal{V}}_{j}}^{'} = \pi_{\mathcal{V}}\left(t_{g_j} \right)$: this new set $arts_{\mathbb{G}}^{maj} = \left\{t_{g_1},\ldots, t_{g_k}\right\}$ is made up of potential merging guides candidates that all verify the property \textbf{(b)}. Results \textbf{(1)} and \textbf{(2)} show that $arts_{\mathbb{G}}^{maj}$ and $arts_{\mathbb{G}}^{'}$ are not disjoint. As a consequence, the set $guides= arts_{\mathbb{G}}^{maj} \cap arts_{\mathbb{G}}^{'}$ of potential merging guides that all verify both property \textbf{(a)} and \textbf{(b)} is not empty.
\end{proof}

\begin{corollary}
	\label{propositionUniqueExpansion}
	For an artifact $t$ being executed in accordance with a GMWf $\mathbb{G}$ of a GMAWfP verifying the axioms' visibility and the non-recursivity assumptions, and an update $t_{\mathcal{V}}^{maj} \geq t_{\mathcal{V}}$ of its partial replica $t_{\mathcal{V}}=\pi_{\mathcal{V}}\left(t\right)$ according to the view $\mathcal{V}$, the expansion of $t_{\mathcal{V}}^{maj}$ contains at least one artifact and the expansion-pruning algorithm presented here returns one and only one artifact.
\end{corollary}

\begin{comment}
\begin{proof}
	The proof of this corollary derives from the proof of the proposition \ref{propositionAtLeastOneGuide} and from the fact that in the last instruction of the algorithm \ref{chap3:algo:search-guide}, an artifact is randomly selected an returned from a non-empty set of potential guides.
\end{proof}
\end{comment}
This result (corollary \ref{propositionUniqueExpansion}) derives from the proof of the proposition \ref{propositionAtLeastOneGuide} (\textit{there is always at least one artifact in the expansion of $t_{\mathcal{V}}^{maj}$ under the conditions of corollary \ref{propositionUniqueExpansion}}) and from the fact that in the last instruction of the algorithm \ref{algo:search-guide}, an artifact is randomly selected an returned from a non-empty set of potential guides (\textit{only one of the expansion artifacts is used in the three-way merging}).

~

\noindent\textbf{\textit{The Diffusion Protocol}}
\label{chap3:sec:execution-protocol-diffusion}

After expansion-pruning, the local workflow engine must examine whether requests need to be sent to other agents (this is the case if $t_i$ still have unlocked buds created on its site\footnote{An agent only requests the execution of a bud if it was created on its site.}, on which the current agent is not accredited in writing) or, if responses are to be returned (this is the case if $t_i$ is complete or semi-complete).

To build the list of requests to diffuse, the local workflow engine scans the artifact $t_i$ produced by expansion-pruning and builds the list of required agents (those to receive a request) from buds\footnote{Normally (due to assumptions of our model) at this stage, for each unlocked bud of $t_i$, agent $i$ is accredited in execution on the associated task. Any other situation would be a design flaw.} (algorithm \ref{algorithmeDiffusion}, lines 2 to 7). If the required agents list is not empty (the artifact is not complete), it sends a request to each agent in the list (algorithm \ref{algorithmeDiffusion}, lines 8 to 12). Otherwise, if there are agents who have previously made requests, it sends responses instead (algorithm \ref{algorithmeDiffusion}, lines 13 to 18).

\begin{algorithm}
\small
\caption{Diffusion protocol executed by an agent $i$.\label{algorithmeDiffusion}}
\begin{algorithmic}[1]
	\Procedure{Diffusion}{}
		\For{$bud : unlockedBuds(t_i)$ and $bud$ have been created by $i$}%\Comment{For each unlocked bud in the updated artefact}
			\If{$bud.type \in \mathcal{A}_{A_i(x)}$}%\Commnent{If peer $i$ can ask the execution of this task}
				\State $agent \gets \textrm{\textit{executorOf}}(bud.type)$
				\State $enqueue(agent, requiredAgents)$
			\EndIf
			\EndFor
			\If {$not~isEmpty(requiredAgents)$}
				\State $req \gets new~Request(i,~t_i)$
			  \For{$agent : requiredAgents$}
					\State $invoqueService("forwardTo",~req,~agent)$
				\EndFor		
			\ElsIf{$not~isEmpty(RET_i)$}
				\State $ans \gets new~Answer(i,~t_i)$
			 	\For{$agent : RET_i$}
					\State $invoqueService("returnTo",~ans,~agent)$
					\State $delete~agent~from~RET_i$
				\EndFor
			\Else
				\State $alert("The~process~execution~is~terminated.")$
			\EndIf
	\EndProcedure
\end{algorithmic}
\end{algorithm}

\mySection{Illustrating the Choreography on the Peer-Review Process}{}
\label{chap3:sec:choreograpghy-illustration}

The execution of an instance of our running example begins when under the editor in chief's action (via the GUI of the specialised editor), an unlocked bud of type $A_{\mathbb{G}}$ is created on his site.
Figure \ref{chap3:fig:execution-figure-1} which must be read following the direction of the arrows it contains, summarises the state\footnote{The state of an agent $i$ at a given moment is given by the values of variables $REQ_i$, $ANS_i$, $RET_i$ and the replicas $t_i$ and $t_{\mathcal{V}_i}$.} of the agent $EC$ (editor in chief site) before and after the event creating the artifact; it also illustrates the running of the five-step protocol on the agent $EC$.
\begin{figure}[ht!]
	\noindent
	\makebox[\textwidth]{\includegraphics[scale=0.28]{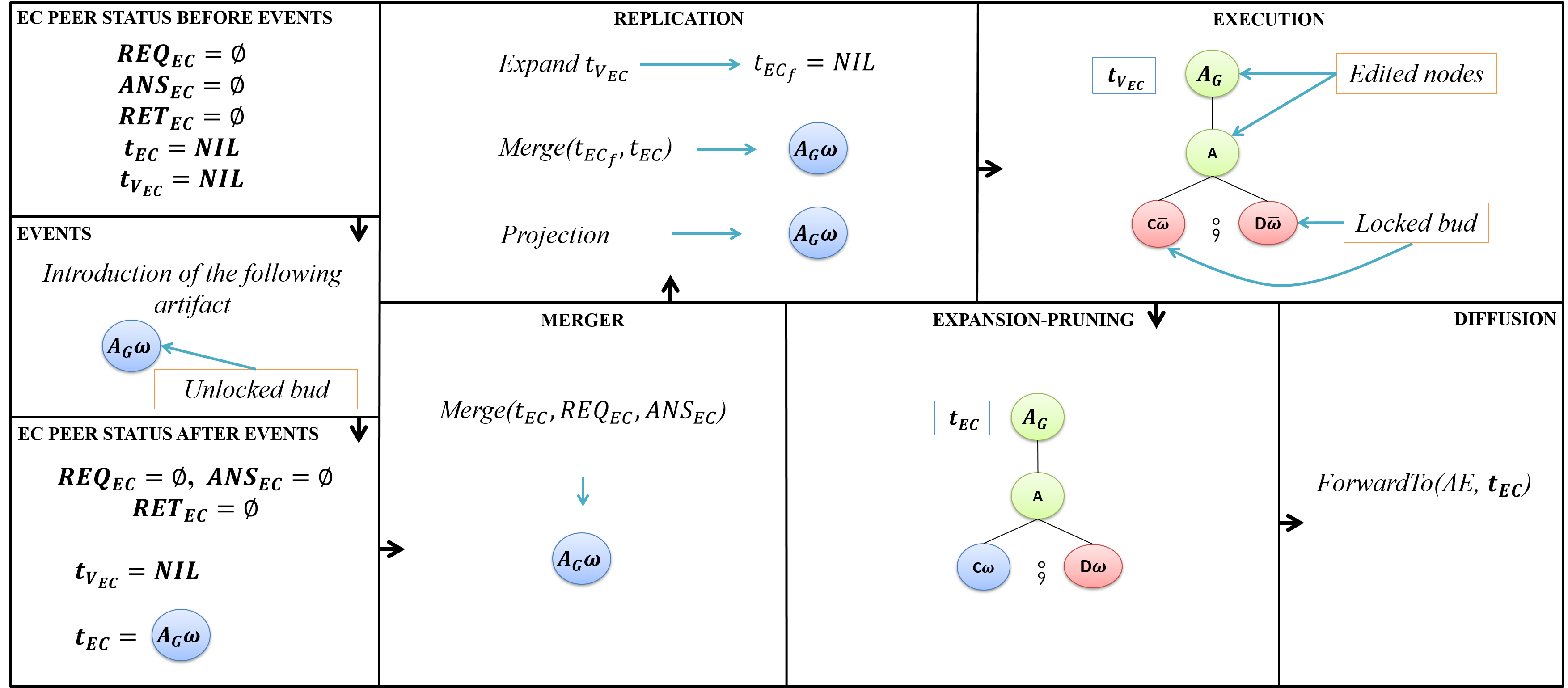}}
	\caption{Beginning of the peer-review process on the editor in chief's site.}
	\label{chap3:fig:execution-figure-1}
\end{figure}

As soon as a bud of type $A_{\mathbb{G}}$ is created, the local workflow engine extends it using the unique $A_{\mathbb{G}}$-production ($P_{1}:A_{\mathbb{G}} \rightarrow A$) of the local GMWf. This results in the creation of a bud of type $A$ that the editor in chief must extend via the specialised editor by choosing an $A$-production.
For this scenario, it is assumed that he chooses the production $P_{3}:A \rightarrow C \fatsemi D$.
The task ($A$) executed by the latter is shown in green colour on figure \ref{chap3:fig:execution-figure-1}.
The newly created tasks ($C$ and $D$) appear in the form of locked buds (the locked buds are shown in red colour) because the editor in chief is not accredited in writing on $C$ and, since $D$ is linked to $C$ by a sequential scheduling constraint, it can only be executed when all tasks ($C$, $E$, $F$, $G1$, $G2$, $H1$, $H2$, $I1$, $I2$) preceding it will have been executed.
After expansion-pruning, the only required agent is the associated editor (responsible for executing task $C$): a request is sent to it by invoking the service \textit{forwardTo}.
\begin{figure}[ht!]
	\noindent
	\makebox[\textwidth]{\includegraphics[scale=0.28]{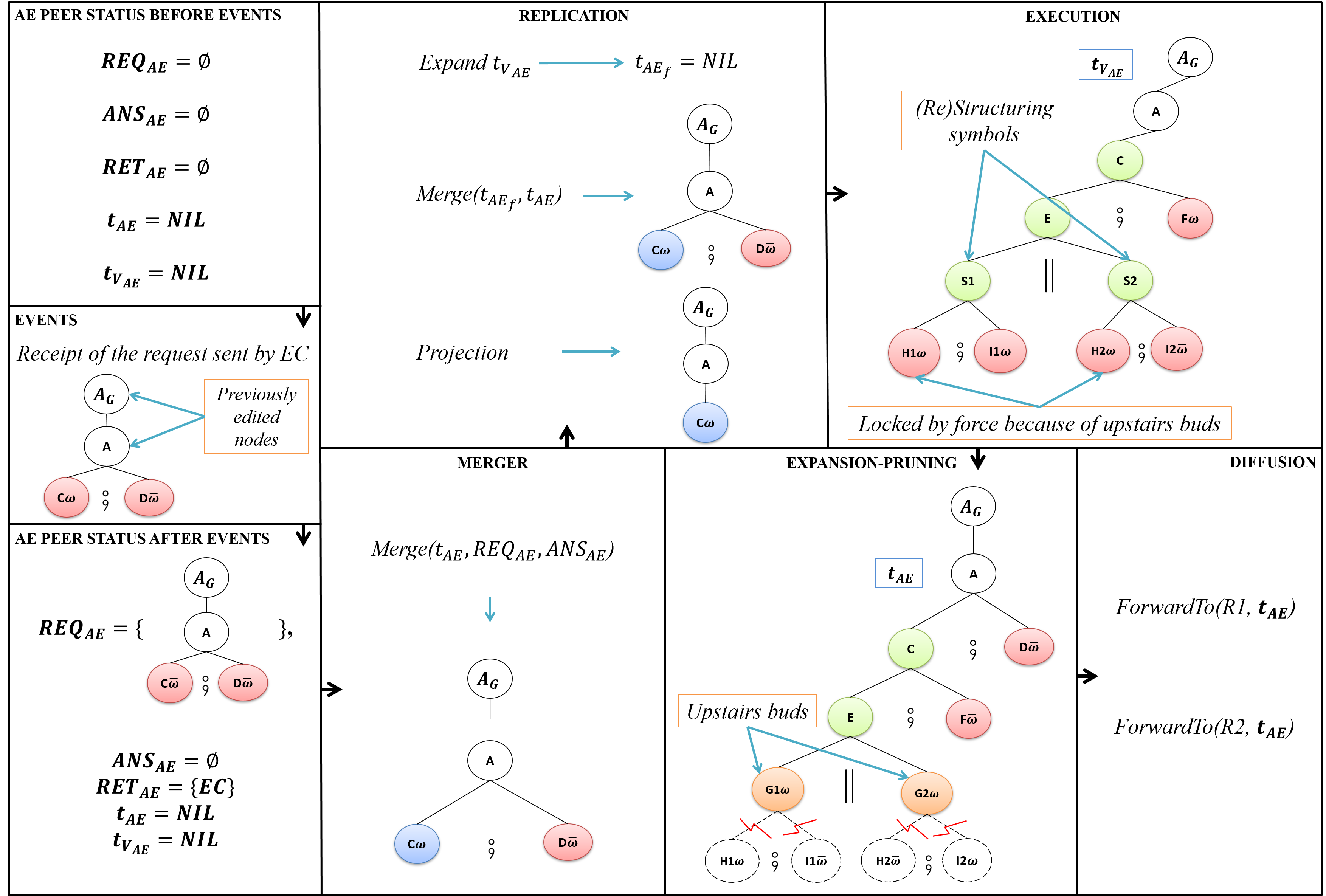}}
	\caption{Continuation of the peer-review process execution on the associated editor's site; the latter receives the request formulated by the editor in chief.}
	\label{chap3:fig:execution-figure-2}
\end{figure}

The event that triggers the workflow execution on the site of associated editor (fig. \ref{chap3:fig:execution-figure-2}) is the receipt of the request sent by the editor in chief.
In the artifact sent by the latter, there are buds ($C_{\overline{\omega}}$ and $D_{\overline{\omega}}$).
After merging, the bud $C_{\omega}$ is unlocked (the unlocked buds are shown in blue). It indicates the only place where the contribution of the associated editor is expected.
During the execution phase, the local artifact partial replica is updated by the associated editor via the productions $P_{3}: C\rightarrow E\fatsemi F$, $P_{4}: E\rightarrow S1\parallel S2$, $P_{5}: S1\rightarrow H1\fatsemi I1$ and $P_{6}: S2\rightarrow H2\fatsemi I2$ of his local GMWf. 
At the end of this phase, buds of types $H1, H2, I1$ and $I2$ appear locked not only because they are constrained by a sequential scheduling (case of $I1$ and $I2$), but especially because of the presence of \textit{upstairs buds} (the upstairs buds are represented in orange colour).
Indeed, tasks $G1$ and $G2$ made visible after the expansion are upstairs buds because they must be executed before the tasks of type $H1$ and $H2$. So, there is pruning at $G1$ and $G2$ before sending (in parallel) the artifact to both referees.
\begin{figure}[ht!]
	\noindent
	\makebox[\textwidth]{\includegraphics[scale=0.28]{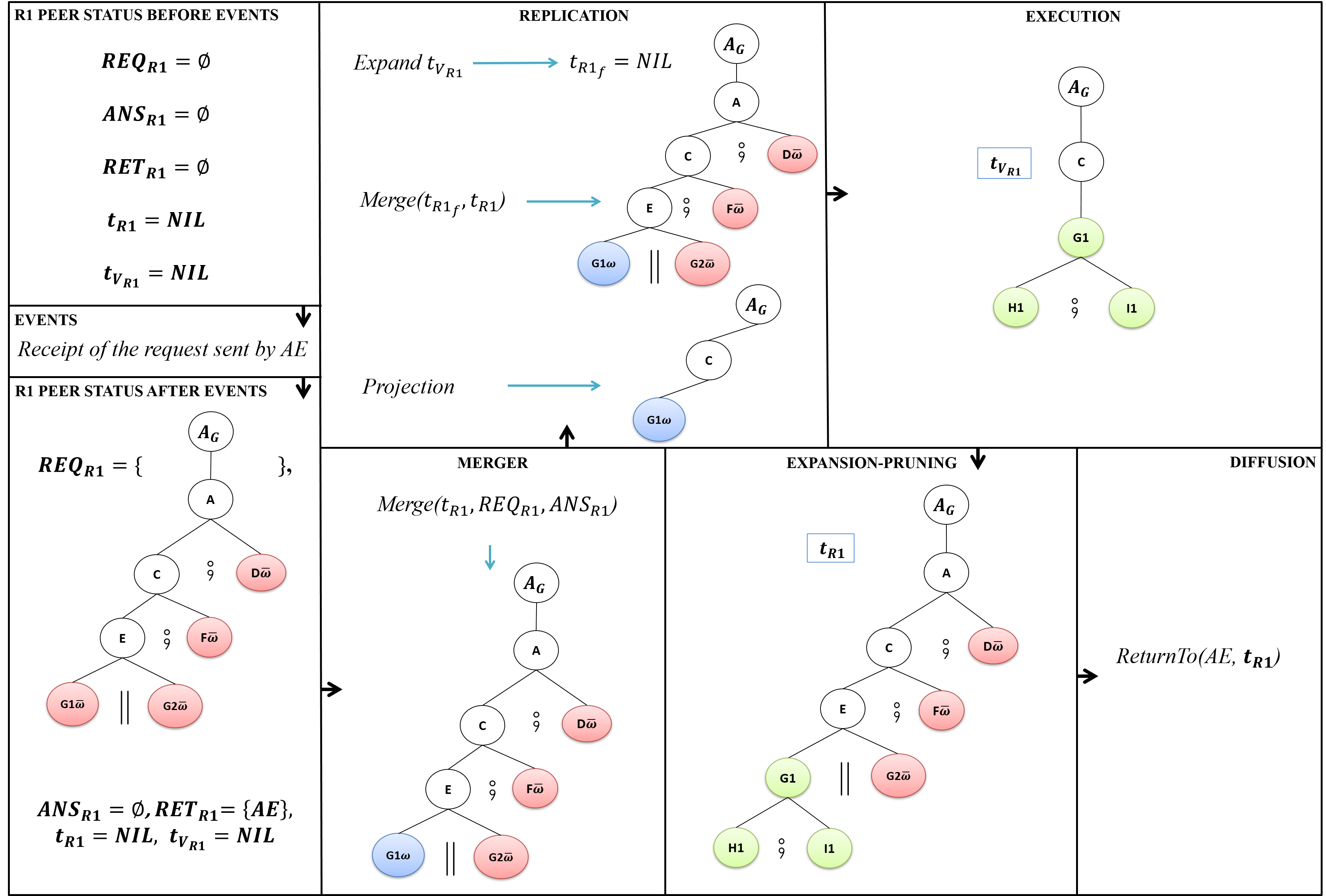}}
	\caption{Continuation of the peer-review process execution on the first referee's site: the request of the associated editor arrives at the first referee.}
	\label{chap3:fig:execution-figure-3}
\end{figure}

Figure \ref{chap3:fig:execution-figure-3} illustrates how the protocol takes place on the site of one of the referees (the first referee). After the contribution of the latter, no new bud is created: no request is formulated. It is rather a response corresponding to the request previously received from the associated editor which is returned by invoking the service $returnTo$.
\begin{figure}[ht!]
	\noindent
	\makebox[\textwidth]{\includegraphics[scale=0.28]{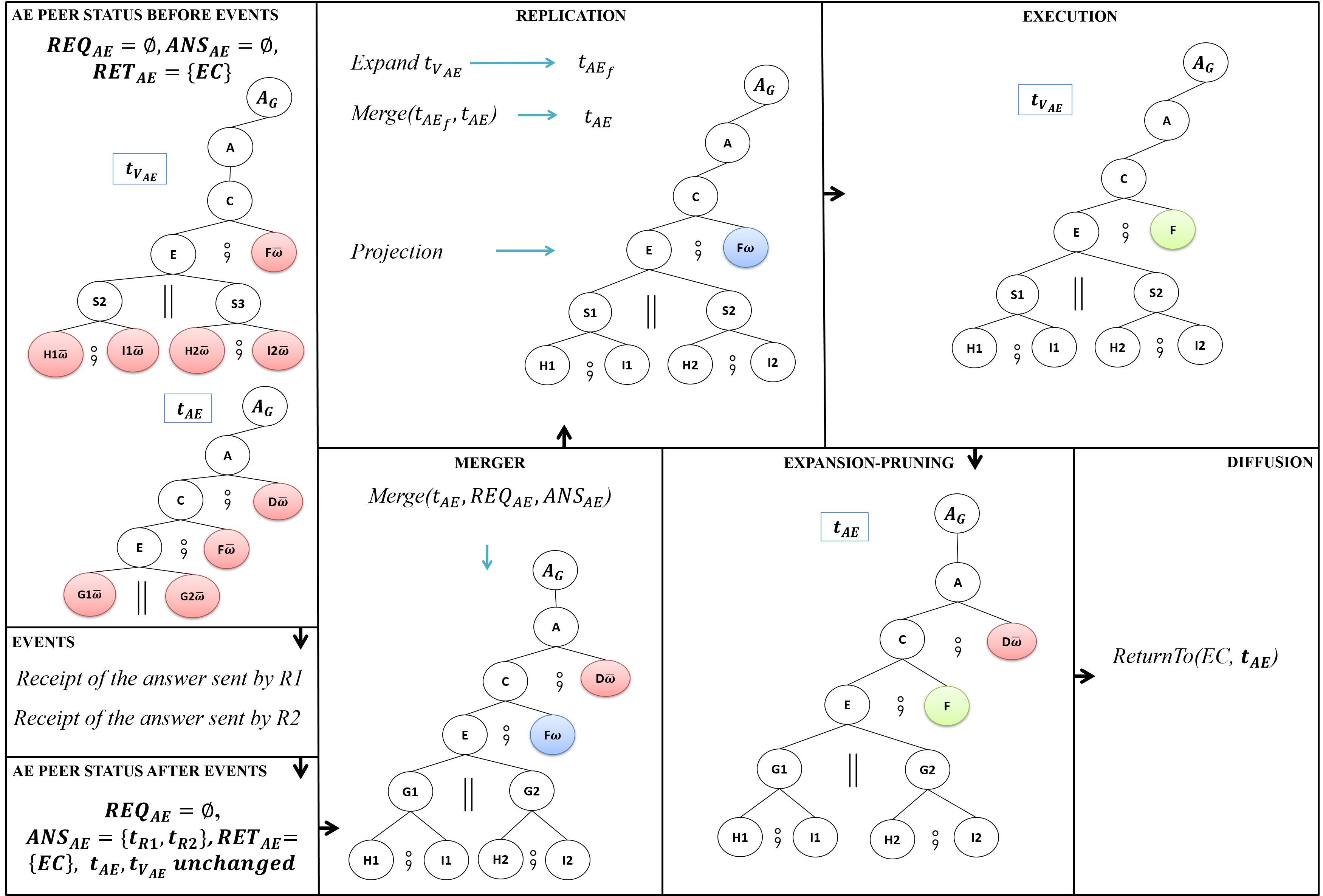}}
	\caption{Continuation of the peer-review process execution: the associated editor receives answers from referees, to requests that he has previously made.}
	\label{chap3:fig:execution-figure-4}
\end{figure}

The execution protocol is unrolled again on the site of the associated editor following events related to the reception of responses from the two referees (fig. \ref{chap3:fig:execution-figure-4}). We choose to treat these responses simultaneously; but we could do otherwise and obtain the same result. 
At merge, since the subtree rooted in $E$ is closed, the bud $F_{\overline{\omega}}$ is unlocked and the associated editor extends it through production $P_{11}: F \rightarrow \varepsilon$. Having no request to make, the answer to the request previously received from the editor in chief is returned.
\begin{figure}[ht!]
	\noindent
	\makebox[\textwidth]{\includegraphics[scale=0.28]{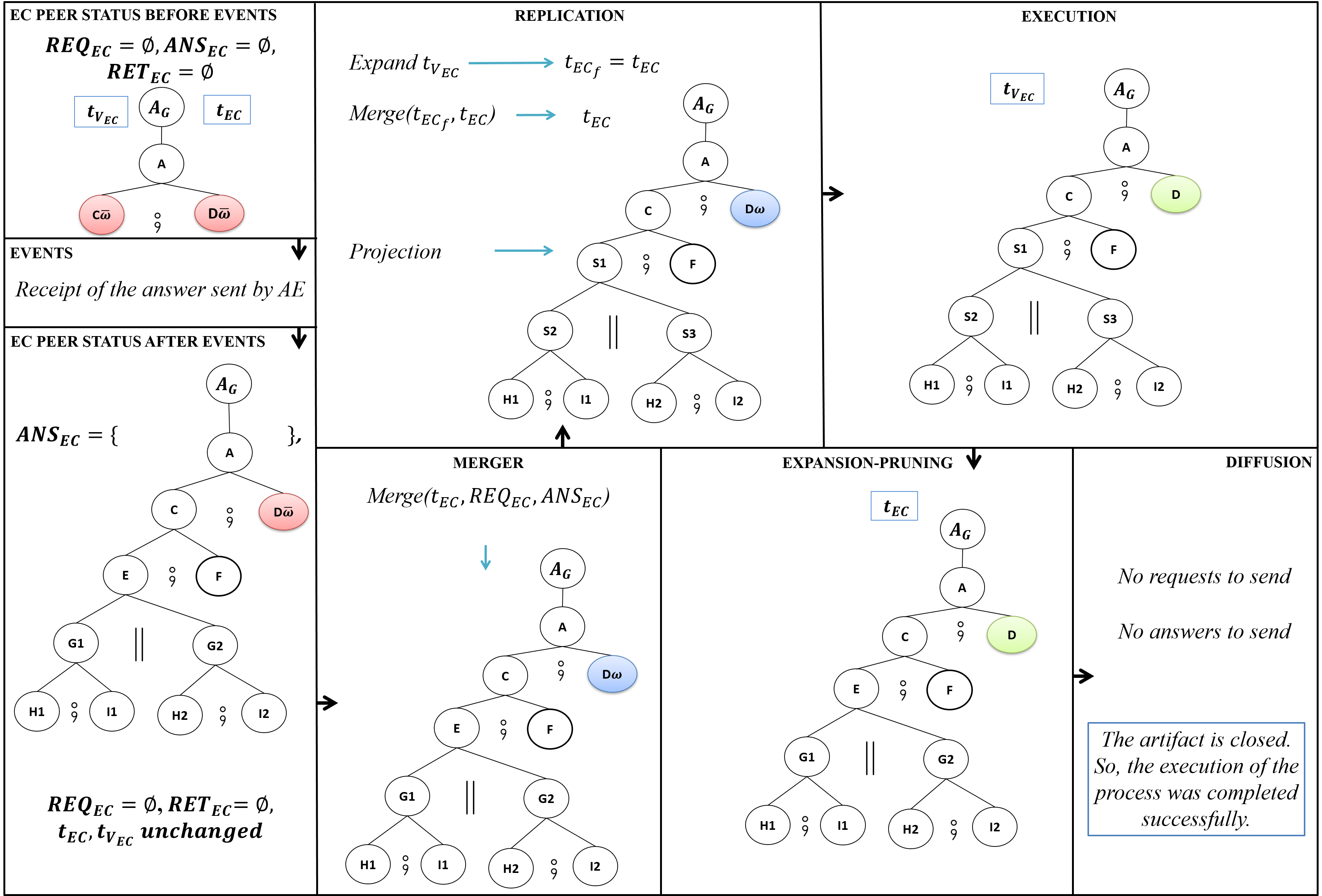}}
	\caption{Continuation and end of the peer-review process execution: the editor in chief receives a response containing referees' contributions, from the associated editor.}
	\label{chap3:fig:execution-figure-5}
\end{figure}

The editor in chief receives the response from associated editor and once again runs the execution protocol (fig. \ref{chap3:fig:execution-figure-5}). 
After its contribution (on the node $D$), the artifact obtained after expansion-pruning is closed and the execution of the process ends successfully.

\mySection{Experimentation}{}
\label{chap3:sec:p2ptinywfms}

In this section, we present and experiment \textit{P2PTinyWfMS} (a Peer-to-Peer Tiny Workflow Management System), an experimental prototype system implemented according to the approach proposed in this chapter.

\mySubSection{P2PTinyWfMS: an Experimental Prototype System}{}
\textit{P2PTinyWfMS} is a tool developed in Java under Eclipse\footnote{Official website of Eclipse: \url{https://www.eclipse.org}, visited the 04/04/2020.} and dedicated to the distributed execution of administrative workflows specified using GMWf. 
In accordance with the agent's architecture of this chapter (see fig. \ref{chap3:fig:peer-architecture}), \textit{P2PTinyWfMS} has a front-end for displaying and graphically editing artifacts manipulated during the execution of a business process (see fig. \ref{chap3:sec:p2ptinywfms-1} and \ref{chap3:sec:p2ptinywfms-3}), as well as a communication module built from SON\footnote{SON is available under Eclipse from a family of SmartTools plugins.}. 

Let's recall that SON (Shared-data Overlay Network) \cite{SON} is a middleware offering several DSL to facilitate the implementation of P2P systems whose components communicate by service invocations. 
Component Description Meta Language (CDML) is the DSL provided by SON to specify among other things the services required and provided by each peers; from a CDML specification, SON generates Java code for allowing peers to communicate.
The following listing shows the contents of the CDML file used in the case of \textit{P2PTinyWfMS} to specify the four services that its instances expose; they are: two input services (\textit{inForwardTo} - lines 6 to 9 -, and \textit{inReturnTo} - lines 10 to 13 -) and two output services (\textit{outForwardTo} - lines 14 to 17 - and \textit{outReturnTo} - lines 18 to 21 -). These services take as argument an artifact corresponding to either a request or a response.

\begin{Verbatim}[frame=lines,fontsize=\scriptsize, numbers=left, numbersep=8pt, label=CDML file: specification of required and provided services of P2PTinyWfMS]
<?xml version="1.0" encoding="ISO-8859-1"?>
<component name="p2pTinyWfMS" type="p2pTinyWfMS" extends="inria.communicationprotocol"
 ns="p2pTinyWfMS">
  <containerclass name="P2pTinyWfMSContainer"/>
  <facadeclass name="P2pTinyWfMSFacade" userclassname="P2pTinyWfMS"/>
  <input name="forwardTo" method="inForwardTo">
    <attribute name="request" 
     javatype="smartworkflow.dwfms.lifa.miu.util.p2pworkflow.PeerToPeerWorkflowRequest"/>
  </input>
  <input name="returnTo" method="inReturnTo">
    <attribute name="response" 
     javatype="smartworkflow.dwfms.lifa.miu.util.p2pworkflow.PeerToPeerWorkflowResponse"/>
  </input>
  <output name="forwardTo" method="outForwardTo">
    <attribute name="request" 
     javatype="smartworkflow.dwfms.lifa.miu.util.p2pworkflow.PeerToPeerWorkflowRequest"/>
  </output>
  <output name="returnTo" method="outReturnTo">
    <attribute name="response" 
     javatype="smartworkflow.dwfms.lifa.miu.util.p2pworkflow.PeerToPeerWorkflowResponse"/>
  </output>
</component>
\end{Verbatim}

\mySubSection{Executing our Running Example under P2PTinyWfMS}{}
SON offers a DSL (the "\textit{.world}" files) for the description of the deployment of a distributed system whose components have been specified by a CDML file. In order to execute our running example, we deployed four instances of \textit{P2PTinyWfMS} identified by $EC$, $AE$, $R1$ and $R2$ respectively. As explained in section \ref{chap3:sec:choreograpghy-illustration}, each instance is initially equipped with the global GMWf as well as accreditations of various agents from which it derives its local GMWf by projection.

Figures \ref{chap3:sec:p2ptinywfms-1}, \ref{chap3:sec:p2ptinywfms-3} and \ref{chap3:sec:p2ptinywfms-6} are screen shots with some highlights of the workflow's distributed execution.
We have the tab "\textit{Workflow overview}" presenting at the beginning of the execution, various tasks, agents, target artifacts etc., on the editor in chief's site (fig. \ref{chap3:sec:p2ptinywfms-1}). We also have the tabs "\textit{Workflow execution}" of the sites of the associated editor (fig. \ref{chap3:sec:p2ptinywfms-3}) and of the editor in chief (fig. \ref{chap3:sec:p2ptinywfms-6}) that present the artifacts resulting from their execution after receiving a request from the editor in chief (resp. after receiving a response from the associated editor).
\begin{figure}[ht!]
	\noindent
	\makebox[\textwidth]{\includegraphics[scale=0.43]{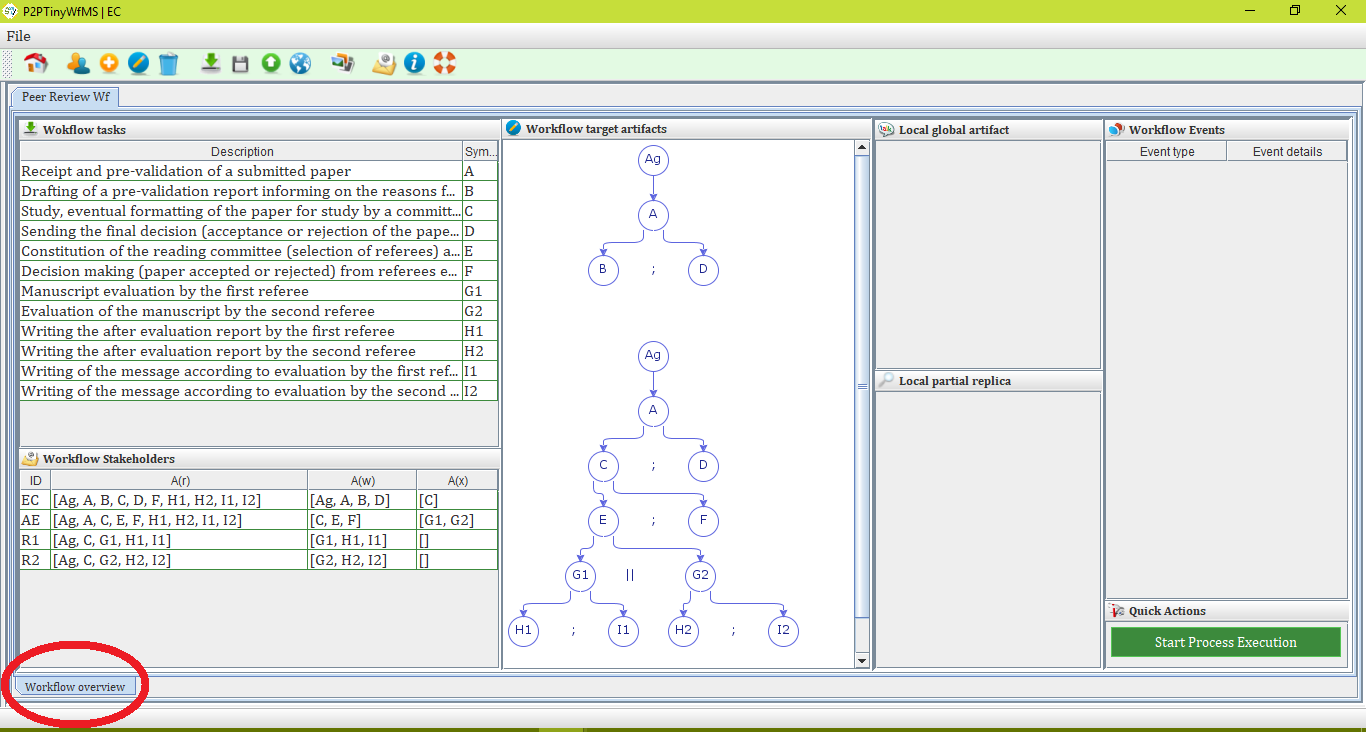}}
	\caption{P2pTinyWfMS on the editor in chief's site: presentation of the GMWf (the tasks and their relations, the actors and their accreditations).}
	\label{chap3:sec:p2ptinywfms-1}
\end{figure}

\begin{figure}[ht!]
	\noindent
	\makebox[\textwidth]{\includegraphics[scale=0.43]{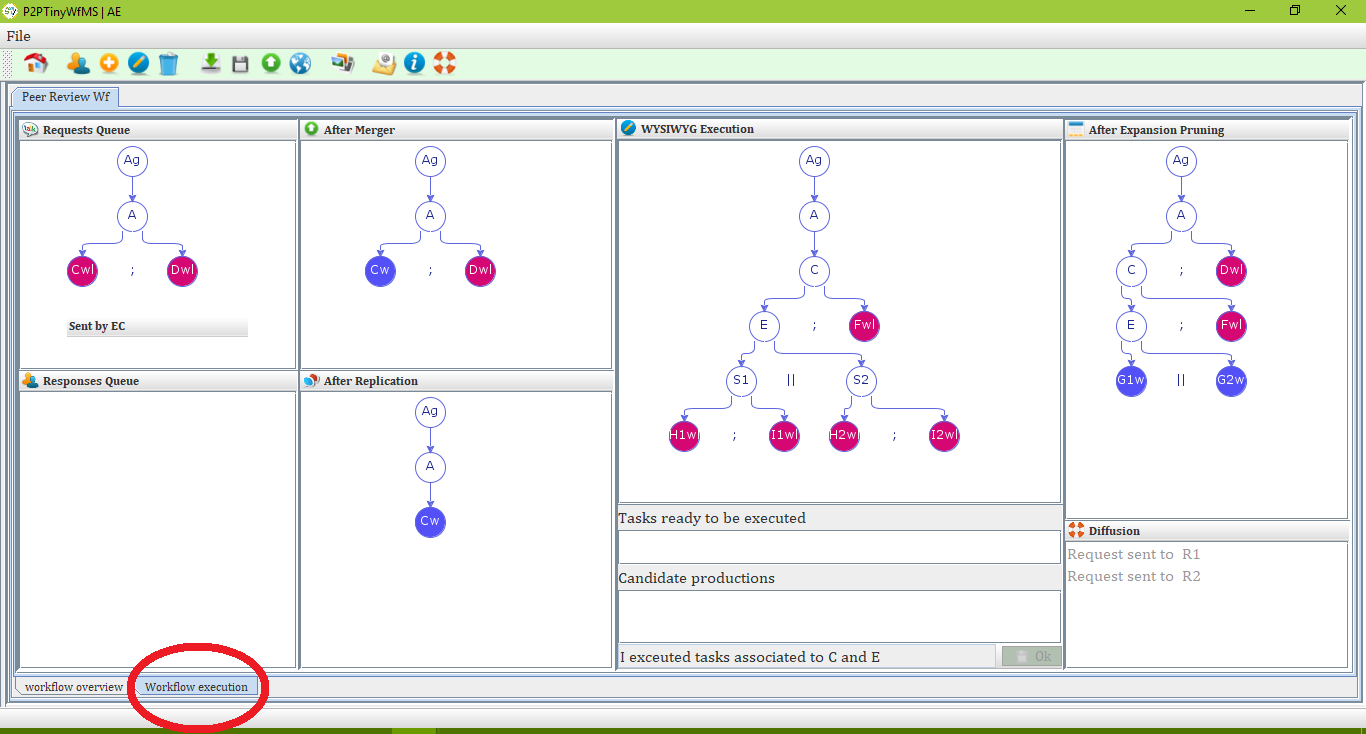}}
	\caption{P2pTinyWfMS on the associated editor's site: receipt of editor in chief's request, execution of tasks, expansion-pruning, and diffusion.}
	\label{chap3:sec:p2ptinywfms-3}
\end{figure}

\begin{figure}[ht!]
	\noindent
	\makebox[\textwidth]{\includegraphics[scale=0.43]{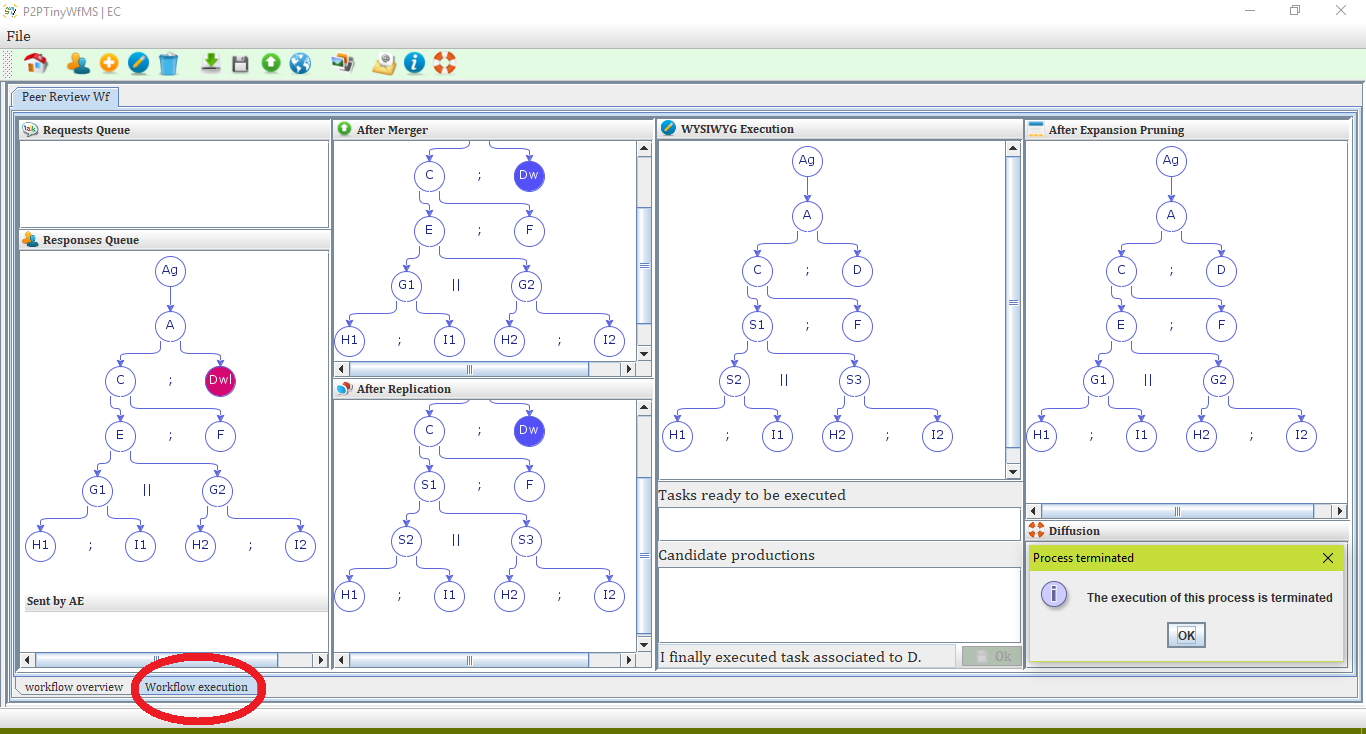}}
	\caption{P2pTinyWfMS on the editor in chief's site: reception of the associated editor's response, execution of tasks, expansion-pruning and end of the case.}
	\label{chap3:sec:p2ptinywfms-6}
\end{figure}

\mySection{Related Works and Discussion}{}
\label{chap3:sec:discussion}

In this section we briefly discuss the similarities and differences of the model presented in this chapter, comparing it with some related work presented earlier (Chapter \ref{chap1:artifact-centric-bpm}). We will mention a few related studies and discuss directly; a more formal comparative study using qualitative and quantitative metrics should be the subject of future work.

Hull et al. \citeyearpar{hull2009facilitating} provide an interoperation framework in which, data are hosted on central infrastructures named \textit{artifact-centric hubs}. As in the work presented in this chapter, they propose mechanisms (including user views) for controlling access to these data. Compared to choreography-like approach as the one presented in this chapter, their settings has the advantage of providing a conceptual rendezvous point to exchange status information. The same purpose can be replicated in this chapter's approach by introducing a new type of agent called "\textit{monitor}", which will serve as a rendezvous point; the behaviour of the agents will therefore have to be slightly adapted to take into account the monitor and to preserve as much as possible the autonomy of agents.

Lohmann and Wolf \citeyearpar{lohmann2010artifact} abandon the concept of having a single artifact hub \cite{hull2009facilitating} and they introduce the idea of having several agents which operate on artifacts. Some of those artifacts are mobile; thus, the authors provide a systematic approach for modelling artifact location and its impact on the accessibility of actions using a Petri net. Even though we also manipulate mobile artifacts, we do not model artifact location; rather, our agents are equipped with capabilities that allow them to manipulate the artifacts appropriately (taking into account their location). Moreover, our approach considers that artifacts can not be remotely accessed, this increases the autonomy of agents.

The process design approach presented in this chapter, has some conceptual similarities with the concept of \textit{proclets} proposed by Wil M. P. van der Aalst et al. \citeyearpar{van2001proclets, van2009workflow}: they both split the process when designing it. In the model presented in this chapter, the process is split into execution scenarios and its specification consists in the diagramming of each of them. Proclets \cite{van2001proclets, van2009workflow} uses the concept of \textit{proclet-class} to model different levels of granularity and cardinality of processes. Additionally, proclets act like agents and are autonomous enough to decide how to interact with each other.

The model presented in this chapter uses an attributed grammar as its mathematical foundation. This is also the case of the AWGAG model by Badouel et al. \citeyearpar{badouel14, badouel2015active}. However, their model puts stress on modelling process data and users as first class citizens and it is designed for Adaptive Case Management.

To summarise, the proposed approach in this chapter allows the modelling and decentralized execution of administrative processes using autonomous agents. In it, process management is very simply done in two steps. The designer only needs to focus on modelling the artifacts in the form of task trees and the rest is easily deduced. Moreover, we propose a simple but powerful mechanism for securing data based on the notion of accreditation; this mechanism is perfectly composed with that of artifacts. The main strengths of our model are therefore : 
\begin{itemize}
	\item The simplicity of its syntax (process specification language), which moreover (well helped by the accreditation model), is suitable for administrative processes;
	\item The simplicity of its execution model; the latter is very close to the blockchain's execution model \cite{hull2017blockchain, mendling2018blockchains}. On condition of a formal study, the latter could possess the same qualities (fault tolerance, distributivity, security, peer autonomy, etc.) that emanate from the blockchain;
	\item Its formal character, which makes it verifiable using appropriate mathematical tools;
	\item The conformity of its execution model with the agent paradigm and service technology.
\end{itemize}
In view of all these benefits, we can say that the objectives set for this thesis have indeed been achieved. However, the proposed model is perfectible. For example, it can be modified to permit agents to respond incrementally to incoming requests as soon as any prefix of the extension of a bud is produced. This makes it possible to avoid the situation observed on figure \ref{chap3:fig:execution-figure-4} where the associated editor is informed of the evolution of the subtree resulting from $C$ only when this one is closed. All the criticisms we can make of the proposed model in particular, and of this thesis in general, have been introduced in the general conclusion (page \pageref{chap5:general-conclusion}) of this manuscript.

\mySection{Summary}{}
\label{chap3:sec:conclusion}
In this chapter, we have proposed a new approach to facilitate workflow design and their execution in a distributed environment. The approach relies on the cooperative editing of a mobile artifact by several agents.
The design of artifacts and their type (artifact type), the architecture and management protocols (choreography) of agents have been presented and discussed. Likewise, an illustration and an experimentation of the whole approach on a real example of an administrative business process (the peer-review process) were presented.

The proposed approach has many advantages, which we have outlined. In particular, it covers all aspects of process automation using BPM technology. 
In addition to the formal and demonstrative presentations, we have conducted a preliminary discussion of the presented approach in relation to existing work, and finally concluded that we have achieved the objective of this thesis.

	\mathversion{normal}
	%\input{Chap4/Chap4}
	
	%Ainsi de suite
	
	%\myChapterStar{Titre}{Titre court}{Ajouter � la table des mati�res? (false|true|chapter|section|subsection|subsubsection -chapter par d�faut-)}
\myChapterStar{General Conclusion}{}{true}
\label{chap5:general-conclusion}
\myMiniToc{section}{Contents}
% If no minitoc then
% \startcontents[chapters]

We hereby summarise the work reported and presented in this document. To this, we associate a critical analysis of our models and methodological choices. Finally, we present some research avenues that can be explored following the work in this thesis manuscript.

\mySectionStar{Recall of this Thesis' Challenge and of our Methodological Choices}{}{true}
In this thesis, we focused on the automation of business processes using the technological framework offered by the BPM domain. 
We have contributed to the ambition of making more accessible, administrative business processes automation through this technology. We were guided by the aim of increasing its success in business sectors using administrative processes as it has been in other sectors such as science, banking and insurance, which are governed by much more programmable processes. 
We thought that a first step towards achieving this great ambition was to rely on the current (most up-to-date) BPM paradigms and tools, to design and implement a new BPM framework that would be tailor-made for the management of administrative business processes. 
Having identified the artifact-centric BPM, structured cooperative edition, P2P computing, multiagent system and SOA concepts as being the hot topics in the implementation of BPM, we set the following goal for this thesis:
\begin{displayquote}
	\textit{The proposal of a new artifact-centric framework, facilitating the modelling of administrative business processes and the completely decentralised execution of the resulting workflows; this completely decentralised execution being provided by a P2P system conceived as a set of agents communicating asynchronously by service invocation so that, the execution of a given workflow instance is technically assimilated to the cooperative edition of (mobile) structured documents called artifacts.}
\end{displayquote}

To achieve this goal, we have chosen to base our work on the structured document cooperative editing model developed by Badouel and Tchoup\'e a decade earlier. They proposed an approach based on grammatical models, to manage the lifecycle of a document collegially edited by actors on geographically distant sites. In their model, each actor has a potentially partial view (obtained by projection) of the edited document. The latter is used as an interface between the different actors of the system. When an actor receives a document he must know from its content, what he can do and/or what he has to do about it. The information contained in a document can only grow over its lifetime in the system. Since the system is distributed and under the assumption that several actors can contribute concurrently, it is possible that at any given time there may be several potentially partial replicas of the document in the system. Therefore, it was necessary not only to address the problems related to the coherence of views in order to ensure the feasibility of synchronisation/merge, and to ensure the system's convergence towards a coherent end-state, but also to provide algorithms for the merging of partial replicas. Badouel and Tchoup\'e did this brilliantly, making their work a solid foundation for modelling workflow systems.

We adapted Badouel and Tchoup\'e's document model to define an artifact model. Like them, we have therefore chosen to use grammars as our basic mathematical instrument. In the same vein, we made the choice to model workflow systems in which actors have potentially only a partial perception of the processes they execute. We believe that this configuration is relevant to many administrative business processes. For example, in a peer-review process, a reviewer does not necessarily need to know if another reviewer has been contacted for the expertise of the article entrusted to him; and even if so, he should not necessarily know if the latter has already returned his report, etc. 
Similarly, when organising a journey for a Head of State, not all actors (secret services, civil office, doctor, presidential guard, etc.) have access to the same information which may include for example, tasks to be executed, their dates and states of execution, etc.

\mySectionStar{A Critical Analysis of the Performed Work}{}{true}
Chronologically, we started by better understanding Badouel and Tchoup\'e's model in order to extend it so that it takes into account, conflicts detection and resolution. We then embarked on the construction of a workflow system based on this model by first proposing a generic architecture of such systems and an experimental system prototype based on it. We finished by proposing an artifact-centric framework for the completely decentralised management of administrative business processes. The results we can claim are the following:

~

\noindent\textbf{An algorithm for reconciling partial replicas of a structured document}: we gave a solution to the conflicts that arose when merging partial replicas of a structured document by developing reconciliation and control techniques adapted to the modelling of Badouel and Tchoup\'e. To address this problem, we have expressed it as that of the consensual merging of $k$ updated partial replica $(t_{\mathcal{V}_i}^{maj})_{1 \leq i\leq k}$ (according to $k$ views $(\mathcal{V}_i)_{1 \leq i\leq k}$) whose global model is given by a grammar $\mathbb{G}=\left(\mathcal{S},\mathcal{P},A\right)$, which consists in: finding the largest documents $t^{maj}_{\mathcal{S}}$ conforming to $\mathbb{G}$ such that, for any document $t$ conforming to $\mathbb{G}$ and admitting $t_{\mathcal{V}_i}^{maj}$ as projection along the view ${\mathcal{V}_i}$, $t^{maj}_{\mathcal{S}}$ and $t$ are eventually updates each for other. 
The solution we have proposed is as follows: (1) We associate a tree automaton with \textit{exit states} $\mathcal{A}^{(i)}$ for each update $t_{\mathcal{V}_i}^{maj}$ of a partial replica $t_{\mathcal{V}_i}$; this automaton recognises the trees (conform to the global model) for which $t_{\mathcal{V}_i}^{maj}$ is a projection. (2) We perform a \textit{synchronous product} of the automata $\mathcal{A}^{(i)}$ with a commutative and associative operator noted $\otimes$ that we define to obtain the consensual automaton $\mathcal{A}_{(sc)}$ generating the consensus documents: $\mathcal{A}_{(sc)}=\otimes\mathcal{A}^{(i)}$. (3) We obtain the consensus documents by generating the set of trees accepted by the automaton $\mathcal{A}_{(sc)}$.

~

\noindent\textbf{A software architecture for the implementation of workflow systems}: we proposed a generic architecture that could facilitate the implementation of workflow systems as modelled by Badouel and Tchoup\'e. The proposed architecture is composed of three tiers: \textit{clients}, a \textit{central server} and \textit{administration tools}. These three tiers are interconnected around a middleware that facilitates service-oriented interfacing between them.

~

\noindent\textbf{TinyCE v2}: based on the proposed system architecture, we have built a workflow system prototype referred to as TinyCE v2. The latter was coded in Java and Haskell following a cross-fertilisation protocol that we presented in this manuscript, and it allowed us to test all the proposed algorithms related to the reconciliation of documents' partial replicas.

~

\noindent\textbf{A grammar-based language for the artifact-centric modelling of administrative processes}: we proposed a new tool (a language) that allows to specify any administrative business process using a triplet $\mathbb{W}_f=\left(\mathbb{G}, \mathcal{L}_{P_k}, \mathcal{L}_{\mathcal{A}_k} \right)$ composed of: 
a grammatical model (GMWf) $\mathbb{G}$, a list of actors $\mathcal{L}_{P_k}$ and a list of accreditations $\mathcal{L}_{\mathcal{A}_k}$. 
The GMWf is an attributed grammar used to describe all the tasks (by means of its symbols or sorts) of the studied process and the precedence of execution between them (by means of its productions); it is used as artifact type.
The list of accreditations provides information on the role played by each actor involved in the process execution; it is through accreditations that one is able to model the potentially partial perceptions that different actors have on the processes and their data.

~

\noindent\textbf{A distributed workflow system and a fully decentralised execution model of administrative business processes}: we proposed a multiagent-like distributed system in which autonomous software agents, based on a same and unique architecture presented in this manuscript, can exchange artifacts to communicate through service invocations, so as to orchestrate the fully decentralised execution of a given administrative business process instance modelled using the proposed grammar-based language. Each time a given agent receives an artifact, it executes a unique five-step protocol allowing it to identify tasks that are ready to be executed on its site, to allow their execution by the local actor and to diffuse if necessary, the updated artifact. To ensure the successful completion of this execution model on a given process case, the initial configuration of agents must be such as to guarantee the coherence of their respective accreditations (views). We have therefore proposed projection algorithms to derive agent-specific models that allow them to control their actions in order to ensure the system's convergence towards a business goal state. We have also investigated some mathematical properties of these algorithms.

~

\noindent\textbf{P2PTinyWfMS}: we have finally produced a prototype of a distributed system providing an artifact-centric management of administrative workflows according to our models. In this one, we have implemented all our algorithms in Java language and we have tested them on some process examples with convincing results.

~

The work we have done and presented in this manuscript is not perfect: neither in the applied methodology, nor in its presentation, and even less in its scientific contributions. The first criticisms that we can make of it are the following:

~

\noindent\textbf{The diversity of our algorithms' presentation formats}: we didn't just use pseudo code to present our algorithms. Sometimes we used code (Haskell and Java) and other times we wrote them as an arbitrary and ordered set of instructions written in natural language; in these cases we have neglected the more frequently used and more precise mathematical notations. This methodological choice of presentation can indeed be confusing for the reader. We justify it, however, by our desire to be precise, concise and as simple to understand as possible. We have written each of our algorithms in all the formats used in this manuscript before selecting for each of them, the format that seemed to us the clearest and simplest to present and understand.

~

\noindent\textbf{Conflict management}: we have chosen to use a single conflict resolution strategy: that of rejecting conflicting contributions and asking for new ones from contributors. This seems to us rather restrictive but it was a necessary (not necessarily wise) choice for a complete automation of the process of merging partial replicas of a structured document. However, in practice, it would be more appropriate to propose, following the example of Git, several conflict management strategies using a participant as an actor (coordinator) of this resolution.

~

\noindent\textbf{Insistence on manipulating user views}: although we have already justified the choice to take user views into consideration in our models by explaining their contribution to both security and accuracy, it is no less true that they do not always have a positive impact on our work. They have made it a little more complex and have led us to make slightly restrictive assumptions such as the non-recursive GWMf assumption that guarantees their projection. We could make the use of these views more flexible by restricting it to only censorship and not to the complete deletion of sensitive data; this would certainly allow us to overcome some assumptions.

~

\noindent\textbf{The weak study of our execution model's properties}: apart from the isolated study of the properties of a few of its algorithms, we did not study some properties of our decentralised workflow execution model as it is usual in similar BPM studies. We probably pay the consequences of our not completely formal and uncommon (but specific to artifact-centric models, it is one of their often mentioned limits) presentation that mixes the artifact modelling with its execution. A clear separation of these two aspects would certainly allow us to better study them in an isolated and more conventional way.

~

\noindent\textbf{Still as theoretical as ever}: obviously, we have the ambition to produce concrete systems that can be used in production environments. We are still a long way from that. For the two types of systems studied in this thesis, we have only produced prototypes that allow us to provide experimental proof of concepts. Theoretical studies on the Badouel and Tchoup\'e model being already quite advanced, it would be time to start implementing these concrete environments.

\mySectionStar{Some Perspectives}{}{true}
The perspectives presented here are classified into categories according to their priority, to better guide the reader wishing to continue the work. The categories includes :
\begin{itemize}
	\item \textbf{Short-term}: to indicate that the perspective is almost unavoidable and its results will be a real plus to the overall vision we have; it is therefore necessary to work on it as soon as possible;
	\item \textbf{Mid-term}: to indicate that the perspective's results will be pratical and usable but not a necessity;
	\item \textbf{Long-term}: to indicate that the perspective is optional and its theoretical results would only help give credibility to our work.
\end{itemize}
Here are now some interesting avenues for the continuation of our work that come to mind:

~

\noindent\textbf{A Language for the Specification of Administrative Workflow based on Attributed Grammars (LSAWfP) (priority: short-term)}: it is obvious that process modelling is a crucial phase of BPM \cite{dumas2018fundamental}. Despite the many efforts made in producing process modelling tools, existing tools (languages) are not commonly accepted. They are mainly criticised for their inability to specify both the tasks making up the processes and their scheduling (their lifecycle models), the data they manipulate (their information models) and their organisational models. Process modelling in these languages often results in a single task graph; such a graph can quickly become difficult to read and maintain. Moreover, these languages are often too general (they have a very high expressiveness); this makes their application to specific types of processes complex: especially for administrative processes. 
One can generalise the artifact specification model presented in this thesis, in order to provide a new language for administrative processes modelling that allows designers to specify the lifecycle, information and organisational models of such processes using a mathematical tool based on a variant of attributed grammars. Therefore, the approach imposed by the new language will certainly require the designer to subdivide his process into scenarios, then to model each scenario individually using a simple task graph (an annotated tree) from which a grammatical model will be further derived. At each moment then, the designer will manipulate only a scenario of the studied process: this seems more intuitive and modular because it will allow to produce task graphs that will be more refined and therefore, more readable and easier to maintain.

~

\noindent\textbf{A Scenario-Oriented Scheme for Administrative Business Processes Modelling (priority: mid-term)}: in the BPM community, researchers and professionals in the field have little interest in the "how" to model business processes to the benefit of the "with what" to model them. As a result, there is a plethora of workflow modelling languages but very few methods \cite{dumas2018fundamental}. The question on the method to be used to successfully model a given process is however crucial when we know that BPM reduces the automation of the said process to its specification in a particular workflow language: a well carried-out specification produces a quality workflow system. Because the modelling language that can be extracted from the work of this thesis seems to be adapted to a process modelling philosophy centered on the notion of scenario, it would be interesting to propose a method that would accompany it. This one would present the steps to be followed to succeed in its scenario-oriented modelling of administrative workflow systems.

~

\noindent\textbf{Verification of workflows specified using our models (priority: short-term)}: one of the BPM activities is the formal analysis/verification of the specifications produced using a given workflow language. Proposing and/or adapting a verification method for workflows designed using our models seems to be an interesting avenue of research especially since several similar works have done the same \cite{badouel2015active, van1997verification, van2000workflow}. To this end, it will certainly be necessary to deeply study our models in order to highlight their mathematical properties. These properties will then make it possible to identify and present the criteria that must be verified by a specification in order to be qualified as correct (sound).

~

\noindent\textbf{A tool to help in the specification of administrative processes with a scenario-oriented approach (priority: short-term)}: since the models we have proposed are new, it would be wise to propose a tool to assist in their use in practice. In addition to being a guide, such a tool should simplify the creation of process models and possibly validate them according to pre-established correction criteria. Moreover, it will be able to provide several DSL for saving the specifications as well as several modules for exporting them in more conventional notations (BPMN, YAWL, etc.).

~

\noindent\textbf{A framework to generate administrative business processes' specific decentralised execution simulators (priority: mid-term)}: as implemented, the WfMS prototype P2PTinyWfMS can be used as a foundation for the production of a tool that generates simulators of the completely decentralised execution of administrative processes specified using our models. The new framework can be based on the models found in \cite{tchembe2019ad}; then, it will be implemented to generate a simulation environment tailored to a given administrative process. Still in a generative logic, another approach would be to use the recent concepts of model-driven engineering to produce a simulator based on integrated development environments such as Eclipse: one could for example use a GEMOC approach \cite{bousse2016execution, combemale2017language}.

~

\noindent\textbf{Extension of our decentralised execution model for recursive GMWf and monitoring support (priority: short-term)}: we already mentioned this in the discussion section of chapter \ref{chap3:choreography-workflow-design-execution}. It would be interesting to introduce into our multiagent system, new types of agents to monitor the execution of processes. In addition, we could also make the use of views more flexible by modifying their impact on projection operations. For example, we could redefine the artifact projection operation so that it no longer erases nodes but censors them; that is, it replaces them with symbols that help in the specification of control flows (restructuring symbols for instance) carrying no information on the process. This will preserve the possibility of offering only a potentially partial view of the processes and their data to actors while allowing the use of recursive symbols in GMWf.

~

\noindent\textbf{Concrete implementation of WfMS supporting our models (priority: short-term)}: eventually, it will also be necessary to propose implementations of the WfMS presented here. Naturally, the implemented system will have to cover all phases and activities of the BPM lifecycle while focusing on the automation of administrative business processes. A study of recent BPM systems to ensure interoperability and an openness of the implemented system on the cloud will certainly be interesting avenues to explore.

~

\noindent\textbf{Study of each proposed formal tool's properties in order to identify more precisely the class of workflows to which they apply (priority: long-term)}: another, more theoretical, line of research would be to formally analyse the proposed BPM approach to characterise the class of workflows to which it can actually be applied. We assumed that we were only interested in administrative workflows; however, the proposed models are quite general and could well be applied to other classes of workflows. The study carried out in this perspective should therefore determine these classes and the conditions under which the proposed model automates them.

\myCleanStarChapterEnd

	%************ Bibliographie ***************
	% La charte de l'école doctorale recommande un style dans lequel les citations seront de la forme (NomAuteur, Année) ou (NomAuteur et al., Année)
	%\myBibliography{style}{url du fichier .bib}
	\myBibliography{apacite}{bibliography}
	
	% *********** Annexes *********************
	\appendix
	
	\myChapter{Implementation of Some Important Algorithms Presented in this Thesis}{}
\label{appendice1:algorithms-implementations}
\mySaveMarks

We had thought to present in this appendix, a Haskell implementation of the projection algorithms proposed in chapter \ref{chap3:choreography-workflow-design-execution} of this thesis. However, these are far too voluminous and their presentation here will not be very readable. We have therefore decided to present only the main data types here. We have hosted the rest of the produced Haskell code on the public Git repository accessible via this link: \url{https://github.com/MegaMaxim10/my-thesis-projection-algorithms}.

\mySectionStar{Haskell Type for Tags}{}{false}
Let's start by defining the tags for the node types (sequential or parallel). More clearly, in a given artifact, a node $A$ is tagged with \Verb|Seq| (resp. \Verb|Par|) when its sub-artifacts are executed in sequence (resp. potentially in parallel), i.e. the production used for its extension is a sequential (resp. parallel) one. A node with at most one sub-artifact is always tagged with \Verb|Seq|.
\begin{Verbatim}[frame=lines, fontsize=\small, numbers=left, numbersep=8pt]
data ProductionTag x = Seq x | Par x deriving (Eq, Show)
\end{Verbatim}
The \Verb|untagProduction| function below clears a given symbol of its tag (\Verb|Seq| or \Verb|Par|):
\begin{Verbatim}[frame=lines, fontsize=\small, numbers=left, numbersep=8pt]
untagProduction:: ProductionTag x -> x
untagProduction (Seq x) = x
untagProduction (Par x) = x
\end{Verbatim}

\mySubSectionStar{Definition of tags (\textit{closed}, \textit{locked}, \textit{unlocked} or \textit{upstair}) for symbols}{}{false}
In an artifact: a closed node is tagged \Verb|Closed|, an unlocked bud is tagged \Verb|Unlocked|, a locked bud is tagged \Verb|Locked| and an upstair bud is tagged \Verb|Upstair| (only found after expansion).
\begin{Verbatim}[frame=lines, fontsize=\small, numbers=left, numbersep=8pt]
data NodeTag x = Closed x | Locked x | Unlocked x | Upstair x deriving (Eq, Show)
\end{Verbatim}
The \Verb|untagNode| function below clears a given symbol of its tag (\Verb|Closed|, \Verb|Unlocked|, \Verb|Locked| or \Verb|Upstair|):
\begin{Verbatim}[frame=lines, fontsize=\small, numbers=left, numbersep=8pt]
untagNode:: NodeTag x -> x
untagNode (Closed x) = x
untagNode (Locked x) = x
untagNode (Unlocked x) = x
untagNode (Upstair x) = x
\end{Verbatim}

\mySubSectionStar{Definition of tags for symbol types (structuring or standard)}{}{false}
The symbols of a given artifact $t$ are either those of the grammatical model $\mathbb{G}$ denoting $t$, or (re)structuring symbols introduced to preserve some important properties of our model (mainly, the form of productions used in GMWf): in this case, the symbols of $\mathbb{G}$ are said to be standard and are tagged with \Verb|Standard| while the (re)structuring symbols are tagged with \Verb|Structural|.
\begin{Verbatim}[frame=lines, fontsize=\small, numbers=left, numbersep=8pt]
data SymbolTag x = Structural x | Standard x deriving (Eq, Show)
\end{Verbatim}
As the previous "untag" functions, the \Verb|untagSymbol| function below clears a given symbol of its tag (\Verb|Structural| or \Verb|Standard|):
\begin{Verbatim}[frame=lines, fontsize=\small, numbers=left, numbersep=8pt]
untagSymbol:: SymbolTag x -> x
untagSymbol (Structural x) = x
untagSymbol (Standard x) = x
\end{Verbatim}

\mySectionStar{Haskell Type for Artifacts}{}{false}
Recursively, we consider that an artifact is given by its root node (\Verb|nodeLabel|) and the list of its sub-artifacts (\Verb|sonsList|) tagged either by \Verb|Seq| (to indicate that they are executed in sequence) or by \Verb|Par| (to indicate that they are potentially executed in parallel). We do not consider empty artifacts. The corresponding Haskell type is as follows:
\begin{Verbatim}[frame=lines, fontsize=\small, numbers=left, numbersep=8pt]
data Artifact a = Node {
                       nodeLabel:: a, 
                       sonsList:: ProductionTag [Artifact a]
                  } deriving Eq
\end{Verbatim}
Here is an example of artifact encoded in this type. It corresponds to the target artifact $art_1$ in the figure \ref{chap3:fig:global-artefacts}:
\begin{Verbatim}[frame=lines, fontsize=\small, numbers=left, numbersep=8pt]
art1 = Node (Closed "Ag") (
           Seq [
             Node (Closed "A") (
               Seq [
                 Node (Closed "B") (Seq []), 
                 Node (Closed "D") (Seq [])
               ])
           ])
\end{Verbatim}

\mySectionStar{Haskell Type for GMWf}{}{false}
Let's start by presenting a type for productions: a production is given by its left hand side (\Verb|lhs|) consisting of one symbol and by its right hand side (\Verb|rhs|) consisting of several symbols.
\begin{Verbatim}[frame=lines, fontsize=\small, numbers=left, numbersep=8pt]
data Production symb = Prod {lhs:: symb, rhs:: [symb]} deriving Eq
\end{Verbatim}

Finally, a GMWf is given by the set of symbols and the set of productions constituting it. The productions are tagged either by \Verb|Seq| or by \Verb|Par|:
\begin{Verbatim}[frame=lines, fontsize=\small, numbers=left, numbersep=8pt]
data GMWf a = GMWf {
                   symbols:: [a], 
                   productions:: [ProductionTag (Production a)]
              } deriving (Eq, Show)
\end{Verbatim}

These are the main data types that we have defined and which are manipulated by the different projection functions that are available in our Git repository\footnote{Our Git repository: \url{https://github.com/MegaMaxim10/my-thesis-projection-algorithms}}. They are included with some test cases that one will be able to immediately experiment.

\myRestoreMarks

	\myChapter{List of Publications Issued from the Work Presented in this Thesis}{}
\label{appendice2:article-appendice}
\mySaveMarks
\mySectionStar{Journal Papers}{}{false}
\mySubSectionStar{Published}{}{false}
\begin{enumerate}
	\item Milliam Maxime Zekeng Ndadji, Maurice Tchoup{\'e} Tchendji, Cl{\'e}mentin Tayou Djamegni and Didier Parigot. "\textit{A new Domain-Specific Language and Methodology based on Scenarios, Grammars and Views, for Administrative Processes Modelling.}" ParadigmPlus, Volume 1, Number 3, 2020, 1-22.
	\item Maurice Tchoup{\'e} Tchendji and Milliam Maxime Zekeng Ndadji. "\textit{Tree Automata for Extracting Consensus from Partial Replicas of a Structured Document.}" Journal of Software Engineering and Applications 10.05 (2017): 432-456.
\end{enumerate}

\mySubSectionStar{Under Review}{}{false}
\begin{enumerate}
	\item Maurice Tchoup{\'e} Tchendji, Milliam Maxime Zekeng Ndadji and Didier Parigot. "\textit{A Grammatical Approach for Administrative Workflow Design and their Distributed Execution using Structured and Cooperatively Edited Mobile Artifacts.}" Software and Systems Modeling, Springer (\textbf{submitted}).
	\item Milliam Maxime Zekeng Ndadji, Maurice Tchoup{\'e} Tchendji, Cl{\'e}mentin Tayou Djamegni and Didier Parigot. "\textit{A Projection-Stable Grammatical Model for the Distributed Execution of Administrative Processes with Emphasis on Actors' Views.}" Journal of King Saud University - Computer and Information Sciences, Elsevier (\textbf{submitted}).
\end{enumerate}

\mySectionStar{Communications in Conferences}{}{false}
\mySubSectionStar{Published}{}{false}
\begin{enumerate}
	\item Milliam Maxime Zekeng Ndadji, Maurice Tchoup{\'e} Tchendji, Cl{\'e}mentin Tayou Djamegni and Didier Parigot. "\textit{A Grammatical Model for the Specification of Administrative Workflow using Scenario as Modelling Unit.}" H. Florez and S. Misra (eds) Applied Informatics. ICAI 2020. Communications in Computer and Information Science, vol 1277, Springer, Cham, 2020. pages 131-145.
	\item Milliam Maxime Zekeng Ndadji, Maurice Tchoup{\'e} Tchendji, Cl{\'e}mentin Tayou Djamegni and Didier Parigot. "\textit{A Language for the Specification of Administrative Workflow Processes with Emphasis on Actors' Views.}" Gervasi O. et al. (eds) Computational Science and Its Applications - ICCSA 2020. ICCSA 2020. Lecture Notes in Computer Science, vol 12254, Springer, Cham, 2020. pages 231-245.
	\item Milliam Maxime Zekeng Ndadji, Maurice Tchoup{\'e} Tchendji and Didier Parigot. "\textit{A Projection-Stable Grammatical Model to Specify Workflows for their P2P and Artifact-Centric Execution.}" CRI'2019 - Conf{\'e}rence de Recherche en Informatique. Dec 2019, Yaound{\'e}, Cameroon. (hal-02375958).
	\item Milliam Maxime Zekeng Ndadji and Maurice Tchoup{\'e} Tchendji. "\textit{A Software Architecture for Centralized Management of Structured Documents in a Cooperative Editing Workflow.}" Innovation and Interdisciplinary Solutions for Underserved Areas. Lecture Notes of the Institute for Computer Sciences, Social Informatics and Telecommunications Engineering (LNICST), Springer, Cham, 2018. pages 279-291.
	\item Maurice Tchoup{\'e} Tchendji and Milliam Maxime Zekeng Ndadji. "\textit{R{\'e}conciliation par consensus des mises {\`a} jour des r{\'e}pliques partielles d'un document structur{\'e}.}" CARI 2016 Proceedings, volume 1, 2016. pages 84-96.
\end{enumerate}

\myRestoreMarks

\begin{comment}
\includepdf[pages={1}, offset=72 -72]{Appendices/Articles/pplus-board.pdf}
\includepdf[pages=-, offset=72 -72]{Appendices/Articles/ICAI-2020-Extended.pdf}

\includepdf[pages={1}, offset=72 -72]{Appendices/Articles/JSEA_10_01_Content_2017011916111320.pdf}
\includepdf[pages=-, offset=72 -72]{Appendices/Articles/JSEA_2017052615402522.pdf}

%\includepdf[pages=-, offset=72 -72]{Appendices/Articles/mainApprocheP2PSOA-Elsevier.pdf}

%\includepdf[pages=-, offset=72 -72]{Appendices/Articles/CRI-2019-Extended.pdf}

\includepdf[pages={1,144-158}, offset=72 -72]{Appendices/Articles/ICAI-2020.pdf}

\includepdf[pages={1,269-283}, offset=72 -72]{Appendices/Articles/ICCSA-2020.pdf}

\includepdf[pages={4}, offset=72 -72]{Appendices/Articles/covers.pdf}
\includepdf[pages=-, offset=72 -72]{Appendices/Articles/mainprojectionmgwfarima.pdf}

\includepdf[pages={1,281-293}, offset=72 -72]{Appendices/Articles/InterSol2017-Book.pdf}

\includepdf[pages={2,97-109}, offset=72 -72]{Appendices/Articles/CARI2016.pdf}
\end{comment}

	%Ainsi de suite
}
\end{document}